%% file: main.tex
\documentclass[envcountsame,runningheads]{llncs}

\newcommand{\arxiv}[1]{#1}
\newcommand{\paper}[1]{#1}

\renewcommand{\paper}[1]{}

\usepackage[numbers,sort&compress]{natbib}
  \usepackage{amsmath,amsthm}

\usepackage[linesnumbered,noend]{algorithm2e}

\usepackage{algpseudocode}

\usepackage{caption}
\usepackage{listings}
\usepackage{lstautogobble}
\lstdefinelanguage{absynth} {
	morekeywords={def,while,var,tick,prob,else}, sensitive=false, morecomment=[l]{//}, morecomment=[s]{/*}{*/}, morestring=[b]",
}
\usepackage{amssymb}
\arxiv{\PassOptionsToPackage{hyphens}{url}}

\makeatletter
\RequirePackage[bookmarks,unicode,colorlinks=true]{hyperref}%
   \def\@citecolor{blue}%
   \def\@urlcolor{blue}%
   \def\@linkcolor{blue}%

\def\orcidID#1{\smash{\href{http://orcid.org/#1}{\protect\raisebox{-1.25pt}{\protect\includegraphics{orcid_color.eps}}}}}
\makeatother

\usepackage{array}
\newcolumntype{L}{>{$}l<{$}} 

\usepackage{siunitx}
\usepackage[capitalize,nameinlink]{cleveref}
\usepackage{mathtools}
\usepackage{multirow}
\usepackage{stmaryrd}
\usepackage{cancel}
\usepackage{thmtools, thm-restate}
\usepackage{nicefrac}
\usepackage[shortlabels]{enumitem}
\usepackage
{todonotes}
\usepackage{microtype}
\usepackage{pbox}
\usepackage{marvosym}
\input{insbox}
\input{util/commands}

\layout{\setlength{\textfloatsep}{8pt}}

\allowdisplaybreaks

\paper{
\usepackage[firstpage]{draftwatermark}
\SetWatermarkText{\hspace*{4.7in}\raisebox{6.3in}{\includegraphics[scale=0.1]{ae-badge-tacas}}}
\SetWatermarkAngle{0}
}

\begin{document}
\usetikzlibrary{shapes,arrows}
\title{Inferring Expected Runtimes of Probabilistic Integer Programs Using Expected
  Sizes\thanks{funded by the Deutsche Forschungsgemeinschaft (DFG, German Research
    Foundation) - 235950644 (Project GI 274/6-2) \&  DFG Research Training Group 2236 UnRAVeL}}
\titlerunning{Inferring Expected Runtimes Using Expected Sizes}

\author{
  Fabian Meyer\paper{\orcidID{0000-0003-1038-4944}} \and Marcel Hark\paper{$^{(\mbox{\Letter})}$\orcidID{0000-0001-5111-3177}} \and J\"urgen Giesl\paper{$^{(\mbox{\Letter})}$\orcidID{0000-0003-0283-8520}}
}
\authorrunning{Fabian Meyer, Marcel Hark, and J\"urgen Giesl}

\institute{
	LuFG Informatik 2, RWTH Aachen University, Aachen, Germany \\
	\email{fabian.niklas.meyer@rwth-aachen.de},
	\email{\{marcel.hark,giesl\}@cs.rwth-aachen.de}
}
\maketitle
\input{abstract}

\section{Introduction}
\label{sec:Intro}
\input{introduction}

\section{Probabilistic Integer Programs}
\label{sec:Preliminaries}
\input{preliminaries}

\section{Complexity Bounds}
\label{sec:Complexity_Bounds}
\input{complexity_bounds}

\section{Computing Expected Runtime Bounds}
\label{sec:Computing_Runtime_Bounds}
\input{comp_time}

\section{Computing Expected Size Bounds}
\label{sec:Computing Size Bounds}
\input{comp_size}

\section{Related Work, Implementation, and Conclusion}
\label{sec:evaluation}
\input{related_work}

\input{implementation}

\input{evaluation}
\input{conclusion}

\bibliographystyle{splncs04}
\bibliography{references}
\paper{
 
\vfill

{\small\medskip\noindent{\bf Open Access}
  This chapter is licensed under the terms of the Creative Commons\break Attribution 4.0
  International License
  (\url{https://creativecommons.org/licenses/by/4.0/}),
which permits use, sharing, adaptation, distribution and
reproduction in any medium or format, as long as you give
appropriate credit to the original author(s) and the
source, provide a link to the Creative Commons license
and indicate if changes were made.}

{\small \spaceskip .28em plus .1em minus .1em The images or other third party material in
  this chapter are included in the chapter's Creative Commons license, unless indicated
  otherwise in a credit line to the material.~If material is not included in the chapter's
  Creative Commons license and your intended\break use is not permitted by statutory
  regulation or exceeds the permitted use, you will need to obtain permission directly
  from the copyright holder.}

\medskip\noindent\includegraphics{cc_by_4-0.eps}
  }
\arxiv{
\clearpage
\begin{appendix}

\section*{Appendix}

\cref{app:all_examples}
presents the detailed results of the tools in our experimental
        evaluation and the collection of (additional) examples that we used for our experiments.
	Afterwards, \cref{app:prelim} recapitulates preliminaries from probability theory
        which are needed for the proofs of our theorems.
	These proofs are given in \cref{app:proofs}.

	\section{Examples}\label{app:all_examples}
We first present additional examples to illustrate certain aspects of our approach in
\cref{app:additional_examples}. Afterwards, we discuss the examples used in our evaluation in \cref{app:examples_evaluation}.

\subsection{Additional Examples}
\label{app:additional_examples}

\input{appendix/additional_examples}

\subsection{Examples for the Evaluation}
\label{app:examples_evaluation}

	\input{appendix/examples}

	\section{Preliminaries from Probability Theory}
	\label{app:prelim}
	In this section, we present preliminaries needed for the proofs of our theorems.
	\cref{app:probtheo} recapitulates basic concepts from probability theory and in \cref{app:probSpace} we recapitulate the (standard) cylinder construction that we use to construct the probability spaces for PIPs.

	\subsection{Basic Concepts of Probability Theory}
	\label{app:probtheo}
	\input{appendix/probtheo}

	\subsection{Cylindrical Construction of Probability Spaces}
	\label{app:probSpace}
	\input{appendix/prob_space_construction}

	\section{Proofs}
	\label{app:proofs}
	This section contains the proofs for all theorems and lemmas of our paper.
	The proofs for \cref{sec:Complexity_Bounds}-\ref{sec:Computing Size Bounds} are presented in \cref{app:Proofsforsec:Complexity_Bounds}-\ref{app:Proofsforsec:Computing Size Bounds}, respectively.

	\subsection{Proofs for \cref{sec:Complexity_Bounds}}
	\label{app:Proofsforsec:Complexity_Bounds}

	\input{proofs/lifting_of_bounds}

	\subsection{Proofs for \cref{sec:Computing_Runtime_Bounds}}
	\input{appendix/details_lprf}

        \input{proofs/concavity_linearity}

	\subsection{Proofs for \cref{sec:Computing Size Bounds}}
	\label{app:Proofsforsec:Computing Size Bounds}

	\input{proofs/elsb_basic}

		\subsubsection{Proofs for Expected Size Bounds of Trivial SCCs}

        \input{proofs/expected_trivial_sizebounds_method}

	\subsubsection{Proofs for Expected Size Bounds of Non-Trivial SCCs}
        It remains to prove \cref{theorem:expectednontrivialsizeboundsmeth} on expected
        size bounds for non-trivial SCCs of the general result variable graph.
	\input{appendix/details_expected_nontrivial}

\end{appendix}
}
\end{document}

%% file: util/commands.tex
\crefname{definition}{Def.}{Def.}
\crefname{example}{Ex.}{Ex.}
\crefname{appendix}{App.}{App.}
\crefname{ex}{Ex.}{Ex.}
\crefname{theorem}{Thm.}{Thm.}
\crefname{lemma}{Lemma}{Lemma}
\crefname{figure}{Fig.}{Fig.}
\crefname{section}{Sect.}{Sect.}
\crefname{subsection}{Sect.}{Sect.}
\crefname{subsubsubappendix}{Sect.}{Sect.}
\crefname{algorithm}{Procedure}{Procs}
\crefname{algocf}{Alg.}{Algorithms}
\crefname{corollary}{Cor.}{Cor.}
\crefname{algorithm}{Alg.}{Alg.}
\crefname{equation}{}{}
\crefname{enumi}{}{}

\newcommand{\true}{\mathbf t}
\newcommand{\false}{\mathbf f}

\DeclareMathOperator{\powerset}{\mathfrak{P}}

\DeclareMathOperator{\expvsign}{{\mathbb E}}
\DeclareMathOperator{\pmeasure}{\mathbb P}
\DeclareMathOperator{\expvsignabs}{{\mathbb E_{abs}}}

\DeclareMathOperator{\BOUND}{\mathcal{B}}
\DeclareMathOperator{\POLYBOUND}{\ZZ[\VV]}

\DeclareMathOperator{\REALPOLYBOUNDPV}{\RR[\PV]}

\newcommand{\abs}[1]{\left |#1 \right|}
\newcommand{\DIST}{\mathcal D}
\newcommand{\expvdist}{\mathfrak E}

\newcommand{\CONSTR}{\mathcal C}
\newcommand{\LOC}{\mathcal L}
\newcommand{\T}{\mathcal{T}}
\newcommand{\GT}{\mathcal{G}\!\T}
\newcommand{\GTINIT}{\GT_0}
\newcommand{\PV}{\mathcal{PV}}

\newcommand{\VV}{\mathcal{V}}
\newcommand{\PIP}{(\PV, \LOC,\GT, \ell_0)}
\newcommand{\CONF}{\mathsf{Conf}}
\newcommand{\RUNS}{\mathsf{Runs}}
\newcommand{\FPATH}{\mathsf{FPath}}

\newcommand{\scheduler}{\mathfrak{S}}
\newcommand{\STATE}{\Sigma}
\newcommand{\state}{s}
\newcommand{\initstate}{{\state_0}}
\newcommand{\loc}{\ell}
\newcommand{\initloc}{\ell_0}

\newcommand{\pip}{\mathcal P}
\newcommand{\pipmeasure}[3]{\pmeasure_{#2,#3}}

\newcommand{\ptransition}[3]{pr_{#2,#3}}
\newcommand{\expv}[3]{\expvsign_{#2,#3}}
\newcommand{\expec}[1]{\expvsign \left( #1 \right)}
\newcommand{\IP}[1]{\pmeasure \left(#1 \right)}
\newcommand{\tin}{t_\text{in}}
\newcommand{\tout}{t_{\bot}}
\newcommand{\gout}{g_{\bot}}

\newcommand{\GRV}{\mathcal{GRV}}
\DeclareMathOperator{\ACTV}{actV}

\DeclareMathOperator{\pre}{pre}

\newcommand{\ZZ}{\mathbb{Z}}
\newcommand{\RR}{\mathbb{R}}
\newcommand{\NN}{\mathbb{N}}

\newcommand{\GTG}{\GT_>}
\newcommand{\GTNI}{{ \GT_{\mathrm{ni}} }}
\newcommand{\GTNIpr}{{ \GT'_{\mathrm{ni}} }}

\newcommand{\natclosure}{\overline {\mathbb{N}}}
\newcommand{\realclosure}{\overline {\mathbb{R}}_{\geq 0}}

\newcommand{\timervar}{\mathcal{R}}
\newcommand{\sizervar}{\mathcal{S}}

\newcommand{\sbound}{{\mathcal{SB}}}
\newcommand{\tbound}{{\mathcal{RB}}}
\newcommand{\sbounde}{{\sbound}_{\expvsign}}
\newcommand{\tbounde}{{\tbound}_{\expvsign}}
\newcommand{\chbounde}{{\mathcal{CB}_{\expvsign}}}
\newcommand{\chsum}{{ch_\Sigma}}

\DeclareMathOperator{\rank}{\mathfrak{r}}

\newcommand{\ENTRYLOC}{\mathcal{E}\!\mathcal{L}}
\newcommand{\ET}{\mathcal{E}\!\T}
\DeclareMathOperator{\incsize}{size}
\DeclareMathOperator{\maxTV}{tv}
\DeclareMathOperator{\ch}{ch}

\newcommand{\eff}[1]{\incsize^{#1}\left(\chbounde\left(#1\right)\right)}

\newcommand{\pr}[4]{\operatorname{q}_{#1,#2,#3} (#4)}

\DeclareMathOperator{\cuts}{cuts}
\newcommand{\CE}{\mathcal{E}}
\newcommand{\CF}{\mathcal{F}}
\newcommand{\CG}{\mathcal{G}}
\newcommand{\ind}{\mathbf{1}}
\newcommand{\run}{\vartheta}
\newcommand{\comment}[1]{\footnote{!!! #1}}
\renewcommand{\comment}[1]{}
\newcommand{\layout}[1]{#1}
\arxiv{\renewcommand{\layout}[1]{}}
\newcommand{\nolayout}[1]{#1}
\paper{\renewcommand{\nolayout}[1]{}}
\newcommand{\mypagebreak}[0]{\pagebreak}
\newcommand{\oldcomment}[1]{}
\newcommand{\eat}[1]{}

\newcommand{\eval}[2]{\ensuremath{\abs{#2}\left(#1\right)}}
\newcommand{\exact}[2]{\ensuremath{#2\left(#1\right)}}
\newcommand{\sigmagen}[1]{\ensuremath{\sigma}\left( #1 \right)}
\newcommand{\is}{=}

\DeclareMathOperator{\UNIFORM}{UNIF}

\DeclareMathOperator{\GEOMETRIC}{GEO}

\newcommand{\overapprox}[1]{\ensuremath{\left\lceil #1\right\rceil}}

\newcommand{\prefix}[1]{#1}

\renewcommand{\emptyset}{\varnothing}

\newcommand{\Inf}[1]{\mathrm{Inf}_{#1}}
\newcommand{\Pre}[1]{\mathrm{Pre}_{#1}}
\newcommand{\PrefC}[2]{\Inf{#1}^{#2}}
\DeclareMathOperator{\bigO}{\mathcal{O}}

\newcommand{\sourcecode}{\url{https://github.com/aprove-developers/KoAT2-Releases/tree/probabilistic}}

\usepackage{xifthen}
\newcommand{\result}[2]{
	\ifthenelse{\isempty{#1}}
	{\infty: \; \SI{#2}{\second}}
	{\bigO (#1): \; \SI{#2}{\second}}
}
\newcommand{\resultamber}[2]{
	\ifthenelse{\isempty{#1}}
        {\textsf{--}}
        {\textsf{#1}: \; \SI{#2}{\second}}
}
\newcommand{\reducedcenter}[1]{
\begin{center}
	\scalebox{0.97}{#1
	}
\end{center}
}
\newcommand{\Romannum}[1]{\uppercase\expandafter{\romannumeral#1\relax}}

\newcommand{\mypar}[1]{\paper{\smallskip\noindent{}\textit{#1}}\arxiv{\paragraph{#1}}}

\newcommand{\automaticcomplanalysis}[0]{\cite{dp-framework,aprove,koat,dblp:conf/rta/avanzinim13,dblp:conf/tacas/avanzinims16,dblp:journals/toplas/0002ah12,dblp:conf/cade/sinnz10,cofloco,cofloco2,costa-complexity,maxcore,campy,rank,pubs,c4b,pastis,pubs-upper-lower,ramlnat,ramlpopl17}}

%% file: abstract.tex
\begin{abstract}
	We present a novel modular approach to infer upper bounds on the expected runtimes of probabilistic integer programs automatically.
	To this end, it computes bounds on the runtimes of program parts and on the sizes of their variables in an alternating way.
	To evaluate its power,~we im\-ple\-men\-ted our approach in a new version of our open-source tool \textsf{KoAT}.\arxiv{\layout{\linebreak[4]
		}}
\end{abstract}

%% file: introduction.tex
There exist several approaches and tools for automatic complexity analysis of non-probabilistic programs, e.g., \automaticcomplanalysis.
While most of them rely on basic techniques like \emph{ranking functions}
(see, e.g., \cite{dblp:conf/vmcai/podelskir04,rank,costa-rf,bradley05,dblp:conf/cav/ben-amramg17,dblp:conf/sas/ben-amramdg19}), they usually combine these basic techniques in sophisticated ways.
For example, in \cite{koat} we developed a modular approach for automated complexity analysis of integer programs, based on an alternation between finding symbolic runtime bounds for program parts and using them to infer bounds on the sizes of variables in such parts.
So each analysis step is restricted to a small part of the program.
The implementation of this approach in \textsf{KoAT} \cite{koat} (which is integrated in \textsf{AProVE} \cite{aprove}) is one of the leading tools for complexity analysis \cite{termcomp}.

While there exist several adaptions of basic techniques like ranking functions to \emph{probabilistic programs} (e.g., \cite{dblp:conf/rta/bournezg06,dblp:conf/popl/chatterjeenz17,rsmchatt,rsm,dblp:conf/pldi/wang0gcqs19,dblp:journals/scp/avanzinily20,lexrsm,dblp:conf/rta/bournezg05,dblp:conf/cav/chakarovs13,dblp:conf/vmcai/fuc19,cade19,dblp:conf/aplas/huangfc18,dblp:journals/pacmpl/huang0cg19,amber,FoundationsTerminationMartingale2020}), most\paper{\layout{\linebreak[2]}} of the sophisticated full approaches for complexity analysis have not been adapted to probabilistic programs yet, and there are only few powerful tools available which analyze the runtimes of probabilistic programs automatically \cite{absynth,dblp:conf/pldi/wang0gcqs19,hoffmannicfp2020,ecoimp}.

We study probabilistic integer programs (\cref{sec:Preliminaries}) and define suitable notions of non-probabilistic and expected runtime and size bounds (\cref{sec:Complexity_Bounds}).
Then, we adapt our modular approach for runtime and size analysis of \cite{koat} to probabilistic programs (\cref{sec:Computing_Runtime_Bounds,sec:Computing Size Bounds}).
So such an adaption is not only possible for \emph{basic techniques}\layout{\linebreak[4]
} like ranking functions, but also for \emph{full approaches} for complexity analysis.

For this adaption, several problems had to be solved.
When computing expected runtime or size bounds for new program parts, the main difficulty is to determine when it is sound to use \emph{expected} bounds on previous program parts and when one has to use \emph{non-probabilistic} bounds instead.
Moreover, the semantics of probabilistic programs is significantly different from classical integer programs.
Thus, \layout{\mypagebreak} the proofs of our techniques differ substantially from the ones in \cite{koat}, e.g., we have to use concepts from measure theory like ranking supermartingales.

In \cref{sec:evaluation}, we evaluate the implementation of our new approach in the tool \textsf{KoAT} \cite{koat,techreport} and compare with related work.
\paper{We refer to \cite{arxivereport} for an appendix of our paper containing}\arxiv{The appendix contains} all proofs, preliminaries from probability and measure theory, and an overview on the benchmark collection used in our evaluation.

%% file: preliminaries.tex
For any set $M\subseteq \overline{\RR}$ (with $\overline{\RR} = \RR \cup \{ \infty \}$)
and $w \in M$, let $M_{\geq w} = \{ v \in M \mid v \geq w \lor v = \infty \}$.
For a set $\PV$ of \emph{program variables}, 
we first introduce the kind of \emph{bounds} that our approach computes.
Similar to \cite{koat}, our bounds represent \emph{weakly monotonically increasing} functions from $\PV \to \realclosure$.
Such bounds have the advantage that they can easily be ``composed'', i.e.,  if $f$ and $g$ are both weakly monotonically increasing upper bounds, then so is $f \circ g$.

\begin{definition}[Bounds]
	\label{def:bounds}
	The set of bounds $\BOUND$ is the smallest set with $\PV \cup \overline{\RR}_{\geq 0}\layout{\linebreak[0]
		} \subseteq \BOUND$, and where $b_1, b_2 \in \BOUND$ and $v \in \RR_{\geq 1}$ imply $b_1 + b_2, \; b_1 \cdot b_2 \in \BOUND$ and $v^{b_1} \in \BOUND$.
\end{definition}

Our notion of probabilistic programs combines classical integer programs (as in,\paper{\layout{\linebreak[2]}} e.g., \cite{koat}) and probabilistic control flow graphs (see, e.g., \cite{lexrsm}).
A \emph{state} $\state$ is a variable assignment $\state \colon \VV \to \ZZ$ for the
(finite) set $\VV$ of all variables in the program, where $\PV \subseteq \VV$,
$\VV\setminus\PV$ is the set of \emph{temporary variables}, and
$\STATE$ is the set of all states.
For any $\state \in \STATE$, the state $\abs{\state}$ is defined by $\abs{\state}(x) = \abs{\state(x)}$ for all $x \in \VV$.
The set $\CONSTR$ of \emph{constraints} is the smallest set con\-taining $e_1 \leq e_2$
for all polynomials $e_1,e_2 \in \POLYBOUND$ and $c_1 \land c_2$ for all $c_1,c_2\in \CONSTR$.
In addition to ``$\leq$'', in examples we also use relations like ``$>$'', which can be simulated by constraints (e.g., $e_1 > e_2$ is equivalent to $e_2 +1 \leq e_1$ when regarding integers).
We also allow the application of states to arithmetic expressions $e$ and constraints $c$.
Then the number $\state(e)$ resp.\ $\state(c) \in \{\true, \false\}$ results from
evaluating the expression resp.\ the constraint when substituting every variable $x$ by
$\state(x)$.
So for bounds $b \in \BOUND$, we have $\abs{\state}(b) \in \realclosure$.

In the transitions of a program, a program variable $x \in \PV$ can also be
updated by adding a value according to a \emph{bounded distribution function}
$d: \STATE \to \mathrm{Dist}(\ZZ)$.
Here, for any state $\state$, $d(\state)$ is the probability distribution of the values that are added to $x$.
As usual, a \emph{probability distribution} on $\ZZ$ is a mapping $pr: \ZZ \to \RR$ with $pr(v) \in [0,1]$ for all $v \in \ZZ$ and $\sum_{v \in \ZZ} \; pr(v) = 1$.
Let $\mathrm{Dist}(\ZZ)$ be the set of distri\-butions $pr$ whose expected value $\expvsign(pr) = \sum_{v \in \ZZ} v \cdot pr(v)$ is well defined and finite, i.e., $\expvsignabs(pr) = \sum_{v \in \ZZ} \; |v| \cdot pr(v) < \infty$.
A distribution function $d:\STATE \to \mathrm{Dist}(\ZZ)$ is \emph{bounded} if there is a finite bound $\expvdist (d) \in \BOUND$ with $\expvsignabs(d(\state)) \leq \eval{\expvdist(d)}{\state}$ for all $\state \in \STATE$.
Let $\DIST$ denote the set of all bounded distribution functions (our implementation supports Bernoulli, uniform, geometric, hypergeometric, and binomial distributions, see \cite{techreport} for details).

\begin{definition}[PIP]
	\label{def:pip}
	$\PIP$ is a \emph{probabilistic integer program} with
	\begin{enumerate}
		\item a finite set of \emph{program variables} $\PV \subseteq \VV$
		\item a finite non-empty set of program \emph{locations} $\LOC$
		\item a finite non-empty set of \emph{general transitions} $\GT$.
		      A general transition $g$ is a finite non-empty set of transitions $t
	= (\loc,p,\tau,\eta,\loc')$, consisting \paper{\layout{\mypagebreak}} of
		      \begin{enumerate}
			      \item the \emph{start} and \emph{target locations} $\loc,\loc' \in \LOC$ of transition $t$,
			      \item the \emph{probability} $p\geq 0$ that transition $t$ is chosen when $g$ is executed,
			      \item the \emph{guard} $\tau\in\CONSTR $ of $t$, and
			      \item the \emph{update function}
			            $\eta\colon\PV\rightarrow \POLYBOUND \cup \;\DIST$ of
	$t$, mapping every program variable to an update polynomial or a bounded distribution function.
		      \end{enumerate}
		      All $t \in g$ must have the same start location $\loc$ and the same guard $\tau$.
		      Thus, we call them the start location and guard \emph{of $g$}, and denote them by $\loc_g$ and $\tau_g$.
		      Moreover, the probabilities $p$ of the transitions in $g$ must add up to $1$.
		\item an \emph{initial location} $\initloc \in \LOC$, where no transition has target location $\initloc$
	\end{enumerate}
\end{definition}

PIPs allow for both probabilistic and non-deterministic branching and sampling.
Probabilistic branching is modeled by selecting a transition out of a non-singleton general transition.
Non-deterministic branching is represented by several general transitions with the same start location and non-exclusive guards.
Probabilistic sampling is realized by update functions that map a program variable to a bounded distribution function.
Non-deterministic sampling is modeled by updating a program variable with an expression
containing temporary
variables from $\VV\setminus\PV$, whose values are non-deterministic (but can be restricted in the guard).
The set of \emph{initial} general transitions 
$\GTINIT\subseteq \GT$ consists of all general transitions with start
location $\ell_0$.

\input{figures/simple_quadratic.tex}

\begin{example}[PIP]
	\label{PIP example}
	{\sl Consider the PIP in \cref{fig:simple_quadratic} with initial location $\initloc$ and the program variables $\PV = \{x, y\}$.
		Here, let $p=1$ and $\tau = \true$ if not stated ex\-plicitly.
		There are four general transitions: $g_0 = \{ t_0 \}$, $g_1 = \{t_1, t_2 \}$, $g_2= \{t_3\}$, and $g_3 = \{ t_4 \}$, where $g_1$ and $g_2$ represent a non-deterministic branching.
		When choosing the general transition $g_1$, the transitions $t_1$ and $t_2$ encode a probabilistic branching.
		If we modified the update $\eta$ and the guard $\tau$ of $t_0$ to $\eta(x) = u \in \VV \setminus \PV$ and $\tau = (u > 0)$, then $x$ would be updated to a non-deterministically chosen positive value.
		In contrast, if $\eta(x) = \GEOMETRIC(\tfrac{1}{2})$, then $t_0$ would update $x$ by adding a value sampled from the geometric distribution with parameter $\tfrac{1}{2}$.}
\end{example}

In the following, we regard a fixed PIP $\pip$ as in \cref{def:pip}.
A \emph{configuration}
is a tuple\paper{\layout{\linebreak[2]}} $(\loc,t,\state)$, with the current location $\loc\in \LOC$, the current state $\state \in \STATE$, and the transition $t$ that was evaluated last and led to the current configuration.
Let $\T = \bigcup_{g \in \GT}\, g$.
Then $\CONF = (\LOC\uplus\{\loc_\bot\}) \times (\T \uplus \{\tin,\tout\})\times \STATE$ is
the set of all configurations, with a special location $\loc_\bot$ indicating the
termination of a run, and special transitions $\tin$\paper{\layout{\linebreak[2]}} (used in the first configuration of a
run) \paper{\layout{\mypagebreak}} and $\tout$ (for the configurations of the run after termination).
The (virtual) general transition $\gout = \{\tout\}$ only contains $\tout$.

A \emph{run} of a PIP is an infinite sequence $\run = \prefix{c_0\, c_1\cdots} \in \CONF^{\omega}$.
Let $\RUNS = \CONF^{\omega}$ and let $\FPATH=\CONF^{*}$ be the set of all \emph{finite paths} of configurations.

In our setting, deterministic Markovian schedulers suffice to resolve all non-determinism (see, e.g.,\ \cite[Prop.\ 6.2.1]{putermanmdp}).
For $c = (\loc,t,\state)\in \CONF$, such a \emph{scheduler} $\scheduler$ yields a pair $\scheduler(c) = (g, \state')$ where $g$ is the next general transition to be taken (with $\loc = \loc_g$) and $\state'$ chooses values for the temporary variables where $\state'(\tau_g) = \true$ and $\state (x) = \state' (x)$ for all $x\in\PV$.
If $\GT$ contains no such $g$, we get $\scheduler(c) = (\gout,s)$.

For each scheduler $\scheduler$ and initial state $\initstate$, we first define a probability mass function $\ptransition{\pip}{\scheduler}{\initstate}$.
For all $c \in \CONF$, $\ptransition{\pip}{\scheduler}{\initstate}(c)$ is the probability that a run starts in $c$.
Thus, $\ptransition{\pip}{\scheduler}{\initstate}(\prefix{c}) = 1$ if $c = (\initloc, \tin,\initstate)$ and $\ptransition{\pip}{\scheduler}{\initstate}(\prefix{c}) = 0$, otherwise.
Moreover, for all $c', c \in \CONF$, $\ptransition{\pip}{\scheduler}{\initstate}(c' \to c)$ is the probability that the configuration $c'$ is followed by the configuration $c$ (see \paper{\cite{arxivereport}}\arxiv{\cref{appendix:prob_space_construction}}
for the formal definition of $\ptransition{\pip}{\scheduler}{\initstate}$).

For any $f = c_0 \cdots c_n \in \FPATH$, let $\ptransition{\pip}{\scheduler}{\initstate}(f) = \ptransition{\pip}{\scheduler}{\initstate}(c_0) \cdot \ptransition{\pip}{\scheduler}{\initstate}(c_0 \to c_1)\cdot \ldots \cdot \ptransition{\pip}{\scheduler}{\initstate}(c_{n-1}
	\to c_n)$. We say that $f$ is \emph{admissible} for $\scheduler$ and $\initstate$ if $\ptransition{\pip}{\scheduler}{\initstate}(f) > 0$.
A run $\run$ is admissible if all its finite prefixes are admissible.
A configuration $c \in \CONF$ is admissible if there is some admissible finite path ending in $c$.

The semantics of PIPs can now be defined by giving a corresponding probability space, which is obtained by a standard cylinder construction (see, e.g., \cite{cylindrical,cylindricalvardi}).
Let $\pipmeasure{\pip}{\scheduler}{\initstate}$ denote the corresponding probability measure which lifts $\ptransition{\pip}{\scheduler}{\initstate}$ %
to cylinder sets: For any $f \in \FPATH$, we have $\ptransition{\pip}{\scheduler}{\initstate}(f) = \pipmeasure{\pip}{\scheduler}{\initstate}(\Pre{f})$ for the set $\Pre{f}$ of all runs with prefix $f$.
So $\pipmeasure{\pip}{\scheduler}{\initstate}(\Theta)$ is the probability that a run from $\Theta \subseteq \RUNS$ is obtained when using the scheduler $\scheduler$ and starting in $\initstate$.

We denote the associated expected value operator by $\expv{\pip}{\scheduler}{\initstate}$.
So for any random variable $X:\RUNS \to \overline{\NN} = \NN \cup \{\infty \}$, we have $\expv{\pip}{\scheduler}{\initstate}(X) = \sum_{n \in \overline{\NN}} \; n \cdot \pipmeasure{\pip}{\scheduler}{\initstate}(X = n)$.
For details on the preliminaries from probability theory we refer to \paper{\cite{arxivereport}}\arxiv{\cref{app:prelim}}.

%% file: figures/simple_quadratic.tex
\begin{figure}[t]
	\centering
	\begin{tikzpicture}[shorten >=1pt,node distance=3.5cm,auto,main node/.style={circle,draw,font=\sffamily\Large\bfseries}]
		\node[main node] (l0) {$\ell_0$}; \node[main node] (l1) [right of = l0] {$\ell_1$}; \node[main node] (l2) [right of = l1] {$\ell_2$};

		\path[->] (l0) edge[align=left] node [below] {$t_0 \in g_0$}
		node [above] {$\eta(x)=x$ \\
				$\eta(y) = y$}
		(l1)

		(l1) edge[loop above, align=left] node [right,yshift=0.1cm] {$t_1 \in g_1$}
		node [above,xshift=0.5cm,yshift=0.25cm] {
                  $\begin{array}{l@{\;\;}l}
                  p=\tfrac{1}{2} & \eta(x)=x-1\\
		  \tau = (x > 0) & \eta(y) = y+x
                  \end{array}$}
		()

		(l1) edge[loop below, align=left] node [right,yshift=0.4cm,xshift=.1cm] {$t_2 \in g_1$}
		node [right,xshift=-1.45cm,yshift=-0.2cm] {
     $\begin{array}{l@{\;\;}l}
                  p=\tfrac{1}{2} & \eta(x)=x\\
		  \tau = (x > 0) &\eta (y) = y+x
                   \end{array}$}
		()

		(l1) edge[align=left] node [below] {$t_3 \in g_2$}
		node [above] {$\eta(x)=x$\\
				$\eta (y) = y$}
		(l2)

		(l2) edge[loop right, align=left] node [right] {$t_4 \in g_3$}
		node [above,xshift=0.5cm,yshift=0.25cm] {$\eta(x)=x$\\
				$\eta(y) = y-1$\\
				$\tau = (y > 0)$}
		();
	\end{tikzpicture}\vspace*{-.2cm}
	\caption{PIP with non-deterministic and probabilistic branching}
	\label{fig:simple_quadratic}
\end{figure}
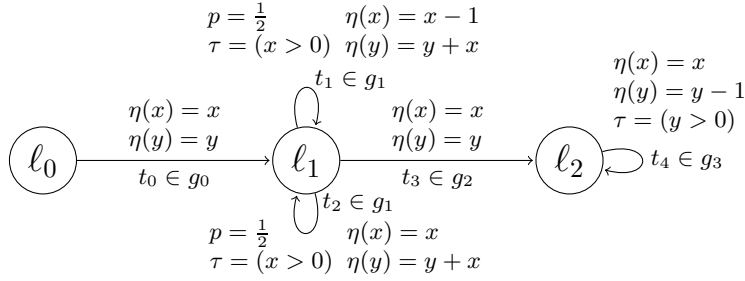

%% file: complexity_bounds.tex
In \cref{Non-Probabilistic Bounds}, we first recapitulate the concepts of (non-probabilistic) runtime and size bounds from \cite{koat}.
Then we introduce \emph{expected} runtime and size bounds in \cref{Probabilistic Bounds}
and connect them to their non-probabilistic counterparts.
%

\subsection{Runtime and Size Bounds}
\label{Non-Probabilistic Bounds}

Again, let $\pip$ denote the PIP which we want to analyze.
\cref{def:time_bounds} recapitulates the notions of runtime and size bounds
from \cite{koat} in our setting.
Recall that bounds from $\BOUND$
do not contain temporary variables, i.e., we always try to infer bounds in
terms of the initial values of the \emph{program variables}.
Let $\sup \emptyset = 0$, as all occurring sets are subsets of $\realclosure$, whose minimal element is $0$.
\begin{definition}[Runtime and Size Bounds \protect{\textnormal{\cite{koat}}}]
	\label{def:time_bounds}
	$\tbound\colon\T\rightarrow \BOUND$ is a \emph{runtime bound}
	and	$\sbound\colon \T\times\VV\rightarrow\BOUND$ is a \emph{size bound}
	if for all transitions $t \in \T$, all variables $x\in\VV$, all schedulers $\scheduler$, and all states $\initstate\in \STATE$, we have
	\[
		\begin{array}{lcll}
			\eval{\tbound (t)}{\initstate}   & \geq                                                                                                                                   & \sup \left\{ \; \abs{\{i \mid t_i = t\}}
			\right.                          & \mid  \left. f
	= \prefix{(\_,t_0,\_) \cdots
	(\_,t_n,\_)} \, \land \, \ptransition{\pip}{\scheduler}{\initstate} (f) >
	0 \, \right\},                \\
			\eval{\sbound (t,x)}{\initstate} & \geq                                                                                                                                   & \sup \left\{ \; |\state (x)| \right.     & \mid  \left. f = \prefix{\cdots (\_,t,\state)} \, \land \, \ptransition{\pip}{\scheduler}{\initstate} (f) > 0 \, \right\}.
		\end{array}
	\]
\end{definition}

So $\tbound (t)$ is a bound on the number of executions of $t$ and $\sbound (t,x)$ over-approximates the greatest absolute value that $x\in \VV$ takes after the application of the transition $t$ in any admissible finite path.
Note that \cref{def:time_bounds} does not apply to $\tin$ and $\tout$, since they are not contained in $\T$.

We call a tuple $(\tbound,\sbound)$ a (non-probabilistic) \emph{bound pair}.
We will use such non-probabilistic bound pairs for an initialization of expected bounds (\cref{theorem:liftingofbounds}) and to compute improved expected runtime and size bounds in \cref{sec:Computing_Runtime_Bounds,sec:Computing Size Bounds}.

\begin{example}[Bound Pair]
	\label{ex:complexity_bounds_nonprob}
	{\sl	The technique of \cite{koat} computes the following bound pair for the PIP of \cref{fig:simple_quadratic} (by ignoring the probabilities of the transitions).

		\vspace*{-.3cm}

		\noindent
		\hspace*{-.3cm}\begin{minipage}{.5\textwidth}
			\[
				\tbound (t) =
				\begin{cases}
					1,      & \text{if $t = t_0$ or $t=t_3$}   \\
					x,      & \text{if $t = t_1$}              \\
					\infty, & \text{if $t = t_2$ or $t = t_4$} \\
				\end{cases}
			\]
		\end{minipage}
                \hspace*{.4cm}
		\begin{minipage}{.4\textwidth}
			\begin{align*}
				\sbound (t,x) & =
				\begin{cases}
					x,         & \text{if $t \in \{t_0, t_1, t_2\}$} \\
					3 \cdot x, & \text{if $t \in \{t_3, t_4 \}$}
				\end{cases}
				\\
				\sbound (t,y) & =
				\begin{cases}
					y,      & \text{if $t = t_0$}                       \\
					\infty, & \text{if $t \in \{t_1, t_2, t_3, t_4 \}$}
				\end{cases}
			\end{align*}
		\end{minipage}

		\vspace*{.1cm}

		\noindent
		Clearly, $t_0$ and $t_3$ can only be evaluated once.
		Since $t_1$ decrements $x$ and no transition increments it, $t_1$'s runtime is bounded by $\abs{\initstate}(x)$.
		However, $t_2$ can be executed arbitrarily often if $\initstate(x) > 0$.
		Thus, the runtimes of $t_2$ and $t_4$ are unbounded (i.e., $\pip$ is not terminating when regarding it as a non-probabilistic program).
		$\sbound (t,x)$ is finite for all transitions $t$, since $x$ is never increased.
		In contrast, the value of $y$ can be arbitrarily large after all transitions but $t_0$.}
\end{example}

\subsection{Expected Runtime and Size Bounds}
\label{Probabilistic Bounds}

We now define the \emph{expected} runtime and size complexity of a PIP $\pip$.

\begin{definition}[Expected Runtime Complexity, PAST \protect{\textnormal{\cite{dblp:conf/rta/bournezg05}}}]
	\label{def:expected_time_complexity}
	For $g \in \GT$, its \emph{runtime} is the random variable $\timervar (g)$ where $\timervar\colon \GT\rightarrow\RUNS\rightarrow \natclosure$ with
	\[
		\timervar (g) (\prefix{\, (\_,t_0,\_)\, (\_,t_1,\_) \, \cdots \,})\; = \; \abs{\, \{i \mid t_i\in g\} \,}.
	\]
	For a scheduler $\scheduler$ and $\initstate \in \STATE$, the \emph{expected runtime complexity} of $g \in \GT$ is $\expv{\pip}{\scheduler}{\initstate} (\timervar (g))$ and
	the \emph{expected runtime complexity} of $\pip$ is $\sum_{g \in \GT} \expv{\pip}{\scheduler}{\initstate}
		(\timervar (g))$.

	If $\pip$'s expected runtime complexity is finite for every scheduler $\scheduler$ and every initial state $\initstate$, then $\pip$ is called \emph{positively almost surely terminating} \emph{(PAST)}.
\end{definition}

\noindent
So $\timervar (g)(\run)$ is the number of executions of a transition from $g$ in the run $\run$.

While non-probabilistic size bounds refer to pairs $(t,x)$ of transitions $t \in \T$ and variables $x \in \VV$ (so-called \emph{result variables} in \cite{koat}), we now introduce expected size bounds for \emph{general result variables} $(g,\ell,x)$, which consist of a general transition $g$, one of its target locations $\loc$, and a program variable $x \in \PV$.
So $x$ must not be a temporary variable (which represents \emph{non-probabilistic} non-determinism), since general result variables are used for \emph{expected} size bounds.

\begin{definition}[Expected Size Complexity]
	\label{def:expected_size_complexity}
	The set of \emph{general result variables} is $\GRV = \{ \, (g,\loc,x) \mid g \in \GT, x \in \PV, (\_,\_,\_,\_,\loc) \in g \,\}$.
	The \emph{size} of $\alpha =$ $(g,\loc,x) \in \GRV$ is the random variable $\sizervar (\alpha)$ where $\sizervar\colon \GRV \rightarrow\RUNS\rightarrow\natclosure$ with
	\[
		\sizervar (g,\loc,x) \, (\prefix{\,(\initloc,t_0,\initstate) \, (\loc_1,t_1,\state_1)\cdots \, }) \; = \; \sup\left\{\, |\state_i (x)| \; \mid \; \loc_i = \loc\land t_i\in g \, \right\}.
	\]
	For a scheduler $\scheduler$ and $\initstate$, the \emph{expected size complexity} of $\alpha\!\in\!\GRV$ is ${\expv{\pip}{\scheduler}{\initstate}} (\sizervar (\alpha))$.
\end{definition}

So for any run $\run$, $\sizervar(g, \loc, x)(\run)$ is the greatest absolute value of $x$ in location $\loc$, whenever $\loc$ was entered with a transition from $g$.
We now define bounds for the expected runtime and size complexity which hold \emph{independent} of the scheduler.

\begin{definition}[Expected Runtime and Size Bounds]
	\label{def:expected_time_bounds}
	\label{def:expected_size_bounds}
	\begin{itemize}
		\item[$\bullet$] $\tbounde\colon \GT\rightarrow \BOUND$ is an \emph{expected runtime bound} if for all $g\in\GT$, all schedulers $\scheduler$, and all $\initstate \in \STATE$, we have $\eval{\tbounde (g)}{\initstate} \geq \expv{\pip}{\scheduler}{\initstate} (\timervar (g))$.
		      \smallskip
		\item[$\bullet$]	$\sbounde\colon \GRV\rightarrow\BOUND$ is an \emph{expected size bound}
		      if for all $\alpha\in\GRV$, all schedulers $\scheduler$, and all $\initstate \in \STATE$, we have $\eval{\sbounde (\alpha)}{\initstate} \geq \expv{\pip}{\scheduler}{\initstate} (\sizervar (\alpha))$.
		      \smallskip
		\item[$\bullet$] A pair $(\tbounde,\sbounde)$ is called an \emph{expected bound pair}.
	\end{itemize}
\end{definition}

\begin{example}[Expected Runtime and Size Bounds]
	\label{ex:probabilistic_bounds}
	\sl Our new techniques from \cref{sec:Computing_Runtime_Bounds,sec:Computing Size Bounds} will derive the following expected bounds for the PIP from \cref{fig:simple_quadratic}.

	\vspace*{-.4cm}

	\begin{align*}
		\tbounde (g)            & =
		\begin{cases}
			1,                      & \hspace*{-.2cm}\text{if
		$g\!\in\!\{g_0,g_2\}$} \\ 
			2\cdot x,               & \hspace*{-.2cm}\text{if $g = g_1$}           \\
			6\cdot x^2 + 2 \cdot y, & \hspace*{-.2cm}\text{if $g = g_3$}
		\end{cases}
		\hspace*{-.8cm}
		                        & \sbounde (g,\_,x) & =
		\begin{cases}
			x,         & \hspace*{-.2cm}\text{if $g = g_0$}              \\
			2\cdot x,  & \hspace*{-.2cm}\text{if $g = g_1$}              \\
			3 \cdot x, & \hspace*{-.2cm}\text{if $g\!\in\!\{ g_2, g_3 \}$}
		\end{cases}
		\\
		\sbounde (g_0,\loc_1,y) & = y               & \sbounde (g_2,\loc_2,y) & =6\cdot x^2 + 2 \cdot y  \\
		\sbounde (g_1,\loc_1,y) & = 	6\cdot x^2 + y  & \sbounde (g_3,\loc_2,y) & =	12\cdot x^2 + 4 \cdot y
	\end{align*}
	While the runtimes of $t_2$ and $t_4$ were unbounded in the non-probabilistic case (\cref{ex:complexity_bounds_nonprob}), we obtain finite bounds on the expected runtimes of $g_1 = \{t_1,t_2\}$ and $g_3 = \{t_4\}$.
	For example, we can expect $x$ to be non-positive after at most $\abs{\initstate}(2\cdot x)$ iterations of $g_1$.
	Based on the above expected runtime bounds, the expected runtime complexity of the PIP is at most $\abs{\initstate}(\tbounde(g_0) + \ldots + \tbounde(g_3)) = \abs{\initstate}(2 + 2 \cdot x + 2 \cdot y + 6 \cdot x ^2)$, i.e., it is in $\mathcal{O}(n^2)$ where $n$ is the maximal absolute value of the program variables at the start of the program.
\end{example}

The following theorem shows that non-probabilistic bounds can be lifted to expected bounds, since they do not only bound the expected value of $\timervar(g)$ resp.\ $\sizervar(\alpha)$, but the whole distribution.
\paper{As mentioned, all proofs can be found in \cite{arxivereport}.}

\begin{restatable}[Lifting Bounds]{theorem}{liftingofbounds}
	\label{theorem:liftingofbounds}
	For a bound pair $(\tbound,\sbound)$, $(\tbounde,\sbounde)$ with $\tbounde (g) = \sum_{t \in g} \tbound (t)$ and $\, \sbounde (g,\loc,x) = \sum_{t = (\_,\_,\_,\_,\loc) \in g} \sbound (t,x)$ is an expected bound pair.
\end{restatable}

Here, we over-approximate the maximum of $\sbound (t,x)$ for $t = (\_,\_,\_,\_,\loc) \in g$ by their sum.
For asymptotic bounds, this does not affect precision, since $\max(f,g)$ and $f + g$ have the same asymptotic growth for any non-negative functions $f, g$.

\begin{example}[Lifting of Bounds]
	\label{ex:Lifting of Bounds}
	{\sl When lifting the bound pair of \cref{ex:complexity_bounds_nonprob} to expected bounds according to \cref{theorem:liftingofbounds}, one would obtain $\tbounde (g_0) = \tbounde (g_2) = 1$ and $\tbounde (g_1) =\tbounde (g_3) = \infty$.
		Moreover, $\sbounde(g_0,\loc_1,x) = x$, $\sbounde(g_1,\loc_1,x) = 2 \cdot x$, $\sbounde(g_2,\loc_2,x) = \sbounde(g_3,\loc_2,x) = 3 \cdot x$, $\sbounde(g_0,\loc_1,y) = y$, and $\sbounde(g,\_,y) = \infty$ whenever $g \neq g_0$.
		Thus, with these lifted bounds one cannot show that $\pip$'s expected runtime complexity is finite, i.e., they are substantially less precise than the finite expected bounds from \cref{ex:probabilistic_bounds}.
		Our approach will compute such finite expected bounds by repeatedly improving the lifted bounds of \cref{theorem:liftingofbounds}.}
\end{example}

%% file: comp_time.tex
We first present a new variant of probabilistic linear ranking functions in \cref{Probabilistic Polynomial Ranking Functions}.
Based on this, in \cref{Inferring Expected Time Bounds} we introduce our modular technique to infer expected runtime bounds by using expected size bounds.

\subsection{Probabilistic Linear Ranking Functions}
\label{Probabilistic Polynomial Ranking Functions}

For probabilistic programs, several techniques based on ranking supermartingales have been developed.
In this section, we define a class of probabilistic ranking functions that will be suitable for our modular analysis.

We restrict ourselves to ranking functions $\rank: \LOC \rightarrow \REALPOLYBOUNDPV_{\mathrm{lin}}$ that map every location to a \emph{linear polynomial} (i.e., of at most degree 1) without temporary variables.
The linearity restriction is common to ease the automated inference of ranking functions.
Moreover, this restriction will be needed for the soundness of our technique.
Nevertheless, our approach of course also infers non-linear expected runtimes (by combining the linear bounds obtained for different program parts).

Let $\exp_{\rank,g,\state}$ denote the expected value of $\rank$ after an execution of $g \in \GT$ in state $\state\in\STATE$. Here, $\state_\eta (x)$ is the expected value of $x \in \PV$ after performing the update $\eta$ in state $\state$. So if $\eta(x) \in \DIST$, then $x$'s expected value after the update results from adding the expected value of the probability distribution $\eta (x) (\state)$: \[\exp_{\rank,g,\state} = \hspace*{-.3cm}	\sum\limits_{(\loc,p,\tau,\eta,\loc') \in g} \hspace*{-.3cm}
	p\cdot \state_\eta(\rank(\loc')) \text{ with }
	\state_\eta (x) =
	\begin{cases}
		\state(\eta (x)),                          & \text{if $\eta (x)\in\POLYBOUND$} \\
		\state(x) + \expvsign (\eta (x) (\state)), & \text{if $\eta (x)\in\DIST$}
	\end{cases}
\]

\begin{definition}[PLRF]
	\label{def:Probabilistic Polynomial Ranking Functions}
	Let $\GTG \subseteq \GTNI \subseteq \GT$.
	Then $\rank \colon \LOC \rightarrow \REALPOLYBOUNDPV_{\mathrm{lin}}$ is a \emph{probabilistic linear ranking func\-tion (PLRF)}
	for $\GTG$ and $\GTNI$ if
for all  $g \in \GTNI \setminus \GTG$ and $c' \in \CONF$
there is a  $\;{\bowtie_{g,c'}} \in \{<,\geq\}$ such that
for all finite paths $\,\cdots \,c' \, c$
that are admissible for some $\scheduler$ and $\initstate \in \STATE$, and where
$c =
	(\loc,t,\state)$ (i.e., where $t$ is the transition that is used in the step from
	$c'$ to $c$), we have:\paper{\layout{\linebreak[4]
		}}
	\paper{\layout{\vspace*{-0.5cm}}}
	\begin{tabbing}
		\textbf{Boundedness (a):}
	 \=If $t \in g$ for a $g \in \GTNI \setminus \GTG$,
then $\state(\rank(\loc)) \bowtie_{g,c'} 0$.\\
	\textbf{Boundedness (b):}\>If $t \in g$ for a $g \in \GTG$,
then $\state(\rank(\loc)) \geq 0$.\\
		\textbf{Non-Increase:}
If $\loc = \loc_g$ for a
$g\in\GTNI$ and $\state(\tau_g) = \true$, then
$\state( \rank(\loc)) \geq \exp_{\rank,g,\state}$.\\
		\textbf{Decrease:}
If $\loc = \loc_g$ for a $g\in\GTG$ and $\state(\tau_g) = \true$, then
$\state( \rank(\loc)) -1 \geq \exp_{\rank,g,\state}$.
	\end{tabbing}
\end{definition}
So if one is restricted to the sub-program with the \underline{n}on-\underline{i}ncreasing transitions $\GTNI$, then $\rank(\loc)$ is an upper bound on the expected number of applications of transitions from $\GTG$ when starting in $\loc$.
Hence, a PLRF for $\GTG =\GTNI =\GT$ would imply that the program is PAST (see, e.g., \cite{dblp:conf/rta/bournezg06,rsmchatt,lexrsm,FoundationsTerminationMartingale2020}).
However, our PLRFs differ from the standard notion of probabilistic ranking functions by considering arbitrary subsets $\GTNI \subseteq \GT$.
This is needed for the modularity of our approach which allows us to analyze program parts separately (e.g., $\GT\setminus\GTNI$ is ignored when\paper{\layout{\linebreak[4]}} inferring a PLRF).
Thus, our ``Boundedness'' conditions differ slightly from the corresponding conditions in other definitions.
Condition (b) requires that $g \in \GTG$ never leads to a configuration where $\rank$ is negative.
Condition (a) states that in an admissible path where
$g=\{t_1,t_2,\ldots\} \in \GTNI\setminus\GTG$ is used for continuing in configuration
$c'$,  \paper{\layout{\mypagebreak}} if executing $t_1$ in $c'$ makes $\rank$ negative, then executing $t_2$ must make $\rank$ negative as well.
Thus,
such a $g$
can
never come before a general transition from $\GTG$ in an admissible path and hence,
$g$ can be ignored when
inferring upper bounds on the runtime.
This increases the power of our approach
and it allows us to consider only \emph{non-negative} random variables in our correctness proofs.

We use SMT solvers to generate PLRFs automatically.
Then for ``Boundedness'', we regard all $\state' \in \STATE$ with $\state'(\tau_g) = \true$ and require ``Boundedness'' for any state $\state$ that is reachable from $\state'$.
\begin{example}[PLRFs]
	\label{ex:PLRF}
	{\sl	Consider again the PIP in \cref{fig:simple_quadratic}
		and the sets $\GTG = \GTNI = \{g_1\}$ and $\GTG' = \GTNIpr = \{g_3\}$, which correspond to its two loops.

		The function $\rank$ with	$\rank (\loc_1) = 2 \cdot x$ and $\rank(\loc_0) = \rank(\loc_2)= 0$ is a PLRF for $\GTG = \GTNI$:
		For every admissible configuration $(\loc,t,\state)$ with $t \in g_1$ we have $\loc = \loc_1$ and $\state(\rank(\loc_1)) = 2 \cdot \state (x) \geq 0$, since $x$ was positive before (due to $g_1$'s guard) and it was either decreased by $1$ or not changed by the update of $t_1$ resp.\ $t_2$.
		Hence $\rank$ is bounded.
		Moreover, for $\state_1(x) = \state(x -1) = \state(x) -1$ and $\state_2(x) = \state(x)$ we have:
		\[
			\exp_{\rank, g, \state} \, = \, \tfrac{1}{2}\cdot \state_1(\rank(\loc_1)) + \tfrac{1}{2}\cdot \state_2(\rank(\loc_1)) \, = \, 2 \cdot \state(x) -1 \, = \, \state(\rank (\loc_1)) - 1
		\]
		So $\rank$ is decreasing on $g_1$ and as $\GTG=\GTNI$, also the non-increase property holds.

		\noindent
		Similarly, $\rank'$ with $\rank' (\loc_2) = y$ and $\rank' (\loc_0) = \rank' (\loc_1) = 0$ is a PLRF for $\GTG' = \GTNIpr$.}
\end{example}

In our implementation, $\GTG$ is always a singleton and we let $\GTNI \subseteq \GT$ be a cycle in the call graph where we find a PLRF for $\GTG \subseteq \GTNI$.
The next subsection shows how we can then obtain an expected runtime bound for the overall program by searching for suitable ranking functions repeatedly.


\subsection{Inferring Expected Runtime Bounds}
\label{Inferring Expected Time Bounds}

Our approach to infer expected runtime bounds is based on an underlying (non-probabilistic) bound pair $(\tbound,\sbound)$ which is computed by existing techniques (in our implementation, we use \cite{koat}).
To do so, we abstract the PIP to a standard integer transition system by ignoring the probabilities of transitions and replacing probabilistic with non-deterministic sampling (e.g., the update $\eta(x) = \GEOMETRIC(\tfrac{1}{2})$ would be replaced by $\eta(x) = x + u$ with $u \in \VV \setminus \PV$, where $u > 0$ is added to the guard).
Of course, we usually have $\tbound(t) = \infty$ for some transitions $t$.

We start with the expected bound pair $(\tbounde,\sbounde)$ that is obtained by lifting $(\tbound,\sbound)$ as in \cref{theorem:liftingofbounds}.
Afterwards, the expected runtime bound $\tbounde$ is improved repeatedly by applying the following \cref{theorem:exptimeboundsmeth} (and similarly, $\sbounde$ is improved repeatedly by applying \cref{thm:Inferring Expected Size Bounds for Trivial SCCs,theorem:expectednontrivialsizeboundsmeth} from \cref{sec:Computing Size Bounds}).
Our approach alternates the improvement of $\tbounde$ and $\sbounde$, and it uses expected size bounds on ``previous'' transitions to improve expected runtime bounds, and vice versa.

To improve $\tbounde$,
we generate a PLRF $\rank$ for a part of the program.
To obtain a bound for the \emph{full} program from $\rank$, one has to determine which transitions can enter the program part and from which locations it can be entered.

\begin{definition}[Entry Locations and Transitions]
	For $\GTNI \subseteq \GT$ and $\loc \in \LOC$, the \emph{entry transitions} are $\ET_{\GTNI}(\loc) = \{g \in \GT\setminus\GTNI \mid \exists t\in g\ldotp t = (\_,\_,\_,\_,\loc)\}$.
	Then the \emph{entry locations} \paper{\layout{\mypagebreak}}
        are all start locations of $\GTNI$ whose entry 
transitions are not empty, 
i.e., $\ENTRYLOC_{\GTNI} = \{\loc \mid \ET_{\GTNI}(\loc)\neq\emptyset \land
	(\loc,\_,\_,\_,\_) \in \bigcup\GTNI\}$.\footnote{For a set of sets like $\GTNI$,
	 $\bigcup\GTNI$ denotes their union,  i.e., $\bigcup\GTNI = \bigcup_{g \in \GTNI}
	g$.}
\end{definition}

\begin{example}[Entry Locations and Transitions]
	\label{ex:entry_transitions_and_locations}
	{\sl For the PIP from \cref{fig:simple_quadratic} and $\GTNI = \{g_1\}$, we have $\ENTRYLOC_{\GTNI} = \{\loc_1\}$ and $\ET_{\GTNI}(\loc_1) = \{ g_0 \}$.
		So the loop formed by $g_1$ is entered at location $\loc_1$ and the general transition $g_0$ has to be executed before.
		Similarly, for $\GTNIpr = \{g_3\}$ we have $\ENTRYLOC_{\GTNIpr} = \{\loc_2\}$ and $\ET_{\GTNIpr}(\loc_2) = \{ g_2 \}$.
	}
\end{example}

Recall that if $\rank$ is a PLRF for $\GTG \subseteq \GTNI$, then in a program that is restricted to $\GTNI$, $\rank(\loc)$ is an upper bound on the expected number of executions of transitions from $\GTG$ when starting in $\loc$.
Since $\rank(\loc)$ may contain negative coefficients, it is not weakly monotonically increasing in general.
To turn expressions $e \in \REALPOLYBOUNDPV$ into bounds from $\BOUND$, let the over-approximation $\overapprox{\cdot}$ replace all coefficients by their absolute value.
So for example, $\overapprox{x-y} =\overapprox{x + (-1)\cdot y} = x + y$.
Clearly, we have $\eval{\overapprox{e}}{\state} \geq \eval{e}{\state}$ for all $s \in \STATE$.
Moreover, if $e \in \REALPOLYBOUNDPV$ then $\overapprox{e} \in \BOUND$.

To turn $\overapprox{\rank(\loc)}$ into a bound for the full program, one has to take into account how often the sub-program with the transitions $\GTNI$ is reached via an entry transition $h \in \ET_{\GTNI}(\loc)$ for some $\loc \in \ENTRYLOC_{\GTNI}$.
This can be over-approximated by $\sum_{t=(\_,\_,\_,\_,\loc) \in h} \tbound(t)$, which is an upper bound on the number of times that transitions in $h$ to the entry location $\loc$ of $\GTNI$ are applied in a full program run.

The bound $\overapprox{\rank(\loc)}$ is expressed in terms of the program variables at the entry location $\loc$ of $\GTNI$.
To obtain a bound in terms of the variables at the start of the program, one has to take into account which value a program variable $x$ may have when the sub-program $\GTNI$ is reached.
For every entry transition $h \in \ET_{\GTNI}(\loc)$, this value can be over-approximated by $\sbounde (h,\loc,x)$.
Thus, we have to instantiate each variable $x$ in $\overapprox{\rank(\loc)}$ by $\sbounde(h,\loc,x)$.
Let $\sbounde(h,\loc,\cdot):\PV \to \BOUND$ be the mapping with $\sbounde(h,\loc,\cdot)(x) = \sbounde(h,\loc,x)$.
Hence, $\sbounde(h,\loc,\cdot)(\overapprox{\rank(\loc)})$ over-approximates the expected number of applications of $\GTG$ if $\GTNI$ is entered in location $\loc$, where this bound is expressed in terms of the input variables of the program.
Here, weak monotonic increase of $\overapprox{\rank(\loc)}$ ensures that instantiating its variables by an over-approximation of their size yields an over-approximation of the runtime.\paper{\layout{\linebreak[4]
		}}
	\paper{\layout{\vspace*{-0.5cm}}}

\begin{restatable}[Expected Runtime Bounds]{theorem}{expectedtimeboundsmethod}
	\label{theorem:exptimeboundsmeth}
	Let $(\tbounde,\sbounde)$ be an expected bound pair, $\tbound$ a (non-probabilistic) runtime bound, and $\rank$ a PLRF for $\GTG \subseteq \GTNI \subseteq \GT$.
	Then $\tbounde'\colon \GT \rightarrow \BOUND$ is an expected runtime bound where
	\reducedcenter{$ \tbounde' (g) =
				\begin{cases}
					\sum\limits_{\substack{\loc \in \ENTRYLOC_{\GTNI}                                                                                                                  \\h \in \ET_{\GTNI}(\loc)}}
					(\sum\limits_{t=(\_,\_,\_,\_,\loc) \in h} \tbound(t))\cdot \left(\exact{\overapprox{\rank(\loc)}}{\sbounde(h,\loc,\cdot)}\right), & \!\!\text{if } g \in \GTG      \\
					\tbounde (g),                                                                                                                     & \!\!\text{if } g \not \in \GTG
				\end{cases}
			$}
\end{restatable}
\begin{example}[Expected Runtime Bounds]
	\label{ex:time_bounds}
	{\sl For the PIP from \cref{fig:simple_quadratic}, our approach starts with $(\tbounde,\sbounde)$ from \cref{ex:Lifting of Bounds} which results from lifting the bound pair from \cref{ex:complexity_bounds_nonprob}.
		To improve the bound $\tbounde(g_1) = \infty$,
we use the PLRF $\rank$ for $\GTG = \GTNI = \{ g_1 \}$ from \cref{ex:PLRF}.
		By \cref{ex:entry_transitions_and_locations}, we have $\ENTRYLOC_{\GTNI} = \{ \loc_1 \}$ and $\ET_{\GTNI}(\loc_1) = \{ g_0 \}$ with $g_0 = \{ t_0 \}$, whose runtime bound is $\tbound(t_0) = 1$, see \cref{ex:complexity_bounds_nonprob}.
		Using the expected size bound $\sbounde (g_0,\loc_1,x) = x$ from \cref{ex:probabilistic_bounds}, \cref{theorem:exptimeboundsmeth} yields
		\[
			\tbounde' (g_1) = \tbound(t_0) \cdot \exact{\overapprox{\rank(\loc_1)}}{\sbounde(g_0,\loc_1,\cdot)}
			= 1 \cdot 2\cdot x = 2 \cdot x.
		\]
		To improve $\tbounde(g_3)$, 
                we use the PLRF $\rank'$ for $\GTG' = \GTNIpr = \{ g_3 \}$ from \cref{ex:PLRF}.
		As $\ENTRYLOC_{\GTNIpr} = \{ \loc_2 \}$ and $\ET_{\GTNIpr}(\loc_2) = \{
	g_2 \}$ by \cref{ex:entry_transitions_and_locations}, where $g_2 = \{t_3\}$ and
	$\tbound(t_3) = 1$ (\cref{ex:complexity_bounds_nonprob}),
with the bound $\sbounde (g_2,\loc_2,y) = 6\cdot x^2 + 2 \cdot y$ from \cref{ex:probabilistic_bounds}, \cref{theorem:exptimeboundsmeth} yields
		\[
			\tbounde' (g_3) \,
	= \, \tbound(t_3) \cdot \exact{\overapprox{\rank'(\loc_2)}}{\sbounde(g_2,\loc_2,\cdot)}
	= 1 \cdot \sbounde(g_2,\loc_2,y) = 6\cdot x^2 + 2 \cdot y. \]
		So based on the expected size bounds of \cref{ex:probabilistic_bounds}, we have shown how to compute the expected runtime bounds of \cref{ex:probabilistic_bounds} automatically.
	}
\end{example}

Similar to \cite{koat}, our approach relies on combining bounds that one has computed earlier in order to derive new bounds.
Here, bounds may be combined \emph{linearly}, bounds may be \emph{multiplied}, and bounds may even be \emph{substituted} into other bounds.
But in contrast to \cite{koat}, sometimes one may combine \emph{expected} bounds that were computed earlier and sometimes it is only sound to combine \emph{non-probabilistic} bounds: If a new bound is computed by \emph{linear combinations} of earlier bounds, then it is sound to use the ``expected versions'' of these earlier bounds.
However, if two bounds are \emph{multiplied}, then it is in general not sound to use their ``expected versions''.
Thus, it would be \emph{unsound} to use the \emph{expected} runtime bounds $\tbounde(h)$ instead of the \emph{non-probabilistic} bounds $\sum_{t=(\_,\_,\_,\_,\loc) \in h} \tbound(t)$ on the entry transitions in \cref{theorem:exptimeboundsmeth} (\paper{a counterexample is given in \cite{arxivereport}}\arxiv{a counterexample is \cref{ex:entry_transition_unsound} in \cref{app:additional_examples}}).%
\footnote{
	An exception is the special case where $\rank(\loc)$ is \emph{constant}.
	Then, our implementa\-tion indeed uses the expected bound $\tbounde(h)$ instead of $\sum_{t=(\_,\_,\_,\_,\loc) \in h} \tbound(t)$\paper{ \cite{arxivereport}}\arxiv{, see \cref{Expected Time Bounds Refined}}.
}

In general, if bounds $b_1,\ldots,b_n$ are \emph{substituted} into another bound $b$, then
it is sound to use ``expected versions'' of the bounds $b_1,\ldots,b_n$ if $b$
is \emph{concave}, see, e.g., \cite{kallenberg_foundations_2002,dblp:journals/scp/avanzinily20,ecoimp}.
Since bounds from $\BOUND$ do not contain negative coefficients, we obtain that a finite\footnote{A bound is \emph{finite} if it does not contain $\infty$.
	We always simplify expressions and thus, a bound like $0 \cdot \infty$ is also finite, because it simplifies to $0$, as usual in measure theory.} bound $b \in \BOUND$ is concave iff it is a linear polynomial (see \paper{\cite{arxivereport}}\arxiv{\cref{lemma:concavity_implies_linearity} in the appendix}).

Thus, in \cref{theorem:exptimeboundsmeth} we may substitute \emph{expected}
size bounds $\sbounde (h,\loc,x)$ into $\overapprox{\rank(\loc)}$, since we restricted ourselves to \emph{linear} ranking functions $\rank$ and hence, $\overapprox{\rank(\loc)}$ is also linear.
Note that in contrast to \cite{dblp:journals/scp/avanzinily20}, where a notion of concavity was used to analyze probabilistic term rewriting, a multilinear expression like $x \cdot y$ is not concave when regarding both arguments simultaneously.
Hence, it is unsound to use such ranking functions in \cref{theorem:exptimeboundsmeth}.
See \paper{\cite{arxivereport}}\arxiv{\cref{ex:concavity_in_expected_size_bounds}} for a counterexample to show why substituting expected bounds into a non-linear bound is incorrect in general.

%% file: comp_size.tex
We first compute \emph{local} bounds for one application of a transition (\cref{sec:Expected Local Change Bound}).
To turn them into \emph{global} bounds, we encode the data flow of a PIP in a graph.
\cref{subsec:expected_size_bounds_method} then presents our technique to compute expected size bounds.
\layout{\vspace*{-.1cm}}
\subsection{Local Change Bounds and General Result Variable Graph}
\label{sec:Expected Local Change Bound}

We first compute a bound on the expected change of a variable during an update.
More precisely, for every general result variable $(g,\loc,x)$ we define a bound
$\chbounde (g,\loc,x)$ \paper{\layout{\mypagebreak}}
on the change of the variable $x$ that we can expect in one
execution of the general transition $g$ when reaching location $\loc$. 
So we consider all $t = (\_,p,\_,\eta,\loc)\in g$ and the expected difference between the current value of $x$ and its update $\eta(x)$.
However, for $\eta(x) \in \POLYBOUND$,
$\eta(x) - x$ is not necessarily from $\BOUND$ because it may contain negative coefficients.
Thus, we use the over-approximation $\overapprox{\eta(x)-x}$ (where we always simplify expressions before applying $\overapprox{\cdot}$, e.g., $\overapprox{x-x}
	= \overapprox{0} = 0$).
Moreover, $\overapprox{\eta(x)-x}$ may contain temporary variables.
Let $\maxTV_t:\VV \to \BOUND$ instantiate all temporary variables by the largest possible value they can have after evaluating the transition $t$.
Hence, we then use $\maxTV_t(\overapprox{\eta(x)-x})$ instead.
For $\maxTV_t$, we have to use the underlying \emph{non-probabilistic} size bound $\sbound$ for the program (since the scheduler determines the values of temporary variables by non-deterministic (\emph{non-probabilistic}) choice).
If $x$ is updated according to a bounded distribution function $d \in \DIST$, 
then as in \cref{sec:Preliminaries},
let  $\expvdist(d) \in \BOUND$  denote a finite bound
on $d$, i.e.,  $\expvsignabs(d(\state)) \leq \eval{\expvdist(d)}{\state}$ for all $\state \in \STATE$.
\renewcommand{\arraystretch}{0.9}

\begin{definition}[Expected Local Change Bound]
	\label{def:Bound for Local Expected Change}
	Let $\sbound$ be a size bound.
	Then $\chbounde \colon \GRV \to \BOUND$ with $\chbounde (g,\loc,x) = \sum\limits_{t=(\_,p,\_,\eta,\loc)\in g}
		p \cdot \ch_t (\eta (x),x)$,\layout{\vspace{-.3cm}}
	where
	\[
		\ch_t (\eta(x),x) = \left\{
		\begin{array}{l}
			\expvdist(d), \quad \quad \; \text{if $\eta(x)=d \in\DIST$} \\
			\maxTV_t(\overapprox{ \eta(x)- x}), \; \text{otherwise}
		\end{array}
		\right. \; \text{and} \; \maxTV_t(y) = \left\{
		\begin{array}{ll}
			\sbound(t,y), & \text{if } y \notin \PV \\
			y,            & \text{if } y \in \PV
		\end{array}
		\right.
	\]
\end{definition}
\renewcommand{\arraystretch}{1.0}

\begin{example}[$\chbounde$]
	\label{ex:comp_size_elsb}
	{\sl For the PIP of \cref{fig:simple_quadratic}, we have $\chbounde(g_0,\_,\_) = \chbounde(g_2,\_,\_) = \chbounde(g_3,\loc_2,x) = 0$, since the respective updates are identities.
		Moreover,
		\[
			\chbounde (g_1,\loc_1,x) \; = \; \tfrac{1}{2} \cdot \overapprox{(x - 1) - x}
			+ \tfrac{1}{2} \cdot \overapprox{x - x}
			\;=\; \tfrac{1}{2} \cdot 1 + \tfrac{1}{2} \cdot 0 \; =\; \tfrac{1}{2}.
		\]
		In a similar way, we obtain $\chbounde (g_1,\loc_1,y) = x$ and $\chbounde (g_3,\loc_2,y) = 1$.}
\end{example}

The following theorem shows that for any admissible configuration in a state $\state'$, $\chbounde(g,\loc,x)$ is an upper bound on the expected value of $|\state(x) - \state'(x)|$ if $\state$ is the next state obtained when applying $g$ in state $\state'$ to reach location $\loc$.

\begin{restatable}[Soundness of $\chbounde$]{theorem}{chboundebasic}
	\label{theorem:elsb_basic}
	For any	 $(g,\loc,x) \in \GRV$, scheduler $\scheduler$, $\initstate \in \STATE$, and admissible configuration $c' = (\_,\_,\state')$, we have
	\[
		\eval{\chbounde (g,\loc,x)}{\state'} \;\; \geq \;\; \sum\nolimits_{c = (\loc,t,\state) \in \CONF, \; t\in g} \;\; \ptransition{\pip}{\scheduler}{\initstate} (c' \to c)\cdot |\state (x) - \state' (x)|.
	\]
\end{restatable}

To obtain \emph{global} bounds from the local bounds $\chbounde (g, \loc, x)$, we construct a \emph{general result variable graph} which encodes the data flow between variables.
Let $\pre(g) = \ET_{\emptyset}(\loc_g)$ be the the set of \emph{pre-transitions} of $g$ which lead into $g$'s start location $\loc_g$.
Moreover, for $\alpha = (g,\loc,x) \in \GRV$ let its \emph{active variables} $\ACTV(\alpha)$ consist of all variables occurring in the bound $x + \chbounde (\alpha)$ for $\alpha$'s expected size.

\begin{definition}[General Result Variable Graph]
	\label{def:general_result_variable_graph}
	The \emph{general result variable graph} has the set of nodes $\GRV$ and the set of \underline{e}dges $\mathcal{GRV\underline{E}}$, where
	\[
		\mathcal{GRVE} = \{\, ((g',\loc',x'), \; (g,\loc,x)) \mid g'\in\pre (g) \land \loc' = \loc_g \land x'\in\ACTV (g,\loc,x) \,\}. \]

\end{definition}
\begin{example}[General Result Variable Graph]
	\label{grvg}
	\sl
	\nolayout{
		\input{figures/comp_size/grvg_simple_quadratic.tex}

	}
	\noindent
	The general result variable graph for the PIP of \cref{fig:simple_quadratic} is shown \layout{below}\nolayout{in \cref{fig:grvg_simple_quadratic}}.
	For $\chbounde$ from \cref{ex:comp_size_elsb}, we have $\ACTV (g_1,\loc_1,x)
	= \{x\}$, as $x + \chbounde(\alpha) = x +\tfrac{1}{2}$ contains
        \paper{\layout{\mypagebreak}} no variable except $x$.
	\layout{
		\InsertBoxR{0}{
			\input{figures/comp_size/grvg_simple_quadratic.tex}
		}
	}
	Similarly, $\ACTV (g_1,\loc_1,y) = \{x,y\}$, as $x$ and $y$ are contained in $y+ \chbounde(g_1,\loc_1,y) = y + x$.
	For all other $\alpha \in \GRV$, we have $\ACTV (\_,\_,x) = \{x\}$ and $\ACTV (\_,\_,y) = \{y\}$.
	As $\pre (g_1) = \{g_0,g_1\}$, the graph captures the dependence of $(g_1,\loc_1,x)$ on $(g_0,\loc_1,x)$ and $(g_1,\loc_1,x)$, and of $(g_1,\loc_1,y)$ on $(g_0,\loc_1,x)$, $(g_0,\loc_1,y)$, $(g_1,\loc_1,x)$, and $(g_1,\loc_1,y)$.
	The other edges are obtained in a similar way.
\end{example}

\subsection{Inferring Expected Size Bounds}
\label{subsec:expected_size_bounds_method}
We now compute global expected size bounds for the general result variables by considering the SCCs of the general result variable graph separately.
As usual, an SCC is a maximal subgraph with a path from each node to every other node.
An SCC is \emph{trivial} if it consists of a single node without an edge to itself.
We first handle trivial SCCs in \cref{subsubsec:Inferring Expected Size Bounds for Trivial SCCs} and consider non-trivial SCCs in \cref{subsubsec:Inferring Expected Size Bounds for Non-Trivial SCCs}.

\subsubsection{Inferring Expected Size Bounds for Trivial SCCs}
\label{subsubsec:Inferring Expected Size Bounds for Trivial SCCs}

By \cref{theorem:elsb_basic}, $x + \chbounde(g,\loc,x)$ is a \emph{local} bound on the expected value of $x$ after applying $g$ once in order to enter $\loc$.
However, this bound is formulated in terms of the values of the variables immediately before applying $g$.
We now want to compute \emph{global} bounds in terms of the \emph{initial} values of the variables at the start of the program.

If $g$ is \emph{initial} (i.e., $g \in \GTINIT$ since $g$ starts in the initial location $\loc_0$), then $x + \chbounde(g,\loc,x)$ is already a global bound,
as the values of the variables before the application of $g$ are the initial values of the variables at the program start.

Otherwise, the variables $y$ occurring in the local bound $x + \chbounde(g,\loc,x)$ have to be replaced by the values that they can take in a full program run before applying the transition $g$.
Thus, we have to consider all transitions $h \in \pre(g)$ and instantiate every variable $y$ by the maximum of the values that $y$ can have after applying $h$.
Here, we again over-approximate the maximum by the sum.

If $\chbounde(g,\loc,x)$ is \emph{concave} (i.e., a \emph{linear} polynomial), then we can instantiate its variables by \emph{expected} size bounds $\sbounde (h,\loc_g,y)$.
However, this is unsound if $\chbounde(g,\loc,x)$ is not linear, i.e., not concave (see \paper{\cite{arxivereport}}\arxiv{\cref{ex:concavity_in_expected_size_bounds}} for a counterexample).
So in this case, we have to use \emph{non-probabilistic}
bounds $\sbound (t,y)$ instead.

As in \cref{Inferring Expected Time Bounds}, we use an underlying non-probabilistic bound pair $(\tbound,\sbound)$ and start with the expected pair $(\tbounde,\sbounde)$ obtained by lifting $(\tbound,\sbound)$ according to \cref{theorem:liftingofbounds}.
While \cref{theorem:exptimeboundsmeth} improves $\tbounde$, we now improve $\sbounde$.
Here, the SCCs of the general result variable graph should be treated in topological order, since then one may first improve $\sbounde$ for result variables corresponding to $\pre(g)$, and use that when improving $\sbounde$ for result variables of the form $(g, \_, \_)$.

\begin{restatable}[Expected Size Bounds for Trivial SCCs]{theorem}{expectedtrivialsizeboundsmethod}
	\label{thm:Inferring Expected Size Bounds for Trivial SCCs}
	Let $\sbounde$ be an expected size bound, $\sbound$ a (non-probabilistic) size bound, and let $\alpha = (g,\loc,x)$ form a trivial SCC of the general result variable graph.
	Let $\incsize^\alpha_{\expvsign}$ and $\incsize^\alpha$ be mappings from $\PV \to \BOUND$ with $\incsize^\alpha_{\expvsign}(y) = \sum\nolimits_{h \in \pre (g)} \sbounde (h,\loc_g,y)$ and $\incsize^\alpha(y) = \sum\nolimits_{h \in \pre (g), \; t = (\_,\_,\_,\_,\loc_g) \in h} \sbound (t,y)$.
	\paper{\layout{\mypagebreak}}
	Then $\sbounde' : \GRV \to \BOUND$ is an expected size bound, where $\sbounde'(\beta) = \sbounde (\beta)$ for $\beta \neq \alpha$ and \reducedcenter{$ \sbounde' (\alpha) =
			\begin{cases}
				x + \chbounde (\alpha),                                               & \text{if $g\in\GTINIT$}                                         \\
				\incsize^\alpha_{\expvsign}(x + \chbounde (\alpha)),                  & \text{if $g\not\in\GTINIT$, $\chbounde (\alpha)$ is linear}     \\
				\incsize^\alpha_{\expvsign}(x) + \incsize^\alpha(\chbounde (\alpha)), & \text{if $g\not\in\GTINIT$, $\chbounde (\alpha)$ is not linear}
			\end{cases}
		$}
\end{restatable}

\begin{example}[$\sbounde$ for Trivial SCCs]
	\label{ex:Expected Size Bounds for Trivial SCCs}
	{\sl	The 
        general result variable graph in \layout{\cref{grvg}}
	\nolayout{\cref{fig:grvg_simple_quadratic}}
	contains
4 trivial SCCs formed by $\alpha_x = (g_0,\loc_1,x)$, $\alpha_y = (g_0,\loc_1,y)$, $\beta_x = (g_2,\loc_2,x)$, and $\beta_y = (g_2,\loc_2,y)$.
	For all these general result variables, the expected local change bound
	$\chbounde$ is $0$ (see \cref{ex:comp_size_elsb}). Thus, it is linear.
	Since $g_0\in \GTINIT$, \cref{thm:Inferring Expected Size Bounds for Trivial SCCs}
	yields $\sbounde' (\alpha_x) = x + \chbounde (\alpha_x) = x$ and $\sbounde' (\alpha_y) = y + \chbounde (\alpha_y) = y$.

	By treating SCCs in topological order, when handling $\beta_x$, $\beta_y$, we can assume that we already have $\sbounde (\alpha_x) = x$, $\sbounde (\alpha_y) = y$ and $\sbounde (g_1,\loc_1,x) = 2 \cdot x$, $\sbounde(g_1,\loc_1,y) = 6\cdot x^2 + y$ (see \cref{ex:probabilistic_bounds}) for the result variables corresponding to $\pre(g_2) = \{ g_0, g_1 \}$.
	We will explain in \cref{subsubsec:Inferring Expected Size Bounds for Non-Trivial SCCs} how to compute such expected size bounds for non-trivial SCCs.
	Hence, by \cref{thm:Inferring Expected Size Bounds for Trivial SCCs} we obtain $\sbounde' (\beta_x) = \incsize^{\beta_x}_{\expvsign} (x + \chbounde(\beta_x)) = \sbounde (\alpha_x) + \sbounde (g_1,\loc_1,x) = 3 \cdot x$ and $\sbounde' (\beta_y) = \incsize^{\beta_y}_{\expvsign} (y + \chbounde(\beta_y)) = \sbounde (\alpha_y) + \sbounde (g_1,\loc_1,y) = 6\cdot x^2 + 2 \cdot y$.
	}
\end{example}

\subsubsection{Inferring Expected Size Bounds for Non-Trivial SCCs}
\label{subsubsec:Inferring Expected Size Bounds for Non-Trivial SCCs}
Now we handle non-trivial SCCs $C$ of the general result variable graph.
An upper bound for the expected size of a variable $x$ when entering $C$ is obtained from $\sbounde (\beta)$ for all general result variables $\beta = (\_,\_,x)$ which have an edge to $C$.

To turn $\chbounde(g,\loc,x)$ into a global bound, as in \cref{thm:Inferring Expected Size Bounds for Trivial SCCs} its variables $y$ have to be instantiated by the values $\incsize^{(g,\loc,x)}(y)$ that they can take in a full program run before applying a transition from $g$.
Thus, $\incsize^{(g,\loc,x)}(\chbounde(g,\loc,x))$ is a global bound on the expected change resulting from \emph{one} application of $g$.
To obtain an upper bound for the whole SCC $C$, we add up these global bounds for all $(g,\_,x) \in C$ and take into account how often the general transitions in the SCC are expected to be executed, i.e., we multiply with their expected runtime bound $\tbounde(g)$.
So while in \cref{theorem:exptimeboundsmeth}
we improve $\tbounde$ using expected size bounds for previous transitions, we now improve $\sbounde(C)$ using expected runtime bounds for the transitions in $C$ and expected size bounds for previous transitions.

\begin{restatable}[Expected Size Bounds for Non-Trivial SCCs]{theorem}{expectednontrivialsizeboundsmethod}
	\label{theorem:expectednontrivialsizeboundsmeth}
	Let $(\tbounde,\allowbreak \sbounde)$ be an expected bound pair, $(\tbound,\sbound)$ a (non-probabilistic) bound pair, and let $C \subseteq \GRV$ form a non-trivial SCC of the general result variable graph where $\GT_C = \{g \in \GT \mid (g,\_,\_) \in C\}$.
	Then $\sbounde'$ is an expected size bound:
        		\reducedcenter{$ \sbounde' (\alpha) =
			\begin{cases}
				\sum_{(\beta,\alpha)\in \mathcal{GRVE}, \; \beta \notin C, \; \alpha \in C, \; \beta = (\_,\_,x)} \; \sbounde (\beta) \quad +                       \\
				\;\;\; \sum\limits_{g \in \GT_C}\tbounde (g) \cdot (\sum\limits_{\alpha' = (g,\_,x) \in C} \eff{\alpha'}), & \!\!\text{if $\alpha=(\_,\_,x) \in C$} \\
				\sbounde (\alpha),                                                                                         & \!\!\text{otherwise}
			\end{cases}
		$}
\end{restatable}
\paper{\layout{\vspace*{-.1cm}}}
Here we really have to use the \emph{non-probabilistic} size bound $\incsize^{\alpha'}$ instead of $\incsize^{\alpha'}_{\mathbb{E}}$, even if $\chbounde(\alpha')$ is linear, i.e., concave.
Otherwise we would multiply the expected values of two random variables which are not independent.

\begin{example}[$\sbounde$ for Non-Trivial SCCs]
	{\sl The
	\paper{\layout{\mypagebreak}}
        general result variable graph in \layout{\cref{grvg}}\nolayout{\cref{fig:grvg_simple_quadratic}}
		contains 
                4 non-trivial SCCs formed by $\alpha_x' = (g_1,\loc_1,x)$, $\alpha_y' = (g_1,\loc_1,y)$, $\beta_x' = (g_3,\loc_2,x)$, and $\beta_y' = (g_3,\loc_2,y)$.
		By the results on $\sbounde$, $\tbounde$, $\chbounde$, and $\sbound$ from \cref{ex:Expected Size Bounds for Trivial SCCs}, \ref{ex:time_bounds}, \ref{ex:comp_size_elsb}, and \ref{ex:complexity_bounds_nonprob}, \cref{theorem:expectednontrivialsizeboundsmeth} yields the expected size bound in \cref{ex:probabilistic_bounds}:
		\paper{\layout{\vspace*{-0.3cm}}}
		\reducedcenter{$
				\begin{array}{rclcl}
					\sbounde'(\alpha_x') & = & \sbounde(\alpha_x) + \tbounde(g_1) \cdot \incsize^{\alpha_x'}(\chbounde(\alpha_x')) & = & x + 2 \cdot x \cdot \tfrac{1}{2} = 2 \cdot x    \\
					\sbounde'(\alpha_y') & = & \sbounde(\alpha_y) + \tbounde(g_1) \cdot \incsize^{\alpha_y'}(\chbounde(\alpha_y')) & = & y + 2 \cdot x \cdot \incsize^{\alpha_y'}(x)     \\
					                     & = & y + 2 \cdot x \cdot \sum\nolimits_{i \in \{0,1,2\}} \sbound(t_i,x)                  & = & 6 \cdot x^2 + y                                 \\
					\sbounde'(\beta_x')  & = & \sbounde(\beta_x) + \tbounde(g_3) \cdot \incsize^{\beta_x'}(\chbounde(\beta_x'))    & = & 3 \cdot x + (6 x^2 + 2 y) \cdot 0 = 3 \cdot x   \\
					\sbounde'(\beta_y')  & = & \sbounde(\beta_y) + \tbounde(g_3) \cdot \incsize^{\beta_y'}(\chbounde(\beta_y'))    & = & 6 \cdot x^2 + 2 \cdot y + (6 x^2 + 2 y) \cdot 1 \\
					                     &   &                                                                                     & = & 12 \cdot x^2 + 4 \cdot y
				\end{array}
			$}
	}
\end{example}

%% file: figures/comp_size/grvg_simple_quadratic.tex
\nolayout{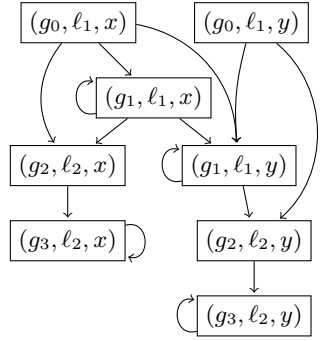
\begin{figure}}
\centering

  \begin{tikzpicture}[shorten >=1pt,node distance=\layout{0.95}\nolayout{1}cm,auto,main node/.style={rectangle,draw,font=\sffamily\small\bfseries},every
    loop/.style={max distance=3.5mm}]
    \node[main node] (g0l1x) {$(g_0,\loc_1,x)$};
    \node[main node] (g1l1x) [below of=g0l1x,xshift=1cm] {$(g_1,\loc_1,x)$};
    \node[main node] (g2l2x) [below left = 0.5cm of g1l1x,xshift=0.75cm] {$(g_2,\loc_2,x)$};
    \node[main node] (g3l2x) [below of = g2l2x] {$(g_3,\loc_2,x)$};
    \node[main node] (g1l1y) [below right = 0.5cm of g1l1x,xshift=-0.75cm] {$(g_1,\loc_1,y)$};
    \node[main node] (g2l2y) [below of = g1l1y,xshift=0.2cm] {$(g_2,\loc_2,y)$};
    \node[main node] (g3l2y) [below of = g2l2y] {$(g_3,\loc_2,y)$};
    \node[main node] (g0l1y) [right of = g0l1x,xshift=1.3cm] {$(g_0,\loc_1,y)$};

    \path [->]
      (g0l1x) edge (g1l1x)
      (g1l1x) edge (g2l2x)
      (g2l2x) edge (g3l2x)
      (g3l2x) edge [loop right] ()
      (g1l1x) edge (g1l1y)
      (g1l1y) edge (g2l2y)
      (g1l1y) edge [loop left] ()
      (g2l2y) edge (g3l2y)
      (g1l1x) edge [loop left] ()
      (g0l1x) edge [bend left=40] (g1l1y)
      (g0l1y) edge [bend right=10] (g1l1y)
      (g0l1y) edge [bend left=50] (g2l2y)
      (g0l1x) edge [bend right=35] (g2l2x)
      (g3l2y) edge [loop left] ();
  \end{tikzpicture}
 \nolayout{\caption{General result variable graph for the PIP in \cref{fig:simple_quadratic}}
   \label{fig:grvg_simple_quadratic}}
  \nolayout{\end{figure}}

%% file: related_work.tex
\paragraph{Related Work}
Our approach adapts techniques from \cite{koat} to probabilistic programs.
As explained in \cref{sec:Intro}, this adaption is not at all trivial (see our proofs in \arxiv{the appendix}\paper{\cite{arxivereport}}).

There has been a lot of work on proving PAST and inferring bounds on expected runtimes using supermartingales, e.g., \cite{dblp:journals/scp/avanzinily20,dblp:conf/rta/bournezg05,dblp:conf/cav/chakarovs13,rsmchatt,dblp:conf/popl/chatterjeenz17,lexrsm,dblp:conf/vmcai/fuc19,dblp:conf/pldi/wang0gcqs19,dblp:conf/rta/bournezg06,cade19,amber,FoundationsTerminationMartingale2020}.
While these techniques infer one (lexicographic) ranking supermartingale to analyze the complete program, our approach deals with information flow between different program parts and analyzes them separately.

There is also work on modular analysis of almost sure termination (AST) \cite{lexrsm,rsm,dblp:journals/pacmpl/huang0cg19,dblp:conf/aplas/huangfc18,amber,FoundationsTerminationMartingale2020}, i.e., termination with probability $1$.
This differs from our results, since AST is compositional, in contrast to PAST (see, e.g., \cite{dblp:journals/jacm/kaminskikmo18,FoundationsExpectedRuntime2020}).

A fundamentally different approach to ranking supermartingales (i.e., to \emph{forward-reasoning}) is \emph{backward-reasoning} by so-called expectation transformers, see, e.g., \cite{absynth,dblp:journals/jacm/kaminskikmo18,dblp:conf/lics/olmedokkm16,dblp:journals/jcss/kozen81,dblp:journals/pacmpl/mcivermkk18,dblp:series/mcs/mciverm05,ecoimp,hoffmannicfp2020,FoundationsExpectedRuntime2020}.
In this orthogonal reasoning, \cite{dblp:journals/jacm/kaminskikmo18,dblp:conf/lics/olmedokkm16,ecoimp,FoundationsExpectedRuntime2020} consider the connection of the expected runtime and size.
While expectation transformers apply backward- instead of forward-reasoning, their correctness can also be justified using supermartingales.
More precisely, Park induction for upper bounds on the expected runtime via expectation transformers essentially ensures that a certain stochastic process is a supermartingale (see \cite{dblp:journals/pacmpl/harkkgk20} for details).

To the best of our knowledge, the only available tools for the inference of upper bounds on the expected runtimes of probabilistic programs are \cite{absynth,dblp:conf/pldi/wang0gcqs19,hoffmannicfp2020,ecoimp}.
The tool of \cite{hoffmannicfp2020} deals with data types and higher order functions in probabilistic \textsf{ML} programs and does not support programs whose complexity depends on (possibly negative) integers (see \cite{ramlweb}).
Furthermore, the tool of \cite{amber} focuses on proving or refuting (P)AST of probabilistic
programs for so-called \emph{Prob-solvable} loops, which do not allow for nested or
sequential loops or non-determinism. So both \cite{hoffmannicfp2020} and
 \cite{amber} are
orthogonal to our work.
We discuss \cite{absynth,dblp:conf/pldi/wang0gcqs19,ecoimp} below.

%% file: implementation.tex
\mypar{Implementation}
We implemented our analysis in a new version of our tool \textsf{KoAT} \cite{koat}.
\textsf{KoAT} is an open-source tool written in \textsf{OCaml}, which can also be
downloaded as a Docker image
and accessed via a web interface \cite{techreport}.

\begin{algorithm}[t]
	\caption{Overall approach to infer bounds on expected runtimes}\label{alg:overall}
	\DontPrintSemicolon
	\KwIn{PIP $\PIP$}
	preprocess the PIP \label{alg:overall_preprocess}\;
	$(\tbound,\sbound) \gets$ perform non-probabilistic analysis using \cite{koat}
	\label{alg:overall_nonprobanalysis}\;
	$(\tbounde,\sbounde) \gets$ lift $(\tbound,\sbound)$ to an expected
	bound pair with \cref{theorem:liftingofbounds} \label{alg:overall_lifting}\;
	\Repeat{{\rm no bound is improved anymore}}{
		\For{{\rm \textbf{all} SCCs $C$ of the general result variable graph in topological order}}{
			\lIf{{\rm $C = \{\alpha\}$ is trivial}}
			{$\sbounde' \gets $ improve $\sbounde$ for $C$ by
				\cref{thm:Inferring Expected Size Bounds for Trivial SCCs}
			}
			\lElse{$\sbounde' \gets $ improve $\sbounde$ for $C$ by
				\cref{theorem:expectednontrivialsizeboundsmeth} \label{alg:nonTrivial}}
			\lFor{{\rm \textbf{all}} $\alpha \in C$}
			{$\sbounde(\alpha) \gets \min \{
					\sbounde(\alpha),\sbounde' (\alpha) \}$} \label{alg:update2}
		}
		\For{{\rm \textbf{all} general transitions $g\in \GT$}}
		{ $\tbounde' \gets $ improve $\tbounde$ for $\GTG = \{g\}$
		by \cref{theorem:exptimeboundsmeth} \;\label{alg:timebound}
			$\tbounde(g) \gets \min \{\tbounde(g),\tbounde'(g)\}$\; \label{alg:update3}
		}}
	\KwOut{$\sum_{g\in\GT} \tbounde (g)$}
	\vspace*{.2cm}
\end{algorithm}
Given a PIP, the analysis proceeds as in \cref{alg:overall}.
The preprocessing in \cref{alg:overall_preprocess} adds invariants to guards
 (using \textsf{APRON} \cite{apron} to generate (non-probabilistic) invariants), unfolds
 transitions \cite{bd77},
and removes unreachable locations, transitions with probability $0$, and transitions with
 unsatisfiable guards
\paper{\layout{\mypagebreak}}
(using \textsf{Z3} \cite{z3}).

We start by a non-probabilistic analysis and lift the resulting bounds
to an initial expected bound pair (\cref{alg:overall_nonprobanalysis,alg:overall_lifting}).
Afterwards, we first try to improve the expected size bounds using \cref{thm:Inferring Expected Size Bounds for Trivial SCCs,theorem:expectednontrivialsizeboundsmeth}, and then we attempt to improve the expected runtime bounds using \cref{theorem:exptimeboundsmeth} (if we find a PLRF using \textsf{Z3}).
To determine the ``minimum'' of the previous and the new bound, we use a heuristic which compares polynomial bounds by their degree.
While we over-approximated the maximum of expressions by their sum to ease readability in
this paper, \textsf{KoAT} also uses bounds containing
``min'' and ``max'' to increase precision.

This alternating modular computation of expected size and runtime bounds is repeated so that one can benefit from improved expected runtime bounds when com\-puting expected size bounds and vice versa.
We abort this improvement of expected bounds in \cref{alg:overall} if they are all finite (or when reaching a timeout).

%% file: evaluation.tex
To assess the power of our approach, we performed an experimental evaluation of our implementation in \textsf{KoAT}.
We did not compare with the tool of \cite{dblp:conf/pldi/wang0gcqs19}, since \cite{dblp:conf/pldi/wang0gcqs19} expects the program to be annotated with already computed invariants.
But for many of the examples in our experiments, the invariant generation tool \cite{stanfordinvariantgenerator} used by \cite{dblp:conf/pldi/wang0gcqs19} did not find invariants strong enough to enable a meaningful analysis (and we could not apply \textsf{APRON} \cite{apron} due to the different semantics of invariants).

Instead, we compare \textsf{KoAT} with the tools \textsf{Absynth} \cite{absynth}
and \textsf{eco-imp} \cite{ecoimp} which are both based on a conceptionally
different \emph{backward-reasoning} approach.
We ran the tools on all 39 examples from  \textsf{Absynth}'s evaluation in \cite{absynth} (except \texttt{recursive}, which contains non-tail-recursion and thus cannot be encoded as a PIP), and on the 8 additional examples from the artifact of \cite{absynth}.
Moreover, our collection has\layout{\linebreak[3]} 29 additional benchmarks:\ \paper{14 examples that illustrate different aspects of PIPs}\arxiv{the 4 examples presented in this work and 10 further examples that illustrate different aspects of PIPs}, 5 PIPs based on examples from \cite{absynth} where we removed as\-sump\-tions, and 10 PIPs based on benchmarks from the \emph{TPDB} \cite{tpdb} where some transitions were enriched with probabilistic behavior.
The \emph{TPDB} is a collection of typical programs used in the annual \emph{Termination and Complexity Competition} \cite{termcomp}.
We ran the experiments on an iMac with an Intel i5-2500S CPU
and $\SI{12}{\giga\byte}$ of
RAM under macOS Sierra
\paper{\layout{\mypagebreak}}
for \textsf{Absynth} and NixOS 20.03 for \textsf{KoAT}
and \textsf{eco-imp}.
A timeout of 5 minutes per example was applied for all tools.
The average runtime of successful runs was $\SI{4.26}{\second}$
for  \textsf{KoAT},
$\SI{3.53}{\second}$ for \textsf{Absynth}, and just $\SI{0.93}{\second}$
for
\textsf{eco-imp}.

	\begin{figure}[t]
		\noindent
		\begin{minipage}{0.5\textwidth}
			\begin{center}

				\arxiv{\vspace*{.3cm}}

				\begin{tabular}{c|c|c|c|}
					\textbf{Bound}          & \textsf{KoAT} & \textsf{Absynth} & \textsf{eco-imp} \\
					\hline $\bigO(1)$       & 6             & 6                & 6               \\
					\hline $\bigO(n)$       & 32            & 32               & 29              \\
					\hline $\bigO(n^2)$     & 3             & 8                & 9               \\
					\hline $\bigO(n^{> 2})$ & 0             & 0                & 0               \\
					\hline \textsf{EXP}     & 0             & 0                & 0               \\
					\hline $\infty$         & 5             & 0                & 2               \\
					\hline \textsf{TO}      & 0             & 0                & 0
				\end{tabular}

				\layout{\vspace*{-.6cm}}
			\end{center}
			\captionof{figure}{{\small Results on benchmarks from \cite{absynth}}}
			\label{fig:comparison_absynth}
		\end{minipage}
		\begin{minipage}{0.5\textwidth}

			\arxiv{\vspace*{.3cm}}

			\begin{center}
				\begin{tabular}{c|c|c|c|}
					\textbf{Bound}          & \textsf{KoAT} & \textsf{Absynth} & \textsf{eco-imp} \\
					\hline $\bigO(1)$       & 2             & 1                & 2               \\
					\hline $\bigO(n)$       & 10            & 3                & 6               \\
					\hline $\bigO(n^2)$     & 12            & 1                & 6               \\
					\hline $\bigO(n^{> 2})$ & 2             & 0                & 0               \\
					\hline \textsf{EXP}     & 1             & 0                & 0               \\
					\hline $\infty$         & 2             & 15               & 12              \\
					\hline \textsf{TO}      & 0             & 9                & 3
				\end{tabular}

				\layout{\vspace*{-.6cm}}
			\end{center}
			\captionof{figure}{{\small Results on our new benchmarks}}
			\label{fig:comparison_ours}
		\end{minipage}
	\end{figure}

\cref{fig:comparison_absynth,fig:comparison_ours} show the generated asymptotic bounds, where $n$ is the maximal absolute value of the program variables at the program start.
Here, ``$\infty$'' indicates that no finite time bound could be computed and ``\textsf{TO}'' means ``timeout''.
The detailed asymptotic results of all tools on all examples can be found in \arxiv{\cref{app:examples}}\paper{\cite{arxivereport,techreport}}.

\textsf{Absynth} and \textsf{eco-imp} slightly outperform \textsf{KoAT} on the examples
from \textsf{Absynth}'s collection, while \textsf{KoAT} is considerably stronger than both tools on the additional benchmarks.
In particular, \textsf{Absynth} and \textsf{eco-imp} outperform our approach on examples with nested probabilistic loops.
While our modular approach can analyze inner loops separately when searching for probabilistic ranking functions, \cref{theorem:exptimeboundsmeth} then requires \emph{non-probabilistic} time bounds for all transitions entering the inner loop.
But these bounds may be infinite if the outer loop has probabilistic behavior itself.
Moreover, in contrast to our work and \cite{ecoimp}, the approach of \cite{absynth} does
not require weakly monotonic bounds.

On the other hand, \textsf{KoAT} is superior to \textsf{Absynth} and \textsf{eco-imp} on large examples with many loops, where only a few transitions have probabilistic behavior (this might correspond to the typical application of randomization in practical programming).
Here, we benefit from the modularity of our approach which treats loops independently and combines their bounds afterwards.
\textsf{Absynth} and \textsf{eco-imp} also fail for our leading example of \cref{fig:simple_quadratic}, while \textsf{KoAT} infers a quadratic bound.
Hence, the tools have particular strengths on orthogonal kinds of examples.

	\paper{\textsf{KoAT}'s source code is available at \sourcecode{}.}
\arxiv{The source code of
\textsf{KoAT} is available at the following website:
	\begin{center}\sourcecode{}\end{center}}
To obtain
a \textsf{KoAT} artifact, see
\begin{center}
\url{https://aprove-developers.github.io/ExpectedUpperBounds/}
\end{center}
for  a static binary and  
Dock\-er image. 
	This web site also provides  all examples from our evaluation\arxiv{ (including graphical
	representations of the corresponding transition systems)}, detailed outputs of our experiments, and a \emph{web interface} to run \textsf{KoAT} directly online.

%% file: conclusion.tex
\mypar{Conclusion}
We presented a new modular approach to infer upper bounds on the expected runtimes of probabilistic integer programs.
To this end, non-probabilistic and expected runtime and size bounds on parts of the 
program are computed in an alternating fashion and then combined to an overall expected
runtime bound. 
In the evaluation, our tool \textsf{KoAT} succeeded on $91\%$ of all examples, while the
main other related tools (\textsf{Absynth} and \textsf{eco-imp}) only inferred
finite bounds for $68\%$ resp.\ $77\%$ of the examples.
In future work,
it would be interesting to consider a modular combination of 
these tools (resp.\ of their underlying approaches).
%

\mypar{Acknowledgements} We thank Carsten Fuhs for discussions on
initial ideas.

%% file: appendix/additional_examples.tex
In the following, we present several examples to illustrate our approach.
We start with a small example with features both non-deterministic and probabilistic branching and both non-deterministic and probabilistic sampling.

\input{figures/preliminaries/pip_demo.tex}

\begin{example}[PIP with all Forms of Branching and Sampling]
	\label{PIP example}
	{\sl Consider
		the PIP in
		\cref{fig:preliminaries_pip_demo} with initial location $\initloc$ and the program variable $\PV = \{x\}$.
		Again, we assume $p=1$ and $\tau = \true$ if not stated explicitly.
		The program has three general transitions: $g_0 = \{ t_0 \}$, $g_1 = \{t_1, t_2 \}$, and $g_2 = \{ t_3 \}$.
		The first two represent a non-deterministic branching.
		If one chooses the general transition $g_1$, then the transitions $t_1$ and $t_2$ encode a probabilistic branching.
		The transition $t_0$ of $g_0$ updates $x$ to a non-deterministically chosen positive value $u$.
		The transition $t_3$ of $g_2$ updates the variable $x$ by sampling from the geometric distribution with parameter $\tfrac{1}{2}$.
		So in state $\state$, $x$ is updated by adding a value sampled from the probability distribution $d(s)\in \mathrm{Dist}(\ZZ)$ where for any $v \in \ZZ$ we have $d(\state)(v) = 0$ if $v \leq 0$, and $d(\state)(v) = (\tfrac{1}{2})^m$ if $v = m$ for some $m \geq 1$.}
\end{example}

The next example illustrates why one must not use \emph{expected} runtime bounds $\tbounde(\widetilde{g})$ instead of the \emph{non-probabilistic} bounds $\sum_{t=(\_,\_,\_,\_,\loc) \in \widetilde{g}} \tbound(t)$ on the entry transitions in \cref{theorem:exptimeboundsmeth}.
Essentially, the reason is that for two random variables $X_1$ and $X_2$, the expected value of $X_1 \cdot X_2$ is not equal to the product of the expected values of $X_1$ and $X_2$ if $X_1$ and $X_2$ are not independent.

\begin{example}[Expected Runtime Bounds for Entry Transitions]
	\label{ex:entry_transition_unsound}
	\begin{figure}[t]
		\centering
		\begin{tikzpicture}[shorten >=1pt,node distance=3.5cm,auto,main node/.style={circle,draw,font=\sffamily\Large\bfseries}]
			\node[main node] (l0) {$\ell_0$}; \node[main node] (l1) [right of = l0] {$\ell_1$}; \node[main node] (l2) [right of = l1,xshift=-1cm]
			{$\ell_2$}; \node[main node] (l3) [right of = l2,xshift=-.5cm] {$\ell_3$};

			\path[->] (l0) edge[align=left, bend left = 45] node [above] {$t_0 \in g_0$ \\
					$\eta(x)=1$\\
					$p = \tfrac{1}{3}$}
			(l1) (l0) edge[align=left] node [above] {$t_1 \in g_0$}
			node [below] {$\eta(x)=2$\\
					$p = \tfrac{1}{3}$}
			(l1) (l0) edge[align=left, bend right = 45] node [below] {$t_2 \in g_0$\\
					$\eta(x)=3$\\
					$p = \tfrac{1}{3}$}
			(l1)

			(l1) edge[align=left] node [below] {$t_3 \in g_1$}
			node [above] {$\eta (z) = x$}
			(l2)

			(l2) edge[align=left, bend right] node [above] {$t_4 \in g_2$}
			node [below] {$\eta(y)=z$\\
					$\tau = (x > 0)$}
			(l3)

			(l3) edge[loop right, align=left] node [right] {$t_5 \in g_3$}
			node [above,xshift=0.5cm,yshift=0.25cm] {$\eta(y) = y-1$\\
					$\tau = (y > 0)$}
			()

			(l3) edge[align=left, bend right] node [below] {$t_6 \in g_4$}
			node [above] {$\eta(x)=x-1$\\
					$\tau = (y \leq 0)$}
			(l2);
		\end{tikzpicture}
		\caption{PIP illustrating the incorrectness of expected time bounds on entry transitions.}
		\label{fig:incorrectness}
	\end{figure}
	{\sl Consider the PIP given in \cref{fig:incorrectness}, where we omitted updates of the form $\eta(x)=x$ to ease readability.
	We consider $\GTNI=\{g_3\}$ and the PLRF $\rank$ with $\rank(\loc_3)=y$.
	Thus, $\GTG=\{g_3\}$.
	We directly obtain $\sbounde(g_0,\loc_1,x) = \tfrac{1}{3}\cdot (1 + 2 + 3) = 2 = \sbounde(g_1,\loc_2,z) =\sbounde(g_2,\loc_3,y)$.
	So for the only entry location $\loc_3$ of $\GTNI$ (i.e., $\ENTRYLOC_{\GTNI} = \{ \loc_3 \}$), we have $\sbounde(g_2,\loc_3,y) = 2$.
	Moreover, $\ET_{\GTNI}(\loc_3) = \{ g_2 \}$ and let us assume that we have already inferred $\tbounde(g_2) = 2$ for the only entry transition $g_2$.
	(Note that this ``outer loop'' decreases $x$ deterministically by $1$ and the expected size of $x$ is $2$).
	If \cref{theorem:exptimeboundsmeth} allowed us to use \emph{expected}
	runtime bounds for the entry transitions on all entry locations, we would get
	\begin{align*}
		\tbounde' (g_3) & = \sum_{\loc \in \ENTRYLOC_{\GTNI},\widetilde{g} \in \ET_{\GTNI}(\loc)} \tbounde(\widetilde{g})\cdot \left(\exact{\overapprox{\rank(\loc)}}{\sbounde(\widetilde{g},\loc,\cdot)}\right) \\
		                & = \tbounde(g_2) \cdot \exact{\overapprox{\rank(\loc_3)}}{\sbounde(g_2,\loc_3,\cdot)}                                                                                                   \\
		                & = 2 \cdot 2                                                                                                                                                                            \\
		                & = 4.
	\end{align*}
	However, if $x$ is set to the value $v \in \{1,2,3\}$ in $g_0$, then $g_3$ is executed $v^2$ times. But this means that the expected number of executions of $g_3$ is \[\tfrac{1}{3} \cdot \left(1^2 + 2^2 + 3^2\right) = \tfrac{14}{3}
		= 4 + \tfrac{2}{3} > 4.
	\]
	Thus, using the \emph{expected} runtime bound $\tbounde(g_2) = 2$ does not lead to an upper bound on the expected runtime of $g_3$, i.e., \cref{theorem:exptimeboundsmeth}
	would become unsound.
	Intuitively, the reason is that the number of executions of $g_2 = \{ t_4 \}$ in a run and the maximal size of $x$ after executing $g_2$ are not independent.
	On the other hand, when using the non-probabilistic runtime bound $\tbound(t_4) = 3$ instead, which arises from the maximal value that $x$ can take, \cref{theorem:exptimeboundsmeth} infers the sound expected runtime bound
	\[
		\tbounde' (g_3) \; = \; \tbound(t_4) \, \cdot \, \exact{\overapprox{\rank(\loc_3)}}{\sbounde(g_2,\loc_3,\cdot)} \; = \; 3 \cdot 2 \; =\; 6 \;>\; 4 + \tfrac{2}{3}. \]}
\end{example}

The next example shows that when improving expected size bounds in \cref{thm:Inferring Expected Size Bounds for Trivial SCCs}, \emph{expected} size bounds must not be substituted into the local change bound $\chbounde(\alpha)$ if $\chbounde(\alpha)$ is not linear, i.e., not concave.

\input{figures/comp_size/concavity_demo.tex}

\begin{example}[Concavity for Expected Size Bounds of Trivial SCCs]
	\sl \label{ex:concavity_in_expected_size_bounds}
	Consider the PIP $\pip$ in \cref{fig:comp_size_concavity_demo}.
	$\pip$ increases $x$ by a random number between $1$ and $15$ (each with probability $\tfrac{1}{15}$) before computing its square.
	So for $\alpha = (g_0,\loc_1,x)$, we have $\sbounde (g_0,\loc_1,x) = x + \chbounde(\alpha) = x + 8$, whereas $\sbound (t_0,x) = x + 15$ is its non-probabilistic counterpart, which is indeed an upper bound on all possible values of $x$ after the update.
	Now we consider $\beta = (g_1,\loc_2,x)$.
	Depending on whether one uses the expected or the non-probabilistic size bound for $g_1$'s pre-transition $g_0$, we get $\incsize^\beta_{\mathbb{E}}(x) = \sbounde (g_0,\loc_1,x) = x+8$ and $\incsize^\beta(x) = \sbound (t_0,x) = x + 15$.
	Note that $\chbounde(\beta) = \overapprox{x^2 - x} = x^2 + x$, which is clearly not concave.
	If we were allowed to use the expected size bound for the pre-transition $g_0$, we would obtain \[\sbounde'(\beta) = \incsize^\beta_{\mathbb{E}}(x + \chbounde(\beta)) = \incsize^\beta_{\mathbb{E}}(x^2 + 2 \cdot x) = (x+8)^2 + 2 \cdot (x + 8) = x^2 + 18 \cdot x + 80.\]
	However, then $\sbounde'$ would not be a correct expected size bound.
	To see this, consider the initial state $\initstate$ with $\initstate(x) = 0$.
	Here, we have $\eval{\sbounde'(\beta)}{\initstate} = 80$, which is smaller than $\expv{\pip}{\scheduler}{\initstate} (\sizervar(\beta)) = \tfrac{1}{15} \cdot (1^2 + 2^2 + \ldots + 15^2) = 82 + \tfrac{2}{3}$.

	In contrast, applying \cref{thm:Inferring Expected Size Bounds for Trivial SCCs} yields
	\[\sbounde'(\beta) =\incsize^\beta_{\mathbb{E}}(x) +\incsize^\beta(\chbounde(\beta)) = x + 8 + (x + 15)^2 + (x + 15) = x^2 + 32 \cdot x + 248.
	\]
	This is a correct expected size bound and indeed, we now have $\eval{\sbounde'(\beta)}{\initstate}
		= 248\geq 82 + \tfrac{2}{3}= \expv{\pip}{\scheduler}{\initstate} (\sizervar(\beta))$.
\end{example}

%% file: figures/preliminaries/pip_demo.tex
\begin{figure}[t]
\begin{center}\begin{tikzpicture}[shorten >=1pt,node distance=3.5cm,auto,main node/.style={circle,draw,font=\sffamily\Large\bfseries}]
		\node[main node] (l0) {$\ell_0$}; \node[main node] (l1) [right of = l0] {$\ell_1$}; \node[main node] (l2) [above of = l0,yshift=-1.5cm] {$\ell_2$};

		\path[->] (l0) edge[align=left] node [right,yshift=0.2cm] {$t_0 \in g_0$}
		node [left,yshift=0.2cm] {$\eta(x)=u$ \\
				$\tau = (u > 0)$} (l2) (l0) edge[align=left, bend left = 20] node [above,xshift=0.5cm]
			{$p=\tfrac{3}{4}$\\
				$\eta(x)=x$} node [below] {$t_1\in g_1$} (l1) (l0) edge[align=left, bend right=20] node [below,xshift=0.5cm]
			{$p=\tfrac{1}{4}$\\
				$\eta(x)=x-1$} node [above] {$t_2\in g_1$} (l1) (l1) edge[align=left,loop right] node [below,xshift=1.2cm,yshift=-0.3cm]
			{$\eta(x)=\GEOMETRIC (\tfrac{1}{2})$\\
				$\tau=(x<10)$} node [right,yshift=0.2cm] {$t_3\in g_2$} ();
	\end{tikzpicture}\end{center}
	\caption{PIP with both forms of branching and both forms of sampling}\label{fig:preliminaries_pip_demo}
\end{figure}

%% file: figures/comp_size/concavity_demo.tex
\begin{figure}[t]
  \centering
  \begin{tikzpicture}[shorten >=1pt,node distance=3.5cm,auto,main node/.style={circle,draw,font=\sffamily\Large\bfseries}]
    \node[main node] (l0) {$\ell_0$};
    \node[main node] (l1) [right of = l0,xshift=0.25cm] {$\ell_1$};
    \node[main node] (l2) [right of = l1] {$\ell_2$};

    \path [->]
      (l0) edge
        node [above] {$\eta (x) = \UNIFORM (1,15)$}
        node [below] {$t_0 \in g_0$}
        (l1)

      (l1) edge
        node [above] {$\eta(x) = x^2$}
        node [below] {$t_1\in g_1$}
        (l2);
  \end{tikzpicture}
  \caption{PIP to illustrate the need for concave change bounds in \cref{thm:Inferring Expected Size Bounds for Trivial SCCs} when substituting expected size bounds}
  \label{fig:comp_size_concavity_demo}
\end{figure}
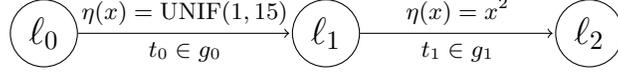

%% file: appendix/examples.tex
\label{app:examples}

As mentioned in \cref{sec:evaluation}, we evaluated our approach using our implementation in \textsf{KoAT} and the tools \textsf{Absynth} \cite{absynth} and \textsf{eco-imp} \cite{ecoimp}.

While we could not compare with the tool of \cite{dblp:conf/pldi/wang0gcqs19}, in the evaluation of \cite{dblp:conf/pldi/wang0gcqs19}
it obtained similar bounds as \textsf{Absynth}
when running it on several benchmarks from \cite{absynth}.
Thus, the comparison of \textsf{KoAT} with \textsf{Absynth} may also give an idea on the general relationship between \textsf{KoAT} and the tool of \cite{dblp:conf/pldi/wang0gcqs19} on the benchmarks from \cite{absynth}.

\begin{figure}
	\centering
	\begin{tabular}{l|c|c|c|}
		\textbf{Example}                                            & \textsf{KoAT}        & \textsf{Absynth}     & \textsf{eco-imp}     \\
		\hline \texttt{2drwalk}                                     & $\result{n}{7.69}$   & $\result{n}{1.55}$   & $\result{n}{0.05}$   \\
		\hline \texttt{bayesian}                                    & $\result{n}{6.90}$   & $\result{n}{0.22}$   & $\result{n}{0.01}$   \\
		\hline \texttt{ber}                                         & $\result{n}{0.19}$   & $\result{n}{0.02}$   & $\result{n}{0.01}$   \\
		\hline \texttt{bin}                                         & $\result{n}{0.20}$   & $\result{n}{0.23}$   & $\result{n}{0.02}$   \\
		\hline \texttt{C4B\_t09}                                    & $\result{n}{0.46}$   & $\result{n}{0.05}$   & $\result{n}{0.02}$   \\
		\hline \texttt{C4B\_t13}                                    & $\result{n}{1.13}$   & $\result{n}{0.03}$   & $\result{n}{0.02}$   \\
		\hline \texttt{C4B\_t15}                                    & $\result{n}{0.43}$   & $\result{n}{0.04}$   & $\result{}{0.34}$    \\
		\hline \texttt{C4B\_t19}                                    & $\result{n}{0.49}$   & $\result{n}{0.04}$   & $\result{n}{0.01}$   \\
		\hline \texttt{C4B\_t30}                                    & $\result{n}{27.49}$  & $\result{n}{0.03}$   & $\result{}{0.03}$    \\
		\hline \texttt{C4B\_t61}                                    & $\result{n}{0.26}$   & $\result{n}{0.03}$   & $\result{n}{0.02}$   \\
		\hline \texttt{complex}                                     & $\result{n^2}{1.52}$ & $\result{n^2}{2.49}$ & $\result{n^2}{0.05}$ \\
		\hline \texttt{condand}                                     & $\result{n}{0.27}$   & $\result{n}{0.01}$   & $\result{n}{0.01}$   \\
		\hline \texttt{cooling}                                     & $\result{n}{0.79}$   & $\result{n}{0.06}$   & $\result{n}{0.03}$   \\
		\hline \texttt{coupon}                                      & $\result{1}{1.66}$   & $\result{1}{0.08}$   & $\result{1}{0.01}$   \\
		\hline \texttt{cowboy\_duel}                                & $\result{1}{0.38}$   & $\result{1}{0.03}$   & $\result{1}{0.01}$   \\
		\hline \texttt{cowboy\_duel\_3way}                          & $\result{1}{2.27}$   & $\result{1}{0.16}$   & $\result{1}{0.01}$   \\
		\hline \texttt{fcall}                                       & $\result{n}{0.64}$   & $\result{n}{0.02}$   & $\result{n}{0.01}$   \\
		\hline \texttt{filling}                                     & $\result{n}{2.67}$   & $\result{n}{0.41}$   & $\result{n}{0.02}$   \\
		\hline \texttt{geo}                                         & $\result{1}{0.35}$   & $\result{1}{0.02}$   & $\result{1}{0.01}$   \\
		\hline \texttt{hyper}                                       & $\result{n}{0.19}$   & $\result{n}{0.03}$   & $\result{n}{0.01}$   \\
		\hline \texttt{linear01}                                    & $\result{n}{0.22}$   & $\result{n}{0.02}$   & $\result{n}{0.01}$   \\
		\hline \texttt{miner}                                       & $\result{n}{1.91}$   & $\result{n}{0.06}$   & $\result{n}{0.01}$   \\
		\hline \texttt{multirace} (\cref{fig:evaluation_multirace}) & $\result{n^2}{0.73}$ & $\result{n^2}{6.72}$ & $\result{n^2}{0.03}$ \\
		\hline \texttt{no\_loop}                                    & $\result{1}{0.33}$   & $\result{1}{0.02}$   & $\result{1}{0.00}$   \\
		\hline \texttt{pol04}                                       & $\result{}{0.87}$    & $\result{n^2}{0.04}$ & $\result{n^2}{0.02}$ \\
		\hline \texttt{pol05}                                       & $\result{}{0.63}$    & $\result{n^2}{1.11}$ & $\result{n^2}{0.02}$ \\
		\hline \texttt{pol06}                                       & $\result{}{2.95}$    & $\result{n^2}{5.72}$ & $\result{n^2}{0.06}$ \\
		\hline \texttt{pol07}                                       & $\result{n^2}{0.66}$ & $\result{n^2}{3.16}$ & $\result{n^2}{0.03}$ \\
		\hline \texttt{prdwalk}                                     & $\result{n}{0.33}$   & $\result{n}{0.05}$   & $\result{n}{0.01}$   \\
		\hline \texttt{prnes}                                       & $\result{n}{1.56}$   & $\result{n}{0.06}$   & $\result{n^2}{0.24}$ \\
		\hline \texttt{prseq}                                       & $\result{n}{0.46}$   & $\result{n}{0.05}$   & $\result{n}{0.02}$   \\
		\hline \texttt{prseq\_bin}                                  & $\result{n}{1.54}$   & $\result{n}{0.09}$   & $\result{n}{0.03}$   \\
		\hline \texttt{prspeed}                                     & $\result{n}{1.43}$   & $\result{n}{0.05}$   & $\result{n}{0.03}$   \\
		\hline \texttt{race}                                        & $\result{n}{0.31}$   & $\result{n}{0.19}$   & $\result{n}{0.05}$   \\
		\hline \texttt{rdbub}                                       & $\result{}{0.59}$    & $\result{n^2}{0.16}$ & $\result{n^2}{0.02}$ \\
		\hline \texttt{rdseql}                                      & $\result{n}{0.44}$   & $\result{n}{0.03}$   & $\result{n}{0.02}$   \\
		\hline \texttt{rdspeed}                                     & $\result{n}{0.58}$   & $\result{n}{0.04}$   & $\result{n}{0.03}$   \\
		\hline \texttt{rdwalk}                                      & $\result{n}{0.39}$   & $\result{n}{0.02}$   & $\result{n}{0.01}$   \\
		\hline \texttt{rfind\_lv}                                   & $\result{1}{0.26}$   & $\result{1}{0.02}$   & $\result{1}{0.01}$   \\
		\hline \texttt{rfind\_mc}                                   & $\result{n}{0.45}$   & $\result{n}{0.04}$   & $\result{n}{0.01}$   \\
		\hline \texttt{robot}                                       & $\result{n}{27.93}$  & $\result{n}{1.76}$   & $\result{n}{0.02}$   \\
		\hline \texttt{roulette}                                    & $\result{n}{1.21}$   & $\result{n}{0.65}$   & $\result{n}{0.09}$   \\
		\hline \texttt{sampling}                                    & $\result{n}{1.78}$   & $\result{n}{2.29}$   & $\result{n}{0.01}$   \\
		\hline \texttt{simple\_recursive}                           & $\result{n}{0.36}$   & $\result{n}{0.02}$   & $\result{n}{0.01}$   \\
		\hline \texttt{sprdwalk}                                    & $\result{n}{0.19}$   & $\result{n}{0.02}$   & $\result{n}{0.01}$   \\
		\hline \texttt{trader}                                      & $\result{}{1.29}$    & $\result{n^2}{0.92}$ & $\result{n^2}{0.87}$
	\end{tabular}
	\caption{Results on examples from \cite{absynth}}
	\label{fig:evaluation_results_absynth_linear}
\end{figure}

\begin{figure}
	\centering \bgroup \def\arraystretch{1.5}
	\begin{tabular}{l|c|c|c|}
		\textbf{Example}                                                               & \textsf{KoAT}
		                                                                               & \textsf{Absynth}                      & \textsf{eco-imp}                             \\
		\hline \cref{fig:preliminaries_pip_demo}                                       & $\result{n}{0.22}$                    & $\result{}{0.46}$     & $\result{}{0.44}$    \\
		\hline \cref{fig:incorrectness}                                                & $\result{1}{0.70}$                    & $\result{}{0.10}$     & $\result{1}{0.03}$   \\
		\hline \cref{fig:comp_size_concavity_demo}                                     & $\result{1}{0.17}$                    & $\result{1}{0.01}$    & $\result{1}{0.00}$   \\
		\hline \texttt{alain.c}                                                        & $\result{n^3}{35.47}$                 & \textsf{TO}           & $\result{}{0.33}$    \\
		\hline \texttt{C4B\_t132}                                                      & $\result{n}{1.03}$                    & $\result{}{0.10}$     & $\result{n}{0.02}$   \\
		\hline \texttt{complex2}                                                       & $\result{n^2}{1.21}$                  & $\result{}{85.78}$    & $\result{n^2}{0.05}$ \\
		\hline \texttt{cousot9}                                                        & $\result{n^2}{1.79}$                  & \textsf{TO}           & \textsf{TO}          \\
		\hline \texttt{ex\_paper1.c}                                                   & $\result{n^2}{28.01}$                 & \textsf{TO}           & $\result{n^2}{0.22}$ \\
		\hline \texttt{fib\_exp\_size} (\cref{fig:evaluation_fib_exp_size})            & $\textsf{EXP}: \; \SI{0.69}{\second}$ & $\result{}{15.32}$    & $\result{}{0.01}$    \\
		\hline \texttt{geo\_race}                                                      & $\result{n}{0.25}$                    & $\result{}{0.14}$     & $\result{}{23.29}$   \\
		\hline \texttt{knuth\_morris\_pratt.c}                                         & $\result{n}{20.79}$                   & $\result{}{0.40}$     & \textsf{TO}          \\
		\hline \texttt{leading} (\cref{fig:simple_quadratic,fig:evaluation_quad_size}) & $\result{n^2}{0.49}$                  & $\result{}{0.74}$     & $\result{}{30.64}$   \\
		\hline \texttt{leading.1}                                                      & $\result{n^2}{0.56}$                  & $\result{}{0.62}$     & $\result{}{1.25}$    \\
		\hline \texttt{multirace2}                                                     & $\result{n^2}{0.68}$                  & \textsf{TO}           & $\result{n^2}{0.03}$ \\
		\hline \texttt{neg\_init\_upd}                                                 & $\result{n}{0.40}$                    & $\result{}{0.09}$     & $\result{}{0.03}$    \\
		\hline \texttt{nested\_break}                                                  & $\result{n^2}{1.17}$                  & $\result{}{0.12}$     & $\result{}{0.14}$    \\
		\hline \texttt{nested\_rdwalk}                                                 & $\result{n}{0.55}$                    & $\result{n}{0.03}$    & $\result{n}{0.02}$   \\
		\hline \texttt{nested\_size}                                                   & $\result{n^5}{1.53}$                  & \textsf{TO}           & $\result{}{3.25}$    \\
		\hline \texttt{nondet\_countdown}
		(\cref{fig:nondet_countdown})                                                  & $\result{n}{0.41}$                    & $\result{}{0.31}$     & $\result{n}{0.02}$   \\
		\hline \texttt{prob\_loop}                                                     & $\result{}{0.72}$                     & $\result{n}{0.43}$    & $\result{}{0.03}$    \\
		\hline \texttt{prseq2}                                                         & $\result{n}{0.64}$                    & $\result{}{0.16}$     & $\result{n}{0.02}$   \\
		\hline \texttt{rank3.c}                                                        & $\result{n^2}{55.88}$                 & \textsf{TO}           & $\result{}{0.14}$    \\
		\hline \texttt{rdseql2}                                                        & $\result{n}{0.64}$                    & $\result{}{0.12}$     & $\result{n}{0.02}$   \\
		\hline \texttt{realheapsort}                                                   & $\result{n^2}{21.47}$                 & \textsf{TO}           & $\result{}{0.81}$    \\
		\hline \texttt{selectsort}                                                     & $\result{n^2}{4.46}$                  & $\result{n^2}{93.77}$ & $\result{n^2}{1.75}$ \\
		\hline \texttt{simple\_nested}                                                 & $\result{}{0.72}$                     & $\result{}{1.85}$     & $\result{n^2}{0.05}$ \\
		\hline \texttt{spctrm}                                                         & $\result{n^2}{21.47}$                 & \textsf{TO}           & \textsf{TO}          \\
		\hline \texttt{trunc\_selectsort}                                              & $\result{n^2}{3.99}$                  & \textsf{TO}           & $\result{n^2}{1.74}$ \\
		\hline \texttt{two\_arrays2}                                                   & $\result{n}{7.76}$                    & $\result{n}{0.61}$    & $\result{n}{0.03}$
	\end{tabular}
	\egroup \caption{Results on our examples}
	\label{fig:evaluation_results_ours}
\end{figure}

In our experiments, we used all examples from the evaluation of \textsf{Absynth} \cite{absynth} except the non-tail-recursive program \texttt{recursive}.
\cref{fig:evaluation_results_absynth_linear} gives the detailed results of the tools on these examples.
Moreover, we used an additional collection of typical PIPs and the results of the tools on these examples are shown in \cref{fig:evaluation_results_ours}.
Here, \texttt{leading} is our leading example from \cref{fig:simple_quadratic}, \cref{fig:preliminaries_pip_demo,fig:incorrectness,fig:comp_size_concavity_demo}
are the examples from \cref{app:additional_examples}, and \texttt{geo\_race}, \texttt{leading.1}, \texttt{neg\_init\_upd}, \texttt{nested\_break}, \texttt{nested\_size}, \texttt{nondet\_countdown}, \texttt{simple\_nested}, \texttt{prob\_loop}, \texttt{nested\_rdwalk}, and \texttt{fib\_exp\_size}
are additional PIPs illustrating different aspects, such as probabilistic loops, nesting of probabilistic, non-probabilistic, and non-deterministic behavior, and expected exponential bounds due to exponential non-probabilistic bounds.
The examples \texttt{complex2}, \texttt{C4B\_t132}, \texttt{multirace2}, \texttt{prseq2}, and \texttt{rdseql2} are based on corresponding examples from the collection of \cite{absynth}, but the original examples contain additional assumptions of the form \texttt{assume x > 0}, which were removed in these variants.
These examples demonstrate that in contrast to \textsf{KoAT} and \textsf{eco-imp}, \textsf{Absynth} relies on such additional assertions in order to find suitable base functions and hence, to compute a finite bound.

The examples \texttt{alain.c}, \texttt{cousot9}, \texttt{ex\_paper1.c}, \texttt{knuth\_morris\_pratt.c},\layout{\linebreak[4]
} \texttt{rank3.c}, \texttt{realheapsort}, \texttt{selectsort}, \texttt{spctrm}, and \texttt{two\_arrays2} are taken from the \emph{Termination Problem Database (TPDB)} \cite{tpdb} and enriched with probabilistic behavior.
Each of these examples contains between 10 and 59 general transitions.
The \emph{TPDB} contains a collection of typical integer programs whose complexity is analyzed by the tools participating in the annual \emph{Termination and Complexity Competition} \cite{termcomp}.
So as mentioned in \cref{sec:evaluation}, these examples illustrate that our approach is particularly suitable for larger programs with several loops or procedures, where only some of the transitions are probabilistic.

The example \texttt{trunc\_selectsort} is almost identical to \texttt{selectsort} except for one removed loop which results in \textsf{Absynth} not being able to derive an expected runtime bound.
This shows the robustness of our method.

\textsf{Absynth}'s heuristic for finding suitable base functions needs a bound on the degree of the overall time bound in advance.
Therefore, in our experiments, whenever \textsf{KoAT} found a polynomial bound of degree $k$, then \textsf{Absynth} was called with the maximum degree set to $k+1$.
Otherwise, when \textsf{KoAT} did not succeed in finding a polynomial bound, this parameter was set to 3.

All example programs were manually translated between the input format of \textsf{KoAT} and the format of \textsf{Absynth} and \textsf{eco-imp}.
\textsf{Absynth} programs can contain statements of the form ``\texttt{tick m}'' indicating that \texttt{m} ``resources'' have been consumed.
In our approach this is encoded by assigning a constant cost to general transitions.
This cost is then incorporated into the analysis by multiplying the expected runtime bound for each general transition with its assigned cost before adding up all these bounds to obtain a bound for the complete program.

In the original formulation of the examples \texttt{pol06} and \texttt{trader}
from \cite{absynth}, the goal was to compute the expected value of the variable $z$ (i.e., these \textsf{Absynth}
programs do not contain ``\texttt{tick}s'' to model the cost of the execution).
For that reason, when running \textsf{KoAT} on these examples we return the sum of all size bounds for $z$, i.e., $\sum_{\alpha = (\_, \_, z)}\sbounde(\alpha)$.

In \cref{fig:evaluation_quad_size,fig:evaluation_multirace,fig:evaluation_fib_exp_size}, we illustrate the different formats for probabilistic programs of \textsf{KoAT} and \textsf{Absynth}
with 3 example programs from our collection and we demonstrate our translation between these formats.
In these \textsf{KoAT} programs, the cost $m$ of a general transition is indicated by adding ``$\{ m \}$''.
Trivial updates, i.e., $\eta(x) = x$, as well as costs of zero are omitted.
To ease readability, we colored all transitions with non-zero cost in red.

\input{figures/evaluation/quad_size.tex}

\cref{fig:evaluation_quad_size} corresponds to the previously introduced \texttt{leading}
example from \cref{fig:simple_quadratic}.
When formulating it in \textsf{KoAT}'s input format, the general transition $g_0$ is required since our PIPs are not allowed to have a loop containing the initial location $f$.
As this restriction is not imposed for \textsf{Absynth}'s programs, the general transition $g_0$ is missing in the input format of \textsf{Absynth}.
To make the programs in the two formats equivalent, we assigned the cost 0 to $g_0$ in the \textsf{KoAT} format (and similarly, $g_3$ also got the cost 0).
The transition $g_1$ is encoded in the \textsf{Absynth} format as a \texttt{while} loop.
The cost of the general transition $g_1$ is expressed by the aforementioned \texttt{tick} statement, whereas the probabilistic choice between $t_1$ and $t_2$ is encoded as a probabilistic branching.
Note that if the guard $x > 0$ of transition $g_1$ is satisfied, then the program makes a non-deterministic choice between $g_1$ and $g_2$.
This is reflected in \textsf{Absynth}'s format by the ``\texttt{if random}'' construct.

\input{figures/evaluation/fib_exp_size.tex}

\cref{fig:evaluation_fib_exp_size} shows an example with expected exponential runtime.
In \textsf{KoAT}'s input format, the transition $t_1$ updates $x_1$ and $x_2$ simultaneously while referring to their previous values.
This cannot be expressed directly in the input format of \textsf{Absynth}.
To this end, we introduced a temporary variable \texttt{tmp} that stores the value of \texttt{x2}
before \texttt{x2} is updated.

\medskip

\input{figures/evaluation/multirace.tex}

The next example depicted in \cref{fig:evaluation_multirace}
is from the collection of \textsf{Absynth}.
Again, in the \textsf{KoAT} format, the general transition $g_0$ is required since the initial location $f$ cannot be part of any loop.
Here, the general transition $g_0$ is also used to encode the assumption ``\texttt{assume m>=0}''.
Line 10 of the \textsf{Absynth} version, respectively transition $t_3$ of the \textsf{KoAT} version show how to express probabilistic sampling in the two input formats.
By our semantics of updates with distribution functions, \texttt{h = h + unif(1,3)} must be encoded as $\eta (h) = \UNIFORM(1,3)$.

\medskip

The input format of \textsf{eco-imp} is almost identical to the one of \textsf{Absynth}. 
The only difference is that \textsf{eco-imp} neither supports a \texttt{break} statement nor recursion.
Thus, in the input for \textsf{eco-imp} we modeled \texttt{break} via an additional Boolean flag.
We only used recursion for the input of \textsf{Absynth} for very large examples.
Their translation to \textsf{Absynth} and \textsf{eco-imp} is explained below.

Moreover, \textsf{eco-imp} has a built-in timeout of 1 minute and does not output the degree of the polynomial bounds it computes. 
Hence, we applied two patches to the implementation of \textsf{eco-imp} to increase the timeout to 5 minutes and to make \textsf{eco-imp} output the degree of the bounds.

\medskip

\begin{figure}[h]
\centering
\begin{minipage}{0.4\textwidth}
      \begin{lstlisting}
        while x<10:
            temp = 1
            while temp > 0:
                x = x + temp
                temp = unif(0,1)
            tick 1
      \end{lstlisting}
    \end{minipage}
\caption{Modeling geometric distributions for \textsf{Absynth} and \textsf{eco-imp}}
\label{fig:geometric_translation}
\end{figure}

As already mentioned, \textsf{KoAT} supports Bernoulli, uniform, hypergeometric, and binomial distributions (which are finite distributions) as well as geometric distributions (which are infinite). 
While \textsf{Absynth} and \textsf{eco-imp} support the same finite distributions, they both do \emph{not} support geometric distributions.
Thus, we encoded geometric distributions via loops. 
For example, to realize $t_3$ in \cref{fig:preliminaries_pip_demo} where $x$ is updated by adding a value drawn from a geometric distribution with parameter $\tfrac{1}{2}$ we used the code presented in \cref{fig:geometric_translation}. 
This is the standard formulation of geometric distributions in probabilistic programs that only support finite distributions (see e.g., \cite{dblp:journals/jacm/kaminskikmo18}).

To keep the programs in \textsf{Absynth}'s input format readable for the larger examples from the TPDB, for \textsf{Absynth} we made use of tail-recursion.
We represent locations as functions and transitions as function calls.
This concerns the examples \texttt{alain.c}, \texttt{cousot9}, \texttt{ex\_paper1.c}, \texttt{knuth\_morris\_pratt.c}, \texttt{rank3.c}, \texttt{realheapsort}, \texttt{spctrm}, and \texttt{two\_arrays2}.
In this way, the input of \textsf{Absynth} is as close as possible to the PIP which \textsf{KoAT} receives as input.
For \textsf{eco-imp}, we used loops instead, since \textsf{eco-imp} does not support recursion.

The modularity of our approach is a significant advantage which allows us to handle larger examples.
However, a consequence of this modularity is that the obtained bounds often contain many nested subterms.
To ease the readability of these bounds, our implementation already performs a rudimentary set of arithmetic simplifications.
However, in particular bounds containing ``min'' and ``max'' terms that may be obtained from the underlying non-probabilistic analysis are difficult to simplify.
Therefore, \textsf{KoAT} also contains an option to use an SMT solver for the simplification of such expressions.
We did not use this option in our experiments since it increases computation time.
But it may sometimes result in asymptotically tighter bounds, since the asymptotic bounds are computed by a naive syntactic over-approximation of the obtained concrete bounds.
The non-probabilistic bounds of \textsf{KoAT} may also be exponential and they are also not necessarily weakly monotonic increasing.
Therefore, before using the non-probabilistic bounds computed by \textsf{KoAT} for our analysis, we apply $\overapprox{\cdot}$ to these bounds to make them weakly monotonically increasing (in case of ``$\min$'' and ``$\max$'' we apply $\overapprox{\cdot}$ to each argument).

\clearpage

%% file: figures/evaluation/quad_size.tex
\begin{figure}[h]
    \begin{minipage}{0.4\textwidth}
      \begin{lstlisting}
        def f():
            var x, y

            while x > 0:
                if random:
                    y = y + x
                    tick 1
                    prob(1,1):
                        x = x - 1
                    else:
                        x = x
                else:
                    break

            while y > 0:
                y = y - 1
                tick 1
      \end{lstlisting}
    \end{minipage}
  \begin{minipage}{0.6\textwidth}
      \begin{tikzpicture}[shorten >=1pt,node distance=3.5cm,auto,main node/.style={circle,draw,font=\sffamily\Large\bfseries}]
        \node[main node] (f) {$f$};
        \node[main node] (g) [right of = f] {$g$};
        \node[main node] (h) [below of = g] {$h$};

        \path[->] (f) edge[align=left] node [below] {$t_0 \in g_0$}
        (g)

        (g) edge[loop above, align=left, color=red]
          node [right,yshift=0.2cm,xshift=-0.5cm] {$t_1 \in g_1$}
          node [above,xshift=0.2cm,yshift=0.35cm] {$p=\tfrac{1}{2}$\\ $\eta(x)=x-1$\\ $\eta(y) = y+x$\\ $\tau = (x > 0)$\\ $\{1\}$}
        ()

        (g) edge[loop right, align=left, color=red]
          node [right,xshift=0.2cm,yshift=-0.4cm] {$t_2 \in g_1$\\$p=\tfrac{1}{2}$\\ $\eta (y) = y+x$
            \\ $\tau = (x > 0)$\\ $\{1\}$}
        ()

        (g) edge[align=left]
          node [right] {$t_3 \in g_2$}
        (h)

        (h) edge[loop left, align=left, color=red]
          node [left] {$t_4 \in g_3$}
          node [above,xshift=-0.1cm,yshift=0.25cm] {$\eta(y) = y-1$\\ $\tau = (y > 0)$\\ $\{1\}$}
        ();
      \end{tikzpicture}
  \end{minipage}
  \caption{The \texttt{leading} example from \cref{fig:simple_quadratic}}
  \label{fig:evaluation_quad_size}
\end{figure}

%% file: figures/evaluation/fib_exp_size.tex
\begin{figure}
    \begin{minipage}{0.4\textwidth}
      \begin{lstlisting}
        def f():
            var x1,x2,tmp,i,y
            x1 = 0
            x2 = 1
            i = 0

            while i < y:
                tmp = x2
                x2 = x1 + x2
                x1 = tmp
                i = i + 1

            while x2 > 0:
                prob(1,1):
                    x2 = x2 - 2
                else:
                    x2 = x2 + 1
                tick 1
      \end{lstlisting}
    \end{minipage}
  \begin{minipage}{0.6\textwidth}
      \begin{tikzpicture}[shorten >=1pt,node distance=3.5cm,auto,main node/.style={circle,draw,font=\sffamily\Large\bfseries}]
        \node[main node] (f) {$f$};
        \node[main node] (g) [right of = f] {$g$};
        \node[main node] (h) [below of = g] {$h$};

        \path[->] (f)
        edge[align=left]
          node [above] {$\eta(x_1) = 0$\\ $\eta (x_2) = 1$\\ $\eta (i) = 0$}
          node [below] {$t_0 \in g_0$}
        (g)

        (g) edge[loop above, align=left]
          node [right,yshift=0.1cm] {$t_1 \in g_1$}
          node [above,xshift=0.2cm,yshift=0.35cm] {$\eta(x_1)=x_2$\\ $\eta(x_2) = x_1+x_2$\\ $\eta(i) = i+1$\\  $\tau = (i < y)$}
        ()

        (g) edge[align=left]
          node [right] {$t_2 \in g_2$}
        (h)

        (h) edge[loop left, align=left, color=red]
          node [left] {$t_3 \in g_3$}
          node [above,xshift=-0.25cm,yshift=0.25cm] {$p=\tfrac{1}{2}$\\ $\eta(x_2) = x_2-2$\\ $\tau = (x_2 > 0)$\\ $\{1\}$}
        ()

        (h) edge[loop below, align=left, color=red]
          node [left,yshift=0.2cm,xshift=-0.1cm] {$t_4 \in g_3$}
          node [below,xshift=0.1cm] {$p=\tfrac{1}{2}$\\ $\eta(x_2) = x_2+1$\\ $\tau = (x_2 > 0)$\\ $\{1\}$}
        ();
      \end{tikzpicture}
  \end{minipage}
  \caption{The \texttt{fib\_exp\_size} example}
  \label{fig:evaluation_fib_exp_size}
\end{figure}

%% file: figures/evaluation/multirace.tex
\begin{figure}[h]
    \begin{minipage}{0.4\textwidth}
      \begin{lstlisting}
        def f():
            var h, t, m, z
            assume m >= 0

            while n > 0:
                h = 0
                t = m
                while h <= t:
                    prob(3,1):
                        h = h + unif(1,3)
                    t = t + 1
                    tick 1
                n = n - 1
      \end{lstlisting}
    \end{minipage}
  \begin{minipage}{0.6\textwidth}
      \begin{tikzpicture}[shorten >=1pt,node distance=3.5cm,auto,main node/.style={circle,draw,font=\sffamily\Large\bfseries}]
        \node[main node] (f) {$f$};
        \node[main node] (g) [below of = f] {$g$};
        \node[main node] (h) [right of = g] {$h$};

        \path[->] (f) edge[align=left]
          node[yshift=1cm,left] {$\tau = (m \geq 0)$}
          node[yshift=1cm,right] {$t_0 \in g_0$}
        (g)

        (g) edge[bend left=30,align=left]
          node [below] {$t_1 \in g_1$}
          node [above] {$\eta(h)=0$\\ $\eta(t) = m$\\ $\tau = (n > 0)$}
        (h)

        (h) edge[loop above, align=left, color=red]
          node [right,yshift=0.1cm] {$t_2 \in g_2$}
          node [above,xshift=0.2cm,yshift=0.35cm] {$p=\tfrac{1}{4}$\\ $\eta(t)=t+1$\\ $\tau = (h \leq t)$\\$\{1\}$}
        (h)

        (h) edge[loop below, align=left, color=red]
          node [right,yshift=-0.1cm] {$t_3 \in g_2$}
          node [below,yshift=-0.35cm,xshift=0.4cm] {$p=\tfrac{3}{4}$\\ $\eta(h)=\UNIFORM(1,3)$\\ $\eta(t) = t+1$\\ $\tau = (h \leq t)$\\$\{1\}$}
        (h)

        (h) edge[bend left=30,align=left]
          node [above] {$t_4 \in g_3$}
          node [below] {$\eta(n)=n-1$}
        (g);

      \end{tikzpicture}
  \end{minipage}
  \caption{The \texttt{multirace} example from \cite{absynth}}
  \label{fig:evaluation_multirace}
\end{figure}

%% file: appendix/probtheo.tex
\label{appendix:probtheo}
\label{app:prelim_probability}
\noindent
This section recapitulates the concepts from probability theory that we use in this work.
We first introduce measurable spaces and $\sigma$-fields in \cref{sigma-Fields} and present probability spaces in \cref{Probability Spaces}.
Integrals of arbitrary measures are defined in \cref{sec:integral_arbitrary_measure}.
Finally, we define random variables and expected values in \cref{Random Variables} and conditional expected values in \cref{Conditional Expected Values}.
\subsubsection{$\sigma$-Fields}
\label{sigma-Fields}
When setting up a probability space over some sample space $\Omega$, which can be any non-empty set, we have to distinguish the sets whose probabilities we want to be able to measure.
The collection of these measurable sets is called a \emph{$\sigma$-field}.
\begin{definition}[$\sigma$-Field, Filtration]
	Let $\Omega$ be an arbitrary set and $\CF\subseteq \powerset(\Omega)$.
	$\CF$ is called a \emph{$\sigma$-field over $\Omega$} if the following three conditions are satisfied.
	\begin{enumerate}
		\item $\Omega \in \CF$,
		\item $A \in \CF \Rightarrow \Omega \setminus A \in \CF$, i.e., $\CF$ is closed under taking the complement,
		\item $A_i \in \CF \Rightarrow \bigcup\limits_{i \in \NN}A_i \in \CF$, i.e., $\CF$ is closed under countable union.
	\end{enumerate}
	The pair $(\Omega, \CF)$ is called a \emph{measurable space}.
	The elements of $\CF$ are called \emph{measurable} sets.

	For a $\sigma$-field $\CF$, a \emph{filtration} is a sequence $(\CF_i)_{i \in \NN}$ of subsets of $\CF$, such that for all $i \in \NN$, $\CF_i$ is a $\sigma$-field over $\Omega$ with $\CF_i \subseteq \CF_{i+1}$.
\end{definition}

To get an intuition for filtrations, consider $\Omega$ to be the set $\RUNS$ of all (infinite) runs and $\CF$ to be the $\sigma$-field created by all cylinder sets.
But after the first $i$ evaluation steps, one only knows the first $i+1$ configurations of the run.
Hence, it also makes sense to consider the $\sigma$-fields $\CF_i$ created by all cylinder sets $\Pre{f}$ where the length of $f$ is $i+1$.
Indeed, $(\CF_i)_{i\in\NN}$ is a filtration.
In this case, we additionally have $\bigcup_{i \in \NN} \CF_i = \CF$, i.e., the filtration ``approximates $\CF$ from below''.

Let $\CF$ and $\mathcal{B}$ be $\sigma$-fields over $\Omega$.
Then $\CF \cap \mathcal{B}$ is a $\sigma$-field over $\Omega$.
Furthermore, if $(\CF_i)_{i \in I}$ is a family of $\sigma$-fields over $\Omega$ then so is $\bigcap\limits_{i \in I}\CF_i$.

For any set $\mathcal{E}$ of subsets of $\Omega$, let $\sigmagen{\mathcal{E}} \subseteq \CF$ consist of all elements that are contained in all $\sigma$-fields that are supersets of $\mathcal{E}$.
The mapping from $\mathcal{E}$ to $\sigmagen{\mathcal{E}}$ is also called \emph{$\sigma$-operator}.

\begin{definition}[Generating $\sigma$-Fields]
	Let $\mathcal{E}\subseteq \powerset(\Omega)$. Then the smallest $\sigma$-field over $\Omega$ containing $\mathcal{E}$ is \[\sigmagen{\mathcal{E}} \is \bigcap\limits_{\mathcal{E}\subseteq\CF,~\CF~\sigma\textnormal{-field over } \Omega} \CF.
	\]
\end{definition}
In our setting, we regard $\sigma$-fields $\CF \subseteq \powerset(\Omega)$ of the form $\sigmagen{\mathcal{E}}$, where $\mathcal{E}$ is a collection of cylinder sets.

It turns out that there is a special case in which the generated $\sigma$-field is easy to describe, namely in the case where a countable covering of the space $\Omega$ is given.

\begin{lemma}[Generating $\sigma$-Fields for Covering of $\Omega$]
	\label{lemma:sigma_field_countable_cover}
	If $\Omega = \biguplus\limits_{i=1}^{\infty} A_i$ for a sequence $A_i \in \powerset(\Omega)$ and $\mathcal{H} \is \sigmagen{\{A_i \mid i \in \NN\}}$ then \[\mathcal{H}=\underbrace{\left\{\biguplus\limits_{i \in J} A_i \mid J \subseteq \NN\right\}}_{ \eqqcolon \mathcal{E}}.
	\]
\end{lemma}

\begin{proof}
	Showing that $\mathcal{E}$ is a $\sigma$-field is enough to prove the desired result: it contains all the sets $A_i$ and every $\sigma$-field containing all the $A_i$ has to contain all their countable unions.
	\begin{item}
	      \item $\Omega \in \mathcal{E}$ by choosing $J=\NN$.
	      \item Let $J\subseteq \NN$.
	      Then $\Omega\setminus\left(\biguplus\limits_{i \in J}A_i\right) = \biguplus\limits_{i \in \NN \setminus J}A_i \in \mathcal{E}.$
	      \item Let $J_n \subseteq \NN$.
	      Then $\bigcup\limits_{n \in \NN} \biguplus\limits_{i \in J_n} A_i = \biguplus\limits_{i \in \bigcup\limits_{n \in \NN}J_n} A_i\in \mathcal{E}.$
	\end{item}
\end{proof}

\noindent
This special type of a $\sigma$-field can be used to describe the elements of the $\sigma$-fields $\CF_{i,j}$ in \cref{appendix:details_lprf}.

The following definition of atoms (similar to the definition of atoms found in the literature on measure theory, e.g.,
\cite[Lemma 5.5.8]{cylindrical}) and its related lemmas will help us to simplify the proofs later on.

\begin{definition}[Atoms of $\sigma$-Fields]
	Let $\CF$ be a $\sigma$-field and $A \in \CF$.
	If $X\in \CF$ and $X\subseteq A$ implies $X = \varnothing$ or $X=A$, then we call $A$ an \emph{atom}.
\end{definition}

This means that atoms are the \emph{smallest} $\CF$-measurable sets.
In general, this does \emph{not} imply that they also generate $\CF$.
However, this holds in the case of \cref{lemma:sigma_field_countable_cover}.

\begin{lemma}[Atoms for Covering of $\Omega$]
	\label{lemma:countable_cover_atoms}
	Let $\mathcal{H}$ be as in \cref{lemma:sigma_field_countable_cover}.
	The atoms of $\mathcal{H}$ are exactly the sets $A_i$, $i \in \NN$.
\end{lemma}

\begin{proof}

	Fix some $j_0 \in \NN$ and let $X \in \mathcal{H}$, i.e., $X$ is a disjoint union of the $A_i$ by \cref{lemma:sigma_field_countable_cover}.
	Then $X\subseteq A_{j_0}$ iff $X\subsetneq A_{j_0}$ or $X = A_{j_0}$.
	But $X\subsetneq A_{j_0}$ implies that $X$ is the empty union $\varnothing$ due to the disjointness of the sets $A_i$.
	So, $A_{j_0}$ is indeed an atom.
\end{proof}
\begin{definition}[Borel-Field]
	If $\Omega = \realclosure$ we use its $\sigma$-field $\mathfrak{B}=\mathfrak{B}\left(\realclosure\right)$, the \emph{Borel-field} with \[\mathfrak{B}\left(\realclosure\right) \is \sigmagen{\{(a,b)\subseteq \realclosure \mid a < b \in \realclosure\}}.
	\]
\end{definition}

\noindent
In this work, we use the concept of measurable maps (or ``measurable mappings'').
Measurable maps are the structure-preserving maps between measurable spaces.
They are defined as follows.

\begin{definition}[Measurable Map]
	\label{defQQQmeasurable_map}
	Let $(\Omega_1,\CF_1)$ and $(\Omega_2,\CF_2)$ be measurable spaces. A function $f\colon \Omega_1 \to \Omega_2$ is called an \emph{$\CF_1$-$\CF_2$ measurable map} or just \emph{measurable} if for all $A \in \CF_2$ \[f^{-1}(A) = \{\run \in \Omega_1 \mid f(\run) \in A\} \in \CF_1.
	\]
	$\sigmagen{f} \is f^{-1}(\CF_2) \is \{f^{-1}(A) \mid A \in \CF_2\}$ is the smallest $\sigma$-field $\CF$ such that $f$ is $\CF_1$-$\CF_2$ measurable.
	Similarly, $\sigmagen{f_0,\ldots,f_n} = \sigmagen{f_0^{-1}(\CF_2) \cup \ldots \cup f_n^{-1}(\CF_2)}$.
	This will become important when talking about random variables and conditional expectations.
\end{definition}

\bigskip

\subsubsection{Probability Spaces}
\label{Probability Spaces}
So far we have only introduced the concept of measurable spaces.
Intuitively, a measurable space provides the structure for defining a \emph{measure}.
\emph{Probability spaces} are measurable spaces combined with a certain measure where the measure of the sample space is $1$.

\begin{definition}[Probability Measure, Probability Space]
	Let $(\Omega, \CF)$ be a measurable space.
	A map $\mu \colon \CF \to \RR_{\geq 0}$ is called a \emph{measure} if
	\begin{enumerate}
		\item $\mu(\emptyset)=0$
		\item $\mu(\biguplus\limits_{i\geq 0} A_i) = \sum\limits_{i\geq 0}\mu(A_i)$.
	\end{enumerate}
	A \emph{probability measure} is a measure $\mathbb{P} \colon \CF \to \RR_{\geq 0}$ with $\mathbb{P}(\Omega)=1$.
	This implies that $\mathbb{P}(A)\in [0,1]$ for every $A \in \CF$.
	If $\mathbb{P}$ is a probability measure, then $(\Omega, \CF, \mathbb{P})$ is called a \emph{probability space}.
	In this setting, a set in $\CF$ is called an \emph{event}.
\end{definition}

\noindent
Here, the intuition for a probability measure $\mathbb{P}$ is that for any set $A \in \CF$, $\mathbb{P}(A)$ is the probability that an element chosen from $\Omega$ is contained in $A$.
Since there is no probability measure on $\emptyset$, we always assume $\Omega \neq \emptyset$.

In the special case where $\Omega$ is countable and $\CF$ is the powerset of $\Omega$, then a probability measure is uniquely determined by the probabilities $\IP{\{a\}}$ for $a \in \Omega$.
Therefore, we call the mapping $pr \colon \Omega \to \RR_{\geq 0}, a \mapsto \IP{\{a\}}$ the \emph{probability mass function} of $\mathbb{P}$.
On the other hand,
if $\Omega$ is countable and $pr \colon \Omega \to \RR_{\geq 0}$ is a mapping such that $\sum_{a \in \Omega} pr(a)=1$, then $pr$ uniquely determines a probability measure on the powerset of $\Omega$.

For a probability space, we can define when a property holds \emph{almost surely}.

\begin{definition}[Almost Sure Properties]
	Let $(\Omega, \CF, \mathbb{P})$ be a probability space and $\alpha$ some property, e.g., a logical formula.
	Let $A_\alpha \subseteq \Omega$ be the set of all elements from $\Omega$ that satisfy the property $\alpha$.
	If $A_\alpha \in \CF$ (i.e., it is measurable) and $\IP{A_\alpha}=1$ then $\alpha$ holds \emph{almost surely}.
\end{definition}

\noindent
If a property $\alpha$ holds almost surely it does not need to hold for all elements of $\Omega$.
However, the measure $\mathbb{P}$ cannot distinguish $A_\alpha=\Omega$ and $A_\alpha\neq \Omega$ if $\IP{A_\alpha}=1$, so w.r.t.\ $\mathbb{P}$, properties that hold almost surely can be considered as holding globally.

\subsubsection{Integrals of Arbitrary Measures}
\label{sec:integral_arbitrary_measure}
We now introduce a notion of an integral with respect to an arbitrary measure.
Therefore, we fix a measurable space $(\Omega, \CF)$ and a measure $\mu$.
The objective is to define a ``mean'' of a measurable function $f$.
The basic idea is to partition the image of $f$ into sets $A_i$ on which $f$ has a constant value $\alpha_i$.
Then we compute the weighted average of the $\alpha_i$, where the weights are the measures of the $A_i$.
This definition is fine if $f$ takes only finitely many values (these functions are called \emph{elementary}).
If $f$ takes infinitely (countable or even uncountable) many values, we have to approximate $f$ step by step by such functions with finite image.
This yields a limit process.
In this work we will only consider the cases where $\mu=\mathbb{P}$ is a probability measure.

\begin{definition}[Elementary Function, \protect{\cite[Def.~2.2.1]{bauer71measure}}]
	An \emph{elementary function} is a non-negative measurable function $f\colon \Omega \to \realclosure$ that takes only finitely many finite values, i.e., there exist $A_1, \dots, A_n \in \CF$ and $\alpha_1, \dots, \alpha_n \geq 0$ such that \[f= \sum \limits_{i=1}^n \alpha_i \cdot \ind_{A_i}.
	\]
\end{definition}

\noindent
Given an elementary function, the decomposition into a linear combination of indicator functions is not unique.
But to define an integral we have to guarantee that its value does not depend on the chosen decomposition.
Fortunately, this can be proved:
\begin{lemma}[Decomposition of Elementary Functions \protect{\cite[Lemma 2.2.2]{bauer71measure}}]
	\label{lemma:elementary_functions_decompositions}
	Let $f=\sum \limits_{i=1}^n \alpha_i \cdot \ind_{A_i}=\sum \limits_{j=1}^n \beta_j \cdot \ind_{B_j}$ be an elementary function, where $\alpha_i, \beta_j \geq 0$, $A_i, B_j \in \CF$ for all $i$ and $j$. Then \[\sum \limits_{i=1}^n \alpha_i \cdot \mu\left(A_i\right)=\sum \limits_{j=1}^n \beta_j \cdot \mu\left(B_j\right).
	\]
\end{lemma}

\begin{definition}[Integral of Elementary Function, \protect{\cite[Def.~2.2.3]{bauer71measure}}]
	\label{def:integral_elementary}
	Let $f= \sum \limits_{i=1}^n \alpha_i \cdot \ind_{A_i}$ be an elementary function. Then we define its integral w.r.t.~$\mu$ as \[\int_{\Omega} f\, d\mu \is \int f \, d\mu \is \sum \limits_{i=1}^n \alpha_i \cdot \mu\left(A_i\right).
	\]
	The well-definedness is justified by \textnormal{\cref{lemma:elementary_functions_decompositions}} as it shows that the chosen decomposition is independent of the elementary function.
\end{definition}

\noindent
However, the measurable functions we use are not elementary.
They take arbitrary (countably or even uncountably) many values.
So we have to generalize \cref{def:integral_elementary}.
It can be shown that any non-negative measurable function is the limit of a monotonic increasing sequence of elementary functions.

\begin{theorem}[Representation by Elementary Functions]
	Let $f\colon \Omega \to \realclosure$ be a non-negative measurable function. Then there exists a monotonic sequence $f_0 \leq f_1 \leq \dots$ of elementary functions such that \[f = \sup_{n \in \NN} f_n = \lim \limits_{n \to \infty} f_n.
	\]
	Furthermore, for any two such sequences $(f_n)_{n \in \NN}$ and $(g_n)_{n \in \NN}$ we have $\sup\limits_{n \in \NN} \int f_n \, d \mu=\sup\limits_{n \in \NN} \int g_n \, d \mu$.
\end{theorem}

\begin{proof}
	See \cite[Cor.
		2.3.2., Thm.
		2.3.6]{bauer71measure}.
\end{proof}

\noindent
This theorem justifies the following definition of an integral for an arbitrary non-negative function.

\begin{definition}[Integral of Arbitrary Functions]
	Let $f\colon \Omega \to \realclosure$ be a non-negative measurable function and $(f_n)_{n \in \NN}$ a monotonic sequence of elementary functions such that~$f= \sup_{n \in \NN} f_n$. Then we define the integral of $f$ w.r.t.~$\mu$ by \[\int_{\Omega} f \, d\mu \is \int f \, d\mu \is \sup_{n \in \NN} \int f_n \, d\mu.
	\]
\end{definition}

\noindent
Before we state the properties of the integral used in this work, we will define the integral on a measurable subset of $\Omega$.

\begin{definition}[Integral on Measurable Subset]
	Let $f\colon \Omega \to \realclosure$ be a non-negative measurable function and $A\in \CF$. Then $f \cdot \ind_{A}$ is non-negative and measurable, and we define \[\int_{A} f \, d\mu \is \int f \cdot \ind_{A} \, d\mu.
	\]
\end{definition}

\begin{lemma}[Properties of the Integral]
	\label{lemma:properties_of_integral}
	Let $a,b \geq 0$, $f,g \colon \Omega \to \realclosure$ be measurable functions, and $A,A_i \in \CF$.
	Then
	\begin{align*}
		\int_{A} a\cdot f + b \cdot g \, d \mu             & = a\cdot\int_{A} f \, d \mu + b \cdot \int_{A} g \, d \mu. & \text{(Linearity)}  \\
		\int_{A} 1 \, d \mu                                & = \mu(A).                                                  & \text{(Measure)}    \\
		\int_{\biguplus\limits_{i \in \NN} A_i} f \, d \mu & = \sum\limits_{i \in \NN}\int_{ A_i} f \, d \mu.           & \text{(Additivity)} \\
	\end{align*}
\end{lemma}
\begin{proof}
	See \cite[(2.2.4.), (2.3.6.), (2.3.7.), Cor.\ 2.3.5.]{bauer71measure}.
\end{proof}

\subsubsection{Random Variables}
\label{Random Variables}
A \emph{random variable} $X$ maps elements of one set $\Omega$ to another set $\Omega'$.
If $\mathbb{P}$ is a probability measure for $\Omega$ (i.e., for $A\subseteq \Omega$, $\mathbb{P}(A)$ is the probability that an element chosen from $\Omega$ is contained in $A$), then one obtains a corresponding probability measure $\mathbb{P}^X$ for $\Omega'$.
For $A' \subseteq \Omega'$, $\mathbb{P}^X(A')$ is the probability that an element chosen from $\Omega$ is mapped by $X$ to an element contained in $A'$.
In other words, instead of regarding the probabilities for choosing elements from $A$, one now regards the probabilities for the values of the random variable $X$.

\begin{definition}[Random Variable]
	\label{defQQQrandom_variable}
	Let $(\Omega, \CF, \mathbb{P})$ be a probability space.
	An $\CF$-$\mathfrak{B}(\realclosure)$ measurable map $X \colon \Omega \to \realclosure$ is a \emph{random variable}.
	Instead of saying ``$\CF$-$\mathfrak{B}\left(\realclosure\right)$ measurable'' we simply use the notion ``$\CF$-measurable''.
	It is called \emph{discrete random variable}, if its image is a \emph{countable} set.

	$\mathbb{P}^X \colon \CF' \to [0,1], A' \mapsto \mathbb{P}(X^{-1}(A'))$ is the probability measure induced by $X$ on $(\realclosure,\mathfrak{B}(\realclosure))$.
	Instead of $\mathbb{P}^X(A)$ the notation $\mathbb{P}(X \in A)$ is common.
	If $A = \{i\}$ is a singleton set, we also write $\mathbb{P}(X = i)$ instead of $\mathbb{P}^X(\{i\})$.
\end{definition}

\begin{definition}[Expected Value]
	Let $X \colon \Omega \to \realclosure$ be a random variable.
	Then its \emph{expected value} is $\expec{X} \is \int X d\mathbb{P}$.
\end{definition}

\begin{lemma}[Expected Value as Sum]
	\label{alternative expected value}
	If $X$ is a discrete random variable we have $\expec{X}=\sum_{r \in \realclosure} r \cdot \IP{X=r}$.
	Note that this series has only countably many non-zero non-negative summands.
	Hence, it either converges or it diverges to infinity.
\end{lemma}

\begin{proof}
	Let $X(\Omega)=\{r_1,r_2,\dots\}$.
	Then $\Omega = \biguplus\limits_{i \in \NN} X^{-1}(\{r_i\})$.
	Hence
	\begin{align*}
		\expec{X} & =\int X d\mathbb{P}                                                                                                                  \\
		          & =\int_{\biguplus\limits_{i \in \NN} X^{-1}(\{r_i\})} X \, d\mathbb{P}                                                                \\
		          & =\sum\limits_{i \in \NN}\int_{X^{-1}(\{r_i\})} X \, d\mathbb{P}          & \tag{by \cref{lemma:properties_of_integral} (Additivity)} \\
		          & = \sum\limits_{i \in \NN}\int_{ X^{-1}(\{r_i\})} r_i \, d\mathbb{P}      & \tag{as $X$ is constant on $X^{-1}(\{r_i\})$}             \\
		          & =\sum\limits_{i \in \NN}r_i\cdot \int_{X^{-1}(\{r_i\})} 1 \, d\mathbb{P} & \tag{by \cref{lemma:properties_of_integral} (Linearity)}  \\
		          & =\sum\limits_{i \in \NN}r_i\cdot \IP{X^{-1}(\{r_i\})}                    & \tag{by \cref{lemma:properties_of_integral} (Measure)}    \\
		          & =\sum\limits_{i \in \NN}r_i\cdot \IP{X = r_i}                                                                                        \\
		          & =\sum\limits_{r \in \realclosure} r \cdot \IP{X = r}.
	\end{align*}
\end{proof}

\begin{definition}[Stochastic Process, Adaptedness to Filtration]
	A family $(X_i)_{i \in \NN}$ of random variables for a probability space $(\Omega, \CF, \mathbb{P})$ is called a \emph{stochastic process}.
	If $(\CF_i)_{i \in \NN}$ is a filtration of $\CF$ and $X_i$ is $\CF_i$-measurable for all $i \in \NN$, then we say that the stochastic process $(X_i)_{i \in \NN}$ is \emph{adapted}
	to the filtration $(\CF_i)_{i \in \NN}$.
\end{definition}

As an example, consider again $\Omega$ to be the set $\RUNS$ of all (infinite) runs, $\CF$ to be the $\sigma$-field created by all cylinder sets, and $\CF_i$ to be the $\sigma$-field created by all cylinder sets $\Pre{f}$ where the length of $f$ is $i+1$.
Let $X_i:\RUNS\to \realclosure$ return the value of a certain program variable in the $i$-th configuration of the run.
To determine this value, one indeed only has to regard the first $i+1$ configurations, i.e., $X_i$ is $\CF_i$-measurable.
Hence, the stochastic process $(X_i)_{i \in \NN}$ is adapted to the filtration $(\CF_i)_{i \in \NN}$.

\subsubsection{Conditional Expected Values w.r.t.\ $\sigma$-Fields}
\label{Conditional Expected Values}
We introduce the notion of \emph{conditional expected value} w.r.t.~a sub-$\sigma$-field on a fixed probability space $(\Omega, \CF,\mathbb{P})$.
The idea is that given a random variable $X$ and a subfield $\CG \subseteq \CF$ we would like to approximate $X$ by another $\CG$-measurable random variable w.r.t.~expectation.
Intuitively, this means that we want to construct a (possibly infinite) non-negative linear combination of the functions $\ind_{G}, G \in \CG$, in such a way that restricted to a set $G \in \CG$, the random variable $X$ and this linear combination have the same average value w.r.t.~$\mathbb{P}$.

\begin{definition}[Conditional Expected Value, \protect{\cite[Def.~10.1.2]{bauer71measure}}]
	\label{defQQQconditional_expectation}
	Let $X \colon \Omega \to \realclosure$ be a random variable and $\CG$ be a sub-$\sigma$-field of $\CF$.
	A random variable $Y \colon \Omega \to \realclosure$ is called a \emph{conditional expected value} of $X$ w.r.t.~$\CG$ if
	\begin{enumerate}
		\item $Y$ is $\CG$-measurable
		\item $\int_G Y \, d \mathbb{P}=\expec{Y \cdot \ind_{G}}= \expec{X \cdot \ind_{G}}=\int_G X \, d \mathbb{P}$ for every $G \in \CG$.
	\end{enumerate}
	If two such random variables $Y$ and $Y'$ exist, the just stated properties already ensure that $\mathbb{P}(Y=Y')=1$.
	Therefore a conditional expected value is almost surely unique which justifies the notation $\expec{X \mid \CG} \is Y$.
\end{definition}

\noindent
If the sub-$\sigma$-field has the special structure described in \cref{lemma:sigma_field_countable_cover} then the just stated property just needs to be checked on the generators.

\begin{lemma}[Conditional Expected Value]
	\label{app_lemma:conditional_expectation_countable_cover}
	Let $X$ be a random variable. If $\Omega = \biguplus\limits_{i=1}^{\infty} A_i$ for a sequence $A_i \in \CF$ and $\CG \is \sigmagen{\{A_i \mid i \in \NN\}}$ then a $\CG$-measurable function $Y$ is a conditional expected value of $X$ iff \[\expec{X\cdot\ind_{A_i}}=\expec{Y\cdot\ind_{A_i}}, \text{ for all } i \in \NN.
	\]
\end{lemma}

\begin{proof}
	By definition, it is left to show that $\expec{X\cdot\ind_{A_i}}=\expec{Y\cdot\ind_{A_i}}$ for all $i \in \NN$ implies $\expec{X\cdot\ind_{G}}=\expec{Y\cdot\ind_{G}}$ for any $G \in \CG$.
	But due to \cref{lemma:sigma_field_countable_cover} it is enough to show this for any disjoint union $\biguplus\limits_{i \in J}A_i$.
	We have $\ind_{\biguplus\limits_{i \in J}A_i}=\sum\limits_{i \in J}\ind_{A_i}$ and this series always converges point-wise as at most one of the summands is non-zero for any $\alpha \in \Omega$.
	Hence we have
	\begin{align*}
		\expec{X\cdot\ind_{\biguplus\limits_{i \in J}A_i}} & =\expec{X\cdot\sum\limits_{i \in J}\ind_{A_i}}                                                             \\
		                                                   & \overset{\textnormal{\cref{lemma:properties_of_integral}}}{=}\sum\limits_{i \in J}\expec{X\cdot\ind_{A_i}} \\
		                                                   & =\sum\limits_{i \in J}\expec{Y\cdot\ind_{A_i}}=\expec{Y\cdot\ind_{\biguplus\limits_{i \in J}A_i}}.
	\end{align*}
\end{proof}

The conditional expected value also has a very special behavior on atoms.

\begin{lemma}[Conditional Expectation is Constant on Atoms]
	\label{lemma:rsm_condexpconstatoms}
	Let $(\Omega, \CF,\layout{\linebreak} \pmeasure)$ be a probability space, $X$ is a $\CF$-measurable function,
	and $\CG \subseteq \CF$ a sub-$\sigma$-field of $\CF$.
	If $A \in\CG$ is an atom, then $\expvsign (X \mid\CG)$ is constant on $A$.
	Moreover, if $\pmeasure(A) > 0$, then for all $a \in A$ we have

	\[\expvsign (X \mid\CG)(a) = \frac {\expvsign (\ind_A \cdot X)}{\pmeasure (A)}
	\]
\end{lemma}

\begin{proof}
	First of all, let $a \in A$ and consider the set
	\[
		M_a\is A \cap \expvsign (X \mid\CG)^{-1}(\expvsign (X \mid\CG)(a)). \]
	Then $a \in M_a \subseteq A$ and as $\expvsign (X \mid\CG)$ is $\CG$-measurable by \cref{defQQQconditional_expectation} we also have $M_a \in\CG$.
	Since $A$ is an atom and $\emptyset \neq M_a \subseteq A$ we already have $M_a =A$.
	So, $\expvsign (X \mid\CG)$ is constant on $\CG$.
	Finally we have by \cref{defQQQconditional_expectation}
	\[
		\expvsign(\ind_A \cdot X) = \expvsign(\ind_A \cdot \expvsign (X \mid\CG)) = \expvsign(\ind_A \cdot const) = const \cdot \pmeasure(A),
	\]
	i.e., if $\pmeasure(A)>0$ we have
	\[
		const = \frac{\expvsign(\ind_A \cdot X)}{\pmeasure(A)}.
	\]
\end{proof}

It turns out that in our setting a conditional expectation always exists as it is almost surely finite.
\begin{theorem}
	[Existence of Conditional Expected Values \protect{\cite[Prop.\ 3.1.]{lexrsm}}]Let \label{thm:existence_cond_exp_nonnegative}
	$X \colon \Omega \to \realclosure$ be a random variable such that~$\IP{X=\infty}=0$ and let $\CG$ be a sub-$\sigma$-field of $\CF$.
	Then $\expec{X \mid \CG}$ exists.
\end{theorem}

This theorem helps us to find bounds on the conditional expected value if we are unable to determine it exactly.

\begin{lemma}
	\label{lem:upper_bound_cond_exp}
	Let $X$ be a random variable with $\IP{X=\infty}=0$ and $\Omega = \biguplus\limits_{i=1}^{\infty} A_i$ for a sequence $A_i \in \CF$, $\CG \is \sigmagen{\{A_i \mid i \in \NN\}}$, and $Y$ a $\CG$-measurable function. We have $\expec{X \mid \CG} \leq Y$ iff \[\expec{X\cdot\ind_{A_i}}\leq\expec{Y\cdot\ind_{A_i}}, \text{ for all } i \in \NN.
	\]
\end{lemma}

\begin{proof}
	We prove the two directions separately.
	\begin{itemize}
		\item[``$\Rightarrow$''] Let $i \in \NN$. Then we have by definition of the conditional expectation \[\expec{X \cdot \ind_{A_i}} = \expec{\expec{X \mid \CG} \cdot \ind_{A_i}} \leq \expec{Y \cdot \ind_{A_i}},
		      \]
		      by monotonicity of the integral (\cref{lemma:properties_of_integral}).
		\item[``$\Leftarrow$''] By \cref{thm:existence_cond_exp_nonnegative} the conditional expected value of $X$ exists and is itself a non-negative $\CG$-measurable random variable.
		      Consider the random variable $Z=\max(\expec{X \mid \CG}-Y,0):\Omega \to \realclosure$.
		      It is a result of measure theory (see e.g., \cite{bauer71measure}) that $Z$ is a $\CG$-measurable random variable.
		      So consider the set $M=Z^{-1}((0,\infty))\in \CG$.
		      Then we know that $M$ is a disjoint union of some of the $A_i$, so w.l.o.g.~let us assume that $A_{i_0}\subseteq M$.
		      Again, a deep result from measure theory shows that $\expec{Z\cdot \ind_{A_{i_0}}}>0$ (see e.g., \cite{bauer71measure}) if $\IP{A_{i_0}}>0$.
		      Then we have
		      \[
			      \begin{array}{rcl}
				      0 & <     & \expec{Z\cdot \ind_{A_{i_0}}}                                                   \\
				        & {}={} & \expec{(\expec{X \mid \CG}-Y)\cdot \ind_{A_{i_0}}}                              \\
				        & {}={} & \expec{\expec{X \mid \CG} \cdot \ind_{A_{i_0}}}-\expec{Y \cdot \ind_{A_{i_0}}},
			      \end{array}
		      \]
		      i.e., $\expec{Y \cdot \ind_{A_{i_0}}} < \expec{\expec{X \mid \CG} \cdot \ind_{A_{i_0}}}$, a contradiction.
		      So, we must have $\IP{A_{i_0}}=0$.
		      As $i_0$ was chosen arbitrarily we almost surely have $\expec{X \mid \CG} \leq Y$.
	\end{itemize}
\end{proof}

\noindent
Recall that $\expec{X \mid \CG}$ is a random variable that is like $X$, but for those elements that are not distinguishable in the sub-$\sigma$-field $\CG$, it ``distributes the value of $X$ equally''.
This statement is formulated by the following lemma.
\begin{lemma}[Expected Value Does Not Change When Regarding Conditional Expected Values]
	\label{lemma:expected_value_does_not_change}
	Let $X$ be a random variable on $(\Omega, \CF, \mathbb{P})$ and let $\CG$ be a sub-$\sigma$-field of $\CF$. Then \[\expec{X}=\expec{\expec{X \mid \CG}}.
	\]
\end{lemma}

\begin{proof}
	\[
		\begin{array}{rcll}
			\expec{\expec{X \mid \CG}} & {}={} & \expec{\expec{X \mid \CG} \cdot \ind_{\Omega}} & \text{as $\ind_{\Omega}\equiv 1$}                                   \\
			                           & {}={} & \expec{X \cdot \ind_{\Omega}}                  & \text{by \cref{defQQQconditional_expectation}, as $\Omega \in \CG$} \\
			                           & {}={} & \expec{X}                                      & \text{as $\ind_{\Omega}\equiv 1$}
		\end{array}
	\]
\end{proof}

\noindent
The following theorem shows (a) that linear operations carry over to conditional expected values w.r.t.\ sub-$\sigma$-fields, (b) that every random variable approximates itself if it is already measurable w.r.t.\ the sub-$\sigma$-field $\CG$, and (c) it allows to simplify multiplications with $\CG$-measurable random variables.

\begin{theorem}[Properties of Conditional Expected Value \protect{\cite[p.~443]{grimmett2001probability}}]
	Let \label{app_thm:propert_conditional_expectation}
	$X,Y$ be random variables on $(\Omega, \CF,\mathbb{P})$ and let $\CG$ be a sub-$\sigma$-field of $\CF$.
	Then the following properties hold.
	\begin{enumerate}[label=(\alph*)]
		\item $\expec{a\cdot X+b \cdot Y \mid \CG}=a \cdot \expec{X\mid \CG}+b \cdot \expec{Y\mid \CG}$.
		\item If $X$ is itself $\CG$-measurable then $\expec{X \mid \CG}=X$.
		\item If $X$ is $\CG$-measurable then $\expec{X \cdot Y\mid \CG}=X \cdot \expec{Y\mid \CG}$.\label{it:cond_exp_multiplicativity}
	\end{enumerate}
\end{theorem}

%% file: appendix/prob_space_construction.tex
\label{appendix:prob_space_construction}

As mentioned in \cref{sec:Preliminaries}, we use \emph{schedulers} to resolve non-deterministic branching and sampling.

\begin{definition}[Scheduler]
	\label{def:Scheduler}
	A function $\scheduler \colon \CONF \to \left(\GT \uplus \{\gout\}\right)\times \STATE$ is a \emph{scheduler} if for every configuration $c = (\loc,t,\state) \in\CONF$, $\scheduler (c) = (g,\state')$ implies:
	\begin{enumerate}[label=(\alph*)]
		\item \label{scheduler:req1}
		      $\state (x) = \state' (x)$ for all $x\in\PV$.
		\item \label{scheduler:req2}
		      $\loc$ is the start location $\loc_g$ of the general transition $g$.
		\item \label{scheduler:req3}
		      $\state'(\tau_g) = \true$ for the guard $\tau_g$ of the general transition $g$.
		\item \label{scheduler:req4}
		      $g = \gout$ and $\state' = \state$, if $\loc = \loc_\bot$, $t = \tout$, or no $g'\in \GT,\state'\in \STATE$ satisfy \cref{scheduler:req1,scheduler:req2,scheduler:req3}.
	\end{enumerate}
\end{definition}
So to continue an evaluation in state $\state$, \cref{scheduler:req1} the scheduler chooses a state $\state'$ that agrees with $\state$ on all program variables, but where the values for temporary variables are chosen non-deterministically.
After instantiating the temporary variables, \cref{scheduler:req2} the scheduler selects a general transition that starts in the current location $\loc$ and \cref{scheduler:req3} whose guard is satisfied, if such a general transition exists.
Otherwise, \cref{scheduler:req4} $\scheduler$ chooses $\gout = \{ \tout \}$ and leaves the state $\state$ unchanged.

To every scheduler $\scheduler$ and initial state $\initstate$, one can assign a unique probability space, which captures the probabilistic evaluation behavior of PIPs.
Such an infinite product probability space $(\RUNS, \CF, \pipmeasure{\pip}{\scheduler}{\initstate})$ can be obtained using a standard cylinder construction (see, e.g., \cite[Theorem, 2.7.2]{cylindrical}), better known under the name \emph{Kolmogorov's Extension Theorem}.
We will present its version for a collection of measurable spaces $(\Omega_i,\powerset(\Omega_i))$ where $\Omega_i$ is \emph{countable} for all $i \in \NN$ as this suits our needs.
For each $(c_0 \cdots c_n) \in \Omega_0\times\dots\times\Omega_n$ with $n\in\NN$ let $pr (\prefix{c_0 \cdots c_n \to \bullet})$ be a probability mass function on $\Omega_{n+1}$.
So in particular, $pr(\prefix{\to \bullet}) = pr(\prefix{\bullet})$ is a probability mass function on $\Omega_0$.
Let $\Omega=\prod_{i=0}^{\infty} \Omega_i$ and let $\CF \subseteq \powerset(\Omega)$ be the $\sigma$-field created by all cylinder sets.
Here, for any $c_0 \in \Omega_0, \ldots, c_n \in \Omega_n$, the set $\Pre{\prefix{c_0 \cdots c_n}} = \{ \prefix{c_0 \cdots c_n,a_{n+1},a_{n+2},\ldots} \mid a_j \in \Omega_j$ for all $j \geq n+1 \, \}$ is the cylinder set of $c_0\cdots c_n$.
So $\CF$ is the smallest set containing $\emptyset$ and all cylinder sets, which is closed under complement and countable unions.

Then, there exists a unique probability measure $\pmeasure$ on $(\Omega,\CF)$ such that for all $B\subseteq \Omega_0 \times \ldots \times \Omega_n$:
\begin{align*}
	      & \sum_{c_0 \in \Omega_0}
	pr(\prefix{c_0}) \cdot \left( \sum_{c_1 \in \Omega_1} pr(\prefix{c_0 \to c_1})\cdots \left( \quad\ \sum_{{c_{n-1} \in \Omega_{n-1}}}
	pr(\prefix{c_0 \cdots c_{n-2} \to c_{n-1}}) \right. \right.                                                                                             \\
	\cdot & \left.\left. \left(\sum_{c_n \in \Omega_n} \ind_{B}(c_0 \cdots c_n) \cdot pr(\prefix{c_0 \cdots c_{n-1} \to c_n}) \right) \right)\cdots \right) \\
	{}={} & \pmeasure (\{\prefix{c_0 \cdots c_n \cdots} \in \Omega \mid c_0 \cdots c_n \in B\})
\end{align*}
Here, $\ind_{B}(c_0 \cdots c_n) = 1$ if $c_0 \cdots c_n \in B$ and $\ind_{B}(c_0 \cdots c_n) = 0$ if $c_0 \cdots c_n \notin B$.
For further details we refer to \cite[Theorem 2.7.2]{cylindrical}.

Let $\scheduler$ be a scheduler and $\initstate$ an initial state.
To use the cylinder construction in our setting, let $\Omega_i = \CONF$ for all $i \in \NN$, which is indeed countable.
We define
\begin{align*}
	\ptransition{\pip}{\scheduler}{\initstate} (\prefix{c}) & =
	\begin{cases}
		1, & \text{if } c = (\initloc, \tin, \initstate) \\
		0, & \text{otherwise}
	\end{cases}
\end{align*}
Let $f = \prefix{\cdots (\loc',t',\state')} \in \FPATH$ and $\scheduler ((\loc',t',\state')) = (g,\tilde{\state})$.
In the following we always have to consider the state $\tilde{\state}$ instead of the state $\state'$ to determine the probability of the next configurations, since the temporary variables are updated by the scheduler.
If $g \in \GT$ and $\loc'$ is the start location of $g$, then for every $t = (\loc',p,\tau,\eta,\loc)\in g$ with $\tilde{\state}(\tau) = \true$ we define:
\[
	\ptransition{\pip}{\scheduler}{\initstate} (\prefix{f \to (\loc,t,\state)}) = p \; \cdot\; \displaystyle\prod_{{x\in\PV}} \pr{\tilde{s}}{t}{s(x)}{x} \; \cdot \; \prod_{{u \in \VV \setminus \PV}} \delta_{\tilde{\state}(u),\state(u)},
\]
where for any transition $t$, $\tilde{s} \in \STATE$, $x \in \PV$, and $v \in \ZZ$, $\pr{\tilde{\state}}{t}{v}{x}$ expresses the probability that during the evaluation of $t$ in state $\tilde{\state}$ the variable $x$ gets updated to the value $v$:
\[
	\pr{\tilde{\state}}{(\_,\_,\_,\eta,\_)}{v}{x}
	=
	\begin{cases}
		1,  & \text{if } \eta (x) \in\POLYBOUND \text{ and } \tilde{\state}(\eta (x)) = v \\
		p', & \text{if } \eta (x) \in\DIST \text{ and
		} \eta(x)(\tilde{\state})(\tilde{\state}(x) - v) = p'                             \\
		0,  & \text{otherwise}
	\end{cases}
\]
Note that if $\eta(x)$ is a distribution function $d$, then the value sampled according to $d(\tilde{s})$ must be added to the current value $\tilde{\state}(x)$ of $x$.

Moreover, for any $z_1,z_2 \in \ZZ$ we have: \[\delta_{z_1,z_2} =
	\begin{cases}
		1, & \text{if } z_1 = z_2    \\
		0, & \text{if } z_1 \neq z_2
	\end{cases}
\]
So we require that the values of the temporary variables which the scheduler chooses in $\tilde{s}$ are ``saved'' by $\state$.
For all other configurations $c$, we define $\ptransition{\pip}{\scheduler}{\initstate}
	(\prefix{f} \to c) = 0$.
If the program has already \emph{terminated} (i.e., $g = g_{\bot}$), we define:

\[
	\ptransition{\pip}{\scheduler}{\initstate} (\prefix{f} \to c) =
	\begin{cases}
		1, & \text{if } \loc = \loc_{\bot},\ c = (\loc_{\bot},\tout,\state) \\
		0, & \text{otherwise}
	\end{cases}
\]

Note that for any non-empty finite path $f=\cdots c'$ and any configuration $c$ the value $\ptransition{\pip}{\scheduler}{\initstate} (\prefix{f} \to c)$ just depends on $c'$ and $c$, i.e., instead of $\ptransition{\pip}{\scheduler}{\initstate} (\prefix{f} \to c)$ we will write $\ptransition{\pip}{\scheduler}{\initstate} (c' \to c)$.

To illustrate $\ptransition{\pip}{\scheduler}{\initstate}$, consider the PIP from \cref{fig:simple_quadratic}, but where the update $\eta(x) = x$ of $t_0$ is changed to $\eta(x) = u \in \VV \setminus \PV$.
Let $\scheduler$ be the scheduler that selects $u = 5$ and that always chooses $g_1$ over $g_2$ (when possible), and let $\state_1 \in \STATE$ with $\state_1(x) = 5$, $\state_1(u) = 5$, and $\state_1(y) = \state_0(y)$.
Then we have:
\begin{align*}
	\ptransition{\pip}{\scheduler}{\initstate} (\, (\initloc,\tin,\initstate) \to (\loc,t,\state)\,) & =
	\begin{cases}
		1, & \text{if } \loc = \loc_1, t = t_0, \state = \state_1 \\
		0, & \text{otherwise}
	\end{cases}
	\\
	\ptransition{\pip}{\scheduler'}{\initstate} (\, (\loc_1,t_0,\state_1) \to (\loc,t,\state) \,)    & =
	\begin{cases}
		\tfrac{1}{2}, & \text{if } \loc = \loc_1, t = t_1, \state (x) = 4, \state (y) = \state_1(y) + 5 \\
		\tfrac{1}{2}, & \text{if } \loc = \loc_1, t = t_2, \state (x) = 5, \state (y) = \state_1(y) + 5 \\
		0,            & \text{otherwise}
	\end{cases}
\end{align*}

If $c'$ has the form $(\loc_2,t,\state)$ with $\state(y) \leq 0$ or $(\loc_\bot,\_,\state)$ for some state $\state$, then the scheduler cannot choose a suitable outgoing transition.
Then $\ptransition{\pip}{\scheduler}{\initstate} (c' \to c) = 1$ if $c = (\loc_\bot, \tout,\state)$ and $\ptransition{\pip}{\scheduler}{\initstate}
	(\prefix{c' \to c}) = 0$, otherwise. So the virtual transition $\tout$ is used for the infinitely many configurations of the run after termination.

Now let $\CF \subseteq \powerset(\RUNS)$ again be the $\sigma$-field created by all cylinder sets (i.e., $\CF$ is the smallest set containing $\emptyset$ and all sets $\Pre{f}$ for $f \in \FPATH$, which is closed under complement and countable unions).
Then, for a fixed scheduler $\scheduler$ and fixed initial state $\state_0$ we get a \emph{unique} probability measure $\pipmeasure{\pip}{\scheduler}{\initstate}\colon \mathcal{F} \to [0,1]$ on $(\RUNS, \CF)$ such that for all $B\subseteq \FPATH$:
\begin{align*}
	      & \sum_{c_0 \in \CONF} \ptransition{\pip}{\scheduler}{\initstate}(\prefix{c_0}) \cdot \left( \sum_{c_1 \in \CONF} \ptransition{\pip}{\scheduler}{\initstate}(\prefix{c_0} \to c_1)\cdots \left( \quad\;\; \sum_{{c_{n-1} \in \CONF}} \ptransition{\pip}{\scheduler}{\initstate}
	(\prefix{c_0 \cdots c_{n-2}} \to c_{n-1}) \right.\right.                                                                                                                                                                                                                              \\
	\cdot & \left.\left. \left(\sum_{c_n \in \CONF} \ind_{B} (c_0 \cdots c_n) \cdot \ptransition{\pip}{\scheduler}{\initstate} (\prefix{ c_0 \cdots c_{n-1}} \to c_n) \right) \right)\cdots \right)                                                                                       \\
	{}={} & \sum_{c_0 \in \CONF} \ptransition{\pip}{\scheduler}{\initstate}(\prefix{c_0}) \cdot \left( \sum_{c_1 \in \CONF} \ptransition{\pip}{\scheduler}{\initstate}(\prefix{c_0} \to c_1)\cdots \left( \quad\;\; \sum_{{c_{n-1} \in \CONF}} \ptransition{\pip}{\scheduler}{\initstate}
	(c_{n-2} \to c_{n-1}) \right.\right.                                                                                                                                                                                                                                                  \\
	\cdot & \left.\left. \left(\sum_{c_n \in \CONF} \ind_{B} (c_0 \cdots c_n) \cdot \ptransition{\pip}{\scheduler}{\initstate} (c_{n-1} \to c_n) \right) \right)\cdots \right)                                                                                                            \\
	{}={} & \pipmeasure{\pip}{\scheduler}{\initstate} (\{\prefix{c_0 \cdots c_n \cdots} \in \Omega \mid \prefix{c_0 \cdots c_n } \in B\})
\end{align*}
by the cylinder set construction.
To ease notation, for a finite path $f=c_0\cdots c_n$ we write

\[
	\ptransition{\pip}{\scheduler}{\initstate}(f) = \ptransition{\pip}{\scheduler}{\initstate}(c_0) \cdot \prod_{i=1}^n \ptransition{\pip}{\scheduler}{\initstate} (c_{i-1} \to c_i). \]

This then yields

\[
	\pipmeasure{\pip}{\scheduler}{\initstate}\left(\Pre{\prefix{f}}\right) = \ptransition{\pip}{\scheduler}{\initstate}(f). \]

Hence, the probability space for $\pip$ is $(\RUNS, \mathcal F, \pipmeasure{\pip}{\scheduler}{\initstate})$.

%% file: proofs/lifting_of_bounds.tex
\liftingofbounds*
\begin{proof}
	Let $(\tbound,\sbound)$ be a (non-probabilistic) bound pair.
	Furthermore, let $g\in\GT$ be some general transition, let $\initstate \in \STATE$, and let $\scheduler$ be a scheduler.
	Finally, let $\run = \prefix{(\loc_0,t_0,\state_0) (\loc_1,t_1,\state_1) \cdots} \in \RUNS$ be an admissible run
	under the chosen scheduler $\scheduler$ and the chosen initial state $\initstate$.
	In particular, this means that $t_0 = t_{in}$.
	Then we have
	\begin{align*}
		 & \eval{\tbounde (g)}{\initstate}                                                                                                                                             \\
		 & = \sum_{t \in g} \eval{\tbound (t)}{\initstate}                                                                                                                             \\
		 & \geq \sum_{t \in g} \sup\left\{\abs{\{i \mid t'_i = t\}} \mid f' = \prefix{(\_,t'_0,\_) \cdots (\_,t'_n,\_)} \land \ptransition{\pip}{\scheduler}{\initstate}(f')>0\right\}
		\\
		 & \geq \abs{\left\{i \in \NN \mid t_i \in g\right\}} \tag{since every prefix of $\run$ is admissible}                                                                         \\
		 & = \timervar(g)(\run).
	\end{align*}

	Non-admissible runs do not influence the expected value $\expv{\pip}{\scheduler}{\initstate} (\timervar(g))$, since the set of all non-admissible runs has probability zero: if a run is non-admissible, by definition it has a non-admissible prefix.
	There are only countably many such non-admissible prefixes $f$.
	By the construction of $\pipmeasure{\pip}{\scheduler}{\initstate}$ in \cref{appendix:prob_space_construction}, we have $\pipmeasure{\pip}{\scheduler}{\initstate}\left(\Pre{f}\right) = \ptransition{\pip}{\scheduler}{\initstate}(f)=0$ and therefore, the countable union 
	\[\bigcup_{f \in \FPATH, f \text{ is non-admissible}}\Pre{f}\] has probability zero.
	But this set is precisely the set of non-admissible runs.
	Note that $\eval{\tbounde(g)}{\initstate} \in \overline{\NN}$ can be viewed as a constant function, i.e., its expected value is again $\eval{\tbounde(g)}{\initstate}$.
	As $\run$ is an arbitrary admissible run, we get by monotonicity of the expected value operator
	\[
		\eval{\tbounde(g)}{\initstate} = \expv{\pip}{\scheduler}{\initstate}(\eval{\tbounde(g)}{\initstate}) \geq \expv{\pip}{\scheduler}{\initstate} (\timervar(g)). \]

	Now, let in addition $\alpha = (g,\loc,x) \in \GRV$ and let $\initstate$, $\scheduler$, and $\run$ be as above.
	Then we have
	\begin{align*}
		 & \eval{\sbounde (\alpha)}{\initstate}                                                                                                                                                                            \\
		 & = \sum_{t = (\_,\_,\_,\_,\loc) \in g} \eval{\sbound (t,x)}{\initstate}                                                                                                                                          \\
		 & \geq \max \left\{\eval{\sbound (t,x)}{\initstate} \mid t = (\_,\_,\_,\_,\loc) \in g \right\}                                                                                                                    \\
		 & = \max \left\{\sup \left\{|\state'_i (x)| \mid f' = \prefix{(\_,\_,\_) \cdots (\_,t'_i,\state'_i)} \land \ptransition{\pip}{\initstate}{\scheduler}(f')>0\right\} \mid t'_i = (\_,\_,\_,\_,\loc) \in g \right\}
		\\
		 & \geq \sup\left\{|\state_i (x)| \mid i \in \NN \land t_i = (\_,\_,\_,\_,\loc) \in g\right\} \tag{since every prefix of $\run$ is admissible}                                                                     \\
		 & = \sizervar(\alpha)(\run).
	\end{align*}

	As $\run$ was an arbitrary admissible run, similar as above we have
	\[
		\eval{\sbounde (\alpha)}{\initstate} = \expv{\pip}{\scheduler}{\initstate} (\eval{\sbounde (\alpha)}{\initstate}) \geq \expv{\pip}{\scheduler}{\initstate} (\sizervar(\alpha)).
	\]
\end{proof}

%% file: appendix/details_lprf.tex
\label{appendix:details_lprf}
At the end of this subsection, we will prove (a refinement of) \cref{theorem:exptimeboundsmeth}.
However, the proof of this theorem is very involved and needs some preparation involving the theory of so-called ranking supermartingales (see \cref{def:rsm}).

Note that our definition of PLRFs differs from the standard definition of probabilistic ranking functions, since we regard distinguished subsets $\GTG \subseteq \GTNI \subseteq \GT$.
For that reason, we have to prove that a stochastic process defined according to our PLRFs is indeed a ranking supermartingale (this will be shown in \cref{ranking corollary}).
Then we can use existing results on ranking supermartingales (\cref{thm:chatt_lexrsm}) to show that our PLRFs yield bounds on a stopping time defined according to $\GTG$ (\cref{def:rsm_stoptime}).
However, in the parts of a run belonging to executions of general transitions from $\GTNI$, we have to take into account that there can occur transitions from $\GTNI\setminus \GTG$ before the first general transition from $\GTG$ is executed.
Thus, we also have to show that a stochastic process corresponding to the part of the run before the first application of a transition from $\GTG$ corresponds to a martingale (this is shown implicitly in \cref{lemma:entry_supermartingale}).

Let $\pip$ be a PIP and $\rank$ be a PLRF for $\GTG \subseteq \GTNI$ as in \cref{def:Probabilistic Polynomial Ranking Functions}.
We define functions $I_0,I_1, \dots \allowbreak\colon \RUNS\rightarrow \powerset (\NN)$ that map a run to the set of indices that correspond to the evaluation steps upon entering $\GTNI$ for the $(i+1)$-th time.
These functions will be used to define random variables $R_{i,j}\colon\RUNS\rightarrow\RR_{\geq 0}$ where $R_{i,j}$ maps a run to $\rank$ evaluated in the configuration before the $(j+1)$-th execution of $g\in\GTG$ upon entering $\GTNI$ for the $(i+1)$-th time.
We start by defining $I_i (\run)$, which is the set of indices of the run $\run$ belonging to the $(i+1)$-th time in $\GTNI$.

\begin{definition}[Index Sets for $\GTNI$]
	$I_0 (\run)$ for a run $\run = \prefix{(\initloc,t_0,\initstate)(\loc_1,t_1,\state_1)\cdots}$ is the smallest set of natural numbers such that:

	\begin{enumerate}
		\item $\min \{n\in \NN \mid t_n\in \bigcup \GTNI \} \in I_0 (\run)$
		\item $n\in I_0 (\run) \land t_{n+1} \in \bigcup \GTNI \implies n+1 \in I_0 (\run)$
	\end{enumerate}
	Moreover, $I_{i+1} (\run)$ for $i \in \NN$ is the smallest set of natural numbers such that:
	\begin{enumerate}
		\item $\min \{n\in \NN \setminus \bigcup_{k=0}^{i} I_k(\run) \mid t_n\in \bigcup \GTNI \} \in I_{i+1} (\run)$
		\item $n\in I_{i+1} (\run) \land t_{n+1} \in \bigcup \GTNI \implies n+1 \in I_{i+1} (\run)$
	\end{enumerate}
\end{definition}

Furthermore, we define mappings $I^{\succ}_{i,j}: \RUNS \to \NN$ for $i, j \in \NN$ such that for a fixed $\run = \prefix{(\initloc,t_0,\initstate)(\loc_1,t_1,\state_1)\cdots}\in \RUNS$ the value $I^{\succ}_{i,j}(\run)$ is the index where some transition from $\GTG$ is evaluated for the $(j+1)$-th time upon entering $\GTNI$ for the $(i+1)$-th time, i.e., $I^{\succ}_{i,j}(\run)$ is the $(j+1)$-th smallest element of $I_i(\run) \cap \{k\in \NN\mid t_k\in\bigcup\GTG\}$.

\begin{definition}[Indices for $\GTG$]
	Let $i, j \in \NN$.
	The mappings $I^{\succ}_{i,j}\colon \RUNS\rightarrow \NN$ are defined as follows:

	\[
		I^{\succ}_{i,j} (\run) =
		\begin{cases}
			\min \left( \cuts^j \left (I_i(\run) \cap \left\{k\in \NN \mid t_{k} \in \bigcup \GTG\right\}\right)\right),                           \\
			\multicolumn{2}{r}{\text{if $\cuts^j \left (I_i(\run) \cap \left\{k\in \NN \mid t_{k} \in \bigcup \GTG\right\}\right) \neq\emptyset$}} \\
			0, & \text{otherwise}
		\end{cases}
	\]
	This means that if the desired configuration does not exist, then the index defaults to $0$.
	Here, $\cuts (X) = X \setminus \{\min X\}$ removes the set's smallest element, i.e.,
	the $j$-th composition $\cuts^j (X)$ is the set $X$ without the first $j$ smallest element(s).

	For $f\in\FPATH$ and $j\in \NN$ we define:

	\[I^{\succ}_{i,j}(f) =
		\begin{cases}
			m,    & \text{if for all } \run \in \Pre{f} \text{ we have } I^{\succ}_{i,j} (\run) = m \\
			\bot, & \text{otherwise}.
		\end{cases}
	\]
\end{definition}

So, the mapping $I^{\succ}_{i,j}$ is defined on a finite path $f$ whenever $f$ already determines the value of $I^{\succ}_{i,j}$ uniquely for every run with prefix $f$.

\begin{definition}[Stochastic Processes $(R_{i,j})_{i,j\in \NN}$]
	\label{definition:rsm_stproc}
	For $i, j \in \NN$ we define the map $R_{i,j}: \RUNS \to \RR_{\geq 0}$ on a run $\run = \prefix{(\initloc,t_0,\initstate)(\loc_1,t_1,\state_1)\cdots} \in \RUNS$ by

	\begin{align*}
		R_{i,j} (\run) & =
		\begin{cases}
			\max \left\{ \exact{\rank (\loc_{I^{\succ}_{i,j} (\run) - 1})}{\state_{I^{\succ}_{i,j}(\run) - 1}}, 0\right\}, & \text{if } I^{\succ}_{i,j} (\run) > 0 \\
			0,                                                                                                             & \text{otherwise}
		\end{cases}
	\end{align*}

	\noindent{}
	For $f\in\FPATH$ and $i,j\in \NN$ we define:

	\begin{align*}
		R_{i,j}(f) & =
		\begin{cases}
			r,    & \text{if for all } \run \in \Pre{f} \text{ we have } I^{\succ}_{i,j}(\run) > 0 \land R_{i,j} (\run) = r \\
			\bot, & \text{otherwise}
		\end{cases}
	\end{align*}
\end{definition}
Hence, $R_{i,j}(f) \in \RR_{\geq 0}$ if $f$ is long enough to determine the value of $R_{i,j}$ uniquely.
Moreover, for any run $\run$, the value $R_{i,j} (\run)$ is the value that $\rank$ takes \emph{before} executing $\GTG$ for the $(j+1)$-th time within the indices $I_i(\run)$ or zero.
We will now prove this ``already determined'' behavior.

\begin{lemma}[$R_{i,j}$ is Already Determined by Long Enough Prefixes]
	\label{lemma:rsm_stprocesslongprefix}
	Let $\run\in\RUNS$.
	If $R_{i,j} (\run) = r > 0$ for some $i, j \in \NN$, then there is a prefix $f\in\FPATH$ such that for all $0 \leq k \leq j$
	\[
		R_{i,k} (f) = R_{i,k}(\run).
	\]
\end{lemma}

\begin{proof}
	Let $\run = \prefix{c_0\cdots}$ be a run with $R_{i,j} (\run) = r > 0$.
	Then $0 < I^{\succ}_{i,j}(\run)$ by definition.
	We choose $f=\prefix{c_0\cdots c_{I^{\succ}_{i,j}(\run)}}$.
	Trivially, we then have $I^{\succ}_{i,k} (\run) = I^{\succ}_{i,k} (f)$ for each $0 \leq k \leq j$.
	Therefore by definition of $R_{i,k}$ we get $R_{i,k} (f) = R_{i,k} (\run)$ for every $0 \leq k \leq j$.
\end{proof}
\cref{lemma:rsm_stprocesslongprefix} eases our reasoning.
First of all, it is clear that $I^{\succ}_{i,j}$ is $\CF$-measurable: If a run has a \emph{positive} value under $I^{\succ}_{i,j}$ then this value only depends on a prefix of a run.
Moreover, since we know that the value of some run $\run \in \RUNS$ under $R_{i,j}$ is either zero or already determined by some prefix of $\run$, this directly implies that $R_{i,j}$ is $\CF$-measurable.

In the next step we construct a suitable filtration $(\CF_{i,j})_{j\in \NN}$ of $\CF$ such that the sequence $(R_{i,j})_{j\in \NN}$ becomes a stochastic process adapted to this filtration for every fixed $i \in \NN$.
Intuitively every $\CF_{i,j}$ should contain all sets associated with a finite path $f$ ending in a configuration with index $I^{\succ}_{i,j}$, such that $R_{i,j}$ evaluates to a value $r \in \RR_{\geq 0}$.
Formally, these sets are defined as follows:
\begin{align}
	S_{i,j,f} (r) = \Pre{f} \cap R_{i,j}^{-1}(\{r\}) \cap \left(I^{\succ}_{i,j}\right)^{-1}(\abs{f}), \label{eq:def_sijfr}
\end{align}
where $\abs{f}$ denotes the length of the finite path $f$.
For any $n \in \NN$, let $\run[n]$ denote the $n$-th configuration in the run $\run$.
So, a run $\run \in \RUNS$ satisfies $\run \in S_{i,j,f} (r)$ whenever

\begin{itemize}
	\item the finite path $f$ is a prefix of $\run$,
	\item $R_{i,j}(\run) = r$, and
	\item $f= \prefix{\run[0]\run[1]\cdots\run[I^\succ_{i,j} (\run) - 1]}$.
\end{itemize}
In the case of $r=0$ we get the runs $\run$ with prefix $f$ where $R_{i,j}$ evaluates to $0$ due to the corresponding configuration.
However, if for some run $\run$ we have $I^{\succ}_{i,j} (\run) = 0$, then also $R_{i,j} (\run) = 0$ but in general there is \emph{no} prefix $f'$ of $\run$ with $R_{i,j} (f') = 0$, i.e., those runs are not covered by the just defined sets.
As a result of this observation only sets $S_{i,j,f} (r)$ with $r > 0$ will be regarded.
We will consider the remaining runs afterwards.
However, a run $\run$ which has $f$ as a prefix and satisfies $R_{i,j}(\run) > 0$ is not necessarily contained in $S_{i,j,f}(R_{i,j}(\run))$.
The reason is that $f$ does \emph{not} necessarily determine the value of $I^{\succ}_{i,j}(\run)$.
To deal with this situation in the following reasoning, we first define two subsets of the configurations, which will increase readability of our reasoning later on.
\begin{align*}
	\CONF_{\GTG}  & = \{c=(\_,t,\_) \in \CONF \mid t \in \bigcup \GTG\},  \\
	\CONF_{\GTNI} & = \{c=(\_,t,\_) \in \CONF \mid t \in \bigcup \GTNI\}.
\end{align*}
We observe the following lemma.
\begin{lemma}[Observations on $S_{i,j,f}$]
	\label{lemma:decomposition}
	Let $f \in \FPATH$ with last configuration $(\loc_f,\_,\state_f)$, $r>0$, and $i,j \in \NN$.
	If $S_{i,j,f} (r) \neq \emptyset$ then
	\begin{enumerate}[a)]
		\item $r = \exact{\rank(\loc_f)}{\state_f}$.
		\item
		      \[
			      S_{i,j,f}(r) = \biguplus_{c \in \CONF_{\GTG}} \Pre{fc}.
		      \]
	\end{enumerate}
\end{lemma}
\begin{proof}

	\begin{enumerate}[a)]
		\item Let $\run \in S_{i,j,f}(r)$.
		      Then $I^{\succ}_{i,j}(\run) = |f|$.
		      By \cref{eq:def_sijfr} we have $r = R_{i,j}(\run)$ and this value is uniquely determined by the configuration $\run[|f|-1]$.
		      But this is the last configuration of $f$, so, $r = \exact{\rank(\loc_f)}{\state_f}$.

		\item Let $\run \in S_{i,j,f}(r)$.
		      By \cref{eq:def_sijfr}, $\run[|f|] = (\_,t,\_)$ such that $t \in \bigcup \GTG$, i.e., $\run[|f|] \in \CONF_{\GTG}$, since $I^{\succ}_{i,j}(\run) = |f|$.

		      For the other direction, since $ S_{i,j,f}(r) \neq \emptyset$ this means that for every run with prefix $f$ we have $I^{\succ}_{i,j-1}(f) = I^{\succ}_{i,j-1}(\run)$.
		      However, $f$ does not determine the value $I^{\succ}_{i,j}$, i.e., $I^{\succ}_{i,j}(f) = 0$.
		      Thus, if $\run \in \Pre{fc}$ for some $c \in \CONF_{\GTG}$, then $I^{\succ}_{i,j}(\run) = |f|$, i.e., $\run \in S_{i,j,f}(r)$.
	\end{enumerate}
\end{proof}
Nevertheless, as already mentioned, not all runs $\run$ with $R_{i,j} (\run) = 0$ can be merged into a single set: in this way we ignore the information that $R_{i,j'}(\run) \neq 0$ is still possible for some $j' < j$.
Therefore we decompose these runs into more fine-grained sets $S^0_{i,j,j',f} (r)$.
For $j'\in\NN$ with $j'<j$ we define
\[
	S^0_{i,j,j',f} (r) =
	\begin{cases}
		\emptyset, \quad \text{if } j' > 0 \land r = 0                                                             \\
		\Pre{f} \cap \left(\bigcap_{k = j' + 1}^{j}
		R_{i,k}^{-1}(\{0\})\right) \cap R_{i,j'}^{-1}(\{r\}) \cap \left(I^{\succ}_{i,j'}\right)^{-1}(\abs{f} - 1), \\
		\qquad \text{otherwise}                                                                                    \\
	\end{cases}
\]
Intuitively, this set contains all runs on which $R_{i,j}$ evaluates to zero but either a previous value at time $j'$ (given by $R_{i,j'}$) evaluates to $r > 0$ (if $j' > 0$ then we must have $r > 0$), or all previous values evaluate to $0$ (if $j' = 0$ and $r = 0$).

The above defined sets can now be used as generating sets for a sequence of $\sigma$-fields $(\CF_{i,j})_{j\in\NN}$ which we will prove to be a filtration in \cref{lemma:rsm_validfiltration}.

\begin{definition}[Filtration for $R_{i,j}$]
	For $i, j \in \NN$, let $\CF_{i,j}$ be the smallest $\sigma$-field containing
	\begin{equation}
		\mbox{\small $
				\begin{cases}
					\left\{S_{i,j,f} (r)\mid f \in \FPATH, r \geq 0\right\},                                                                          & \text{if } j = 0  \\
					\left\{S_{i,j,f} (r)\mid f \in \FPATH, r > 0\right\} \cup \left\{S^0 _{i,j,j',f} (r) \mid f \in \FPATH, j'< j, r \geq 0 \right\}, & \text{if } j > 0.
				\end{cases}
			$}
	\end{equation}
	We denote this generating system of $\CF_{i,j}$ by $\CE_{i,j}$.
\end{definition}

A $\sigma$-field which is generated by a countable partition of the sample space is easier to handle, e.g., when it comes to computing conditional expected values (see \cref{app_lemma:conditional_expectation_countable_cover,lem:upper_bound_cond_exp}).
Hence, we prove that our just defined sequence of $\sigma$-fields has this property w.r.t.\ the respective generating system $\CE_{i,j}$.

\begin{lemma}[$\CF_{i,j}$ is Generated by a Countable Partition of $\RUNS$]
	\label{lemma:rsm_gensetscountable}
	For every $i,j\in\NN$ the generating system $\CE_{i,j}$ of $\CF_{i,j}$ is a countable partition of $\RUNS$.\footnote{Here, we only take non-empty sets into account.}
\end{lemma}

\begin{proof}
	Let $i, j \in \NN$ be fixed.
	By definition, the sets in $\CE_{i,j}$ are pairwise disjoint.
	Moreover, using \cref{lemma:rsm_stprocesslongprefix}, and the countability of $\FPATH$ we already have that the functions $R_{i,j'}$ only take countably many values, for any $j' \in \NN$.
	So, $\CE_{i,j}$ contains only countably many non-empty sets.

	Let $\run \in \RUNS$.
	Whenever $R_{i,j}(\run) > 0$, then \cref{lemma:rsm_stprocesslongprefix} ensures that $\run \in S_{i,j,f}(R_{i,j}(\run))$ for some prefix $f$ of $\run$.

	Now assume $R_{i,j}(\run) = 0$.
	If $j = 0$, then $\run \in S_{i,j,f}(0) = S_{i,0,f}(0)$ for $f = \prefix{\run[0]\cdots\run[I_{i,0}(\run)-1]}$ if $I_{i,0}(\run) > 0$ and $\run \in S_{i,j,\varepsilon}(0) = S_{i,0,\varepsilon}(0)$ if $I_{i,0}(\run) = 0$.
	So, assume $j > 0$.
	Whenever $R_{i,j'}(\run) > 0$ for some $0 \leq j' < j$, then again \cref{lemma:rsm_stprocesslongprefix} ensures that $\run \in S^0_{i,j,j',f}(R_{i,j'}(\run))$ for some prefix $f$ of $\run$.
	So, the only remaining case is $R_{i,j'}(\run) = 0$ for all $0 \leq j' < j$.
	Then, $\run \in S^0_{i,j,0,f}(0)$ for $f = \prefix{\run[0]\cdots\run[I_{i,0}(\run)]}$.
	So, $\CE_{i,j}$ is indeed a partition of $\RUNS$.
\end{proof}
Note that \cref{lemma:rsm_gensetscountable} yields that all elements of $\CF_{i,j}$ are just countable disjoint unions of the just described generators (see
\cref{lemma:sigma_field_countable_cover}).

We will now prove that $(\CF_{i,j})_{j\in \NN}$ indeed forms a filtration of $\CF$.
\begin{lemma}[$(\CF_{i,j})_{j\in \NN}$ is a Filtration of $\CF$]
	\label{lemma:rsm_validfiltration}
	For all $i, j \in \NN$ we have $\CF_{i,j}\subseteq \CF_{i,j+1}$ and $\CF_{i,j}\subseteq\mathcal{F}$, i.e., $(\CF_{i,j})_{j\in\NN}$ is a filtration of $\mathcal{F}$.\footnote{See \cref{appendix:prob_space_construction} for the definition and construction of $\mathcal{F}$.}
\end{lemma}

\begin{proof}
	First of all, we have already seen that for any $i, j \in \NN$ both $R_{i,j}$ and $I^{\succ}_{i,j}$ are $\CF$-measurable.
	Hence, by definition we have $S_{i,j,f}(r) \in \CF$ and $S^0_{i,j,j',f}(r) \in \CF$, i.e., $\CF_{i,j} \subseteq \CF$.

	Now we show that for any $f\in\FPATH$ and $j \in \NN$ such that $S_{i,j,f}(r) \neq \emptyset$ for some $r>0$ we have $S_{i,j,f}(r) \in \CF_{i,j+1}$.
	By \cref{lemma:decomposition} b) we have
	\[
		S_{i,j,f} (r) = \biguplus_{c \in \CONF_{\GTG}} \Pre{fc}. \]
	We partition $S_{i,j,f}(r)$ into $S_{i,j,f}(r)=A \uplus B$ such that $R_{i,j+1}(\run_a) > 0$ for every $\run_a \in A$ and $R_{i,j+1}(\run_b) = 0$ for every $\run_b \in B$.
	Then we have
	\[
		A=\bigcup_{f' \in \FPATH} \bigcup_{c \in \CONF_{\GTG}} \bigcup_{r' > 0} S_{i,j+1,\prefix{fcf'}} (r') \in \CF_{i,j+1},
	\]
	where the innermost union leads to only countably many non-empty sets and
	\[
		B = \bigcup_{c \in \CONF_{\GTG}}S^0_{i,j+1,j,fc} (r) \in \CF_{i,j+1}.
	\]
	So we conclude $S_{i,j,f} (r) = A \uplus B \in \CF_{i,j+1}$.

	Now we consider the sets $S^0_{i,j,j',f} (r)$.
	We fix $j,j'\in\NN$, $f\in\FPATH$, and $r\in\RR_{\geq 0}$.
	Again, for every run $\run$ we either have $R_{i,j+1} (\run) > 0$ or $R_{i,j+1} (\run) = 0$.
	Consider the partition $S^0_{i,j,j',f} (r) = D \uplus E$ where $R_{i,j+1}(\run_d) > 0$ for every $\run_d \in D$ and $R_{i,j+1}(\run_e) = 0$ for every $\run_e \in E$.
	Since $R_{i,j+1} (\run_d) > 0$ for runs $\run_d \in D$, by \cref{lemma:rsm_stprocesslongprefix}
	there is a finite prefix of $\run_d$ that is long enough to determine the value of $R_{i,k}(\run_d)$ for any $k \leq j+1$.
	Moreover, by definition of $S^0_{i,j,j',f} (r)$ this prefix has $f$ as a prefix itself and it has to ensure that $R_{i,k}$ evaluates to zero for $j'+1 \leq k \leq j$.
	Using all these requirements, we can write $D$ as follows.
	\[
		D = \bigcup_{\substack{f' \in \FPATH \\
				R_{i,k}(ff')=0, j' < k \leq j}}\bigcup_{r' > 0} S_{i,j+1,\prefix{ff'}} (r'),\]
	i.e., $D \in \CF_{i,j+1}.$
	On the other hand we know that for all runs $\run_e\in E$ we have $R_{i,j+1} (\run_e) = 0$ and therefore \[E=S^0_{i,j+1,j',f} (r) \in \CF_{i,j+1}.\]
	Finally $S^0_{i,j,j',f} = D \uplus E\in \CF_{i,j+1}$.

	There still remains the case of showing that $S_{i,0,f}(0) \in \CF_{i,1}$, since this set is not covered by the cases above.
	However, the reasoning required for this case is similar to the reasoning above:
	We partition $S_{i,0,f}(0)$ such that on the first part $R_{i,1}$ takes strictly positive values and on the other part $R_{i,1}$ is constant zero.

	\[
		S_{i,0,f}(0) = \bigcup_{\substack{f' \in \FPATH \\
				R_{i,0}(ff')=0}}\bigcup_{r' > 0} S_{i,1,\prefix{ff'}} (r') \cup \bigcup_{c \in \CONF} S^0_{i,1,0,fc}(0) \in \CF_{i,1}. \]

	Therefore, $\CF_{i,j} \subseteq \CF_{i,j+1} \subseteq \CF$ which concludes the proof.
\end{proof}

For our filtration to be useful in combination with the stochastic process $(R_{i,j})_{j\in \NN}$, the process has to be adapted to $(\CF_{i,j})_{j\in \NN}$.

\begin{lemma}[$(R_{i,j})_{j\in \NN}$ is Adapted to $(\CF_{i,j})_{j\in \NN}$]
	\label{lem:adaptedness}
	For every $i, j\in \NN$ the function $R_{i,j}$ is measurable w.r.t.
	the $\sigma$--field $\CF_{i,j}$.
\end{lemma}

\begin{proof}
	For the proof we fix $i, j \in \NN$.
	We have already seen that $R_{i,j}$ takes only countably many values.
	The measurabilty is shown by proving $\{\run \in \RUNS \mid R_{i,j}(\run) = r\} \in \CF_{i,j}$ for every $r \geq 0$.
	We begin with the case $r > 0$:

	\[
		\{\run\in\RUNS \mid R_{i,j} (\run) = r\} = \bigcup_{f \in \FPATH} S_{i,j,f}(r) \in \CF_{i,j}
	\]
	Now using the above statement we can also show the remaining case $R_{i,j}^{-1} (\{0\}) \in \CF_{i,j}$.
	\begin{align*}
		R_{i,j}^{-1} (\{0\}) & = \{\run\in\RUNS \mid R_{i,j} (\run) = 0\}                                                                 \\
		                     & =\RUNS\setminus \bigcup_{r'\in \RR_{> 0}} \underbrace{R_{i,j}^{-1}(\{r'\})}_{\in \CF_{i,j}} \in \CF_{i,j}.
	\end{align*}
	Hence every $R_{i,j}$ is measurable w.r.t.
	$\CF_{i,j}$ and so, the process $(R_{i,j})_{j \in \NN}$ is adapted to $(\CF_{i,j})_{j \in \NN}$.
\end{proof}

The point of interest in the above defined stochastic process $(R_{i,j})_{j\in \NN}$ is the first index where it takes the value zero.
At this point no more transition $t\in\bigcup\GTG$ can be executed by the properties of $\rank$ (see \cref{def:Probabilistic Polynomial Ranking Functions}).
To introduce this ``point of interest'' more formally, we first present the notion of a stopping time.

\begin{definition}[Stopping Time]
	\label{def:stopping_time}
	Let $(\Omega, \mathcal{F}, \pmeasure)$ be a probability space with a filtration $(\mathcal{F}_j)_{j \in \NN}$ and $T: \Omega \to \natclosure$.
	We call $T$ a \emph{stopping time}
	for $(\mathcal{F}_j)_{j \in \NN}$ if $T^{-1}(\{j\}) \in \mathcal{F}_j$ holds for all $j \in \NN$.
\end{definition}

\begin{definition}[Stopping Time for $(R_{i,j})_{j\in \NN}$]
	\label{def:rsm_stoptime}
	For every $i \in \NN$, let $T'_i:\RUNS \to \overline{\NN}$ be defined by $T'_i (\run) = \min \{j\in \NN\mid R_{i,j} (\run) = 0\}$.
\end{definition}
\noindent
As usual, we define $\min \emptyset = \infty$.

Due to the adaptedness of $(R_{i,j})_{j\in \NN}$ to the filtration, $T'_i$ is indeed a stopping time for the filtration $(\CF_{i,j})_{j\in \NN}$.

The next lemma states that as long as our process has not terminated it is expected to decrease by at least $1$.
This will allow us to prove that $(R_{i,j})_{j \in \NN}$ forms a ranking supermartingale.
So let us first define this notion.

\begin{definition}[Ranking Supermartingale \protect{\textnormal{\cite{lexrsm}}}]
	\label{def:rsm}
	Let $(\Omega, \mathcal{F}, \pmeasure)$ be a probability space with a filtration $(\mathcal{F}_j)_{j \in \NN}$ and $T$ a stopping time for $(\CF_{i,j})_{j\in \NN}$.
	A stochastic process $(X_j)_{j \in \NN}$ is a (strict) \emph{ranking supermartingale for $T$} iff for all $j \in \NN$
	\begin{enumerate}
		\item $X_j$ is adapted to $\mathcal{F}_j$
		\item $X_j(\Omega) \subseteq \RR_{\geq 0}$
		\item We have
		      \[
			      \expvsign(X_{j+1} \cdot \ind_{T^{-1}((j, \infty))} \mid \mathcal{F}_j) \leq (X_{j} - 1) \cdot \ind_{T^{-1}((j, \infty))}
		      \]
	\end{enumerate}
\end{definition}

We now have to prove that the process $(R_{i,j})_{j\in \NN}$ indeed satisfies the (third) ranking condition of \cref{def:rsm}.
However, the proof of \cref{lemma:rs_rankcond} is very intricate.
To increase readability we will formulate it as a corollary and split its proof into several lemmas (\cref{lemma:rs_rankcond,lemma:rank_cond_on_generators,lem:ranking_condition_A_m}) which we prove separately.
Combining these lemmas then yields the complete proof for the statement.

\begin{lemma}[Ranking Condition for $(R_{i,j})_{j\in \NN}$]
	\label{lemma:rs_rankcond}
	Let $\scheduler$ be an arbitrary scheduler, $\initstate$ an arbitrary initial state, and $i,j \in \NN$.
	Then on the set $\left(T'_i\right)^{-1}((j,\infty))$ we have:

	\[
		\expv{\pip}{\scheduler}{\initstate} (R_{i,j+1}\mid \CF_{i,j}) \leq R_{i,j} - 1
	\]
\end{lemma}

\begin{proof}
	The inequation given above is equivalent to:

	\[
		\ind_{T'_i > j} \cdot \expv{\pip}{\scheduler}{\initstate} (R_{i,j+1} \mid \CF_{i,j}) \leq \ind_{T'_i > j} \cdot (R_{i,j} - 1), \]
	where $T'_i > j$ stands for the set $\left(T'_i\right)^{-1}((j,\infty))$.

	Since $\left(T'_i\right)^{-1}((j,\infty)) \in \CF_{i,j}$ (i.e., $\ind_{T'_i > j}$ is $\CF_{i,j}$-measurable), by \cref{app_thm:propert_conditional_expectation} \ref{it:cond_exp_multiplicativity}, this inequation is itself equivalent to
	\begin{align}
		\expv{\pip}{\scheduler}{\initstate} \left(\ind_{T'_i > j} \cdot R_{i,j+1} \mid \CF_{i,j}\right) \leq \ind_{T'_i > j} \cdot (R_{i,j} - 1),\label{eq:proof_rankcond}
	\end{align}

	First of all, due to the properties ``Decrease'' and ``Boundedness'' of $\rank$ in \cref{def:Probabilistic Polynomial Ranking Functions}, we have $R_{i,j}(\run) \geq 1$ for any admissible run $\run \in \left(T'_i\right)^{-1}((j,\infty))$.

	By \cref{lem:upper_bound_cond_exp}, it suffices to only consider the generators $\CE_{i,j}$ of $\CF_{i,j}$ to prove this statement.
	Note that on the generators of the form $S^0_{i,j,j',f}(r)$, by \cref{def:rsm_stoptime} we have $T'_i \leq j$.
	So in this case, both sides of \cref{eq:proof_rankcond} become zero and the inequation holds trivially.
	Similarly, if $r=j=0$, then $T'_i = 0 = j$ on the sets $S_{i,0,f}(0)$, i.e., both sides of \cref{eq:proof_rankcond} become zero, too.
	Therefore, we do not need to consider those generators here.
	So, we can assume $r > 0$ for the generators $S_{i,j,f}(r)$.
	Moreover, if a generator is in fact a null-set, the left-hand side becomes zero and we can omit it as well.
	Furthermore, for any admissible run $\run$ with $\run \in S_{i,j,f}(r)$ we have by definition of $\rank$ that $ \run \in \left(T'_i\right)^{-1}((j,\infty))$.
	So by \cref{lemma:rsm_condexpconstatoms}, it is sufficient to prove
	\[
		\expv{\pip}{\scheduler}{\initstate} \left(\ind_{S_{i,j,f}(r)} \cdot \ind_{T'_i > j} \cdot R_{i,j+1} \right) \leq \expv{\pip}{\scheduler}{\initstate} \left(\ind_{S_{i,j,f}(r)}\cdot \ind_{T'_i > j} \cdot (R_{i,j} - 1)\right) \]
	whenever $\pipmeasure{\pip}{\scheduler}{\initstate} (S_{i,j,f}(r))>0$.
	Moreover, as discussed above we have
	\[
		\ind_{T'_i > j} \cdot \ind_{S_{i,j,f}(r)}
		= \ind_{S_{i,j,f}(r)} \qquad \text{almost surely.} \]
	Note that this equation only holds for admissible runs and hence, it only holds almost surely.

	So the above inequation is equivalent to
	\[
		\expv{\pip}{\scheduler}{\initstate} \left(\ind_{S_{i,j,f}(r)} \cdot R_{i,j+1} \right) \leq \expv{\pip}{\scheduler}{\initstate} \left(\ind_{S_{i,j,f}(r)} \cdot (R_{i,j} - 1)\right) = \pipmeasure{\pip}{\scheduler}{\initstate}(S_{i,j,f}(r)) \cdot (r-1),\]
	since $R_{i,j}$ is constant with value $r$ on $S_{i,j,f}(r)$.
	Note that this means, as $R_{i,j} \geq 1$ on the set $S_{i,j,f}(r) \subseteq (T'_i > j) = \left(T'_i\right)^{-1}((j,\infty))$, that $\pipmeasure{\pip}{\scheduler}{\initstate}(S_{i,j,f}(r)) > 0$ already implies $r \geq 1$.
	That this inequation holds is exactly the statement of \cref{lemma:rank_cond_on_generators}.
\end{proof}

\begin{lemma}
	\label{lemma:rank_cond_on_generators}
	Let $\scheduler$ be an arbitrary scheduler, $\initstate$ an arbitrary initial state, and $i,j \in \NN$.
	Let $f \in \FPATH$ and $r \geq 1$ such that $\pipmeasure{\pip}{\scheduler}{\initstate}\left( S_{i,j,f}(r)\right) > 0$.
	Then we have:

	\[
		\expv{\pip}{\scheduler}{\initstate} (\ind_{S_{i,j,f}(r)} \cdot R_{i,j+1}) \leq \pipmeasure{\pip}{\scheduler}{\initstate}\left( S_{i,j,f}(r)\right) \cdot (r - 1)
	\]
\end{lemma}

\begin{proof}
	Let $f \in \FPATH$ and $r \geq 1$ such that $\pipmeasure{\pip}{\scheduler}{\initstate}\left(S_{i,j,f}(r)\right) > 0$.
	We have
	\begin{align*}
		      & \expv{\pip}{\scheduler}{\initstate} (\ind_{S_{i,j,f}(r)} \cdot R_{i,j+1})                                                \\
		{}={} & \sum_{r'\in\RR_{\geq 0}}
		r' \cdot \pipmeasure{\pip}{\scheduler}{\initstate} (\ind_{S_{i,j,f}(r)} \cdot R_{i,j+1} = r') \tag{by definition of $\expvsign$} \\
	\end{align*}
	Note that the set of positive real numbers $\RR_{\geq 0}$ is uncountable but we have seen that all except countably many addends in the above sum are zero.
	This justifies why we can use the sum for the expected value and can avoid the integral.

	We are interested in the value of $R_{i,j+1}$.
	However, in a run $\run \in S_{i,j,f}(r)$ the position $I^{\succ}_{i,j+1}(\run)$ might occur arbitrarily far after $I^{\succ}_{i,j}(\run)$ (if it occurs at all), whereas in the meantime transitions from $\GTNI\setminus \GTG$ are evaluated.
	To simplify the analysis, our goal is to ignore all runs where this next evaluation happens after more than $m \geq 1$ steps.
	To achieve this, we consider the sets $A_m$:
	\[
		A_m \is S_{i,j,f}(r) \cap (I^{\succ}_{i,j+1})^{-1}([I^{\succ}_{i,j}(f)+1,I^{\succ}_{i,j}(f)+m]).
	\]
	So, $A_m$ contains all runs $\run \in S_{i,j,f}(r)$ such that the configuration with index $I^{\succ}_{i,j+1}(\run)$ occurs within the next $m \geq 1$ evaluation steps after $I^{\succ}_{i,j}(\run)$.
	By definition, we have $A_{m} \subseteq A_{m+1}$ for all $m \geq 1$.
	Therefore, $A_\infty= \lim_{m \to \infty} A_m = \bigcup_{m \geq 1} A_m = S_{i,j,f}(r) \cap (I^{\succ}_{i,j+1})^{-1}([I^{\succ}_{i,j}(f)+1,\infty))$.
	Since $I^{\succ}_{i,j+1}(\run) \neq 0$ already implies $I^{\succ}_{i,j+1}(\run) > I^{\succ}_{i,j}(\run)$ we observe
	\[
		\run \in S_{i,j,f}(r)\setminus A_{\infty} \implies I^{\succ}_{i,j+1}(\run) =0.
	\]
	But $I^{\succ}_{i,j+1}(\run) = 0$ implies $R_{i,j+1}(\run) = 0$, so we have
	\begin{align*}
		\expv{\pip}{\scheduler}{\initstate} (\ind_{S_{i,j,f}(r)} \cdot R_{i,j+1}) {}={} & \expv{\pip}{\scheduler}{\initstate} ((\ind_{A_{\infty}}+\ind_{S_{i,j,f}(r)\setminus A_{\infty}}) \cdot R_{i,j+1})                                                                         \\
		{}={}                                                                           & \expv{\pip}{\scheduler}{\initstate} (\ind_{A_{\infty}} \cdot R_{i,j+1}) + \expv{\pip}{\scheduler}{\initstate} (\underbrace{\ind_{S_{i,j,f}(r)\setminus A_{\infty}} \cdot R_{i,j+1}}_{=0}) \\
		{}={}                                                                           & \expv{\pip}{\scheduler}{\initstate} (\ind_{A_{\infty}} \cdot R_{i,j+1}).                                                                                                                  \\
	\end{align*}
	By continuity of the probability measure we can then conclude

	\begin{align}
		\lim_{m\to\infty} \expv{\pip}{\scheduler}{\initstate} (\ind_{A_m} \cdot R_{i,j+1}) & = \expv{\pip}{\scheduler}{\initstate} (\ind_{A_\infty} \cdot R_{i,j+1}) \nonumber                     \\
		                                                                                   & = \expv{\pip}{\scheduler}{\initstate} (\ind_{S_{i,j,f}(r)} \cdot R_{i,j+1})\label{eq:limit_reasoning}
	\end{align}
	By \cref{lem:ranking_condition_A_m} we have for all $m\geq 1$

	\[
		\expv{\pip}{\scheduler}{\initstate} (\ind_{A_m} \cdot R_{i,j+1}) \leq \pipmeasure{\pip}{\scheduler}{\initstate} \left(S_{i,j,f}(r)\right) \cdot (r-1). \]
	So using \cref{eq:limit_reasoning}, we have therefore proven
	\[
		\expv{\pip}{\scheduler}{\initstate} (\ind_{S_{i,j,f}(r)} \cdot R_{i,j+1}) \leq \pipmeasure{\pip}{\scheduler}{\initstate}(S_{i,j,f}(r)) \cdot (r -1),
	\]
	which concludes the proof.
\end{proof}

\begin{lemma}
	\label{lem:ranking_condition_A_m}
	Let $\scheduler$ be an arbitrary scheduler, $\initstate$ an arbitrary initial state, and $j \in \NN$.
	Let $f \in \FPATH$, $m \in \NN_{\geq 1}$, and $r \geq 1$ such that $\pipmeasure{\pip}{\scheduler}{\initstate}\left( S_{i,j,f}(r)\right) > 0$ and let $A_m \is S_{i,j,f}(r) \cap (I^{\succ}_{i,j+1})^{-1}([I^{\succ}_{i,j}(f)+1,I^{\succ}_{i,j}(f)+m]).$
	Then we have:
	\begin{gather}
		\expv{\pip}{\scheduler}{\initstate} (\ind_{A_m} \cdot R_{i,j+1}) \leq \pipmeasure{\pip}{\scheduler}{\initstate} \left(S_{i,j,f}(r)\right) \cdot (r-1).\label{eq:proofexptimebounds_prA}
	\end{gather}
\end{lemma}

\begin{proof}
	When we reason about measurable sets, the easiest way is to write them as a disjoint union of cylinder sets (if possible).
	For $A_m$ this is indeed possible, since we know that the runs must have the prefix $f$, right after $f$ there must be the position of $I^{\succ}_{i,j}$ and then within at most $m$ steps we must reach the position of $I^{\succ}_{i,j+1}$.
	In other words

	\[
		A_m = \biguplus_{1 \leq m_1 \leq m} \biguplus_{\substack{c_0,c_{m_1} \in \CONF_{\GTG} \\
		c_1,\ldots,c_{m_1-1} \in \CONF_{\GTNI}\setminus \CONF_{\GTG} \\
		f' = \prefix{f c_0 c_1 \cdots c_{m_1}}}} \Pre{f'}
	\]
	Using this decomposition we then get
	\begin{align*}
		      & \expv{\pip}{\scheduler}{\initstate} (\ind_{A_{m}} \cdot R_{i,j+1}) \\
		{}={} & \sum_{m_1 = 1}^m \sum_{\substack{c_0,c_{m_1} \in \CONF_{\GTG}      \\
		c_1,\ldots,c_{m_1-1} \in \CONF_{\GTNI}\setminus \CONF_{\GTG}               \\
		f' = \prefix{f c_0 c_1 \cdots c_{m_1}}}}
		R_{i,j+1} (f') \cdot \ptransition{\pip}{\scheduler}{\initstate} (f') \tag{$\dagger$} \label{eq:important_sum}
		\\
		{}={} & \sum_{m_1 = 1}^{m-1} \sum_{\substack{c_0,c_{m_1} \in \CONF_{\GTG}  \\
		c_1,\ldots,c_{m_1-1} \in \CONF_{\GTNI}\setminus \CONF_{\GTG}               \\
		f' = \prefix{f c_0 c_1 \cdots c_{m_1}}}}
		R_{i,j+1} (f') \cdot \ptransition{\pip}{\scheduler}{\initstate} (f')       \\
		      & + \sum_{\substack{c_0,c_{m} \in \CONF_{\GTG}                       \\
		c_1,\ldots,c_{m-1} \in \CONF_{\GTNI}\setminus \CONF_{\GTG}                 \\
		f' = \prefix{f c_0 c_1 \cdots c_{m}}}}
		R_{i,j+1} (f') \cdot \ptransition{\pip}{\scheduler}{\initstate} (f')       \\
	\end{align*}
	Let us now consider the very last addend:
	\[
		\sum_{\substack{c_0,c_{m} \in \CONF_{\GTG} \\
		c_1,\ldots,c_{m-1} \in \CONF_{\GTNI}\setminus \CONF_{\GTG} \\
		f' = \prefix{f c_0 c_1 \cdots c_{m}}}}
		R_{i,j+1} (f') \cdot \ptransition{\pip}{\scheduler}{\initstate} (f'). \]
	Let $c_k = (\loc_k',t_k',\state_k')$ for every $k \in \{0,\dots,m\}$.

	Note that in the following sums, all addends including a term of the form $\exact{\rank (\loc')}{\state'}$ are non-negative.
	The reason is as follows: Let $f'$ be an \emph{admissible prefix} such that $f' = \prefix{f c_0 c_1\cdots c_{m}}$ with $c_0,c_{m} \in \CONF_{\GTG}$, $c_1,\ldots,c_{m-1} \in \CONF_{\GTNI}\setminus \CONF_{\GTG}$, and $c_k = (\loc'_k,t'_k,\state'_k)$.
	Then, any admissible path ending in one of the $c_0,\ldots,c_{m-1}$ is continued by the scheduler $\scheduler$ with some transition $t \in \bigcup \GTNI$.
	Moreover, from any of the $c_0,\ldots,c_{m-1}$ we can reach the configuration $c_{m}$ with a positive probability.
	By the properties ``Boundedness (b)'' and ``Decrease'', we must have $\exact{\rank (\loc_{m-1}')}{\state_{m-1}'}\geq 0$ (since $c_{m-1}$ is continued with a transition from $\bigcup\GTG$ by the scheduler $\scheduler$).
	Then, by the properties ``Boundedness (a)'' and ``Non-Increase'', we must have $\exact{\rank (\loc_{k}')}{\state_{k}'}\geq 0$ for $0 \leq k \leq m-2$, which can be proven by induction.
	Consequently, whenever for some $0 \leq k \leq m-2$ we have $\ptransition{\pip}{\scheduler}{\initstate}(c_k \to (\loc',t',\state')) > 0$, then $t' \in \bigcup \GTNI$ due to the scheduler $\scheduler$ and hence, $\exact{\rank (\loc')}{\state'}\geq 0$, again by the property ``Boundedness (a)''.
	\begin{align*}
		      & \smashoperator[r]{\sum_{\substack{c_0,c_{m} \in \CONF_{\GTG} \\
		c_1,\ldots,c_{m-1} \in \CONF_{\GTNI}\setminus \CONF_{\GTG}           \\
		f' = \prefix{f c_0 c_1 \cdots c_{m}}}}}
		R_{i,j+1} (f') \cdot \ptransition{\pip}{\scheduler}{\initstate} (f') \\
		{}={} & \smashoperator[r]{\sum_{\substack{c_0,c_{m} \in \CONF_{\GTG} \\
		c_1,\ldots,c_{m-1} \in \CONF_{\GTNI}\setminus \CONF_{\GTG}           \\
		f' = \prefix{f c_0 c_1 \cdots c_{m}}}}}
		\exact{ \rank (\loc_{m-1}')}{\state_{m-1}'}\cdot \ptransition{\pip}{\scheduler}{\initstate}
		(f') \tag{by the definition of $R_{i,{j+1}}$}                        \\
		{}={} & \smashoperator[r]{\sum_{\substack{c_0,c_{m} \in \CONF_{\GTG} \\
		c_1,\ldots,c_{m-1} \in \CONF_{\GTNI}\setminus \CONF_{\GTG} }}}
		\exact{ \rank (\loc_{m-1}')}{\state_{m-1}'}\cdot \ptransition{\pip}{\scheduler}{\initstate} (\prefix{f c_0 c_1 \cdots c_{m-1}}) \cdot \ptransition{\pip}{\scheduler}{\initstate} (c_{m-1} \to c_{m})
	\end{align*}
	Now the idea is to propagate the value of $\rank$ back to the last configuration of the finite path $f$.
	\begin{align*}
		      & \smashoperator[r]{\sum_{\substack{c_0,c_{m} \in \CONF_{\GTG}                                                                                                                                \\
		c_1,\ldots,c_{m-1} \in \CONF_{\GTNI}\setminus \CONF_{\GTG}}}}
		\exact{ \rank (\loc_{m-1}')}{\state_{m-1}'}\cdot \ptransition{\pip}{\scheduler}{\initstate} (\prefix{f c_0 c_1\cdots c_{m-1}}) \cdot \ptransition{\pip}{\scheduler}{\initstate} (c_{m-1} \to c_{m}) \\
		{}={} & \smashoperator[r]{\sum_{\substack{c_0,c_{m} \in \CONF_{\GTG}                                                                                                                                \\
		c_1,\ldots,c_{m-1} \in \CONF_{\GTNI}\setminus \CONF_{\GTG}}}}
		\exact{ \rank (\loc_{m-1}')}{\state_{m-1}'}\cdot \ptransition{\pip}{\scheduler}{\initstate} (\prefix{f c_0 c_1 \cdots c_{m-2}})\cdot                                                                \\&\qquad \phantom{\sum_{\substack{1\\2}}}\ptransition{\pip}{\scheduler}{\initstate} (c_{m-2} \to c_{m-1}) \cdot \ptransition{\pip}{\scheduler}{\initstate} (c_{m-1} \to c_{m}) \\
		{}={} & \smashoperator[r]{\sum_{\substack{c_0 \in \CONF_{\GTG}                                                                                                                                      \\
		c_1,\ldots,c_{m-2} \in \CONF_{\GTNI}\setminus \CONF_{\GTG}}}}
		\ptransition{\pip}{\scheduler}{\initstate} (\prefix{f c_0 c_1 \cdots c_{m-2}}) \cdot                                                                                                                \\
		      & \qquad \smashoperator[r]{\sum_{\substack{c_{m} \in \CONF_{\GTG}                                                                                                                             \\
					c_{m-1} \in \CONF_{\GTNI}\setminus \CONF_{\GTG}}}}
		\exact{ \rank (\loc_{m-1}')}{\state_{m-1}'}
		\cdot \ptransition{\pip}{\scheduler}{\initstate} (c_{m-2} \to c_{m-1}) \cdot \ptransition{\pip}{\scheduler}{\initstate} (c_{m-1} \to c_{m})                                                         \\
		{}={} & \smashoperator[r]{\sum_{\substack{c_0 \in \CONF_{\GTG}                                                                                                                                      \\
		c_1,\ldots,c_{m-2} \in \CONF_{\GTNI}\setminus \CONF_{\GTG}}}}
		\ptransition{\pip}{\scheduler}{\initstate} (\prefix{f c_0 c_1 \cdots c_{m-2}})\cdot                                                                                                                 \\
		      & \qquad \smashoperator[lr]{\sum_{\substack{c_{m-1} \in \CONF_{\GTNI}\setminus \CONF_{\GTG}}}}
		\exact{ \rank (\loc_{m-1}')}{\state_{m-1}'}
		\cdot \ptransition{\pip}{\scheduler}{\initstate} (c_{m-2} \to c_{m-1}) \cdot \left(\sum_{c_{m} \in \CONF_{\GTG} }\ptransition{\pip}{\scheduler}{\initstate} (c_{m-1} \to c_{m})\right)              \\
		{}={} & \smashoperator[r]{\sum_{\substack{c_0 \in \CONF_{\GTG}                                                                                                                                      \\
		c_1,\ldots,c_{m-2} \in \CONF_{\GTNI}\setminus \CONF_{\GTG}}}}
		\ptransition{\pip}{\scheduler}{\initstate} (\prefix{f c_0 c_1 \cdots c_{m-2}})\cdot \smashoperator[lr]{\sum_{\substack{c_{m-1} \in \CONF_{\GTNI}\setminus \CONF_{\GTG}                              \\
					\ptransition{\pip}{\scheduler}{\initstate} (c_{m-1} \to \CONF_{\GTG})>0}}}
		\exact{ \rank (\loc_{m-1}')}{\state_{m-1}'}
		\cdot \ptransition{\pip}{\scheduler}{\initstate} (c_{m-2} \to c_{m-1})
	\end{align*}
	The very last step holds, since from a configuration we either end with probability $1$ in $\CONF_{\GTG}$ or we end there with probability $0$ due to the definition of the scheduler.
	Here, we stored in the summation index-set the fact that we of course only take those configurations into account from which we can continue to a configuration in $\CONF_{\GTG}$ with a positive probability (which by definition directly implies that we do so with probability $1$ as explained before).

	In the sum above, we have $t_i'\in\bigcup\GTNI$ for every addend.
	Using these observations combined with the linearity of the expected value operator as well as the property ``Non-Increase'' of PLRFs one obtains:

	\begin{align*}
		         & \smashoperator[r]{\sum_{\substack{c_0 \in \CONF_{\GTG}                                                                                                      \\
		c_1,\ldots,c_{m-2} \in \CONF_{\GTNI}\setminus \CONF_{\GTG}}}}
		\ptransition{\pip}{\scheduler}{\initstate} (\prefix{f c_0 c_1 \cdots c_{m-2}})\cdot \smashoperator[lr]{\sum_{\substack{c_{m-1} \in \CONF_{\GTNI}\setminus \CONF_{\GTG} \\
					\ptransition{\pip}{\scheduler}{\initstate} (c_{m-1} \to \CONF_{\GTG})>0}}}
		\exact{ \rank (\loc_{m-1}')}{\state_{m-1}'}
		\cdot \ptransition{\pip}{\scheduler}{\initstate} (c_{m-2} \to c_{m-1})                                                                                                 \\
		{}\leq{} & \smashoperator[r]{\sum_{\substack{c_0 \in \CONF_{\GTG}                                                                                                      \\
		c_1,\ldots,c_{m-2} \in \CONF_{\GTNI}\setminus \CONF_{\GTG}}}}
		\ptransition{\pip}{\scheduler}{\initstate} (\prefix{f c_0 c_1 \cdots c_{m-2}})\cdot \smashoperator[lr]{\sum_{\substack{c_{m-1} \in \CONF_{\GTNI}\setminus \CONF_{\GTG}}}}
		\exact{ \rank (\loc_{m-1}')}{\state_{m-1}'}
		\cdot \ptransition{\pip}{\scheduler}{\initstate} (c_{m-2} \to c_{m-1})                                                                                                 \\
		{}\leq{} & \smashoperator[r]{\sum_{\substack{c_0 \in \CONF_{\GTG}                                                                                                      \\
		c_1,\ldots,c_{m-2} \in \CONF_{\GTNI}\setminus \CONF_{\GTG}                                                                                                             \\
		\ptransition{\pip}{\scheduler}{\initstate}(c_{m-2}\to \CONF_{\GTNI}\setminus \CONF_{\GTG}) > 0}}}
		\ptransition{\pip}{\scheduler}{\initstate} (\prefix{f c_0 c_1 \cdots c_{m-2}}) \cdot \exact{\rank (\loc_{m-2}') }{\state_{m-2}'}\tag{Non-Increasing property of $\rank$}
	\end{align*}
	Again, we stored in the summation index-set the fact that we only take those configurations into account from which we can continue to a configuration in $\CONF_{\GTNI}\setminus \CONF_{\GTG}$ with a positive probability (i.e., we continue to $\CONF_{\GTNI}\setminus \CONF_{\GTG}$ with probability $1$).

	Now let us have a look at the addend for $m_1 = m-1$ of \cref{eq:important_sum}.
	Similar to the reasoning for the summand with $m_1 = m$ we have

	\begin{align*}
		         & \sum_{\substack{c_0,c_{m-1} \in \CONF_{\GTG}                                                                                                                                               \\
		c_1,\ldots,c_{m-2} \in \CONF_{\GTNI}\setminus \CONF_{\GTG}                                                                                                                                            \\
		f' = \prefix{f c_0 c_1 \cdots c_{m-1}}}}
		R_{i,j+1} (f') \cdot \ptransition{\pip}{\scheduler}{\initstate} (f')                                                                                                                                  \\
		{}={}    & \smashoperator[r]{\sum_{\substack{c_0,c_{m-1} \in \CONF_{\GTG}                                                                                                                             \\
		c_1,\ldots,c_{m-2} \in \CONF_{\GTNI}\setminus \CONF_{\GTG}}}}
		\exact{ \rank (\loc_{m-2}')}{\state_{m-2}'}\cdot \ptransition{\pip}{\scheduler}{\initstate} (\prefix{f c_0 c_1\cdots c_{m-2}}) \cdot \ptransition{\pip}{\scheduler}{\initstate} (c_{m-2} \to c_{m-1}) \\
		{}\leq{} & \smashoperator[r]{\sum_{\substack{c_0 \in \CONF_{\GTG}                                                                                                                                     \\
		c_1,\ldots,c_{m-2} \in \CONF_{\GTNI}\setminus \CONF_{\GTG}                                                                                                                                            \\
		\ptransition{\pip}{\scheduler}{\initstate} (c_{m-2} \to \CONF_{\GTG})>0}}}
		\exact{ \rank (\loc_{m-2}')}{\state_{m-2}'}\cdot \ptransition{\pip}{\scheduler}{\initstate} (\prefix{f c_0 c_1 \cdots c_{m-2}})                                                                       \\
	\end{align*}

	So, when we combine these two steps we get the following for the sum in \cref{eq:important_sum}.

	\begin{align*}
		         & \expv{\pip}{\scheduler}{\initstate} (\ind_{A_{m}} \cdot R_{i,j+1})                                                   \\
		{}={}    & \sum_{m_1 = 1}^m \sum_{\substack{c_0,c_{m_1} \in \CONF_{\GTG}                                                        \\
		c_1,\ldots,c_{m_1-1} \in \CONF_{\GTNI}\setminus \CONF_{\GTG}                                                                    \\
		f' = \prefix{f c_0 c_1 \cdots c_{m_1}}}}
		R_{i,j+1} (f') \cdot \ptransition{\pip}{\scheduler}{\initstate} (f')                                                            \\
		{}\leq{} & \sum_{m_1 = 1}^{m-2} \sum_{\substack{c_0,c_{m_1} \in \CONF_{\GTG}                                                    \\
		c_1,\ldots,c_{m_1-1} \in \CONF_{\GTNI}\setminus \CONF_{\GTG}                                                                    \\
		f' = \prefix{f c_0 c_1 \cdots c_{m_1}}}}
		R_{i,j+1} (f') \cdot \ptransition{\pip}{\scheduler}{\initstate} (f')                                                            \\
		         & \qquad+ \smashoperator[r]{\sum_{\substack{c_0 \in \CONF_{\GTG}                                                       \\
		c_1,\ldots,c_{m-2} \in \CONF_{\GTNI}\setminus \CONF_{\GTG}}}}
		\exact{ \rank (\loc_{m-2}')}{\state_{m-2}'}\cdot \ptransition{\pip}{\scheduler}{\initstate} (\prefix{f c_0 c_1 \cdots c_{m-2}}) \\
	\end{align*}

	\noindent
	Proceeding in this manner, i.e., iteratively propagating the point where to evaluate $\rank$ back by using the property ``Non-Increase'' of $\rank$, we step by step over-approximate this sum by

	\[\smashoperator[r]{\sum_{\substack{c_0 \in \CONF_{\GTG}}}}
		\ptransition{\pip}{\scheduler}{\initstate} (\prefix{f c_0}) \cdot \exact{\rank (\loc'_{0}) }{\state'_{0}}.\]

	Finally, we can proceed with the crucial step in this part of the proof.
	To do so, let us assume, that the last configuration in $f$ is $c_f = (\loc_f,t_f,\state_f)$.

	\begin{align*}
		         & \smashoperator[r]{\sum_{\substack{c_0 \in \CONF_{\GTG}}}}
		\ptransition{\pip}{\scheduler}{\initstate} (\prefix{f c_0}) \cdot \exact{\rank (\loc'_{0}) }{\state'_{0}}                                                            \\
		{}={}    & \smashoperator[r]{\sum_{\substack{c_0 \in \CONF_{\GTG}}}}
		\ptransition{\pip}{\scheduler}{\initstate} (\prefix{f}) \cdot \ptransition{\pip}{\scheduler}{\initstate} (c_f \to c_0) \cdot \exact{\rank (\loc'_{0}) }{\state'_{0}} \\
		{}={}    & \ptransition{\pip}{\scheduler}{\initstate} (\prefix{f}) \cdot \smashoperator[r]{\sum_{\substack{c_0 \in \CONF_{\GTG}}}}
		\ptransition{\pip}{\scheduler}{\initstate} (c_f \to c_0) \cdot \exact{\rank (\loc'_{0}) }{\state'_{0}}                                                               \\
		{}\leq{} & \ptransition{\pip}{\scheduler}{\initstate} (\prefix{f}) \cdot (\exact{\rank (\loc_f) }{\state_f} -1) \tag{by property ``Decrease'' of $\rank$}            \\
		{}={}    & \ptransition{\pip}{\scheduler}{\initstate} (\prefix{f}) \cdot (r -1) \tag{by \cref{lemma:decomposition}}                                                  \\
		{}={}    & \pipmeasure{\pip}{\scheduler}{\initstate}(S_{i,j,f}(r)) \cdot (r -1).
	\end{align*}
	Let us elaborate on the last equation.
	We assumed that $\pipmeasure{\pip}{\scheduler}{\initstate}(S_{i,j,f}(r))>0$.
	Moreover, we have seen in \cref{lemma:decomposition} that

	\[
		S_{i,j,f}(r) = \biguplus_{c_0 \in \CONF_{\GTG}} \Pre{fc}
	\]
	By definition of the scheduler $\scheduler$ this means that it chooses to continue $f$ with a general transition from $\GTG$ (all other runs are not admissible).
	But this then means that
	\begin{align*}
		      & \pipmeasure{\pip}{\scheduler}{\initstate}(S_{i,j,f}(r))                                         \\
		{}={} & \pipmeasure{\pip}{\scheduler}{\initstate}\left(\biguplus_{c_0 \in \CONF_{\GTG}} \Pre{fc}\right) \\
		{}={} & \pipmeasure{\pip}{\scheduler}{\initstate}\left(\biguplus_{c_0 \in \CONF} \Pre{fc}\right)        \\
		{}={} & \pipmeasure{\pip}{\scheduler}{\initstate}\left(\Pre{f}\right)                                   \\
		{}={} & \ptransition{\pip}{\scheduler}{\initstate} (\prefix{f}).
	\end{align*}
	So, we have shown \cref{eq:proofexptimebounds_prA}.
\end{proof}

\begin{corollary}
	\label{ranking corollary}
	Let $(\RUNS,\mathcal{F},\pmeasure)$ be the probability space of a PIP $\pip$, $i \in \NN$, and $(\CF_{i,j})_{j \in \NN}$ the filtration from \cref{lemma:rsm_validfiltration}.
	Then, $(R_{i,j})_{j \in \NN}$ is a ranking supermartingale for $T_i'$.
\end{corollary}

\begin{proof}
	First of all, we have seen in \cref{lem:adaptedness} that $(R_{i,j})_{j \in \NN}$ is adapted to $(\CF_{i,j})_{j \in \NN}$.
	Secondly, by definition, $R_{i,j}$ takes only non-negative real values.
	Thirdly, we have proven in \cref{lemma:rs_rankcond} that the third condition required in \cref{def:rsm} is satisfied.
	So $(R_{i,j})_{j \in \NN}$ is indeed a ranking supermartingale.
\end{proof}

\begin{theorem}[Ranking Supermartingale Bounds Stopping Time \protect{\textnormal{\cite[Thm.\ 7.2]{lexrsm}}}]
	\label{thm:chatt_lexrsm}
	Let $(\Omega, \mathcal{F}, \pmeasure)$ be a probability space with a filtration $(\mathcal{F}_j)_{j \in \NN}$ and $T$ a stopping time.
	If $(X_j)_{j \in \NN}$ is a ranking supermartingale for $T$ then

	\[\expvsign(T) \leq \expvsign(X_0).
	\]
\end{theorem}

Applying \cref{thm:chatt_lexrsm} to the process $(R_{i,j})_{j \in \NN}$ yields the following result.

\begin{corollary}
	\label{thm:rsm_timebound}
	Let $\scheduler$ be an arbitrary scheduler, $\initstate$ an arbitrary initial state, and $i \in \NN$.
	Then \[\expv{\pip}{\scheduler}{\initstate} (T'_i) \leq \expv{\pip}{\scheduler}{\initstate} (R_{i,0}).
	\]
\end{corollary}

However, so far we have only considered the case of $\GTG$ transitions.
While this is a crucial ingredient of the proof of \cref{theorem:exptimeboundsmeth}, we also need to take the position into account when $\GTNI$ is entered for the first time.
So we define a random variable similar to $R_{i,0}$, which we will call $E_i$ because it refers to the $(i+1)$-th \underline{e}ntry into $\GTNI$.
To ease the definition, we first define the index $I^{e}_{i}$, i.e., the index of the $(i+1)$-th entry into $\GTNI$.

\begin{definition}[Index for Entry Into $\GTNI$]
	For any $i \in \NN$ we define the mapping $I^{e}_{i}\colon \RUNS\rightarrow \NN$ as follows:

	\[
		I^{e}_{i} (\run) = \min \left(I_i(\run)\right)
	\]
	with $\min (\varnothing) = 0$.
	Again, this means that if the desired configuration does not exist, then the index defaults to $0$.
	For $f\in\FPATH$ we define:

	\[I^{e}_{i}(f) =
		\begin{cases}
			m, & \text{if for all } \run \in \Pre{f} \text{ we have } I^{e}_{i} (\run) = m \\
			0, & \text{otherwise}.
		\end{cases}
	\]
\end{definition}

In analogy to the process $R_{i,j}$ we now define the random variable $E_i$ based on $I^{e}_i$.

\begin{definition}[Stochastic Process $(E_{i})_{i \in \NN}$]
	\label{definition:entry_stproc}
	For any $i \in \NN$ we define the map $E_i: \RUNS \to \RR_{\geq 0}$ as follows.
	Let $\run = \prefix{(\initloc,t_0,\initstate)(\loc_1,t_1,\state_1)\cdots} \in \RUNS$.

	\begin{align*}
		E_i (\run) & =
		\begin{cases}
			\max \left\{\exact{\rank (\loc_{I^{e}_i(\run) - 1})}{\state_{I^{e}_i(\run) - 1}}, 0 \right\}, & \text{if } I^{e}_i(\run) > 0 \\
			0,                                                                                            & \text{otherwise}
		\end{cases}
	\end{align*}

	\noindent{}
	For $f\in\FPATH$ and $i\in \NN$ we define:

	\begin{align*}
		E_{i}(f) & =
		\begin{cases}
			r,    & \text{if for all } \run \in \Pre{f} \text{ we have } I^{e}_{i}(\run) > 0 \land E_{i} (\run) = r \\
			\bot, & \text{otherwise}
		\end{cases}
	\end{align*}
\end{definition}

Clearly, $E_i$ is a non-negative random variable.
We aim to prove \cref{lemma:entry_supermartingale} for a yet to define $\sigma$-field $\CG_i$.
Here, the idea is the following: as already mentioned, so far we have only considered the first occurrence of a general transition from $\GTG$ in a part of the run belonging to $\GTNI$.
However, before the first such occurrence there can occur general transitions from $\GTNI\setminus \GTG$ which is handled by the random variable $E_i$.
Then \cref{lemma:entry_supermartingale} implies that $\expv{\pip}{\scheduler}{\initstate}\left(E_i\right)$ is \emph{at least} $\expv{\pip}{\scheduler}{\initstate}\left(R_{i,0}\right)$.
This will be crucial for the proof of \cref{theorem:exptimeboundsmeth}.
Nevertheless, while proving \cref{lemma:entry_supermartingale} is cumbersome, it is still easier than proving the statement $\expv{\pip}{\scheduler}{\initstate}\left(R_{i,0}\right) \leq \expv{\pip}{\scheduler}{\initstate}\left(E_i\right)$ for the expected values directly, as with the conditional expectation we only have to take certain parts of a run into account whereas for the general expected value we would have to take the whole run into account.
\begin{lemma}[Relating $R_{i,0}$ and $E_i$]
	\label{lemma:entry_supermartingale}
	Let $\scheduler$ be an arbitrary scheduler, $\initstate$ an arbitrary initial state, and $i \in \NN$.
	Then
	\[
		\expv{\pip}{\scheduler}{\initstate}\left(R_{i,0} \mid \CG_i\right) \leq E_i.
	\]
\end{lemma}

The construction of $\CG_i$ will be similar to the construction of $\CF_{i,0}$, except that we concentrate on $E_i$ instead of $R_{i,0}$.
So, in analogy to $S_{i,0,f}(r)$ we define the set

\[
	S^{e}_{i,f}(r) = \Pre{f} \cap E_i^{-1}(\{r\}) \cap \left(I^{e}_i\right)^{-1}(\{|f|\}), \]
where $i \in \NN$, $f \in \FPATH$, and $r \geq 0$.
Clearly, there are only countably many non-empty $S^{e}_{i,f}(r)$.
By definition, these sets are also pairwise disjoint.
We make the following observation similar to \cref{lemma:decomposition}.

\begin{corollary}[Observations on $S^{e}_{i,f}$]
	\label{coro:decomposition}
	Let $f \in \FPATH$ with last configuration $(\loc_f,\_,\state_f)$, $r > 0$, and $i,j \in \NN$.
	If $S^{e}_{i,f} (r) \neq \emptyset$ then
	\begin{enumerate}[a)]
		\item $r = \exact{\rank(\loc_f)}{\state_f}$.
		\item
		      \[
			      S^{e}_{i,f} (r) = \biguplus_{c \in \CONF_{\GTNI}} \Pre{fc}.
		      \]
	\end{enumerate}
\end{corollary}

We then give the following definition.

\begin{definition}
	For $i \in \NN$ we define
	\[
		\CG_i = \sigmagen{S^{e}_{i,f}(r) \mid f \in \FPATH, r \geq 0}.
	\]
\end{definition}
By \cref{coro:decomposition} b), we directly obtain the following corollary.

\begin{corollary}
	For every $i \in \NN$ the $\sigma$-field $\CG_i$ is a sub-$\sigma$-field of $\CF$.
\end{corollary}

We can now prove \cref{lemma:entry_supermartingale}.

\begin{proof}[Proof of \cref{lemma:entry_supermartingale}]
	Let $f\in \FPATH$ and $r \geq 0$ such that $\pipmeasure{\pip}{\scheduler}{\initstate}\left(S^{e}_{i,f}(r)\right) > 0$.
	Due to \cref{lem:upper_bound_cond_exp}, it is enough to prove
	\begin{equation}
		\expv{\pip}{\scheduler}{\initstate}\left(R_{i,0} \cdot \ind_{S^{e}_{i,f}(r)}\right) \leq \expv{\pip}{\scheduler}{\initstate}\left(E_i \cdot \ind_{S^{e}_{i,f}(r)}\right). \label{eq:cond_exp_entry}
	\end{equation}
	By definition of $R_{i,0}$, for any admissible run $\run \in \RUNS$ we have $R_{i,0}(\run) = 0$ whenever $I^{\succ}_{i,0}(\run) = 0$, since $0 \not \in I_i(\run)$.
	So to prove \eqref{eq:cond_exp_entry}, it is enough to consider the runs in $S^{e}_{i,f}(r) \cap I_{i,0}^{-1}\left(\NN\setminus\{0\}\right)$ (where $I_{i,0}^{-1}$ abbreviates $(I^{\succ}_{i,0})^{-1}$), i.e., we have
	\[
		\expv{\pip}{\scheduler}{\initstate}\left(R_{i,0} \cdot \ind_{S^{e}_{i,f}(r)}\right) = \expv{\pip}{\scheduler}{\initstate}\left(R_{i,0} \cdot \ind_{S^{e}_{i,f}(r) \cap I_{i,0}^{-1}\left(\NN\setminus\{0\}\right)}\right). \]
	By definition of $S^{e}_{i,f}(r)$ and the simple fact that $I^{e}_i \leq I^{\succ}_{i,0}$ we have $S^{e}_{i,f}(r) \cap I_{i,0}^{-1}\left(\NN\setminus\{0\}\right) = S^{e}_{i,f}(r) \cap I_{i,0}^{-1}\left([|f|,\infty)\right)$.

	So we consider the sets $S^{e}_{i,f}(r) \cap I_{i,0}^{-1}\left([|f|,|f| + m]\right)$ for $m \in \NN$ which satisfy

	\begin{equation}
		\expv{\pip}{\scheduler}{\initstate}\left(R_{i,0} \cdot \ind_{S^{e}_{i,f}(r)\cap I_{i,0}^{-1}\left([|f|,|f| + m]\right)}\right) \leq \expv{\pip}{\scheduler}{\initstate}\left(E_i \cdot \ind_{S^{e}_{i,f}(r)}\right),
	\end{equation}
	as we will prove in \cref{lem:essential_proof_entry_supermartingale}.
	In contrast to \cref{lem:ranking_condition_A_m}, here, we include the case $m = 0$ since $I^{e}_i(\run) = I^{\succ}_{i,0}(\run)$ is indeed possible.

	Since $S^{e}_{i,f}(r) \cap I_{i,0}^{-1}\left([|f|,|f| + m]\right) \subseteq S^{e}_{i,f}(r) \cap I_{i,0}^{-1}\left([|f|,|f| + m + 1]\right)$, we can then conclude

	\begin{align*}
		         & \expv{\pip}{\scheduler}{\initstate}\left(R_{i,0} \cdot \ind_{S^{e}_{i,f}(r)}\right)                                                                \\
		{}={}    & \expv{\pip}{\scheduler}{\initstate}\left(R_{i,0} \cdot \ind_{S^{e}_{i,f}(r) \cap I_{i,0}^{-1}\left(\NN\setminus\{0\}\right)}\right)                \\
		{}={}    & \expv{\pip}{\scheduler}{\initstate}\left(\lim_{m \to \infty} R_{i,0} \cdot \ind_{S^{e}_{i,f}(r)\cap I_{i,0}^{-1}\left([|f|,|f| + m]\right)}\right) \\
		{}={}    & \lim_{m \to \infty} \expv{\pip}{\scheduler}{\initstate}\left(R_{i,0} \cdot \ind_{S^{e}_{i,f}(r)\cap I_{i,0}^{-1}\left([|f|,|f| + m]\right)}\right) \\
		{}\leq{} & \expv{\pip}{\scheduler}{\initstate}\left(E_i \cdot \ind_{S^{e}_{i,f}(r)}\right).
	\end{align*}
	This proves \eqref{eq:cond_exp_entry}.
\end{proof}

\begin{lemma}
	\label{lem:essential_proof_entry_supermartingale}
	Let $\scheduler$ be an arbitrary scheduler, $\initstate$ an arbitrary initial state, and $i\in \NN$.
	Let $f\in \FPATH$ and $r \geq 0$ such that $\pipmeasure{\pip}{\scheduler}{\initstate}\left(S^{e}_{i,f}(r)\right) > 0$.
	Then for any $m \in \NN$ we have
	\[
		\expv{\pip}{\scheduler}{\initstate}\left(R_{i,0} \cdot \ind_{S^{e}_{i,f}(r)\cap I_{i,0}^{-1}\left([|f|,|f| + m]\right)}\right) \leq \expv{\pip}{\scheduler}{\initstate}\left(E_i \cdot \ind_{S^{e}_{i,f}(r)}\right).
	\]
\end{lemma}

\begin{proof}
	This proof is very similar to the proof of \cref{lem:ranking_condition_A_m}.
	However, it is not exactly the same, so we state it here for the sake of completeness.
	We have
	\[
		S^{e}_{i,f}(r)\cap I_{i,0}^{-1}\left([|f|,|f| + m]\right) = \biguplus_{0 \leq m_1 \leq m} \biguplus_{\substack{c_{m_1} \in \CONF_{\GTG} \\
		c_0, c_1,\ldots,c_{m_1-1} \in \CONF_{\GTNI}\setminus \CONF_{\GTG} \\
		f' = \prefix{f c_0 c_1 \cdots c_{m_1}}}} \Pre{f'}
	\]
	Using this decomposition we then get

	\begin{align*}
		      & \expv{\pip}{\scheduler}{\initstate} (\ind_{S^{e}_{i,f}(r)\cap I_{i,0}^{-1}\left([|f|,|f| + m]\right)} \cdot R_{i,0}) \\
		{}={} & \sum_{m_1 = 0}^m \sum_{\substack{c_{m_1} \in \CONF_{\GTG}                                                            \\
		c_0,c_1,\ldots,c_{m_1-1} \in \CONF_{\GTNI}\setminus \CONF_{\GTG}                                                             \\
		f' = \prefix{f c_0 c_1 \cdots c_{m_1}}}}
		R_{i,0} (f') \cdot \ptransition{\pip}{\scheduler}{\initstate} (f') \tag{$\dagger$} \label{eq:important_sum_2}
		\\
		{}={} & \sum_{m_1 = 0}^{m-1} \sum_{\substack{c_{m_1} \in \CONF_{\GTG}                                                        \\
		c_0,c_1,\ldots,c_{m_1-1} \in \CONF_{\GTNI}\setminus \CONF_{\GTG}                                                             \\
		f' = \prefix{f c_0 c_1 \cdots c_{m_1}}}}
		R_{i,0} (f') \cdot \ptransition{\pip}{\scheduler}{\initstate} (f')                                                           \\
		      & + \sum_{\substack{c_{m} \in \CONF_{\GTG}                                                                             \\
		c_0,c_1,\ldots,c_{m-1} \in \CONF_{\GTNI}\setminus \CONF_{\GTG}                                                               \\
		f' = \prefix{f c_0 c_1 \cdots c_{m}}}}
		R_{i,0} (f') \cdot \ptransition{\pip}{\scheduler}{\initstate} (f')
	\end{align*}

	Let us now consider the very last addend:
	\[
		\sum_{\substack{c_{m} \in \CONF_{\GTG} \\
		c_0,c_1,\ldots,c_{m-1} \in \CONF_{\GTNI}\setminus \CONF_{\GTG} \\
		f' = \prefix{f c_0 c_1 \cdots c_{m}}}}
		R_{i,0} (f') \cdot \ptransition{\pip}{\scheduler}{\initstate} (f'). \]
	Let $c_k = (\loc_k',t_k',\state_k')$ for every $k \in \{0,\dots,m\}$.

	Again, in the following sums, all addends including a term of the form $\exact{\rank (\loc')}{\state'}$ are non-negative.
	The reason is as follows: Let $f'$ be an \emph{admissible prefix} such that $f' = \prefix{f c_0 c_1\cdots c_{m}}$ with $c_{m} \in \CONF_{\GTG}$, $c_0,\ldots,c_{m-1} \in \CONF_{\GTNI}\setminus \CONF_{\GTG}$, and $c_k = (\loc'_k,t'_k,\state'_k)$.
	Then, any admissible path ending in one of the $c_0,\ldots,c_{m-1}$ is continued by the scheduler $\scheduler$ with some transition $t \in \bigcup \GTNI$.
	Moreover, from any of the $c_0,\ldots,c_{m-1}$ we can reach the configuration $c_{m}$ with a positive probability.
	By the properties ``Boundedness (b)'' and ``Decrease'', we must have $\exact{\rank (\loc_{m-1}')}{\state_{m-1}'}\geq 0$ (since $c_{m-1}$ is continued with a transition from $\bigcup\GTG$ by the scheduler $\scheduler$).
	Then, by the properties ``Boundedness (a)'' and ``Non-Increase'', we must have $\exact{\rank (\loc_{k}')}{\state_{k}'}\geq 0$ for all $0 \leq k \leq m-2$, which can be proven by induction.
	Consequently, whenever for some $0 \leq k \leq m-2$ we have $\ptransition{\pip}{\scheduler}{\initstate}(c_k \to (\loc',t',\state')) > 0$, then $t' \in \bigcup \GTNI$ due to the scheduler $\scheduler$ and hence, $\exact{\rank (\loc')}{\state'}\geq 0$, again by the property ``Boundedness (a)''.
	\begin{align*}
		      & \smashoperator[r]{\sum_{\substack{c_{m} \in \CONF_{\GTG}                                                                                                                                     \\
		c_0,c_1,\ldots,c_{m-1} \in \CONF_{\GTNI}\setminus \CONF_{\GTG}                                                                                                                                       \\
		f' = \prefix{f c_0 c_1 \cdots c_{m}}}}}
		R_{i,0} (f') \cdot \ptransition{\pip}{\scheduler}{\initstate} (f')                                                                                                                                   \\
		{}={} & \smashoperator[r]{\sum_{\substack{c_{m} \in \CONF_{\GTG}                                                                                                                                     \\
		c_0,c_1,\ldots,c_{m-1} \in \CONF_{\GTNI}\setminus \CONF_{\GTG}                                                                                                                                       \\
		f' = \prefix{f c_0 c_1 \cdots c_{m}}}}}
		\exact{ \rank (\loc_{m-1}')}{\state_{m-1}'}\cdot \ptransition{\pip}{\scheduler}{\initstate}
		(f') \tag{by the definition of $R_{i,0}$}                                                                                                                                                            \\
		{}={} & \smashoperator[r]{\sum_{\substack{c_{m} \in \CONF_{\GTG}                                                                                                                                     \\
		c_0,c_1,\ldots,c_{m-1} \in \CONF_{\GTNI}\setminus \CONF_{\GTG} }}}
		\exact{ \rank (\loc_{m-1}')}{\state_{m-1}'}\cdot \ptransition{\pip}{\scheduler}{\initstate} (\prefix{f c_0 c_1 \cdots c_{m-1}}) \cdot \ptransition{\pip}{\scheduler}{\initstate} (c_{m-1} \to c_{m}) \\
	\end{align*}
	Now the idea is to propagate the value of $\rank$ back to the last configuration of the finite path $f$.
	\begin{align*}
		      & \smashoperator[r]{\sum_{\substack{c_m \in \CONF_{\GTG}                                                                                                                                      \\
		c_0,c_1,\ldots,c_{m-1} \in \CONF_{\GTNI}\setminus \CONF_{\GTG}}}}
		\exact{ \rank (\loc_{m-1}')}{\state_{m-1}'}\cdot \ptransition{\pip}{\scheduler}{\initstate} (\prefix{f c_0 c_1\cdots c_{m-1}}) \cdot \ptransition{\pip}{\scheduler}{\initstate} (c_{m-1} \to c_{m}) \\
		{}={} & \smashoperator[r]{\sum_{\substack{c_{m} \in \CONF_{\GTG}                                                                                                                                    \\
		c_0,c_1,\ldots,c_{m-1} \in \CONF_{\GTNI}\setminus \CONF_{\GTG}}}}
		\exact{ \rank (\loc_{m-1}')}{\state_{m-1}'}\cdot \ptransition{\pip}{\scheduler}{\initstate} (\prefix{f c_0 c_1 \cdots c_{m-2}}) \cdot                                                               \\
		      & \qquad\phantom{\sum{\substack{1                                                                                                                                                             \\2}}}\ptransition{\pip}{\scheduler}{\initstate} (c_{m-2} \to c_{m-1}) \cdot \ptransition{\pip}{\scheduler}{\initstate} (c_{m-1} \to c_{m}) \\
		{}={} & \smashoperator[r]{\sum_{c_0, c_1,\ldots,c_{m-2} \in \CONF_{\GTNI}\setminus \CONF_{\GTG}}}
		\ptransition{\pip}{\scheduler}{\initstate} (\prefix{f c_0 c_1 \cdots c_{m-2}}) \cdot                                                                                                                \\
		      & \qquad \smashoperator[r]{\sum_{\substack{c_{m} \in \CONF_{\GTG}                                                                                                                             \\
					c_{m-1} \in \CONF_{\GTNI}\setminus \CONF_{\GTG}}}}
		\exact{ \rank (\loc_{m-1}')}{\state_{m-1}'}
		\cdot \ptransition{\pip}{\scheduler}{\initstate} (c_{m-2} \to c_{m-1}) \cdot \ptransition{\pip}{\scheduler}{\initstate} (c_{m-1} \to c_{m})                                                         \\
		{}={} & \smashoperator[r]{\sum_{c_0,c_1,\ldots,c_{m-2} \in \CONF_{\GTNI}\setminus \CONF_{\GTG}}}
		\ptransition{\pip}{\scheduler}{\initstate} (\prefix{f c_0 c_1 \cdots c_{m-2}})\cdot                                                                                                                 \\
		      & \qquad \smashoperator[r]{\sum_{\substack{c_{m-1} \in \CONF_{\GTNI}\setminus \CONF_{\GTG}}}}
		\exact{ \rank (\loc_{m-1}')}{\state_{m-1}'}
		\cdot \ptransition{\pip}{\scheduler}{\initstate} (c_{m-2} \to c_{m-1}) \cdot                                                                                                                        \\&\qquad\left(\sum_{c_{m} \in \CONF_{\GTG} }\ptransition{\pip}{\scheduler}{\initstate} (c_{m-1} \to c_{m})\right) \\
		{}={} & \smashoperator[r]{\sum_{c_0,c_1,\ldots,c_{m-2} \in \CONF_{\GTNI}\setminus \CONF_{\GTG}}}
		\ptransition{\pip}{\scheduler}{\initstate} (\prefix{f c_0 c_1 \cdots c_{m-2}})\cdot                                                                                                                 \\&\qquad\smashoperator[r]{\sum_{\substack{c_{m-1} \in \CONF_{\GTNI}\setminus \CONF_{\GTG} \\
					\ptransition{\pip}{\scheduler}{\initstate} (c_{m-1} \to \CONF_{\GTG})>0}}}
		\exact{ \rank (\loc_{m-1}')}{\state_{m-1}'}
		\cdot \ptransition{\pip}{\scheduler}{\initstate} (c_{m-2} \to c_{m-1})
	\end{align*}
	Again, the very last step holds, since from a configuration we either end with probability $1$ in $\CONF_{\GTG}$ or we end there with probability $0$ due to the definition of the scheduler.
	Here, we stored in the summation index-set the fact we of course only take those configurations into account from which we can continue to a configuration in $\CONF_{\GTG}$ with a positive probability (which by definition directly implies that we do so with probability $1$ as explained before).
	We then obtain:

	\begin{align*}
		         & \smashoperator[r]{\sum_{c_0,c_1,\ldots,c_{m-2} \in \CONF_{\GTNI}\setminus \CONF_{\GTG}}}
		\ptransition{\pip}{\scheduler}{\initstate} (\prefix{f c_0 c_1 \cdots c_{m-2}})\cdot \smashoperator[lr]{\sum_{\substack{c_{m-1} \in \CONF_{\GTNI}\setminus \CONF_{\GTG} \\
					\ptransition{\pip}{\scheduler}{\initstate} (c_{m-1} \to \CONF_{\GTG})>0}}}
		\exact{ \rank (\loc_{m-1}')}{\state_{m-1}'}
		\cdot \ptransition{\pip}{\scheduler}{\initstate} (c_{m-2} \to c_{m-1})                                                                                                 \\
		{}\leq{} & \smashoperator[r]{\sum_{c_0,c_1,\ldots,c_{m-2} \in \CONF_{\GTNI}\setminus \CONF_{\GTG}}}
		\ptransition{\pip}{\scheduler}{\initstate} (\prefix{f c_0 c_1 \cdots c_{m-2}})\cdot \smashoperator[lr]{\sum_{\substack{c_{m-1} \in \CONF_{\GTNI}\setminus \CONF_{\GTG}}}}
		\exact{ \rank (\loc_{m-1}')}{\state_{m-1}'}
		\cdot \ptransition{\pip}{\scheduler}{\initstate} (c_{m-2} \to c_{m-1})                                                                                                 \\
		{}\leq{} & \smashoperator[r]{\sum_{\substack{c_0,c_1,\ldots,c_{m-2} \in \CONF_{\GTNI}\setminus \CONF_{\GTG}                                                            \\
		\ptransition{\pip}{\scheduler}{\initstate}(c_{m-2}\to \CONF_{\GTNI}\setminus \CONF_{\GTG}) > 0}}}
		\ptransition{\pip}{\scheduler}{\initstate} (\prefix{f c_0 c_1 \cdots c_{m-2}}) \cdot \exact{\rank (\loc_{m-2}') }{\state_{m-2}'}\tag{Non-Increasing property of $\rank$}
	\end{align*}
	Again, we stored in the summation index-set the fact that we only take those configurations into account from which we can continue to a configuration in $\CONF_{\GTNI}\setminus \CONF_{\GTG}$ with a positive probability.

	Now let us have a look at the addend for $m_1 = m-1$ of \cref{eq:important_sum_2}.
	Similar to the reasoning for the summand with $m_1 = m$ we have

	\begin{align*}
		         & \sum_{\substack{c_{m-1} \in \CONF_{\GTG}                                                                                                                                                   \\
		c_0,c_1,\ldots,c_{m-2} \in \CONF_{\GTNI}\setminus \CONF_{\GTG}                                                                                                                                        \\
		f' = \prefix{f c_0 c_1 \cdots c_{m-1}}}}
		R_{i,0} (f') \cdot \ptransition{\pip}{\scheduler}{\initstate} (f')                                                                                                                                    \\
		{}={}    & \smashoperator[r]{\sum_{\substack{c_{m-1} \in \CONF_{\GTG}                                                                                                                                 \\
		c_0,c_1,\ldots,c_{m-2} \in \CONF_{\GTNI}\setminus \CONF_{\GTG}}}}
		\exact{ \rank (\loc_{m-2}')}{\state_{m-2}'}\cdot \ptransition{\pip}{\scheduler}{\initstate} (\prefix{f c_0 c_1\cdots c_{m-2}}) \cdot \ptransition{\pip}{\scheduler}{\initstate} (c_{m-2} \to c_{m-1}) \\
		{}\leq{} & \smashoperator[r]{\sum_{\substack{c_0,c_1,\ldots,c_{m-2} \in \CONF_{\GTNI}\setminus \CONF_{\GTG}                                                                                           \\
		\ptransition{\pip}{\scheduler}{\initstate} (c_{m-2} \to \CONF_{\GTG})>0}}}
		\exact{ \rank (\loc_{m-2}')}{\state_{m-2}'}\cdot \ptransition{\pip}{\scheduler}{\initstate} (\prefix{f c_0 c_1 \cdots c_{m-2}})                                                                       \\
	\end{align*}

	So, when we combine these two steps we get the following for the sum in \cref{eq:important_sum_2}.

	\begin{align*}
		         & \expv{\pip}{\scheduler}{\initstate} (\ind_{S^{e}_{i,f}(r)\cap I_{i,0}^{-1}\left([|f|,|f| + m]\right)} \cdot R_{i,0}) \\
		{}={}    & \sum_{m_1 = 1}^m \sum_{\substack{c_{m_1} \in \CONF_{\GTG}                                                            \\
		c_0,c_1,\ldots,c_{m_1-1} \in \CONF_{\GTNI}\setminus \CONF_{\GTG}                                                                \\
		f' = \prefix{f c_0 c_1 \cdots c_{m_1}}}}
		R_{i,0} (f') \cdot \ptransition{\pip}{\scheduler}{\initstate} (f')                                                              \\
		{}\leq{} & \sum_{m_1 = 1}^{m-2} \sum_{\substack{c_{m_1} \in \CONF_{\GTG}                                                        \\
		c_0,c_1,\ldots,c_{m_1-1} \in \CONF_{\GTNI}\setminus \CONF_{\GTG}                                                                \\
		f' = \prefix{f c_0 c_1 \cdots c_{m_1}}}}
		R_{i,0} (f') \cdot \ptransition{\pip}{\scheduler}{\initstate} (f')                                                              \\
		         & \qquad+ \smashoperator[r]{\sum_{c_0,c_1,\ldots,c_{m-2} \in \CONF_{\GTNI}\setminus \CONF_{\GTG}}}
		\exact{ \rank (\loc_{m-2}')}{\state_{m-2}'}\cdot \ptransition{\pip}{\scheduler}{\initstate} (\prefix{f c_0 c_1 \cdots c_{m-2}}) \\
	\end{align*}

	\noindent
	Proceeding in this manner, i.e., iteratively propagating the point where to evaluate $\rank$ back by using the property ``Non-Increase'' of $\rank$, we step by step over-approximate this sum by

	\[\smashoperator[r]{\sum_{\substack{c_0 \in \CONF_{\GTNI}}}}
		\ptransition{\pip}{\scheduler}{\initstate} (\prefix{f c_0}) \cdot \exact{\rank (\loc'_{0}) }{\state'_{0}}.\]

	Finally, we can proceed with the crucial step in this part of the proof.
	To do so, let us assume, that the last configuration in $f$ is $c_f = (\loc_f,t_f,\state_f)$.

	\begin{align*}
		         & \smashoperator[r]{\sum_{\substack{c_0 \in \CONF_{\GTNI}}}}
		\ptransition{\pip}{\scheduler}{\initstate} (\prefix{f c_0}) \cdot \exact{\rank (\loc'_{0}) }{\state'_{0}}                                                            \\
		{}={}    & \smashoperator[r]{\sum_{\substack{c_0 \in \CONF_{\GTNI}}}}
		\ptransition{\pip}{\scheduler}{\initstate} (\prefix{f}) \cdot \ptransition{\pip}{\scheduler}{\initstate} (c_f \to c_0) \cdot \exact{\rank (\loc'_{0}) }{\state'_{0}} \\
		{}={}    & \ptransition{\pip}{\scheduler}{\initstate} (\prefix{f}) \cdot \smashoperator[r]{\sum_{\substack{c_0 \in \CONF_{\GTNI}}}}
		\ptransition{\pip}{\scheduler}{\initstate} (c_f \to c_0) \cdot \exact{\rank (\loc'_{0}) }{\state'_{0}}                                                               \\
		{}\leq{} & \ptransition{\pip}{\scheduler}{\initstate} (\prefix{f}) \cdot (\exact{\rank (\loc_f) }{\state_f}) \tag{by property ``Non-Increase'' of $\rank$}           \\
		{}={}    & \ptransition{\pip}{\scheduler}{\initstate} (\prefix{f}) \cdot r \tag{by \cref{coro:decomposition}}                                                        \\
		{}={}    & \pipmeasure{\pip}{\scheduler}{\initstate}(S^{e}_{i,f}(r)) \cdot r \tag{see proof of \cref{lem:ranking_condition_A_m}}                                     \\
		{}={}    & \expv{\pip}{\scheduler}{\initstate}\left(E_i\cdot \ind_{S^{e}_{i,f}(r)}\right).
	\end{align*}
	This proves the statement of the lemma.
\end{proof}
Finally, we have all information at hand to prove \cref{theorem:exptimeboundsmeth}.
This proof will combine the results we have obtained on the connection between $R_{i,0}$ and $T'_i$ (\cref{thm:rsm_timebound}) and the just proved connection between $R_{i,0}$ and $E_i$ (\cref{lemma:entry_supermartingale}).

In fact, we will even prove the soundness of a refinement of \cref{theorem:exptimeboundsmeth} which allows us to use \emph{expected}
instead of \emph{non-probabilistic} runtime bounds on the entry transitions of $\GTNI$ in special cases.
If $\rank(\loc)$ is \emph{constant} then its value is \emph{independent} of the values of the program variables at the location $\loc$.
In this special case, for every $h \in \ET_{\GTNI}(\loc)$ instead of the non-probabilistic runtime bounds $\tbound(t)$ for $t \in h$, we can use the \emph{expected}
runtime bound $\tbounde (h)$.
To this end we consider a partition $\{\ENTRYLOC_{\mathrm{c}, \GTNI, \rank}, \ENTRYLOC_{\mathrm{nc}, \GTNI, \rank}\} $ of the entry locations $\ENTRYLOC_\GTNI$.
\begin{definition}[Constant \& Non-Constant Entry Locations]
	Let $\rank$ be a PLRF and $\GTNI \subseteq \GT$.
	Then we define the set of constant (resp.\ non-constant) entry locations $\ENTRYLOC_{\mathrm{c}, \GTNI, \rank}$ (resp.\ $\ENTRYLOC_{\mathrm{nc},\GTNI,\rank}$) by
	\begin{align}
		\ENTRYLOC_{\mathrm{c}, \GTNI, \rank}
		 & = \{\loc \mid \loc \in \ENTRYLOC_{\GTNI}, \rank(\loc) \in \RR\} \quad \text{ and
		}                                                                                   \\
		\ENTRYLOC_{\mathrm{nc}, \GTNI, \rank}
		 & = \{\loc \mid \loc \in \ENTRYLOC_{\GTNI}, \rank(\loc) \not\in \RR\}.
	\end{align}
	If the ranking function $\rank$ is clear from the context we omit it in the index.
\end{definition}

This special treatment of locations $\loc$ where $\rank(\loc)$ is constant can allow the propagation of already computed expected runtime bounds.
For example, let $\GTNI=\GTG=\{g\}$ where $g = \{t\}$ with $t = (\loc,\_,\_,\_,\loc')$ and $\loc \neq \loc'$.
Then a trivial ranking function $\rank$ with $\rank(\loc) = 1$, and $\rank (\loc') = 0$ allows us to compute an expected time bound for $g$ by adding the expected time bounds of $\GTNI$'s entry transitions.

Nevertheless, constant ranking functions are only of limited use.
For example, if $g$ has the same start and target location one can\emph{not} decrease $g$ with a constant ranking function.
Similarly, if one uses constant ranking functions for the transitions $g$ and $g'$, where $g$ is an entry transition w.r.t.
$\GTNI'=\{g'\}$ and $g'$ is an entry transition w.r.t.
$\GTNI=\{g\}$, then this propagation of expected bounds cannot turn infinite expected runtime bounds for $g$ and $g'$ into finite ones.

\begin{theorem}[Expected Time Bounds]
	\label{Expected Time Bounds Refined}
	Let $(\tbounde,\sbounde)$ be an expected bound pair, $\tbound$ a (non-probabilistic) runtime bound, and $\rank$ a PLRF for some $\GTG \subseteq \GTNI \subseteq \GT$.
	Then $\tbounde'\colon \GT \rightarrow \BOUND$ is an expected runtime bound where
	\[
		\tbounde' (g) =
		\begin{cases}
			\tbounde (g), & \text{if } g \not \in \GTG                              \\
			\sum\limits_{\substack{\loc \in \ENTRYLOC_{\mathrm{nc},\GTNI}           \\h \in \ET_{\GTNI}(\loc)}} \left(\sum\limits_{t=(\_,\_,\_,\_,\loc) \in h} \tbound(t)\right)\cdot \left(\exact{\overapprox{\rank(\loc)}}{\sbounde(h,\loc,\cdot)}\right)\\
			\quad \quad+ \sum\limits_{\substack{\loc\in\ENTRYLOC_{\mathrm{c},\GTNI} \\
					h\in\ET_\GTNI (\loc)}} \tbounde (h) \cdot \overapprox{\rank(\loc)}
			,             & \text{if } g \in \GTG
		\end{cases}
	\]
\end{theorem}
\begin{proof}
	Let $\sbounde$, $\tbound$, $\tbounde$, and $\rank$ be according to the theorem.
	For $g\not\in\GTG$ we have $\tbounde' (g) = \tbounde (g)$ and since $\tbounde$ is required to be a valid expected time bound, $\tbounde'$ itself is a valid expected time bound for $g\not\in\GTG$.
	So, assume $g\in\GTG$ and fix some $\initstate \in \STATE$.

	Moreover, consider

	\[
		n = \eval{\sum_{\substack{\loc \in \ENTRYLOC_{\mathrm{nc},\GTNI} \\h \in \ET_{\GTNI}(\loc)}} \sum_{t = (\_,\_,\_,\_,\loc)\in h} \tbound (t)}{\initstate} \in \natclosure. \]
	If $n = \infty$, then, as the sum has only finitely many addends, there has to be some $\loc \in \ENTRYLOC_{\mathrm{nc},\GTNI}$, some $h \in \ET_{\GTNI}(\loc)$, and some $t' = (\_,\_,\_,\_,\loc)\in h$ with $\tbound(t')=\infty$.
	If $\eval{\tbounde'(g)}{\initstate}=\infty$, we infer a valid expected time bound.
	Thus, let us consider the case where $\eval{\tbounde'(g)}{\initstate} < \infty$.
	This is only possible if $\eval{\exact{\overapprox{\rank(\loc)}}{\sbounde(h,\loc,\cdot)}}{\initstate} = 0$.
	Since $\rank(\loc)$ is a non-constant linear polynomial with real coefficients in the program variables (as $\loc \in \ENTRYLOC_{\mathrm{nc},\GTNI}$), by definition of $\overapprox{\cdot}$ we can only have $\eval{\exact{\overapprox{\rank(\loc)}}{\sbounde(h,\loc,\cdot)}}{\initstate}
		= 0$ if $\rank(\loc) = \sum_{x \in \PV} r_x \cdot x$ for real numbers $r_x \in \RR$ for $x \in \PV$ (i.e., the constant coefficient is zero), such that $r_x \neq 0$ implies $\eval{\sbounde(h,\loc,x)}{\initstate} = 0$.
	By definition of $\sbounde$ and the random variable $\sizervar(h,\loc,x)$, this means that in any admissible run for all $x \in \PV$ with $r_x \neq 0$, the value of $x$ is zero after entering $\loc$ with the general transition $h$ where the program was started in the initial state $\initstate$.
	Thus, in any admissible run the following holds: if a part belonging to $\GTNI$ is entered via the location $\loc$ by $h$, then, since $\rank(\loc)$ evaluates to zero in this configuration, no transition from $\GTG$ occurs in this part (as to execute such a general transition, $\rank$ must evaluate to a positive value, but only non-increasing transitions precede the one from $\GTG$).
	Thus, if $\eval{\tbounde'(g)}{\initstate} \neq \infty$ but $n = \infty$, then the probability is zero that $\GTNI$ is entered infinitely often and a general transition from $\GTG$ is executed in each of these infinite parts belonging to $\GTNI$.
	Similarly, we can assume $\eval{\sum_{\loc\in\ENTRYLOC_{\mathrm{c},\GTNI}}\sum_{h\in \ET_\GTNI(\loc)}\tbounde(h)}{\initstate} < \infty$.
	Combining these two assumptions implies that the set of all runs entering $\GTNI$ \emph{infinitely} often such that the value of $E_i$ (i.e., the value of the ranking function when entering $\GTNI$) is positive for infinitely many $i$ has probability $0$.
	Therefore, we will not consider those runs in the following, i.e., in the following for an admissible run there exist only finitely many $i$ such that $E_i$ is positive.
	Note that if $E_i = 0$ for some visit to $\GTNI$, then in this visit of $\GTNI$ \emph{no} general transition from $\GTG$ can be executed.

	Now we will show that for every $g \in \GTG$ we have
	\begin{align}
		\label{eq:soundness_timebound}
		         & \expv{\pip}{\scheduler}{\initstate} (\timervar(g))\nonumber                    \\
		{}\leq{} & \eval{
		\sum_{\substack{\loc \in \ENTRYLOC_{\mathrm{nc},\GTNI}                                    \\h \in \ET_\GTNI(\loc)}} \left(\sum_{t=(\_,\_,\_,\_,\loc) \in h} \tbound(t)\right) \cdot \left( \exact{\overapprox{\rank(\loc)}}{ \sbounde(h,\loc,\cdot)}\right)\right.\nonumber\\
		         & \quad \quad \quad + \left.\sum_{\substack{\loc\in\ENTRYLOC_{\mathrm{c}, \GTNI} \\
					h \in\ET_\GTNI(\loc)}} \tbounde(h) \cdot \overapprox{\rank(\loc)}}{\initstate}
	\end{align}
	To ease notation we define $\T^{in}_{\GTNI}
		= \bigcup_{\loc\in \ENTRYLOC_\GTNI} \bigcup_{h \in \ET_{\GTNI}(\loc)}
		h \subseteq \T$.
	To prove \eqref{eq:soundness_timebound}, we will partition the set $\RUNS$ such that for each set in the partition we have all information at hand to prove the desired statement.
	The crucial information is in which of the visits of $\GTNI$ the value of $E_i$ was positive, in which of these visits at least one transition from $\GTG$ was used, and which transition was used to enter $\GTNI$ in this case.
	We store this information in sets of the form $B_{N,M,(t_i)_{i \in N}}$, where $M \subseteq N \subseteq \NN$ are \emph{finite}, and $t_i \in \T^{in}_{\GTNI}$ for all $i \in N$.
	If $N = \emptyset$, then $M = \emptyset$ and by default $(t_i)_{i \in N}=()$.
	In this case, we define
	\[
		B_{N,M,(t_i)_{i \in N}} = B_{\emptyset,\emptyset,()} = \{\run \in \RUNS \mid E_i(\run) = 0 \text{ for all } i \in \NN\}. \]
	For example, this set contains all admissible runs that never visit $\GTNI$.

	For $N \neq \emptyset$ where $M \subseteq N \subseteq \NN$ and $t_i \in \T^{in}_{\GTNI}$ for $i \in N$, we define
	\begin{align*}
		B_{N,M,(t_i)_{i \in N}} = \bigl \{\run \in \RUNS \mid \quad & E_i(\run)> 0 \land \run[I^{e}_i(\run)-1] = (\_,t_i,\_) \text{ for } i \in N \\
		{}\land{}                                                   & E_i(\run) = 0 \text{ for } i \not \in N                                     \\
		{}\land{}                                                   & I^{\succ}_{i,0}(\run) > 0 \iff i \in M \bigr\},
	\end{align*}
	i.e., this set consists of all admissible runs where the visits to $\GTNI$ where $E_i$ evaluates to a positive value are exactly the ones with $i \in N$, the transition used directly \emph{before} the $(i+1)$-th visit is $t_i$, and the indices in $M$ are exactly those visits in which $\GTG$ is executed at least once.
	In total, we have \emph{countably} many sets $B_{N,M,(t_i)_{i \in N}}$, since $N$ and $M$ are from the countable set of \emph{finite} subsets of $\NN$, and $(t_i)_{i \in N}$, when interpreted as a finite sequence, is from the countable set $\T^{*}$.
	We have
	\begin{enumerate}
		\item $B_{N,M,(t_i)_{i \in N}}$ is measurable for all $M \subseteq N \subseteq \NN$ and $t_i \in \T^{in}_{\GTNI}$ for all $i \in N$.
		\item The sets $B_{N,M,(t_i)_{i \in N}}$ are pairwise disjoint.
		\item \label{it:covering}
		      \[\biguplus_{\substack{M \subseteq N \subseteq \NN \\
					      |N| < \infty\\
					      t_i \in \T^{in}_{\GTNI} \text{ for all } i \in N}} B_{N,M,(t_i)_{i \in N}}= \RUNS.
		      \]
	\end{enumerate}
	Here, \cref{it:covering}.\ holds only almost surely, i.e., there are runs which are not contained in any of the $B_{N,M,(t_i)_{i \in N}}$ but these runs are \emph {not} admissible under the assumption that $E_i$ only evaluates to a positive value finitely often.

	Let us fix some $N,M,(t_i)_{i \in N}$.
	Moreover, we assume $t_i=(\_,\_,\_,\_,\loc_i)$ with $\loc_i \in \ENTRYLOC_\GTNI$ and let $g_i$ be the unique general transition with $t_i \in g_i$.
	We then have
	\begin{align*}
		         & \expv{\pip}{\scheduler}{\initstate} \left(\ind_{B_{N,M,(t_i)_{i \in N}}} \cdot \sum_{i \in N} E_i\right)                                                                                                    \\
		{}={}
		         & \expv{\pip}{\scheduler}{\initstate} \left(\ind_{B_{N,M,(t_i)_{i \in N}}} \cdot \sum_{i \in N} \max \{\exact{\rank(\loc_{I^{e}_i-1})}{\state_{I^{e}_i-1}},0\}\right) \tag{by \cref{definition:entry_stproc}} \\
		{}\leq{} & \expv{\pip}{\scheduler}{\initstate} \left(\ind_{B_{N,M,(t_i)_{i \in N}}} \cdot\sum_{i \in N} \eval{\overapprox{\rank(\loc_i)}}{\state_{I^{e}_i-1}}\right)                                                   \\
		{}\leq{} & \expv{\pip}{\scheduler}{\initstate} \left(\ind_{B_{N,M,(t_i)_{i \in N}}} \cdot \sum_{i \in N} \eval{\exact{\overapprox{\rank(\loc_i)}}{\sizervar (g_i,\loc_i,\cdot)}}{\initstate}\right)                    \\
		{}={}    & \expv{\pip}{\scheduler}{\initstate} \left(\ind_{B_{N,M,(t_i)_{i \in N}}} \cdot \sum_{\substack{\loc\in\ENTRYLOC_{\GTNI}                                                                                     \\
				h \in \ET_\GTNI(\loc)}}
		\left(\sum_{t = (\_,\_,\_,\_,\loc) \in h} \mathcal C_{(t_i)_{i \in N}} (t)\right) \cdot \eval{\exact{\overapprox{\rank(\loc)}}{\sizervar (h,\loc, \cdot)}}{\initstate}\right)                                          \\
	\end{align*}
	where $\mathcal C_{(t_i)_{i \in N}} (t)$ denotes the number of occurrences of transition $t$ in $(t_i)_{i \in N}$.
	We continue by splitting the sum:

	\begin{align*}
		         & \expv{\pip}{\scheduler}{\initstate} \left(\ind_{B_{N,M,(t_i)_{i \in N}}} \cdot \sum_{\substack{\loc\in\ENTRYLOC_{\GTNI}                                              \\
				h \in \ET_\GTNI(\loc)}}
		\left(\sum_{t=(\_,\_,\_,\_,\loc)\in h}\mathcal C_{(t_i)_{i \in N}}(t) \right) \cdot \eval{\exact{\overapprox{\rank(\loc)}}{\sizervar (h,\loc,\cdot)}}{\initstate}\right)        \\
		{}={}    & \expv{\pip}{\scheduler}{\initstate} \left(\ind_{B_{N,M,(t_i)_{i \in N}}} \cdot \left(\sum_{\substack{\loc\in\ENTRYLOC_{\mathrm{nc},\GTNI}                            \\
				h \in \ET_\GTNI(\loc)}}
		\left(\sum_{t=(\_,\_,\_,\_,\loc)\in h}\mathcal C_{(t_i)_{i \in N}}(t)\right) \cdot \eval{\exact{\overapprox{\rank(\loc)}}{\sizervar (h,\loc,\cdot)}}{\initstate} \right.\right. \\
		         & \left.\left. + \sum_{\substack{\loc\in\ENTRYLOC_{\mathrm{c},\GTNI}                                                                                                   \\
				h\in\ET_\GTNI(\loc)}}
		\left(\sum_{t=(\_,\_,\_,\_,\loc)\in h} \mathcal C_{(t_i)_{i \in N}}(t)\right)\cdot \eval{\exact{\overapprox{\rank(\loc)}}{\sizervar(h,\loc, \cdot)}}{\initstate}
		\right)\right)                                                                                                                                                                  \\
		{}\leq{} & \expv{\pip}{\scheduler}{\initstate} \left(\ind_{B_{N,M,(t_i)_{i \in N}}} \cdot \left(\sum_{\substack{\loc\in\ENTRYLOC_{\mathrm{nc},\GTNI}                            \\
				h \in \ET_\GTNI(\loc)}}
		\left(\sum_{t=(\_,\_,\_,\_,\loc)\in h} \eval{\tbound (t)}{\initstate}\right) \cdot \eval{\exact{\overapprox{\rank(\loc)}}{\sizervar (h,\loc,\cdot)}}{\initstate} \right.\right. \\
		         & \left.\left. + \sum_{\substack{\loc\in\ENTRYLOC_{\mathrm{c},\GTNI}                                                                                                   \\
		h \in\ET_\GTNI(\loc)}}\timervar(h) \cdot(\overapprox{\rank(\loc)}) \right)\right)                                                                                               \\
		{}={}    & \sum_{\substack{\loc \in \ENTRYLOC_{\mathrm{nc},\GTNI}                                                                                                               \\h \in \ET_\GTNI(\loc)}} \left(\sum_{t=(\_,\_,\_,\_,\loc) \in h} \eval{\tbound(t)}{\initstate}\right)\cdot \expv{\pip}{\scheduler}{\initstate} \left(\ind_{B_{N,M,(t_i)_{i \in N}}} \cdot \eval{\exact{\overapprox{\rank(\loc)}}{\sizervar (h,\loc,\cdot)}}{\initstate}\right) \\
		         & + \sum_{\substack{\loc\in\ENTRYLOC_{\mathrm{c},\GTNI}                                                                                                                \\
				h \in\ET_\GTNI(\loc)}} \overapprox{\rank(\loc)} \cdot \expv{\pip}{\scheduler}{\initstate} \left(\ind_{B_{N,M,(t_i)_{i \in N}}} \cdot \timervar(h)\right) \tag{$\ddagger$}
	\end{align*}
	Here, in the last step we have used that $\overapprox{\rank(\loc)}$ is constant in the second part of the sum.
	Note that the penultimate inequation holds since $\run \in B_{N,M,(t_i)_{i \in N}}$ implies that the number of positions at which a transition $(\_,\_,\_,\_,\loc)\in h$ appears in $\run$, and therefore also in the tuple $(t_i)_{i \in N}$, is bounded by the minimum of $\sum_{t=(\_,\_,\_,\_,\loc)} (\eval{\tbound (t)}{\initstate})$ and $\timervar(h)(\run)$.
	Hence, it is bounded by both values.

	Recall the definition of $T'_i$ from \cref{def:rsm_stoptime}:
	\[
		T'_i (\run) = \min \{j\in \NN\mid R_{i,j} (\run) = 0\} \]
	By definition we have for every $\run \in \RUNS$
	\[
		I^{\succ}_{i,0}(\run) = 0 \implies R_{i,0}(\run) = 0 \implies T'_i (\run) = 0. \]
	Moreover, the properties ``Boundedness'', ``Non-Increase'', and ``Decrease'' of PLRFs together imply that the set of runs where a transition $g\in\GTG$ is executed when $\rank$ is already $0$ or less is a null set.
	So, for all $i\in M$ and all admissible runs $\run\in B_{N,M,(t_i)_{i \in N}}$ we have

	\[
		R_{i,0} (\run) > 0. \]
	Therefore $T'_i (\run) \geq |I_i (\run) \cap \{m\in \NN \mid t'_m \in \bigcup\GTG\}|$ for any admissible run $\run = (\initloc,t'_0,\initstate)\cdots\in B_{N,M,(t_i)_{i \in N}}$.
	Then for $g\in\GTG$ we have
	\begin{align*}
		         & \ind_{B_{N,M,(t_i)_{i \in N}}}(\run) \cdot \timervar (g)(\run)                                                      \\
		{}={}    & \ind_{B_{N,M,(t_i)_{i \in N}}}(\run) \cdot\abs{\biguplus_{i \in M} (I_i (\run) \cap \{m\in\NN \mid t'_{m} \in g\})} \\
		{}\leq{} & \ind_{B_{N,M,(t_i)_{i \in N}}}(\run) \cdot\sum_{i \in M} T'_i (\run)                                                \\
		{}\leq{} & \ind_{B_{N,M,(t_i)_{i \in N}}}(\run) \cdot \sum_{i \in \NN}
		T'_i (\run) \tag{$\dagger\dagger$}
	\end{align*}
	The last step holds because $M \subseteq N \subseteq \NN$.

	Since the sets $B_{N,M,(t_i)_{i \in N}}$ form (almost surely) a countable partition of $\RUNS$ we have
	\begin{align*}
		         & \expv{\pip}{\scheduler}{\initstate} (\timervar(g))                                                                                                    \\
		{}\leq{} & \expv{\pip}{\scheduler}{\initstate}\left(\sum_{i \in \NN}
		T'_i\right) \tag{by $(\dagger\dagger)$}                                                                                                                          \\
		{}={}    & \sum_{i \in \NN} \expv{\pip}{\scheduler}{\initstate}\left(T'_i\right)                                                                                 \\
		{}\leq{}
		         & \sum_{i \in \NN} \expv{\pip}{\scheduler}{\initstate}\left(R_{i,0}\right) \tag{by \cref{thm:rsm_timebound}}                                            \\
		{}\leq{}
		         & \sum_{i \in \NN} \expv{\pip}{\scheduler}{\initstate}\left(E_i\right) \tag{by \cref{lemma:expected_value_does_not_change,lemma:entry_supermartingale}} \\
		{}={}    & \expv{\pip}{\scheduler}{\initstate}\left(\sum_{i \in \NN}
		E_i\right) . \tag{$\ddagger\ddagger$}
	\end{align*}

	\noindent
	With $(\ddagger)$ and $(\ddagger\ddagger)$ we can conclude \eqref{eq:soundness_timebound}:
	\begin{align*}
		         & \expv{\pip}{\scheduler}{\initstate} (\timervar(g))                                                                                                                                                                                    \\
		{}\leq{} & \expv{\pip}{\scheduler}{\initstate} \left(\sum_{i \in \NN}
		E_i\right) \tag{by $(\ddagger\ddagger)$}                                                                                                                                                                                                         \\
		{}={}    & \expv{\pip}{\scheduler}{\initstate} \left(\ind_{\biguplus_{M \subseteq N \subseteq \NN,\, |N|<\infty, \, t_i \in \T^{in}_{\GTNI} \text{ for all } i \in N} B_{N,M,(t_i)_{i \in N}}} \cdot \sum_{i \in \NN} E_i\right)                 \\
		{}={}    & \sum_{\substack{M \subseteq N \subseteq \NN                                                                                                                                                                                           \\
		|N|<\infty                                                                                                                                                                                                                                       \\
		t_i \in \T^{in}_{\GTNI} \text{ for all } i \in N}} \expv{\pip}{\scheduler}{\initstate} \left(\ind_{B_{N,M,(t_i)_{i \in N}}} \cdot \sum_{i \in N} E_i\right)                                                                                      \\
		{}={}    & \sum_{\substack{M \subseteq N \subseteq \NN                                                                                                                                                                                           \\
		0<|N|<\infty                                                                                                                                                                                                                                     \\
		t_i \in \T^{in}_{\GTNI} \text{ for all } i \in N}} \expv{\pip}{\scheduler}{\initstate} \left(\ind_{B_{N,M,(t_i)_{i \in N}}} \cdot \sum_{i \in N} E_i\right)\tag{as $\ind_{B_{\emptyset,\emptyset,()}} \cdot \sum_{i\in \emptyset} E_i \equiv 0$} \\
		{}\leq{} & \sum_{\substack{M \subseteq N \subseteq \NN                                                                                                                                                                                           \\
		0<|N|<\infty                                                                                                                                                                                                                                     \\
				t_i \in \T^{in}_{\GTNI} \text{ for all } i \in N}}
		\left(\sum_{\substack{\loc \in \ENTRYLOC_{\mathrm{nc},\GTNI}                                                                                                                                                                                     \\
					h \in \ET_\GTNI(\loc)}}
		\left(\sum_{t=(\_,\_,\_,\_,\loc) \in h} \eval{\tbound(t)}{\initstate}\right)\cdot \right.                                                                                                                                                        \\
		         & \left.\expv{\pip}{\scheduler}{\initstate} \left(\ind_{B_{N,M,(t_i)_{i \in N}}}
		\cdot \eval{\exact{\overapprox{\rank(\loc)}}{\sizervar (h,\loc,\cdot)}}{\initstate}\right) + \right. \tag{by $(\ddagger)$}                                                                                                                       \\
		         & \left. \sum_{\substack{\loc\in\ENTRYLOC_{\mathrm{c},\GTNI}                                                                                                                                                                            \\h\in\ET_\GTNI(\loc)}} \overapprox{\rank(\loc)}\cdot \expv{\pip}{\scheduler}{\initstate} \left(\ind_{B_{N,M,(t_i)_{i \in N}}}\cdot \timervar(h)\right)\right) \\
		{}={}    & \sum_{\substack{\loc \in \ENTRYLOC_{\mathrm{nc},\GTNI}                                                                                                                                                                                \\h \in \ET_\GTNI(\loc)}} \left(\sum_{t=(\_,\_,\_,\_,\loc) \in h} \eval{\tbound(t)}{\initstate}\right)\cdot \expv{\pip}{\scheduler}{\initstate} \left(\eval{\exact{\overapprox{\rank(\loc)}}{\sizervar (h,\loc,\cdot)}}{\initstate}\right) \\
		         & + \sum_{\substack{\loc\in\ENTRYLOC_{\mathrm{c},\GTNI}                                                                                                                                                                                 \\
				h\in\ET_\GTNI(\loc)}}
		\expv{\pip}{\scheduler}{\initstate} (\timervar (h)) \cdot \overapprox{\rank(\loc)}                                                                                                                                                               \\
		{}\leq{} & \sum_{\substack{\loc \in \ENTRYLOC_{\mathrm{nc},\GTNI}                                                                                                                                                                                \\h \in \ET_\GTNI(\loc)}} \left(\sum_{t=(\_,\_,\_,\_,\loc) \in h} \eval{\tbound(t)}{\initstate}\right)\cdot \left(\exact{\overapprox{\rank(\loc)}}{\expv{\pip}{\scheduler}{\initstate}\left(\sizervar(h,\loc,\cdot)\right)}\right) \\
		         & + \sum_{\substack{\loc\in\ENTRYLOC_{\mathrm{c}, \GTNI}                                                                                                                                                                                \\
				h \in\ET_\GTNI(\loc)}} \eval{\tbounde(h)}{\initstate} \cdot \overapprox{\rank(\loc)}
		\tag{since $\overapprox{\rank(\loc)}$ is linear}                                                                                                                                                                                                 \\
		{}\leq{} & \sum_{\substack{\loc \in \ENTRYLOC_{\mathrm{nc},\GTNI}                                                                                                                                                                                \\h \in \ET_\GTNI(\loc)}} \left(\sum_{t=(\_,\_,\_,\_,\loc) \in h} \eval{\tbound(t)}{\initstate}\right)\cdot \left(\exact{\overapprox{\rank(\loc)}}{ \eval{\sbounde(h,\loc,\cdot)}{\initstate}}\right) 	\tag{since $\overapprox{\rank(\loc)}$ is linear} \\
		         & + \sum_{\substack{\loc\in\ENTRYLOC_{\mathrm{c}, \GTNI}                                                                                                                                                                                \\
				h \in\ET_\GTNI(\loc)}} \eval{\tbounde(h)}{\initstate} \cdot \overapprox{\rank(\loc)}
		\\
		{}={}    & \eval{
		\sum_{\substack{\loc \in \ENTRYLOC_{\mathrm{nc},\GTNI}                                                                                                                                                                                           \\h \in \ET_\GTNI(\loc)}} \left(\sum_{t=(\_,\_,\_,\_,\loc) \in h} \tbound(t)\right)\cdot \left( \exact{\overapprox{\rank(\loc)}}{ \sbounde(h,\loc,\cdot)}\right) \right. \\
		&\qquad\left.+ \sum_{\substack{\loc\in\ENTRYLOC_{\mathrm{c}, \GTNI}\\
					h \in\ET_\GTNI(\loc)}} \tbounde(h) \cdot \overapprox{\rank(\loc)}}{\initstate}
	\end{align*}
	i.e., $\tbounde'$ is indeed a valid expected time bound.
\end{proof}

\begin{example}[Constant Ranking Functions]
	\input{figures/nondet_countdown}
	{\sl Consider the example \textnormal{\texttt{nondet\_count\-down}}
		from our additional benchmark suite.
		Its graph is depicted in \cref{fig:nondet_countdown}.
		Here, we have $\PV = \{
			x \}$, i.e., $y$ is a temporary variable.
		As usual, we omit trivial updates and probabilities that are $1$.

		We start with finding a PLRF for $g_1$.
		Here, we compute the ranking function $\rank$ with $\rank(\ell_1)= 4 \cdot x$ and $\rank(\ell_2) = 4 \cdot x - 4$ for $\GTG = \{g_1\}$ and $\GTNI = \{g_1,g_2\}$.
		First of all, $g_1$ satisfies the condition ``Boundedness (b)'', since for any admissible configuration $(\loc,t,\state)$ with $t\in g_1$ we have $\exact{\rank(\ell)}{\state} \geq 0$.
		Moreover, $g_1$ satisfies the condition ``Decrease'' since for any state $\state$ such that $\state(\tau_{g_1}) = \true$, i.e., $\state(x) > 0$, we have
		\begin{align*}
			         & \exp_{\rank,g_1,\state}                                                                                                                                  \\
			{}={}    & \tfrac{1}{2}\cdot 4 \cdot \state(x) + \tfrac{1}{4}\cdot (4 \cdot (\state(x) + 1) - 4) + \tfrac{1}{4}\cdot (4 \cdot \state(x) - 4) = 4 \cdot \state(x) -1 \\
			{}\leq{} & \exact{\rank(\loc_1)}{\state} -1.
		\end{align*}
		The transition $g_2$ trivially satisfies ``Boundedness (a)'' since it only consists of one transition.
		More precisely, for any admissible path $\, \cdots \, c' \, c$ there is a unique $\bowtie_{g,c'} \in \{<,\geq\}$ such that for all $c = (\loc,t,\state)$ with $t \in g_2$, i.e., $t = t_4$ and $\loc = \loc_1$, we have $\state(\rank(\loc_1)) \bowtie_{g,c'} 0$, since the value for temporary variable $y$ which the scheduler chooses depends only on $c'$.

		The transition $g_2$ also satisfies ``Non-Increase'', since for any state $\state$ such that $\state(\tau_{g_2}) = \true$, i.e., $\state(y) > 0$, we have
		\begin{align*}
			      & \exp_{\rank,g_2,\state}         \\
			{}={} & 4 \cdot (\state(x) - \state(y)) \\
			{}<{} & 4 \cdot \state(x)               \\
			{}={} & \exact{\rank(\loc_2)}{\state}.
		\end{align*}

		Here, $\ENTRYLOC_{\GTNI} = \{\loc_1\}$ and $\ET_{\GTNI}(\loc_1) = \{g_0\}$.
		So let us assume that we have already computed $\tbound(t_0) = 1$ and $\sbounde(g_0,\loc_1,x) = x$.
		By using \cref{theorem:exptimeboundsmeth}, we infer
		\[
			\tbounde(g_1) = \tbound(t_0) \cdot \exact{\overapprox{\rank(\ell_1)}}{\sbounde(g_0,\loc_1,\cdot)} = 4 \cdot x. \]

		However, there is no PLRF $\rank'$ for $\GTG' = \{g_2\}$ and $\GTNIpr = \{g_1,g_2\}$ since $t_3$ increases the value of $x$, and $t_1$ and $t_2$ leave the value of $x$ unchanged whereas $t_4$ decreases the value of $x$.
		Thus, we consider $\GTG' = \{g_2\} = \GTNIpr$.
		In this case, $\ENTRYLOC_{\GTNI'} = \{\loc_2\}$ and $\ET_{\GTNI}(\loc_2) = \{g_1\}$.
		Unfortunately, we cannot compute finite \emph{non-probabilistic} runtime bounds for $t_2$ and $t_3$, since they might occur arbitrarily often in an admissible run, and we cannot compute a finite expected size bound $\sbounde(g_1,\loc_2,x)$ without a finite expected runtime bound for $g_2$, as $\sbounde(g_1,\loc_2,x)$ and $\sbounde(g_2,\loc_1,x)$ are members of the same non-trivial SCC of the general result variable graph (see \cref{theorem:expectednontrivialsizeboundsmeth}).
		Thus, by using \cref{theorem:exptimeboundsmeth}, we would always obtain
		\[
			\tbounde(g_2) = \left(\sum_{t \in g_1} \tbound(t)\right) \cdot \exact{\overapprox{\rank'(\ell_2)}}{\sbounde(g_1,\loc_2,\cdot)} = \infty, \]
		since $\rank'(\ell_2)$ is non-zero.
		However, we can infer the \emph{constant} ranking function $\rank'$ with $\rank'(\loc_2) = 1$ and $\rank'(\loc_1) = 0$ for $\GTG' = \{g_2\} = \GTNIpr$.
		Using \cref{Expected Time Bounds Refined} we obtain
		\[
			\tbounde(g_2) = \tbounde(g_1) \cdot \rank(\ell_2) = 4 \cdot x. \]
	}
\end{example}

%% file: figures/nondet_countdown.tex
\begin{figure}[t]
  \begin{tikzpicture}[shorten >=1pt,node distance=3.5cm,auto,main node/.style={circle,draw,font=\sffamily\Large\bfseries}]
      \node[main node] (f) {$\ell_0$};
      \node[main node] (g) [right of = f] {$\ell_1$};
      \node[main node] (h) [right of = g] {$\ell_2$};

      \path[->] (f)
      edge[align=left]
      node [above] {$\tau = x > 0$}
        node [below] {$t_0 \in g_0$}
      (g)

      (g) edge[loop above, align=left]
        node [above] {$p = \tfrac{1}{2}$ \\ $\tau = x > 0$ \\ $t_1 \in g_1$}
      ()

      (g) edge[align=left, bend left = 45]
        node{$p = \tfrac{1}{4}$ \\ $\tau = x > 0$ \\ $t_2 \in g_1$}
      (h)
      
      (g) edge[align=left]
        node [above] {$p = \tfrac{1}{4}$ \\ $\eta(x) = x + 1$}
        node [below] {$\tau = x > 0$ \\$t_3 \in g_1$}
      (h)
      
      (h) edge[align=left, bend left = 45]
        node [below] {$\eta(x) = x - y$ \\$\tau = y > 0$ \\ $t_4 \in g_2$}
      (g);
  \end{tikzpicture}
  \caption{The \texttt{nondet\_countdown} example}
  \label{fig:nondet_countdown}
\end{figure}

%% file: proofs/concavity_linearity.tex
As discussed in \cref{Inferring Expected Time Bounds}, if bounds $b_1,\ldots,b_n$ are \emph{substituted} into another bound $b$, then it is sound to use ``expected versions'' of the bounds $b_1,\ldots,b_n$ if $b$ is \emph{concave}.
The essential reason is that if a function $f$ is concave, then the expected value of the random variable $f(X_1,\ldots,X_n)$ is greater or equal than applying the function $f$ to the expected values of the random variables $X_1,\ldots,X_n$, see
\cite[Lemma 3.5]{kallenberg_foundations_2002}.

A unary function $f:\RR \to \RR$ is \emph{concave} if for any two numbers $v, v' \in \RR$ and any $r$ between $v$ and $v'$ (i.e., $r = t \cdot v + (1-t) \cdot v'$ for $t \in [0,1]$), the point $(r, f(r))$ is above the line connecting the points $(v, f(v))$ and $(v', f(v'))$.
The following definition extends this concept to functions with several arguments.

\begin{definition}[Concavity]
	\label{def:concavity}
	A finite bound $b \in \BOUND$ is \emph{concave} if for any $w,w':\PV \to \RR_{\geq 0}$ and $t \in [0,1]$, we have $(t \cdot w + (1-t) \cdot w')(b) \geq t \cdot w(b) + (1-t) \cdot w'(b)$.
\end{definition}

It turns out that a finite bound is concave iff it is a \emph{linear} polynomial.
(The ``only if'' direction is due to the fact that bounds from $\BOUND$ do not contain negative coefficients, i.e., they do not contain non-linear concave polynomials like $- x^2$.)

\begin{lemma}
	\label{lemma:concavity_implies_linearity}
	A finite bound is concave iff it is a linear polynomial.
\end{lemma}
\begin{proof}
	Every linear polynomial is clearly concave.
	For the other direction, by \cref{def:bounds}, if $b \in \BOUND$ does not contain $\infty$, then $b$ is built from variables in $\PV$, non-negative constants in $\RR_{\geq_0}$, and the operators $+$, $\cdot$, and $v^{\bullet}$ for $v \geq 1$.
	Thus, $b$ defines a function which is infinitely often continously differentiable.

	It is a standard result from analysis (see, e.g., \cite{boyd_vandenberghe_2004}), that such a $b$ is concave as a function $(\PV \to \RR_{\geq 0}) \to \RR$ iff its Hessian matrix $H$, i.e., the matrix of its second derivatives, is negative semi-definite on $(\PV \to \RR_{>0})$.
	A quadratic real matrix $A \in \RR^{n \times n}$ is \emph{negative semi-definite} if for all $\vec{v} \in \RR^n$ we have $\vec{v}^{T} \cdot A \cdot \vec{v} \leq 0$.

	So, let $b$ be a bound which is not a linear polynomial.
	Then its Hessian matrix $H$ contains only terms built from variables in $\PV$, constants in $\RR_{\geq_0}$, and the operators $+$, $\cdot$, and $v^{\bullet}$ for $v \geq 1$.
	Moreover, $H$ is not the zero matrix.
	Now let $A$ be $H$ evaluated at the assignment from $(\PV \to \RR_{>0})$ which assigns every program variable the value $1$.
	Then each entry of $A$ is non-negative and at least one entry is positive (since only non-negative numbers, addition, multiplication and exponentiation with strictly positive basis are involved, no expressions can cancel each other out).
	Now let $\vec{v}$ be the vector where each entry is $1$.
	Then $ A \cdot \vec{v}$ is a vector with only non-negative entries and at least one positive entry.

	Thus
	\[
		\vec{v}^{T} \cdot A \cdot \vec{v} > 0, \]
	since this value is just the sum of the entries of $ A \cdot \vec{v}$.

	Hence, $A$ is \emph{not} negative semi-definite, i.e., $H$ is \emph{not} negative semi-definite on $(\PV \to \RR_{>0})$.
	Thus, $b$ is not concave.
\end{proof}

%% file: proofs/elsb_basic.tex
\chboundebasic*

\begin{proof}
	Let $\state'_{\scheduler}$ be the state chosen by the scheduler to continue $c'$, i.e.,
	$\scheduler (c') = (\_,\state'_{\scheduler})$, and $\state'(x)=\state'_{\scheduler}(x)$ for all $x \in \PV$.
	By definition of the underlying probability mass function $\ptransition{\pip}{\scheduler}{\initstate}$ (see \cref{appendix:prob_space_construction}) we have:
	\begin{align*}
		         & \sum_{{\substack{c = (\loc,t,\state) \in \CONF                              \\
						t\in g}}}
		\ptransition{\pip}{\scheduler}{\initstate} (c' \to c) \cdot |\state (x) - \state' (x)| \\
		{}={}    & \sum_{{\substack{c = (\loc,t,\state) \in \CONF                              \\
						t = (\_,p,\_,\_,\loc)\in g}}}
		|\state (x) - \state' (x)| \cdot p \cdot \prod_{{y\in\PV}} \pr{\state'_{\scheduler}}{t}{\state(y)}{y}
		\cdot \prod_{{u \in \VV \setminus \PV}} \delta_{\state'_{\scheduler}(u),\state(u)}     \\
		{}={}    & \sum_{{\substack{v\in\mathbb Z                                              \\
						t = (\_,p,\_,\eta,\loc)\in g}}}
		|v| \cdot p \cdot
		\begin{cases}
			\delta_{\state'_{\scheduler}(\eta (x)) - \state' (x), v}, & \text{if } \eta(x) \in \POLYBOUND \\
			d(\state'_{\scheduler})(v),                               & \text{if } \eta(x)= d \in \DIST   \\
		\end{cases}
		\\
		{}={}    & \sum_{{\substack{v\in\NN                                                    \\
						t = (\_,p,\_,\eta,\loc)\in g}}}
		v \cdot p \cdot
		\begin{cases}
			\delta_{|\state'_{\scheduler}(\eta (x)) - \state' (x)|, v}, & \text{if } \eta(x) \in \POLYBOUND \\
			d(\state'_{\scheduler})(v)+d(\state'_{\scheduler})(-v),     & \text{if } \eta(x) = d \in \DIST  \\
		\end{cases}
		\\
		{}={}    & \sum_{t = (\_,p,\_,\eta,\loc)\in g} p \cdot \left(\sum_{v\in\NN} v \cdot
		\begin{cases}
			\delta_{|\state'_{\scheduler}(\eta (x)) - \state' (x)|, v}, & \text{if } \eta(x) \in \POLYBOUND \\
			d(\state'_{\scheduler})(v)+d(\state'_{\scheduler})(-v),     & \text{if } \eta(x) = d \in \DIST  \\
		\end{cases}
		\right)                                                                                \\
		{}\leq{} & \sum_{{\substack{t = (\_,p,\_,\eta,\loc)\in g}}}
		p \cdot
		\begin{cases}
			\eval{\overapprox{\eta(x) - x}}{\state'_{\scheduler}}, & \text{if } \eta(x) \in \POLYBOUND \\
			\eval{\expvdist(d)}{\state'_{\scheduler}},             & \text{if } \eta(x) = d \in \DIST  \\
		\end{cases}
		\tag{since $\state'_{\scheduler}(x)=\state'(x)$ for all $x \in \PV$ and $\eval{\overapprox{\eta(x) - x}}{\state'_{\scheduler}} \geq \eval{|\eta(x) - x|}{\state'_{\scheduler}}$
		}                                                                                      \\
		{}\leq{} & \eval{\chbounde(g,\loc,x)}{\state'_{\scheduler}}                            \\
		{}={}    & \eval{\chbounde(g,\loc,x)}{\state'}
		\tag{since $\state'_{\scheduler}(x)=\state'(x)$ for all $x \in \PV$}
	\end{align*}
\end{proof}

%% file: proofs/expected_trivial_sizebounds_method.tex
\expectedtrivialsizeboundsmethod*
\begin{proof}
	First of all, for a general result variable $\beta \neq \alpha$ there is nothing to prove, since $\sbounde$ is already a valid expected size bound.

	\paragraph{Initial Transition:}
	We now consider the case
	$g\in \GTINIT$.
	It suffices to consider admissible runs $\run = \prefix{(\initloc,\tin,\initstate) (\loc_1,t_1,s_1)\cdots}$ with $t_1\in g$ and $\loc_1 = \loc$, since there is no transition leading back to the initial location.
	Therefore the transition $g$ will always be the general transition that is executed first (if it is executed at all).
	Hence, for all considered runs $\run$ we have:
	\[
		\sizervar (\alpha) (\run) = \sup \{|\state_1 (x)|\} = |\state_1 (x)| \]

	Thus, we will use \cref{theorem:elsb_basic} with the (trivially) admissible configuration $c_0 = (\initloc,\tin,\initstate)$.
	Note that $\prefix{c_0}$ also forms an admissible finite path.
	Moreover, any admissible finite path starts with $c_0$.
	So we have
	\begin{align*}
		         & \expv{\pip}{\scheduler}{\initstate} (\sizervar (\alpha))                                                                                                   \\
		{}={}    & \ptransition{\pip}{\scheduler}{\initstate}(c_0) \cdot \sum_{{\substack{c = (\loc,t,\state') \in \CONF                                                      \\
						t\in g}}}
		\ptransition{\pip}{\scheduler}{\initstate}(c_0 \to c) \cdot |\state' (x)|                                                                                             \\
		{}={}    & \ptransition{\pip}{\scheduler}{\initstate}
		(\prefix{c_0}) \cdot \sum_{{\substack{c = (\loc,t,\state') \in \CONF                                                                                                  \\
						t\in g}}}
		\ptransition{\pip}{\scheduler}{\initstate}(c_0 \to c) \cdot |\state' (x) - \initstate (x) + \initstate (x)|                                                           \\
		{}\leq{} & \ptransition{\pip}{\scheduler}{\initstate} (\prefix{c_0}) \cdot \sum_{{\substack{c = (\loc,t,\state') \in \CONF                                            \\
						t\in g}}}
		\ptransition{\pip}{\scheduler}{\initstate} (c_0 \to c) \cdot (|\initstate (x)| + |\state' (x) - \initstate (x)|)                                                      \\
		{}\leq{} & |\initstate (x)| + \ptransition{\pip}{\scheduler}{\initstate} (\prefix{c_0}) \cdot \sum_{{\substack{c = (\loc,t,\state') \in \CONF                         \\
						t\in g}}}
		\ptransition{\pip}{\scheduler}{\initstate} (c_0 \to c) \cdot |\state' (x) - \initstate (x)|                                                                           \\
		{}\leq{} & |\initstate (x)| + \ptransition{\pip}{\scheduler}{\initstate} (\prefix{c_0}) \cdot \eval{\chbounde (\alpha)}{\initstate}\tag{by \cref{theorem:elsb_basic}} \\
		{}\leq{} & \eval{x + \chbounde (\alpha)}{\initstate} \tag{since $\ptransition{\pip}{\scheduler}{\initstate} (\prefix{c_0}) = 1$}                                      \\
		{}={}    & \eval{\sbounde' (\alpha)}{\initstate}
	\end{align*}

	\paragraph{Non-Initial Transition:}
	Now let
	$g\not\in\GTINIT$.
	This case is similar to the previous case except that we will look at all finite paths $f$ leading to an execution of $g$.
	This means we have to look for all $f$'s ending in a configuration $c = (\loc_g,t_n,\state_n)$ with $t_n\in h$ and $h \in \pre (g)$.

	Recall that the active variables consist of all variables that influence $x + \chbounde (\alpha)$.
	So the set of active variables of $\alpha = (g,\loc,x)$ is never empty, since we always have $x \in \ACTV(\alpha)$.
	Hence, as $g$ is not an initial general transition, it has a predecessor in the general result variable graph.
	Since $\alpha$ forms a trivial SCC of the program, this means that $g$ cannot be applied more than once in any program run.
	In other words, we can restrict ourselves to finite paths $f$ where $g$ has not been applied before the last evaluation step.
	Therefore, for all considered runs $\run$ that start with the prefix $\prefix{f c}$ for some configuration $c = (\loc', t, s)$ with $t \in g$, we get
	\[
		\sizervar (\alpha) (\run) = \sup \{|\state (x)|\} = |\state (x)|. \]

	Moreover, we define $\sizervar_{in}^\alpha(\cdot) = \sum\limits_{h \in \pre(g)} \sizervar(h,\loc_g,\cdot)$, i.e., for every variable $y \in \PV$ we have
	\[
		\sizervar_{in}^\alpha(y) = \sum\limits_{h \in \pre(g)} \sizervar(h,\loc_g,y).
	\]
	To ease notation, we allow to apply $\sizervar_{in}^\alpha(y)$ also to finite prefixes instead of just runs.

	So we get
	\begin{align*}
		         & \expv{\pip}{\scheduler}{\initstate} (\sizervar (\alpha))                                                                                \\
		{}={}    & \sum_{\substack{f = \prefix{(\loc_0,t_0,\state_0) \cdots (\loc_n,t_n,\state_n)}                                                         \\
		t_n \in h                                                                                                                                          \\
				h \in \pre (g)}}
		\ptransition{\pip}{\scheduler}{\initstate} (f)                                                                                                     \\
		         & \quad\quad \cdot \sum_{\substack{c = (\loc',t,\state) \in \CONF                                                                         \\
		t\in g}} \ptransition{\pip}{\scheduler}{\initstate} ((\loc_n,t_n,\state_n) \to c) \cdot |\state (x)|                                               \\
		{}={}    & \sum_{\substack{f = \prefix{(\loc_0,t_0,\state_0) \cdots (\loc_n,t_n,\state_n)}                                                         \\
		t_n \in h                                                                                                                                          \\
				h \in \pre (g)}}
		\ptransition{\pip}{\scheduler}{\initstate} (f)                                                                                                     \\
		         & \quad\quad \cdot \sum_{\substack{c = (\loc',t,\state) \in \CONF                                                                         \\
		t\in g}} \ptransition{\pip}{\scheduler}{\initstate} ((\loc_n,t_n,\state_n) \to c) \cdot |\state (x)-\state_n(x)+\state_n(x)|                       \\
		{}\leq{} & \sum_{\substack{f = \prefix{(\loc_0,t_0,\state_0) \cdots (\loc_n,t_n,\state_n)}                                                         \\
		t_n \in h                                                                                                                                          \\
				h \in \pre (g)}}
		\ptransition{\pip}{\scheduler}{\initstate} (f)                                                                                                     \\
		         & \quad\quad \cdot \sum_{\substack{c = (\loc',t,\state) \in \CONF                                                                         \\
				t\in g}} \ptransition{\pip}{\scheduler}{\initstate}
		((\loc_n,t_n,\state_n) \to c) \cdot (|\state_n(x)| + |\state (x)-\state_n(x)|)                                                                     \\
		{}\leq{} & \sum_{\substack{f = \prefix{(\loc_0,t_0,\state_0) \cdots (\loc_n,t_n,\state_n)}                                                         \\
		t_n \in h                                                                                                                                          \\
				h \in \pre (g)}}
		\ptransition{\pip}{\scheduler}{\initstate} (f) \cdot ( |\state_n (x)| + \eval{\chbounde (\alpha)}{\state_n}) \tag{by \cref{theorem:elsb_basic}}    \\
		{}\leq{} & \sum_{\substack{f = \prefix{(\loc_0,t_0,\state_0) \cdots (\loc_n,t_n,\state_n)}                                                         \\
		t_n \in h                                                                                                                                          \\
				h \in \pre (g)}}
		\ptransition{\pip}{\scheduler}{\initstate} (f) \cdot ( |\sizervar_{in}^\alpha(x)(f)| + \eval{\chbounde (\alpha)}{\sizervar_{in}^\alpha(\cdot)(f)}) \\
		{}={}    & \expv{\pip}{\scheduler}{\initstate}\left(\sizervar_{in}^\alpha(x) + \eval{\chbounde (\alpha)}{\sizervar_{in}^\alpha(\cdot)}\right)
	\end{align*}

	Now, Jensen's inequality allows us to yield better bounds by considering expected bounds on the sizes of the variables \emph{if $\chbounde (\alpha)$ is concave}:
	\begin{align*}
		         & \expv{\pip}{\scheduler}{\initstate}\left(\sizervar_{in}^\alpha(x) + \eval{\chbounde (\alpha)}{\sizervar_{in}^\alpha(\cdot)}\right)                                                 \\
		{}\leq{} & \expv{\pip}{\scheduler}{\initstate}\left(\sizervar_{in}^\alpha(x)\right) + \eval{\chbounde (\alpha)}{\expv{\pip}{\scheduler}{\initstate}\left(\sizervar_{in}^\alpha(\cdot)\right)} \\
		{}\leq{} & \eval{\incsize^\alpha_{\mathbb{E}}(x + \chbounde (\alpha))}{\initstate}                                                                                                            \\
		{}={}    & \eval{\sbounde' (\alpha)}{\initstate}
	\end{align*}
	This corresponds to the second case of \cref{thm:Inferring Expected Size Bounds for Trivial SCCs}, i.e.,
	$g\not\in\GTINIT$, $\beta = \alpha$, and $\chbounde (\alpha)$ is concave.

	\bigskip

	If $\chbounde (\alpha)$ \emph{is not concave}, then we get:
	\begin{align*}
		         & \expv{\pip}{\scheduler}{\initstate}\left(\sizervar_{in}^\alpha(x) + \eval{\chbounde (\alpha)}{\sizervar_{in}^\alpha(\cdot)}\right)                                                 \\
		{}\leq{} & \expv{\pip}{\scheduler}{\initstate}\left(\sizervar_{in}^\alpha(x)\right) + \expv{\pip}{\scheduler}{\initstate}\left(\eval{\chbounde (\alpha)}{\sizervar_{in}^\alpha(\cdot)}\right) \\
		{}\leq{} & \eval{\incsize^\alpha_{\mathbb{E}}(x) + \incsize^\alpha(\chbounde (\alpha))}{\initstate}                                                                                           \\
		{}={}    & \eval{\sbounde' (\alpha)}{\initstate}
	\end{align*}
\end{proof}

%% file: appendix/details_expected_nontrivial.tex
Let $C \subseteq \GRV$ form a fixed non-trivial SCC of the general result variable graph.
Furthermore, we fix an initial state $\initstate$.
To abbreviate notation, we write $\GT_C = \{g \in \GT \mid (g,\_,\_) \in C\}$, $\CONF_{C} = \{(\loc,t,\_) \mid t \in g, (g, \loc,\_) \in C\}$ and $\CONF_{\neg C} = \CONF \setminus \CONF_{C}$.

Moreover, for any $x \in \PV$ we define the random variable
\[
	\sizervar_C(x)\is \sup_{\alpha' = (\_,\_,x) \in C} \sizervar(\alpha').
\]
So for any $\run \in \RUNS$, $\sizervar_C(x)(\run)$ is the maximal absolute value that $x$ takes in a location $\loc$ in $\run$ that was entered with $g$ for general result variables $(g, \loc, x)$ from $C$.
Then we immediately get $\sizervar(\alpha') \leq \sizervar_C(x)$ for each $\alpha' = (\_,\_,x) \in C$, so in particular
\begin{equation}
	\label{observe3}
	\expv{\pip}{\scheduler}{\initstate}\left(\sizervar(\alpha')\right) \leq \expv{\pip}{\scheduler}{\initstate}\left(\sizervar_C(x)\right) \quad \text{for all $\alpha' = (\_,\_,x) \in C$,}
\end{equation}
by monotonicity of the expected value operator.

We will fix $x \in \PV$ from now on.
To over-approximate the expected value of $\sizervar_C(x)$, in analogy to the size random variable $\sizervar(\alpha)$ for $\alpha \in \GRV$, we define a random variable $S(x)_n$ corresponding to the maximal size of the variable $x$ after evaluating general transitions $g \in \GT_C$ at most $n$ times.
By definition of $\CONF_C$, we also take the locations into account, i.e., we only consider evaluations of $g \in \GT_C$ ending in a location $\loc$ such that $(g,\loc,\_) \in C$.
The special case $S(x)_0$ defaults to the value before the first such general transition is executed.
The underlying idea is similar to the stochastic process $R_{i,j}$ which we defined in \cref{definition:rsm_stproc} to prove the correctness of our technique for inferring expected time bounds.

By $I(\run)\subseteq \NN$ we denote the set containing all indices corresponding to a configuration from $\CONF_C$ in a run $\run\in\RUNS$ which we will make more formal in the following definition.

\begin{definition}[Indices for $\CONF_C$]
	Let $\run=\prefix{c_0 \, c_1 \,\cdots} \in\RUNS$.
	Then we define
	\[
		I (\run) = \{i \in \NN \mid c_i \in \CONF_C\}. \]
	The $k$-th index (for $k>0$) is then given by:
	\[
		I_k (\run) =
		\begin{cases}
			\min (\cuts^{k-1} (I (\run))), & \text{if $\cuts^{k-1} (I (\run)) \neq \emptyset$} \\
			0,                             & \text{otherwise}
		\end{cases}
	\]
	with $\cuts(X) = X\setminus\{\min X\}$, i.e., $I_k (\run)$ is the $k$-th-smallest element of $I (\run)$.

	As usual, for some finite path $f$ we define $I_k(f)=i_k \in \NN$ iff for all runs $ \run = \prefix{f \cdots}$ we have $I_k(\run)=i_k$.
\end{definition}
It is worth noting that in any admissible run $\run$ we have $t_0=\tin$, i.e., if $I(\run)$ is non-empty, then its minimal element is \emph{positive}.

Moreover, we can now define the function $S(x)_n$ mapping a run $\run$ to the maximal absolute value that $x$ takes in the first $n$ configurations from $\CONF_C$ and the function $S(x)_0$ returning the size immediately before entering the SCC $C$.

\begin{definition}[$S(x)_0,S(x)_n$ for Runs]
	\label{definition:expnontrivsnfinitepath}
	For a run $\run = \prefix{(\initloc,t_0,\initstate) \cdots}$ and $n \in \NN_{\geq 1}$ we define:
	\begin{align*}
		S(x)_0 (\run) & =
		\begin{cases}
			\abs{\state_{I_1 (\run) - 1}(x)}, & \text{if } I_1(\run) > 0 \\
			0,                                & \text{otherwise}
		\end{cases}
		\\
		S(x)_n (\run) & =
		\begin{cases}
			\max (S_{n-1}(x)(\run),\abs{\state_{I_n(\run)}(x)}), & \text{if } I_n(\run) > 0 \\
			S_{n-1} (x) (\run),                                  & \text{otherwise}         \\
		\end{cases}
	\end{align*}
\end{definition}
By construction, for every run $\run$ we have
\begin{align}
	S(x)_n(\run) \leq S_{n+1}(x)(\run) \label{eq:monotonicity}
\end{align}
and
\begin{align}
	\sizervar_C(x)(\run) \leq \lim_{n \to \infty} S(x)_n(\run) \label{eq:limit_size}
	.
\end{align}
Note that in general we do not have equality in \eqref{eq:limit_size}
since $\sizervar_C(x)$ does \emph{not} take the value \emph{before} entering $\CONF_C$ into account.

We generalize the just defined random variables to finite paths in the standard way.
\begin{definition}[$S(x)_n$ for Finite Paths]
	\label{definition:expnontrivsnruns}
	The functions $S(x)_n$ for $n \in \NN$ are generalized from runs to finite paths $f\in \FPATH$ as follows:

	\begin{align*}
		S(x)_n (f) & =
		\begin{cases}
			z, & \text{if } S(x)_n (\run) = z \neq 0 \text{ for all runs } \run = \prefix{f \cdots} \\
			0, & \text{otherwise}
		\end{cases}
	\end{align*}
	for all $n \in \NN$.
\end{definition}
We will now define a filtration $(\CF_n)_{n \in \NN}$ such that the stochastic process $(S(x)_n)_{n \in \NN}$ becomes adapted to this filtration.
However, there is one thing which we have to keep in mind when doing so.

By \cref{def:general_result_variable_graph}, the set $\CONF_C$ can only be entered once.
So, if in an \emph{admissible} run $\run$ a configuration $c' \in \CONF_C$ is followed by some configuration $c \in \CONF_{\neg C}$ then we know that in the remainder of $\run$ there cannot occur another configuration from $\CONF_C$ afterwards.
Hence, in an admissible run $\run$, the part related to executions in $\CONF_C$ is an infix of $\run$.
So, we now distinguish the runs by the ``shape'' of the infix belonging to the first execution in $\CONF_C$ (remember that admissibility is related to a probability measure which we do not know yet).
Therefore, for some $n>0$ and $f = c_1\cdots c_n \in \CONF^n$ consider the set $\Inf{f}$ defined by
\[
	\Inf{f} \is \left\{ \run \in \RUNS \mid \run = \prefix{\cdots c_1 \cdots c_n \cdots}\right\}
\]
i.e., the set $\Inf{f}$ consists exactly of the runs in which $f$ occurs as an infix.

However, taking $f\neq f' \in \CONF_{C}^n$ then does \emph{not} yield $\Inf{f} \cap \Inf{f'} = \emptyset$.
The problem that arises here is that whereas in admissible runs $\CONF_C$ is only entered once, there are still non-admissible runs in which $\CONF_C$ can be entered multiple times which then are contained in the intersection above.
Furthermore, this intersection contains (for example) $\Inf{ff'}$, which can also contain admissible runs.
Nevertheless, we are only interested in the \emph{first} occurrence of $\CONF_C$, i.e., for $f = c_1\cdots c_n \in \CONF_{C}^n$ we are interested in
\begin{equation}
	\label{PrefC old}
	\PrefC{f}{C} \is \left\{ \run \in \Inf{f} \mid c_i = \run[I_1(\run) + i - 1] \text{ for } 1 \leq i \leq n\right\}
\end{equation}
Then, taking $f\neq f' \in \CONF_{C}^n$ indeed yields $\PrefC{f}{C} \cap \PrefC{f'}{C} = \emptyset$.
Thus, it seems natural to define $\CF_n$ as the smallest $\sigma$-field containing all $\PrefC{f}{C}$ for $f \in \CONF_{C}^n$ for some fixed $n>0$.

However, in this case we ignore that the complement of $\biguplus_{f \in \CONF_{C}^n} \PrefC{f}{C}$ is more fine-grained, since it does not only contain those runs where $\CONF_C$ is not visited at all but it also contains the runs in which the infix belonging to the first occurrence of $\CONF_C$ has a length $0 < k < n$.
Hence, we also include these runs by adapting our definition of $\PrefC{f}{C}$.
Moreover, the definition of $\PrefC{f}{C}$ in \eqref{PrefC old} also has the drawback that we do \emph{not} know the configuration before entering $\CONF_C$, i.e., we cannot determine the value of $S(x)_0$ and thus, we cannot determine the value of any of the $S(x)_n$.
To abbreviate notation, we first define
\[
	A_{\neg C}\is \{\run \in \RUNS \mid \run=\prefix{c_0 c_1 \cdots} \land c_k \in \CONF_{\neg C} \text{ for all } k \in \NN\}, \]
i.e., $A_{\neg C}$ is the set of runs that never visit $\CONF_C$.

\begin{definition}[Prefix Sets for $\CONF_C$]
	\label{Prefix Sets for CONF}
	Let $c_0 \in \CONF_{\neg C}$ and let $f=c_1 \cdots c_n \in \CONF_{C}^k \, \CONF_{\neg C}^{n-k}$ for some $0 < k \leq n$.
	We then define
	\[
		\PrefC{c_0f}{C} \is \bigl\{ \run \in \Inf{c_0f} \mid \\
		c_i = \run[I_1(\run)+i-1] \text{ for } 0 \leq i \leq n \bigr\}
	\]
\end{definition}
Intuitively, $\PrefC{c_0f}{C}$ is the set of runs where the first visit of $\CONF_C$ is entered from $c_0$ and right after $c_0$ the infix part $f$ occurs.
Now we have the right concepts at hand to define the filtration.
\begin{definition}[Filtration for Size Process]
	\label{def:size_filtration}
	For $n>0$ we define
	\[
		\CF_{n} \is \sigmagen{\left\{\PrefC{c_0f}{C} ~\middle \vert~ f \in \CONF_{C}^k\,\CONF_{\neg C}^{n-k}, 0 < k \leq n \land c_0 \in \CONF_{\neg C}\right\} \cup \{A_{\neg C}\}}
	\]
	and
	\[
		\CF_0 \is \sigmagen{\left\{\biguplus_{c_1 \in \CONF_{C}}\PrefC{c_0c_1}{C} \mid c_0 \in \CONF_{\neg C} \right\} \cup \{A_{\neg C}\}}.
	\]
\end{definition}
In this way, whenever $f ,f' \in \CONF_{C}^k\,\CONF_{\neg C}^{n-k}$ and $c_0, c'_0 \in \CONF$ with $(f,c_0) \neq (f',c'_0)$ we have $\PrefC{c_0f}{C} \cap \PrefC{c'_0f'}{C} = \emptyset$.
Furthermore, we have for any $n > 0$
\[
	\biguplus_{k=1}^{n}\quad \biguplus_{f \in \CONF_{C}^k\,\CONF_{\neg C}^{n-k}} \PrefC{f}{C} = \RUNS\setminus A_{\neg C}. \]

For $n \in \NN$, the $\sigma$-field $\CF_n$ captures exactly the information that we need to evaluate $S_n(x)$ on a given $\run$.
Intuitively, two runs $\run_1, \run_2 \in \RUNS$ are distinguishable in $\CF_n$ for $n > 0$, if they do not coincide on the infix of length $n+1$ starting from the point directly before entering $\CONF_C$.
Here, $\CF_0$ is a special case: it can only distinguish between two runs whenever their first visits of $\CONF_C$ are entered from different configurations.

\begin{lemma}[$(S(x)_n)_{n \in \NN}$ is Adapted to the Filtration $(\CF_n)_{n \in \NN}$]
	\label{S(x) Adapted}
	Consider $(\CF_n)_{n \in \NN}$ as defined in \cref{def:size_filtration}.

	\begin{enumerate}[a)]
		\item The sequence of $\sigma$-fields $(\CF_n)_{n \in \NN}$ forms a filtration of $\CF$.

		\item The random variable $S(x)_n$ is $\CF_n$-measurable.
	\end{enumerate}
\end{lemma}
\begin{proof}
	We prove the two statements separately.

	\begin{enumerate}[a)]
		\item Obviously, it is enough to prove $\CF_n \subseteq \CF_{n+1}$ for all $n>0$ since the case $n=0$ is trivial by construction.
		      It is enough to prove that each of the generators given in \cref{def:size_filtration} for $\CF_{n}$ is contained in $\CF_{n+1}$.
		      By construction, $A_{\neg C} \in \CF_{n+1}$.
		      So, let $c_1\cdots c_n \in \CONF_{C}^k \,\CONF_{\neg C}^{n-k}$ for some $0 < k \leq n$ and $c_0 \in \CONF_{\neg C}$.
		      Then we have
		      \[
			      \PrefC{c_0f}{C} = \biguplus_{c \in \CONF} \PrefC{c_0fc}{C} \in \CF_{n+1}, \]
		      so we indeed have $\CF_n \subseteq \CF_{n+1}$.

		\item To prove that $S(x)_n$ is $\CF_n$-measurable we proceed as follows.
		      First of all, $S(x)_n$ is a discrete random variable, i.e., its image is countable.

		      Secondly, if $n = 0$ then for any $z \in \NN$ we have

		      \begin{align*}
			      S(x)_0^{-1}(\{z\}) = \biguplus_{(\_,\_,\state_0)=c_0 \in \CONF_{\neg C}, \abs{\state_0(x)} = z} \left(\biguplus_{c_1 \in \CONF_C} \PrefC{c_0c_1}{C}\right) \in \CF_0.
		      \end{align*}

		      Thirdly, to evaluate $S(x)_n$ for $n>0$ on some $\run \in \RUNS$ we need to know the states in the run which occur at the positions $I_1(\run)-1, I_1(\run),\ldots, I_1(\run)+n-1$.
		      But this is exactly the information provided by $\CF_n$.
		      So, we have for $z \in \NN_{\geq 1}$
		      \begin{align*}
			            & (S(x)_n)^{-1}(\{z\})                                                                                                                        \\
			      {}={} & \biguplus_{c_0 \in \CONF_{\neg C}}\quad\biguplus_{k = 1}^n\quad\biguplus_{\substack{f = c_1\cdots c_n \in \CONF_{C}^k\,\CONF_{\neg C}^{n-k} \\
					      \max\{\abs{\state_0(x)},\ldots,\abs{\state_k(x)}\}=z}} \PrefC{c_0f}{C} \in \CF_n.
		      \end{align*}
	\end{enumerate}
\end{proof}
Still, in this section we want to over-approximate the expected value of $\sizervar_C(x)$.
We will do this by using the stochastic process $(S(x)_n)_{n \in \NN}$ and its adaptedness to the filtration $(\CF_{n})_{n \in \NN}$.

In \cref{eq:monotonicity} we have seen that the $S(x)_n$ form a monotonic sequence whose limit over-approximates $\sizervar_C(x)$ (see
\cref{eq:limit_size}).
So, to achieve our goal, we use the famous monotonic convergence theorem, which we will restate here.

\begin{theorem}[Monotonic Convergence Theorem, see e.g., \protect{\cite[Thm.\ 5.6.12]{grimmett2001probability}}]
	Let $(X_n)_{n \in \NN}$ be a non-negative, monotonic sequence of random variables on some probability space $(\Omega, \mathcal{A},\mathbb{P})$. Then \[\lim_{n \to \infty} \expvsign(X_n) = \expvsign\left( \lim_{n \to \infty} X_n \right).
	\]
\end{theorem}
Using this theorem, from \eqref{eq:limit_size} we directly obtain the following corollary for our process $(S(x)_n)_{n \in \NN}$.

\begin{corollary}[Over-Approximating the Expected Value of $\sizervar_C(x)$]
	\label{coro:size_limit}
	Let $\initstate \in \STATE$ and $\scheduler$ a scheduler. Then we have \[\expv{\pip}{\scheduler}{\initstate}\left(\sizervar_C(x)\right) \leq \lim_{n \to \infty} \expv{\pip}{\scheduler}{\initstate}\left(S(x)_n\right).
	\]
\end{corollary}
We will now over-approximate the expected value $\expv{\pip}{\scheduler}{\initstate}\left(S(x)_n\right)$ in such a way that we can over-approximate the limit $\lim_{n \to \infty} \expv{\pip}{\scheduler}{\initstate}\left(S(x)_n\right)$.
Since $S(x)_n$ is $\CF_n$-measurable, we get the following relation between $S(x)_n$ and $S(x)_{n+1}$ for all $n \geq 0$:

\begin{align}
	\nonumber      & \expv{\pip}{\scheduler}{\initstate}\left(S(x)_{n+1}\right)                                                                                                                                        \\
	\nonumber	{}={} & \expv{\pip}{\scheduler}{\initstate}\left(\expv{\pip}{\scheduler}{\initstate}\left(S(x)_{n+1} \mid \CF_n\right)\right)\tag{by \cref{lemma:expected_value_does_not_change}}                         \\
	\nonumber	{}={} & \expv{\pip}{\scheduler}{\initstate}\left(\expv{\pip}{\scheduler}{\initstate}\left(S(x)_{n+1} + S(x)_n - S(x)_n \mid \CF_n\right)\right)                                                           \\
	\nonumber	{}={} & \expv{\pip}{\scheduler}{\initstate}\left(\expv{\pip}{\scheduler}{\initstate}\left(S(x)_n\mid \CF_n\right) + \expv{\pip}{\scheduler}{\initstate}\left(S(x)_{n+1} - S(x)_n \mid \CF_n\right)\right) \\
	{}={}          & \expv{\pip}{\scheduler}{\initstate}\left(S(x)_n + \expv{\pip}{\scheduler}{\initstate}\left(S(x)_{n+1} - S(x)_n \mid \CF_n\right)\right). \label{reformulation1}
\end{align}
The last step is done by \cref{app_thm:propert_conditional_expectation} (b), since $\expv{\pip}{\scheduler}{\initstate}(S(x)_n)$ is $\CF_n$-measurable by \cref{S(x) Adapted} (b).

So, to begin with, we over-approximate $\expv{\pip}{\scheduler}{\initstate}\left(S(x)_{n+1} - S(x)_n \mid \CF_n\right)$.
Note that this conditional expected value exists since $S(x)_{n+1} \geq S(x)_n$, i.e., $S(x)_{n+1} - S(x)_n$ is non-negative and $S(x)_n$ always takes a finite value for every $n \in \NN$ (see \cite[Prop.\ 3.1]{lexrsm}).

To do so, we introduce the following random variables.
Let $n > 0$ and let $\run \in \RUNS$ be some run.
We define
\begin{align*}
	      & \left(\sum_{\alpha = (g_{I_{n+1}},\_,x) \in C} \eval{\eff{\alpha}}{\initstate}\right) (\run) \\
	{}={} &
	\begin{cases}
		\sum_{\alpha = (g,\_,x) \in C} \eval{\eff{\alpha}}{\initstate}, & \text{if
		} \run = \prefix{(\_,\_,\state_0) \cdots},I_{n+1}(\run) > 0,                                                            \\
		                                                                & \run[I_{n+1}(\run)] = (\_,t,\_) \text{ with } t \in g \\
		0,                                                              & \text{otherwise}
	\end{cases}
\end{align*}
This expression defines a sequence of random variables.
It approximates the expected change of $x$ that can happen when entering the $(n+1)$-th configuration of $C$ during the run $\run$.
Using this definition we can state the following lemma.

\begin{lemma}
	\label{lem:cond_exp_diff}
	For any $n \in \NN$ we have
	\[
		\expv{\pip}{\scheduler}{\initstate}\left(S(x)_{n+1} - S(x)_n \mid \CF_n\right) \leq \expv{\pip}{\scheduler}{\initstate}\left(\sum_{\alpha = (g_{I_{n+1}},\_,x) \in C} \eval{\eff{\alpha}}{\initstate} \middle\vert\ \CF_n\right)
	\]
\end{lemma}

\begin{proof}
	First note that the conditional expectation \[\expv{\pip}{\scheduler}{\initstate}\left(\sum_{\alpha = (g_{I_{n+1}},\_,x) \in C} \eval{\eff{\alpha}}{\initstate} \middle\vert\ \CF_n\right)
	\]
	is $\CF_n$-measurable by definition.
	We prove the result by analyzing the cases $n=0$ and $n>0$ separately.

	\paragraph{Case $n = 0$:}

	By \cref{lem:upper_bound_cond_exp}, we only need to consider the sets
	\[
		A_{c_0}\is \biguplus_{c_1 \in \CONF_{C}}\PrefC{c_0c_1}{C}
	\]
	for some fixed $(\_,\_,\state)=c_0 \in \CONF_{\neg C}$ as both sides are zero everywhere else.
	Moreover, w.l.o.g.\ we can assume that $A_{c_0}$ is not a null-set w.r.t.\ the chosen initial state $\initstate$ and scheduler $\scheduler$, i.e., $\pipmeasure{\pip}{\scheduler}{\initstate}\left(A_{c_0}\right) > 0$.
	We have to prove
	\begin{eqnarray*}
		\lefteqn{\expv{\pip}{\scheduler}{\initstate}\left(\left(S(x)_{1} - S(x)_0\right) \cdot \ind_{A_{c_0}}\right) \;\; \leq}\\
		&& \expv{\pip}{\scheduler}{\initstate}\left( \expv{\pip}{\scheduler}{\initstate}\left(\sum_{\alpha = (g_{I_{1}},\_,x) \in C} \eval{\eff{\alpha}}{\initstate} \middle\vert\ \CF_0\right)\quad \cdot \quad \ind_{A_{c_0}} \right).
	\end{eqnarray*}
	By the definition of conditional expected values (\cref{defQQQconditional_expectation}), $A_{c_0} \in \CF_0$ implies
	\begin{align*}
		      & \expv{\pip}{\scheduler}{\initstate}\left( \expv{\pip}{\scheduler}{\initstate}\left(\sum_{\alpha = (g_{I_{1}},\_,x) \in C} \eval{\eff{\alpha}}{\initstate} \middle\vert\ \CF_0\right)\quad \cdot \quad \ind_{A_{c_0}} \right) \\
		{}={} & \expv{\pip}{\scheduler}{\initstate}\left(\left(\sum_{\alpha = (g_{I_1},\_,x) \in C}
			\eval{\eff{\alpha}}{\initstate}\right) \cdot \ind_{A_{c_0}}\right).
	\end{align*}
	Thus, we now prove that
	\begin{eqnarray*}
		\lefteqn{\expv{\pip}{\scheduler}{\initstate}\left(\left(S(x)_{1} - S(x)_0\right) \cdot \ind_{A_{c_0}}\right) \;\;	\leq}\\
		&& \expv{\pip}{\scheduler}{\initstate}\left(\left(\sum_{\alpha = (g_{I_1},\_,x) \in C}
			\eval{\eff{\alpha}}{\initstate}\right) \cdot \ind_{A_{c_0}}\right).
	\end{eqnarray*}

	We have
	\begin{align*}
		         & \expv{\pip}{\scheduler}{\initstate}\left(\left(S(x)_{1} - S(x)_0\right) \cdot \ind_{A_{c_0}}\right)                                                 \\
		{}={}    & \expv{\pip}{\scheduler}{\initstate}\left(\left(S(x)_{1} - S(x)_0\right) \cdot \ind_{\biguplus_{c_{1} \in \CONF_C} \PrefC{c_0c_1}{C}}\right)         \\
		{}={}    & \expv{\pip}{\scheduler}{\initstate}\left(\sum_{c_{1} \in \CONF_C}\left(S(x)_{1} - S(x)_0\right) \cdot \ind_{\PrefC{c_0c_1}{C}}\right)               \\
		{}={}    & \expv{\pip}{\scheduler}{\initstate}\left(\sum_{c_{1}=(\_,\_,s_{1}) \in \CONF_C}\left(S(x)_{1} - S(x)_0\right) \cdot \ind_{\PrefC{c_0c_1}{C}}\right) \\
		{}={}    & \expv{\pip}{\scheduler}{\initstate}\left(\sum_{\substack{c_{1}=(\_,\_,s_{1}) \in \CONF_C                                                            \\
		\abs{s_{1}(x)} \geq \abs{\state(x)}}}\left(S(x)_{1} - S(x)_0\right) \cdot \ind_{\PrefC{c_0c_1}{C}}\right)                                                      \\
		{}={}    & \expv{\pip}{\scheduler}{\initstate}\left(\sum_{\substack{c_{1}=(\_,\_,s_{1}) \in \CONF_C                                                            \\
		\abs{s_{1}(x)} \geq \abs{\state(x)}}}\left(\abs{s_{1}(x)} - \abs{\state(x)}\right) \cdot \ind_{\PrefC{c_0c_1}{C}}\right)                                       \\
		{}={}    & \expv{\pip}{\scheduler}{\initstate}\left(\sum_{\substack{c_{1}=(\_,\_,s_{1}) \in \CONF_C                                                            \\
		\abs{s_{1}(x)} \geq \abs{\state(x)}}}\left(\abs{\abs{s_{1}(x)} - \abs{\state(x)}}\right) \cdot \ind_{\PrefC{c_0c_1}{C}}\right)                                 \\
		{}\leq{} & \expv{\pip}{\scheduler}{\initstate}\left(\sum_{\substack{c_{1}=(\_,\_,s_{1}) \in \CONF_C                                                            \\
		\abs{s_{1}(x)} \geq \abs{\state(x)}}}\left(\abs{s_{1}(x) - \state(x)}\right) \cdot \ind_{\PrefC{c_0c_1}{C}}\right)                                             \\
		{}={}    & \sum_{\substack{c_{1}=(\_,\_,s_{1}) \in \CONF_C                                                                                                     \\
		\abs{s_{1}(x)} \geq \abs{\state(x)}}}\left(\abs{s_{1}(x) - \state(x)}\right) \cdot \expv{\pip}{\scheduler}{\initstate}\left(\ind_{\PrefC{c_0c_1}{C}}\right)    \\
		{}={}    & \sum_{\substack{c_{1}=(\_,\_,s_{1}) \in \CONF_C                                                                                                     \\
		\abs{s_{1}(x)} \geq \abs{\state(x)}}}\left(\abs{s_{1}(x) - \state(x)}\right) \cdot \pipmeasure{\pip}{\scheduler}{\initstate}\left(\PrefC{c_0c_1}{C}\right)
	\end{align*}

	Now we would like to ``split'' the set $\PrefC{c_0c_1}{C}$.
	We are only interested in runs where $c_0$ is the configuration directly before entering $\CONF_C$.
	However, if we know that a run $\run$ satisfies $\run[I_1(\run)-1] = c_0$ it could still be the case that $c_0$ occurs multiple times in $\run$ and, e.g., the third or fourth occurrence leads into $\CONF_C$.
	Thus, we now explain why this cannot happen in any admissible run before we continue with the proof.

	So let $\run$ be an admissible run which enters $C$ and let $\run[I_1(\run) - 1] = c_0 = (\loc,t,\state)$ such that $\run$ contains multiple occurrences of $c_0$.
	Moreover, let $\run[I_1(\run)] = c_1 = (\loc',t_1,\_)$.
	By definition of the scheduler, we have $\scheduler(c_0)=(g_1,\_)$, where $t_1 \in g_1$.
	Since from $c_0$ we can visit $C$ but do not need to, there are locations $\loc'\neq \loc''$ such that $(g_1,\loc',x) \in C$ for some variable $x$ but $(g_1,\loc'',y) \not \in C$ for every $y \in \PV$.
	Since $(g_1,\loc',x) \in C$ and $C$ is an SCC, there is some general transition $g_2$ reachable from $g_1$ (in possibly several steps) such that after $g_2$ the general transition $g_1$ can be executed again.
	Thus, $\loc$ is a target location of $g_2$.
	Hence there is an edge from $(g_2,\loc,x)$ to $(g_1,\loc',x)$ in the general result variable graph.
	Moreover, there is also an edge from $(g_2,\loc,x)$ to $(g_1,\loc'',x)$, since $g_2$ can be executed before $g_1$.
	So there is a path from $(g_1,\loc',x)$ via $(g_2,\loc,x)$ to $(g_1,\loc'',x)$ in the graph.

	Since $c_0$ occurs multiple times in the run before $C$ is entered, there is also a general transition $g_3$ with target location $\loc$ which is reachable from $g_1$ when ending with $g_1$ in $\loc''$.
	Thus, there is some path from $(g_1,\loc'',x)$ via $(g_3,\loc,x)$ to $(g_1,\loc',x)$.
	Combining both parts, there is a path from $(g_1,\loc',x)$ to $(g_1,\loc'',x)$ in the general result variable graph and vice versa, i.e., $(g_1,\loc'',x) \in C$, which is a contradiction.
	Therefore, $c_0$ can occur at most once in any admissible run.

	So, we can now continue our reasoning from above.
	Again, $g_1$ is the unique general transition which the scheduler chooses in $c_0$, i.e., $\scheduler(c_0)=(g_1,\_)$.
	\begin{align*}
		         & \sum_{\substack{c_{1}=(\_,\_,s_{1}) \in \CONF_C                                                                                                                                                                        \\
		\abs{s_{1}(x)} \geq \abs{\state(x)}}}\left(\abs{s_{1}(x) - \state(x)}\right) \cdot \pipmeasure{\pip}{\scheduler}{\initstate}\left(\PrefC{c_0c_1}{C}\right)                                                                        \\
		{}={}    & \sum_{
		\substack{c_{1}=(\_,\_,s_{1}) \in \CONF_C                                                                                                                                                                                         \\
		\abs{s_{1}(x)} \geq \abs{\state(x)}}
		}\left(\abs{s_{1}(x) - \state(x)}\right) \cdot \ptransition{\pip}{\scheduler}{\initstate}(\prefix{c_0 \to c_1}) \cdot \pipmeasure{\pip}{\scheduler}{\initstate} (A_{c_0})                                                         \\
		{}={}    & \left(\sum_{\substack{c_{1}=(\_,\_,s_{1}) \in \CONF_C                                                                                                                                                                  \\
		\abs{s_{1}(x)} \geq \abs{\state(x)}}}\left(\abs{s_{1}(x) - \state(x)}\right) \cdot \ptransition{\pip}{\scheduler}{\initstate}\left(\prefix{c_0 \to c_{1}}\right)\right) \cdot \pipmeasure{\pip}{\scheduler}{\initstate} (A_{c_0}) \\
		{}={}    & \left(\sum_{\substack{\loc\in \LOC                                                                                                                                                                                     \\c_{1}=(\loc,\_,s_{1}) \in \CONF_C \\
		\abs{s_{1}(x)} \geq \abs{\state(x)}}}\left(\abs{s_{1}(x) - \state(x)}\right) \cdot \ptransition{\pip}{\scheduler}{\initstate}\left(\prefix{c_0 \to c_{1}}\right)\right) \cdot \pipmeasure{\pip}{\scheduler}{\initstate}
		(A_{c_0})                                                                                                                                                                                                                         \\
		{}\leq{} & \sum\limits_{\loc \in \LOC}\; \left(\sum_{\substack{c_{1}=(\loc,t,s_{1}) \in \CONF_C                                                                                                                                   \\
		t \in g_1}}
		\ptransition{\pip}{\scheduler}{\initstate}\left(\prefix{c_0 \to c_{1}}\right) \cdot \left(\abs{s_{1}(x) - \state(x)}\right) \right) \quad \cdot \quad \pipmeasure{\pip}{\scheduler}{\initstate}
		(A_{c_0}) \tag{$\dagger$}                                                                                                                                                                                                         \\
		{}\leq{} & \left(\sum_{\alpha = (g_{1},\_,x) \in C} \eval{\eff{\alpha}}{\initstate}\right) \cdot \pipmeasure{\pip}{\scheduler}{\initstate}
		(A_{c_0}) \tag{by \cref{theorem:elsb_basic}} \label{eq:step_condeff}
		\\
		{}={}    & \expv{\pip}{\scheduler}{\initstate}\left(\left(\sum_{\alpha = (g_{I_1},\_,x) \in C}
			\eval{\eff{\alpha}}{\initstate}\right) \cdot \ind_{A_{c_0}}\right)\tag{by definition }
	\end{align*}

	Note that in $(\dagger)$ we have used that the general transition $g_1$ which is used to continue $c_0$ is unique.
	This proves the case $n = 0$.

	\paragraph{Case $n > 0$:}

	Since we have constructed $\CF_n$ using a generating system of countably many atoms which cover the sample space $\RUNS$, by \cref{lem:upper_bound_cond_exp} it is sufficient to only consider the generators of $\CF_n$, similar to the case $n = 0$.
	Thus, let $0< k \leq n$, $c_0 \in \CONF_{\neg C}$, and $c_1 \, \cdots \, c_n=f \in \CONF_{C}^k\,\CONF_{\neg C}^{n-k}$.
	We have to prove that
	\begin{align*}
		         & \expv{\pip}{\scheduler}{\initstate}\left(\left(S(x)_{n+1} - S(x)_n\right) \cdot \ind_{\PrefC{c_0f}{C}}\right)                                                                                                                          \\
		{}\leq{} & \expv{\pip}{\scheduler}{\initstate}\left(\expv{\pip}{\scheduler}{\initstate}\left(\sum_{\alpha = (g_{I_{n+1}},\_,x) \in C} \eval{\eff{\alpha}}{\initstate} \middle\vert\ \CF_n\right) \quad \cdot \quad \ind_{\PrefC{c_0f}{C}}\right).
	\end{align*}
	By the definition of conditional expected values (\cref{defQQQconditional_expectation}), $\PrefC{c_0f}{C} \in \CF_n$ implies
	\begin{align*}
		      & \expv{\pip}{\scheduler}{\initstate}\left(\expv{\pip}{\scheduler}{\initstate}\left(\sum_{\alpha = (g_{I_{n+1}},\_,x) \in C} \eval{\eff{\alpha}}{\initstate} \middle\vert\ \CF_n\right) \quad \cdot \quad \ind_{\PrefC{c_0f}{C}}\right) \\
		{}={} & \expv{\pip}{\scheduler}{\initstate}\left(\left(\sum_{\alpha = (g_{I_{n+1}},\_,x) \in C} \eval{\eff{\alpha}}{\initstate}\right) \cdot \ind_{\PrefC{c_0f}{C}}\right).
	\end{align*}
	Hence, it suffices to prove that
	\begin{align*}
		         & \expv{\pip}{\scheduler}{\initstate}\left(\left(S(x)_{n+1} - S(x)_n\right) \cdot \ind_{\PrefC{c_0f}{C}}\right)                                                       \\
		{}\leq{} & \expv{\pip}{\scheduler}{\initstate}\left(\left(\sum_{\alpha = (g_{I_{n+1}},\_,x) \in C} \eval{\eff{\alpha}}{\initstate}\right) \cdot \ind_{\PrefC{c_0f}{C}}\right).
	\end{align*}
	First of all, we left out the generator $A_{\neg C}$ since here the claim holds trivially, as both sides are zero.
	Secondly, whenever $k < n$, then $S(x)_{n+1} - S(x)_n$ is zero on all admissible runs in $\PrefC{c_0f}{C}$, i.e., the left-hand side of the inequation is zero and the claim holds trivially.
	So, w.l.o.g., we can assume that $n = k$.
	Let $c_i=(\_,\_,\state'_i)$ for $0 \leq i \leq n$.
	Due to the chosen scheduler $\scheduler$, the transition that is executed to continue in an admissible run from $c_n$ is an element of a \emph{unique} general transition $g_{n+1}$.
	Furthermore, if a run $\run$ leaves $\CONF_C$ after $c_n$ then $S(x)_{n+1} (\run) - S(x)_n (\run)$ is zero as well.

	\begin{align*}
		         & \expv{\pip}{\scheduler}{\initstate}\left(\left(S(x)_{n+1} - S(x)_n\right) \cdot \ind_{\PrefC{c_0f}{C}}\right)                                                                                                                                                                \\
		{}={}    & \expv{\pip}{\scheduler}{\initstate}\left(\left(S(x)_{n+1} - S(x)_n\right) \cdot \ind_{\biguplus_{c_{n+1} \in \CONF} \PrefC{c_0fc_{n+1}}{C}}\right)                                                                                                                           \\
		{}={}    & \expv{\pip}{\scheduler}{\initstate}\left(\sum_{c_{n+1} \in \CONF}\left(S(x)_{n+1} - S(x)_n\right) \cdot \ind_{\PrefC{c_0fc_{n+1}}{C}}\right)                                                                                                                                 \\
		{}={}    & \expv{\pip}{\scheduler}{\initstate}\left(\sum_{c_{n+1}=(\_,\_,\state'_{n+1}) \in \CONF_C}\left(S(x)_{n+1} - S(x)_n\right) \cdot \ind_{\PrefC{c_0fc_{n+1}}{C}}\right)                                                                                                         \\
		{}={}    & \expv{\pip}{\scheduler}{\initstate}\left(\sum_{\substack{c_{n+1}=(\_,\_,\state'_{n+1}) \in \CONF_C                                                                                                                                                                           \\
		\abs{\state'_{n+1}(x)} \geq \max_{0 \leq i \leq n}\abs{\state'_i(x)}}}\left(S(x)_{n+1} - S(x)_n\right) \cdot \ind_{\PrefC{c_0fc_{n+1}}{C}}\right)                                                                                                                                       \\
		{}={}    & \expv{\pip}{\scheduler}{\initstate}\left(\sum_{\substack{c_{n+1}=(\_,\_,\state'_{n+1}) \in \CONF_C                                                                                                                                                                           \\
		\abs{\state'_{n+1}(x)} \geq \max_{0 \leq i \leq n}\abs{\state'_i(x)}}}\left(\abs{\state'_{n+1}(x)} - \max_{0 \leq i \leq n}\abs{\state'_i(x)}\right) \cdot \ind_{\PrefC{c_0fc_{n+1}}{C}}\right)                                                                                         \\
		{}\leq{} & \expv{\pip}{\scheduler}{\initstate}\left(\sum_{\substack{c_{n+1}=(\_,\_,\state'_{n+1}) \in \CONF_C                                                                                                                                                                           \\
		\abs{\state'_{n+1}(x)} \geq \max_{0 \leq i \leq n}\abs{\state'_i(x)}}}\left(\abs{\state'_{n+1}(x)} - \abs{\state'_n(x)}\right) \cdot \ind_{\PrefC{c_0fc_{n+1}}{C}}\right)                                                                                                               \\
		{}={}    & \expv{\pip}{\scheduler}{\initstate}\left(\sum_{\substack{c_{n+1}=(\_,\_,\state'_{n+1}) \in \CONF_C                                                                                                                                                                           \\
		\abs{\state'_{n+1}(x)} \geq \max_{0 \leq i \leq n}\abs{\state'_i(x)}}}\left(\abs{\abs{\state'_{n+1}(x)} - \abs{\state'_n(x)}}\right) \cdot \ind_{\PrefC{c_0fc_{n+1}}{C}}\right)                                                                                                         \\
		{}\leq{} & \expv{\pip}{\scheduler}{\initstate}\left(\sum_{\substack{c_{n+1}=(\_,\_,\state'_{n+1}) \in \CONF_C                                                                                                                                                                           \\
		\abs{\state'_{n+1}(x)} \geq \max_{0 \leq i \leq n}\abs{\state'_i(x)}}}\left(\abs{\state'_{n+1}(x) - \state'_n(x)}\right) \cdot \ind_{\PrefC{c_0fc_{n+1}}{C}}\right)                                                                                                                     \\
		{}={}    & \sum_{\substack{c_{n+1}=(\_,\_,\state'_{n+1}) \in \CONF_C                                                                                                                                                                                                                    \\
		\abs{\state'_{n+1}(x)} \geq \max_{0 \leq i \leq n}\abs{\state'_i(x)}}}\left(\abs{\state'_{n+1}(x) - \state'_n(x)}\right) \cdot \expv{\pip}{\scheduler}{\initstate}\left(\ind_{\PrefC{c_0fc_{n+1}}{C}}\right)                                                                            \\
		{}={}    & \sum_{\substack{c_{n+1}=(\_,\_,\state'_{n+1}) \in \CONF_C                                                                                                                                                                                                                    \\
		\abs{\state'_{n+1}(x)} \geq \max_{0 \leq i \leq n}\abs{\state'_i(x)}}}\left(\abs{\state'_{n+1}(x) - \state'_n(x)}\right) \cdot \pipmeasure{\pip}{\scheduler}{\initstate}\left(\PrefC{c_0fc_{n+1}}{C}\right)                                                                             \\
		{}={}    & \left(\sum_{\substack{c_{n+1}=(\_,\_,\state'_{n+1}) \in \CONF_C                                                                                                                                                                                                              \\
		\abs{\state'_{n+1}(x)} \geq \max_{0 \leq i \leq n}\abs{\state'_i(x)}}}\left(\abs{\state'_{n+1}(x) - \state'_n(x)}\right) \cdot \ptransition{\pip}{\scheduler}{\initstate} (\prefix{c_n \to c_{n+1}})\right) \cdot \pipmeasure{\pip}{\scheduler}{\initstate}\left(\PrefC{c_0f}{C}\right) \\
		{}\leq{} & \left(\sum_{\alpha = (g_{n+1},\_,x) \in C} \eval{\chbounde(\alpha)}{\state'_n}\right) \cdot \pipmeasure{\pip}{\scheduler}{\initstate}\left(\PrefC{c_0f}{C}\right) \tag{by \cref{theorem:elsb_basic}}                                                                         \\
		{}\leq{} & \left(\sum_{\alpha = (g_{n+1},\_,x) \in C} \eval{\eff{\alpha}}{\initstate}\right) \cdot \pipmeasure{\pip}{\scheduler}{\initstate}\left(\PrefC{c_0f}{C}\right)                                                                                                                \\
		{}={}    & \expv{\pip}{\scheduler}{\initstate}\left(\left(\sum_{\alpha = (g_{n+1},\_,x) \in C} \eval{\eff{\alpha}}{\initstate}\right) \cdot \ind_{\PrefC{c_0f}{C}}\right)                                                                                                               \\
		{}={}    & \expv{\pip}{\scheduler}{\initstate}\left(\left(\sum_{\alpha = (g_{I_{n+1}},\_,x) \in C} \eval{\eff{\alpha}}{\initstate}\right) \cdot \ind_{\PrefC{c_0f}{C}}\right),                                                                                                          \\
	\end{align*}
	which proves the case $n>0$.
\end{proof}
Moreover, we can over-approximate the expected value
\[
	\expv{\pip}{\scheduler}{\initstate}\left(\sum_{\alpha = (g_{I_{n+1}},\_,x) \in C} \eval{\eff{\alpha}}{\initstate}\right). \]
For this expected value, we can ignore the set of non-admissible runs w.r.t.\ the chosen scheduler $\scheduler$ and the chosen initial state $\initstate$, since this set has probability zero.
Furthermore, we can ignore those runs, in which the random variable takes the value $0$, i.e., especially those ones, which visit $\CONF_C$ less then $n+1$ times.
Thus, to compute the expected value we only need to consider the set $\biguplus_{c_0 \in \CONF_{\neg C},\ f \in \CONF_C^{n+1}} \PrefC{c_0f}{C}$, since $\sum_{\alpha = (g_{I_{n+1}},\_,x) \in C} \eval{\eff{\alpha}}{\initstate}$ takes the value zero on any admissible run in the complement.
Then we have
\begin{align}
	\nonumber      & \expv{\pip}{\scheduler}{\initstate}\left(\sum_{\alpha = (g_{I_{n+1}},\_,x) \in C} \eval{\eff{\alpha}}{\initstate}\right)                                                        \\
	\nonumber	{}={} & \expv{\pip}{\scheduler}{\initstate}\left(\ind_{\biguplus_{\substack{c_0 \in \CONF_{\neg C},                                                                                     \\
	f = c_1 \cdots c_{n+1} \in \CONF_C^{n+1}} }\PrefC{c_0f}{C}} \cdot \sum_{\alpha = (g_{I_{n+1}},\_,x) \in C} \eval{\eff{\alpha}}{\initstate}\right)                                                \\
	\nonumber	{}={} & \sum_{\substack{c_0 \in \CONF_{\neg C},                                                                                                                                         \\
	f = c_1 \cdots c_{n+1} \in \CONF_C^{n+1}}}\expv{\pip}{\scheduler}{\initstate}\left(\ind_{ \PrefC{c_0f}{C}} \cdot \sum_{\alpha = (g_{I_{n+1}},\_,x) \in C} \eval{\eff{\alpha}}{\initstate}\right) \\
	\nonumber	{}={} & \sum_{\substack{c_0 \in \CONF_{\neg C},                                                                                                                                         \\
	f = c_1 \cdots c_{n+1} \in \CONF_C^{n+1}}}\expv{\pip}{\scheduler}{\initstate}\left(\ind_{ \PrefC{c_0f}{C}} \cdot \sum_{\alpha = (g_{n+1},\_,x) \in C} \eval{\eff{\alpha}}{\initstate}\right)     \\
	\nonumber	{}={} & \sum_{\substack{c_0 \in \CONF_{\neg C},                                                                                                                                         \\
	f = c_1 \cdots c_{n+1} \in \CONF_C^{n+1}}}\expv{\pip}{\scheduler}{\initstate}\left(\ind_{ \PrefC{c_0f}{C}}\right) \cdot \sum_{\alpha = (g_{n+1},\_,x) \in C} \eval{\eff{\alpha}}{\initstate}     \\
	{}={}          & \sum_{\substack{c_0 \in \CONF_{\neg C},                                                                                                                                         \\
			f = c_1 \cdots c_{n+1} \in \CONF_C^{n+1}}}\pipmeasure{\pip}{\scheduler}{\initstate}\left( \PrefC{c_0f}{C}\right) \cdot \sum_{\alpha = (g_{n+1},\_,x) \in C} \eval{\eff{\alpha}}{\initstate} \label{reformulation2}
\end{align}

Finally, we have all information at hand to prove \cref{theorem:expectednontrivialsizeboundsmeth}.

\expectednontrivialsizeboundsmethod*
\begin{proof}
	Let $n \geq 0$.
	Then using our results obtained so far in this subsection we have
	\begin{align*}
		         & \expv{\pip}{\scheduler}{\initstate}\left(S(x)_{n+1}\right)                                                                                                                                                                                     \\
		{}={}    & \expv{\pip}{\scheduler}{\initstate}\left(S(x)_{n}
		+ \expv{\pip}{\scheduler}{\initstate}\left(S(x)_{n+1}-S(x)_n \middle\vert\ \CF_n\right)\right) \tag{by \eqref{reformulation1}}                                                                                                                            \\
		{}\leq{} & \expv{\pip}{\scheduler}{\initstate}\left(S(x)_{n}
		+\expv{\pip}{\scheduler}{\initstate}\left(\sum_{\alpha = (g_{I_{n+1}},\_,x) \in C} \eval{\eff{\alpha}}{\initstate} \middle\vert\ \CF_n\right)\right) \tag{by \cref{lem:cond_exp_diff}}                                                                    \\
		{}={}    & \expv{\pip}{\scheduler}{\initstate}\left(S(x)_{n}\right) \; + \; \expv{\pip}{\scheduler}{\initstate}\left(\sum_{\alpha = (g_{I_{n+1}},\_,x) \in C} \eval{\eff{\alpha}}{\initstate}\right) \tag{by \cref{lemma:expected_value_does_not_change}} \\
		{}={}    & \expv{\pip}{\scheduler}{\initstate}\left(S(x)_{n}\right)                                                                                                                                                                                       \\
		         & + \sum_{\substack{c_0 \in \CONF_{\neg C},                                                                                                                                                                                                      \\
		f = c_1\cdots c_{n+1} \in \CONF_C^{n+1}}}\pipmeasure{\pip}{\scheduler}{\initstate}\left( \PrefC{c_0f}{C}\right) \cdot \sum_{\alpha = (g_{n+1},\_,x) \in C} \eval{\eff{\alpha}}{\initstate} \tag{by \eqref{reformulation2}}                                \\
	\end{align*}

	Iterating this argument $(n+1)$-times then yields:

	\begin{align}
		\nonumber & \expv{\pip}{\scheduler}{\initstate}\left(S(x)_{n+1}\right)         \\
		{}\leq{}  & \expv{\pip}{\scheduler}{\initstate}\left(S(x)_{0}\right) \nonumber \\&+ \sum_{i = 1}^{n+1}\sum_{\substack{c_0 \in \CONF_{\neg C},\\
				f = c_1\cdots c_{i} \in \CONF_C^{i}}}\pipmeasure{\pip}{\scheduler}{\initstate}\left( \PrefC{c_0f}{C}\right) \cdot \sum_{\alpha = (g_{i},\_,x) \in C} \eval{\eff{\alpha}}{\initstate}. \label{observe1}
	\end{align}

	Remember that with $g_i$ we denote the general transition which leads to configuration $c_i$, i.e., $c_i = (t_i,\_,\_)$ and $t_i \in g_i$.
	This notation will be used again in the next calculations.

	Let us elaborate on
	\[
		\sum_{i = 1}^{n+1} \; \sum_{c_0 \in \CONF_{\neg C},\ f = c_1\cdots c_{i} \in \CONF_C^{i}}\pipmeasure{\pip}{\scheduler}{\initstate}\left( \PrefC{c_0f}{C}\right) \cdot \sum_{\alpha = (g_{i},\_,x) \in C} \eval{\eff{\alpha}}{\initstate}. \]

	To increase readability of the following reasoning, we introduce a shorthand notation.

	\[
		\chsum^C (g) =
		\begin{cases}
			\sum_{\alpha = (g,\_,x)\in C} \eval{\eff{\alpha}}{\initstate}, & \text{if $g\in \GT_C$} \\
			0,                                                             & \text{otherwise}
		\end{cases}
		. \]
	Then we have
	\begin{align*}
		      & \sum_{i = 1}^{n+1}\; \sum_{\substack{c_0 \in \CONF_{\neg C},                                                                                                                                                            \\
				f = c_1\cdots c_{i} \in \CONF_C^{i}}}
		\pipmeasure{\pip}{\scheduler}{\initstate}\left( \PrefC{c_0f}{C}\right) \cdot \sum_{\alpha = (g_{i},\_,x) \in C} \eval{\eff{\alpha}}{\initstate}                                                                                 \\
		{}={} & \sum_{i = 1}^{n+1} \; \sum_{\substack{c_0 \in \CONF_{\neg C},                                                                                                                                                           \\
				f = c_1\cdots c_{i} \in \CONF_C^{i}}}
		\pipmeasure{\pip}{\scheduler}{\initstate}\left( \PrefC{c_0f}{C}\right) \cdot \chsum^C (g_i)                                                                                                                                     \\
		{}={} & \sum_{i = 1}^{n+1}\; \sum_{\substack{c_0 \in \CONF_{\neg C},                                                                                                                                                            \\
				f = c_1\cdots c_{i} \in \CONF_C^{i}}}
		\pipmeasure{\pip}{\scheduler}{\initstate}\left( \biguplus_{c_{i+1}\cdots c_{n+1} \in \CONF^{n+1-i}}\PrefC{c_0fc_{i+1}\cdots c_{n+1}}{C}\right) \cdot \chsum^C (g_i)                                                             \\
		{}={} & \sum_{i = 1}^{n+1}\; \sum_{\substack{c_0 \in \CONF_{\neg C},                                                                                                                                                            \\
				f = c_1\cdots c_{i} \in \CONF_C^{i}}}
		\;	\sum_{c_{i+1}\cdots c_{n+1} \in \CONF^{n+1-i}}\pipmeasure{\pip}{\scheduler}{\initstate}\left(\PrefC{c_0fc_{i+1}\cdots c_{n+1}}{C}\right) \cdot \chsum^C (g_i)                                                                 \\
		{}={} & \sum_{i = 1}^{n+1}\; \sum_{\substack{c_0 \in \CONF_{\neg C},                                                                                                                                                            \\
				f = c_1\cdots c_{i-1} \in \CONF_C^{i-1}}}
		\;	\sum_{c_i\cdots c_{n+1} \in \CONF^{n+1-(i-1)}}\pipmeasure{\pip}{\scheduler}{\initstate}\left(\PrefC{c_0fc_i\cdots c_{n+1}}{C}\right) \cdot \chsum^C (g_i) \tag{$\chsum^C (g_i) = 0$ if $g_i \not \in \GT_C$}                  \\
		{}={} & \sum_{i = 1}^{n+1}\; \sum_{c_0 \in \CONF_{\neg C},\ c_1 \in \CONF_C}
		\;	\sum_{c_2\cdots c_{n+1} \in \CONF^{n}}\pipmeasure{\pip}{\scheduler}{\initstate}\left(\PrefC{c_0c_1c_2\cdots c_{n+1}}{C}\right) \cdot \chsum^C (g_i) \tag{since $\CONF_C$ can only be entered once in an admissible run}       \\
		{}={} & \sum_{c_0 \in \CONF_{\neg C},\ c_1 \in \CONF_C}\; \sum_{c_2\cdots c_{n+1} \in \CONF^{n}}\pipmeasure{\pip}{\scheduler}{\initstate}\left(\PrefC{c_0c_1c_2\cdots c_{n+1}}{C}\right) \cdot \sum_{i = 1}^{n+1}\chsum^C (g_i)
	\end{align*}

	Since we fixed $c_1\in\CONF_C$, for $f \neq f' \in \CONF^n$ the intersection $\PrefC{c_0c_1f}{C} \cap \PrefC{c_0c_1f'}{C}$ only contains non-admissible runs.
	This is the reason for the slight abuse of notation of $\PrefC{}{C}$ in the reasoning above which increases readability.
	Here, if for no $k\geq 0$ we have $f \in \CONF_C^k\CONF_{\neg C}^{n-k}$, then the set $\PrefC{f}{C}$ just contains non-admissible runs, i.e., it is a null-set.
	Again, we assume that any configuration $c_i$ has the form $(\_,t_i,\_)$.
	Hence, we conclude
	\begin{align}
		\nonumber         & \sum_{c_0 \in \CONF_{\neg C},\ c_1 \in \CONF_C} \; \sum_{c_2\cdots c_{n+1} \in \CONF^{n}}\pipmeasure{\pip}{\scheduler}{\initstate}\left(\PrefC{c_0c_1c_2\cdots c_{n+1}}{C}\right)\cdot \sum_{i = 1}^{n+1}\chsum^C (g_i) \\
		\nonumber		{}={}    & \sum_{\substack{c_0\in\CONF_{\neg C}                                                                                                                                                                                    \\
		c_1\in\CONF_C                                                                                                                                                                                                                               \\
				c_2\cdots c_{n+1}\in\CONF^n}}
		\left( \pipmeasure{\pip}{\scheduler}{\initstate} (\PrefC{c_0\cdots c_{n+1}}{C}) \cdot \sum_{g\in\GT_C} \sum_{1\leq i\leq n+1, t_i \in g} \chsum^C(g) \right)                                                                                \\
		\nonumber	{}\leq{} & \sum_{\substack{c_0\in\CONF_{\neg C}                                                                                                                                                                                    \\
		c_1\in\CONF_C                                                                                                                                                                                                                               \\
				c_2\cdots c_{n+1}\in\CONF^n}}
		\left( \pipmeasure{\pip}{\scheduler}{\initstate} (\PrefC{c_0\cdots c_{n+1}}{C}) \cdot \sum_{g\in\GT_C} \min \{\timervar (g)(\run) \mid \run \in \PrefC{c_0\cdots c_{n+1}}{C}\} \cdot \chsum^C(g) \right)                                    \\
		\nonumber		{}={}    & \sum_{\substack{c_0\in\CONF_{\neg C}                                                                                                                                                                                    \\
		c_1\in\CONF_C                                                                                                                                                                                                                               \\
				c_2\cdots c_{n+1}\in\CONF^n}}
		\left(\sum_{g\in\GT_C} \left( \pipmeasure{\pip}{\scheduler}{\initstate} (\PrefC{c_0\cdots c_{n+1}}{C}) \cdot \min \{\timervar (g)(\run) \mid \run \in \PrefC{c_0\cdots c_{n+1}}{C}\}\right) \cdot \chsum^C(g) \right)                       \\
		\nonumber	{}\leq{} & \sum_{\substack{c_0\in\CONF_{\neg C}                                                                                                                                                                                    \\
		c_1\in\CONF_C                                                                                                                                                                                                                               \\
				c_2\cdots c_{n+1}\in\CONF^n}}
		\left( \sum_{g\in\GT_C} \expv{\pip}{\scheduler}{\initstate} \left(\timervar (g)\cdot \ind_{\PrefC{c_0\cdots c_{n+1}}{C}}\right) \cdot \chsum^C(g) \right)                                                                                   \\
		\nonumber	{}={}    & \sum_{g\in\GT_C} \sum_{\substack{c_0\in\CONF_{\neg C}                                                                                                                                                                   \\
		c_1\in\CONF_C                                                                                                                                                                                                                               \\
				c_2\cdots c_{n+1}\in\CONF^n}}
		\left(\expv{\pip}{\scheduler}{\initstate} \left(\timervar (g)\cdot \ind_{\PrefC{c_0\cdots c_{n+1}}{C}}\right) \cdot \chsum^C(g) \right)                                                                                                     \\
		\nonumber		{}={}    & \sum_{g\in\GT_C}
		\left(\expv{\pip}{\scheduler}{\initstate} \left(\timervar (g)\cdot \ind_{\biguplus_{\substack{c_0\in\CONF_{\neg C}                                                                                                                          \\
		c_1\in\CONF_C                                                                                                                                                                                                                               \\
		c_2\cdots c_{n+1}\in\CONF^n}}\PrefC{c_0\cdots c_{n+1}}{C}}\right) \cdot \chsum^C(g) \right)                                                                                                                                                 \\
		\nonumber	{}\leq{} & \sum_{g\in \GT_C} (\chsum^C (g) \cdot \expv{\pip}{\scheduler}{\initstate} (\timervar (g)))                                                                                                                              \\
		{}={}             & \sum_{g\in \GT_C} \left( \expv{\pip}{\scheduler}{\initstate} (\timervar (g)) \cdot \sum_{\alpha = (g,\_,x)\in C} \eval{\eff{\alpha}}{\initstate} \right) \label{observe2}
	\end{align}

	Let us abbreviate $\sum_{(\beta,\alpha)\in \mathcal{GRVE}, \; \beta \notin C, \; \alpha \in C, \; \beta = (\_,\_,x)} \; \sbounde (\beta)$ by $\incsize^C_{\mathbb{E}}(x)$.
	The value $ \expv{\pip}{\scheduler}{\initstate}(S (x)_0)$ is bounded by $\eval{\incsize^C_{\mathbb{E}}(x)}{\initstate}$, since $\sbounde$ is a valid expected size bound.
	By combining our previous observations we get:
	\begin{align*}
		         & \expv{\pip}{\scheduler}{\initstate}(S(x)_{n+1})                                                                                                                                                 \\
		{}\leq{} & \expv{\pip}{\scheduler}{\initstate} (S(x)_{0}) + \sum_{i = 1}^{n+1}
		\sum_{\substack{c_0 \in \CONF_{\neg C}                                                                                                                                                                     \\
				f = c_1\cdots c_{i} \in \CONF_C^{i}}}
		\pipmeasure{\pip}{\scheduler}{\initstate}\left( \PrefC{c_0f}{C}\right) \cdot \sum_{\alpha = (g_{i},\_,x) \in C} \eval{\eff{\alpha}}{\initstate} \tag{by \eqref{observe1}}                                  \\
		{}\leq{} & \expv{\pip}{\scheduler}{\initstate} (S(x)_{0}) + \sum_{g\in\GT_C} \left ( \expv{\pip}{\scheduler}{\initstate}(\timervar (g)) \cdot \sum_{\alpha = (g,\_,x)\in C}\eval{\eff{\alpha}}{\initstate}
		\right) \tag{by \eqref{observe2}}                                                                                                                                                                          \\
		{}\leq{} & \eval{\incsize^C_{\mathbb{E}}(x)}{\initstate}
		+ \eval{\sum_{g \in \GT_C} \left( \tbounde (g)\cdot \sum_{\alpha = (g,\_,x)\in C}\eff{\alpha} \right)}	{\initstate}
		\tag{as $(\tbounde, \sbounde)$ is an expected bound pair}                                                                                                                                                  \\
		{}={}    & \eval{\incsize^C_{\mathbb{E}}(x) + \sum_{g \in \GT_C} \left( \tbounde (g)\cdot \sum_{\alpha = (g,\_,x)\in C}\eff{\alpha} \right)}{\initstate}
	\end{align*}
	So we have
	\[
		\expv{\pip}{\scheduler}{\initstate} (S(x)_{n+1}) \leq \eval{\incsize^C_{\mathbb{E}}(x) + \sum_{g \in \GT_C}\tbounde (g)\cdot \sum_{\alpha = (g,\_,x)\in C} \eff{\alpha}}{\initstate}. \]
	Hence, we also have
	\[
		\lim_{n \to \infty} \expv{\pip}{\scheduler}{\initstate} (S(x)_{n+1}) \leq \eval{\incsize^C_{\mathbb{E}}(x) + \sum_{g \in \GT_C}\tbounde (g) \cdot \sum_{\alpha = (g,\_,x)\in C} \eff{\alpha}}{\initstate}. \]
	By \cref{coro:size_limit} this means that \[\expv{\pip}{\scheduler}{\initstate}\left(\sizervar_C(x)\right) \leq \eval{\incsize^C_{\mathbb{E}}(x) + \sum_{g \in \GT_C}\tbounde (g) \cdot \sum_{\alpha = (g,\_,x)\in C} \eff{\alpha}}{\initstate}.\]
	So for any $\alpha' = (\_,\_,x) \in C$, \eqref{observe3} implies
	\begin{align*}
		         & \expv{\pip}{\scheduler}{\initstate}\left(\sizervar(\alpha')\right)                                                              \\
		{}\leq{} & \eval{\incsize^C_{\mathbb{E}}(x) + \sum_{g \in \GT_C}\tbounde (g) \cdot \sum_{\alpha = (g,\_,x)\in C} \eff{\alpha}}{\initstate} \\
		{}={}    & \eval{\sbounde'(\alpha')}{\initstate}
	\end{align*}
	Hence, $\sbounde'$ is indeed an expected size bound.
\end{proof}

%% file: main.bbl
\providecommand{\noopsort}[1]{}
\begin{thebibliography}{10}
\providecommand{\url}[1]{\texttt{#1}}
\providecommand{\urlprefix}{URL }
\providecommand{\doi}[1]{https://doi.org/#1}

\bibitem{lexrsm}
Agrawal, S., Chatterjee, K., Novotn\'{y}, P.: Lexicographic ranking
  supermartingales: An efficient approach to termination of probabilistic
  programs. Proc.\ ACM Program.\ Lang.  \textbf{2}(POPL) (2017),
  \url{https://doi.org/10.1145/3158122}

\bibitem{pubs}
Albert, E., Arenas, P., Genaim, S., Puebla, G.: Closed{-}form upper bounds in
  static cost analysis. J.\ Autom.\ Reasoning  \textbf{46}(2),  161--203
  (2011), \url{https://doi.org/10.1007/s10817-010-9174-1}

\bibitem{costa-complexity}
Albert, E., Arenas, P., Genaim, S., Puebla, G., Zanardini, D.: Cost analysis of
  object{-}oriented bytecode programs. Theor.\ Comput.\ Sci.  \textbf{413}(1),
  142--159 (2012), \url{https://doi.org/10.1016/j.tcs.2011.07.009}

\bibitem{pubs-upper-lower}
Albert, E., Genaim, S., Masud, A.N.: On the inference of resource usage upper
  and lower bounds. ACM Trans.\ Comput.\ Log.  \textbf{14}(3) (2013),
  \url{https://doi.org/10.1145/2499937.2499943}

\bibitem{maxcore}
Albert\noopsort{1}, E., Bofill, M., Borralleras, C., Martin-Martin, E., Rubio,
  A.: Resource analysis driven by (conditional) termination proofs. Theory
  Pract.\ Log.\ Program.  \textbf{19}(5-6),  722--739 (2019),
  \url{https://doi.org/10.1017/S1471068419000152}

\bibitem{rank}
Alias, C., Darte, A., Feautrier, P., Gonnord, L.: Multi{-}dimensional rankings,
  program termination, and complexity bounds of flowchart programs. In: Proc.\
  SAS~'10. LNCS, vol.~6337, pp. 117--133 (2010),
  \url{https://doi.org/10.1007/978-3-642-15769-1_8}

\bibitem{cylindrical}
Ash, R.B., Dol\'{e}ans-Dade, C.A.: {Probability and Measure Theory}. Harcourt
  Academic Press, 2nd edn. (2000)

\bibitem{dblp:conf/rta/avanzinim13}
Avanzini, M., Moser, G.: A combination framework for complexity. In: Proc.\
  {RTA}~13. LIPIcs, vol.~21, pp. 55--70 (2013),
  \url{https://doi.org/10.4230/LIPIcs.RTA.2013.55}

\bibitem{dblp:conf/tacas/avanzinims16}
Avanzini, M., Moser, G., Schaper, M.: \textsf{TcT}: Tyrolean {C}omplexity
  {T}ool. In: Proc.\ {TACAS}~'16. LNCS, vol.~9636, pp. 407--423 (2016),
  \url{https://doi.org/10.1007/978-3-662-49674-9\_24}

\bibitem{ecoimp}
Avanzini, M., Moser, G., Schaper, M.: A modular cost analysis for probabilistic
  programs. Proc. {ACM} Program. Lang.  \textbf{4}({OOPSLA}) (2020),
  \url{https://doi.org/10.1145/3428240}

\bibitem{dblp:journals/scp/avanzinily20}
Avanzini\noopsort{1}, M., {Dal Lago}, U., Yamada, A.: On probabilistic term
  rewriting. Sci.\ Comput.\ Program.  \textbf{185} (2020),
  \url{https://doi.org/10.1016/j.scico.2019.102338}

\bibitem{bauer71measure}
Bauer, H.: Probability Theory and Elements of Measure Theory. Holt, Rinehart
  and Winston, Inc., New York (1971)

\bibitem{costa-rf}
Ben{-}Amram, A.M., Genaim, S.: Ranking functions for linear-constraint loops.
  J.\ {ACM}  \textbf{61}(4) (2014), \url{https://doi.org/10.1145/2629488}

\bibitem{dblp:conf/cav/ben-amramg17}
Ben{-}Amram, A.M., Genaim, S.: On multiphase-linear ranking functions. In:
  Proc.\ {CAV}~'17. LNCS, vol. 10427, pp. 601--620 (2017),
  \url{https://doi.org/10.1007/978-3-319-63390-9\_32}

\bibitem{dblp:conf/sas/ben-amramdg19}
Ben{-}Amram\noopsort{1}, A.M., Dom{\'{e}}nech, J.J., Genaim, S.:
  Multiphase-linear ranking functions and their relation to recurrent sets. In:
  Proc.\ SAS~'19. LNCS, vol. 11822, pp. 459--480 (2019),
  \url{https://doi.org/10.1007/978-3-030-32304-2\_22}

\bibitem{dblp:conf/rta/bournezg05}
Bournez, O., Garnier, F.: Proving positive almost-sure termination. In: Proc.\
  {RTA}~'05. LNCS, vol.~3467, pp. 323--337 (2005),
  \url{https://doi.org/10.1007/978-3-540-32033-3\_24}

\bibitem{dblp:conf/rta/bournezg06}
Bournez, O., Garnier, F.: Proving positive almost sure termination under
  strategies. In: Proc.\ {RTA}~'06. LNCS, vol.~4098, pp. 357--371 (2006),
  \url{https://doi.org/10.1007/11805618\_27}

\bibitem{boyd_vandenberghe_2004}
Boyd, S., Vandenberghe, L.: Convex Optimization. Cambridge University Press
  (2004), \url{https://doi.org/10.1017/CBO9780511804441}

\bibitem{bradley05}
Bradley, A.R., Manna, Z., Sipma, H.B.: Linear ranking with reachability. In:
  Proc.\ CAV~'05. LNCS, vol.~3576, pp. 491--504 (2005),
  \url{https://doi.org/10.1007/11513988_48}

\bibitem{koat}
Brockschmidt, M., Emmes, F., Falke, S., Fuhs, C., Giesl, J.: Analyzing runtime
  and size complexity of integer programs. ACM Trans.\ Program.\ Lang.\ Syst.
  \textbf{38}(4) (2016), \url{https://doi.org/10.1145/2866575}

\bibitem{bd77}
Burstall, R.M., Darlington, J.: A transformation system for developing
  recursive programs. J.\ ACM  \textbf{24}(1),  44--67 (1977),
  \url{https://doi.org/10.1145/321992.321996}

\bibitem{c4b}
Carbonneaux, Q., Hoffmann, J., Shao, Z.: Compositional certified resource
  bounds. In: Proc.\ PLDI~'15. pp. 467--478 (2015),
  \url{https://doi.org/10.1145/2737924.2737955}

\bibitem{pastis}
Carbonneaux\noopsort{1}, Q., Hoffmann, J., Reps, T.W., Shao, Z.: Automated
  resource analysis with \textsf{Coq} proof objects. In: CAV~'17. LNCS, vol.
  10427, pp. 64--85 (2017), \url{https://doi.org/10.1007/978-3-319-63390-9_4}

\bibitem{dblp:conf/cav/chakarovs13}
Chakarov, A., Sankaranarayanan, S.: Probabilistic program analysis with
  martingales. In: Proc.\ CAV~'13. LNCS, vol.~8044, pp. 511--526 (2013),
  \url{https://doi.org/10.1007/978-3-642-39799-8\_34}

\bibitem{dblp:conf/popl/chatterjeenz17}
Chatterjee, K., Novotn{\'{y}}, P., Zikelic, D.: Stochastic invariants for
  probabilistic termination. In: Proc.\ POPL~'17. pp. 145--160 (2017),
  \url{https://doi.org/10.1145/3093333.3009873}

\bibitem{rsmchatt}
Chatterjee\noopsort{1}, K., Fu, H., Novotn\'{y}, P., Hasheminezhad, R.:
  Algorithmic analysis of qualitative and quantitative termination problems for
  affine probabilistic programs. ACM Trans. Program. Lang. Syst.
  \textbf{40}(2) (2018), \url{https://doi.org/10.1145/3174800}

\bibitem{FoundationsTerminationMartingale2020}
Chatterjee\noopsort{2}, K., Fu, H., Novotn\'{y}, P.: Termination analysis of
  probabilistic programs with martingales. In: Barthe, G., Katoen, J., Silva,
  A. (eds.) Foundations of Probabilistic Programming, pp. 221--–258.
  Cambridge University Press (2020),
  \url{https://doi.org/10.1017/9781108770750.008}

\bibitem{rsm}
Ferrer~Fioriti, L.M., Hermanns, H.: Probabilistic termination: Soundness,
  completeness, and compositionality. In: Proc.\ POPL~'15. pp. 489--501 (2015),
  \url{https://doi.org/10.1145/2676726.2677001}

\bibitem{cofloco}
Flores-Montoya, A., H\"ahnle, R.: Resource analysis of complex programs with
  cost equations. In: Proc.\ APLAS~'14. LNCS, vol.~8858, pp. 275--295 (2014),
  \url{https://doi.org/10.1007/978-3-319-12736-1_15}

\bibitem{cofloco2}
Flores{-}Montoya, A.: Upper and lower amortized cost bounds of programs
  expressed as cost relations. In: Proc.\ FM~'16. LNCS, vol.~9995, pp. 254--273
  (2016), \url{https://doi.org/10.1007/978-3-319-48989-6_16}

\bibitem{dblp:conf/vmcai/fuc19}
Fu, H., Chatterjee, K.: Termination of nondeterministic probabilistic programs.
  In: Proc.\ VMCAI~'19. LNCS, vol. 11388, pp. 468--490 (2019),
  \url{https://doi.org/10.1007/978-3-030-11245-5_22}

\bibitem{aprove}
Giesl, J., Aschermann, C., Brockschmidt, M., Emmes, F., Frohn, F., Fuhs, C.,
  Hensel, J., Otto, C., Pl{\"{u}}cker, M., Schneider{-}Kamp, P., Str{\"{o}}der,
  T., Swiderski, S., Thiemann, R.: Analyzing program termination and complexity
  automatically with \textsf{AProVE}. J.\ Autom.\ Reasoning  \textbf{58}(1),
  3--31 (2017), \url{https://doi.org/10.1007/s10817-016-9388-y}

\bibitem{termcomp}
Giesl, J., Rubio, A., Sternagel, C., Waldmann, J., Yamada, A.: The termination
  and complexity competition. In: Proc.\ TACAS~'19. LNCS, vol. 11429, pp.
  156--166 (2019), \url{https://doi.org/10.1007/978-3-030-17502-3\_10}

\bibitem{cade19}
Giesl\noopsort{1}, J., Giesl, P., Hark, M.: Computing expected runtimes for
  constant probability programs. In: Proc.\ CADE~'19. LNAI, vol. 11716, pp.
  269--286 (2019), \url{https://doi.org/10.1007/978-3-030-29436-6_16}

\bibitem{grimmett2001probability}
Grimmett, G., Stirzaker, D.: Probability and Random Processes. Oxford
  University Press, Oxford, New York (2001)

\bibitem{dblp:journals/pacmpl/harkkgk20}
Hark, M., Kaminski, B.L., Giesl, J., Katoen, J.: Aiming low is harder:
  Induction for lower bounds in probabilistic program verification. Proc.\ ACM
  Program.\ Lang.  \textbf{4}(POPL) (2020),
  \url{https://doi.org/10.1145/3371105}

\bibitem{dblp:journals/toplas/0002ah12}
Hoffmann, J., Aehlig, K., Hofmann, M.: Multivariate amortized resource
  analysis. {ACM} Trans. Program. Lang. Syst.  \textbf{34}(3) (2012),
  \url{https://doi.org/10.1145/2362389.2362393}

\bibitem{ramlnat}
Hoffmann, J., Shao, Z.: Type{-}based amortized resource analysis with integers
  and arrays. J.\ Funct.\ Program.  \textbf{25} (2015),
  \url{https://doi.org/10.1017/S0956796815000192}

\bibitem{ramlpopl17}
Hoffmann\noopsort{1}, J., Das, A., Weng, S.C.: Towards automatic resource bound
  analysis for \textsf{OCaml}. In: Proc.\ POPL~'17. pp. 359--373 (2017),
  \url{https://doi.org/10.1145/3009837.3009842}

\bibitem{dblp:conf/aplas/huangfc18}
Huang, M., Fu, H., Chatterjee, K.: New approaches for almost-sure termination
  of probabilistic programs. In: Proc.\ APLAS~'18. LNCS, vol. 11275, pp.
  181--201 (2018), \url{https://doi.org/10.1007/978-3-030-02768-1_11}

\bibitem{dblp:journals/pacmpl/huang0cg19}
Huang, M., Fu, H., Chatterjee, K., Goharshady, A.K.: Modular verification for
  almost-sure termination of probabilistic programs. Proc.\ {ACM} Program.\
  Lang.  \textbf{3}(OOPSLA) (2019), \url{https://doi.org/10.1145/3360555}

\bibitem{apron}
Jeannet, B., Min{\'e}, A.: \textsf{Apron}: A library of numerical abstract
  domains for static analysis. In: Proc.\ CAV~'09. pp. 661--667 (2009),
  \url{https://doi.org/10.1007/978-3-642-02658-4_52}

\bibitem{kallenberg_foundations_2002}
Kallenberg, O.: Foundations of {Modern} {Probability}. Springer, New York
  (2002), \url{https://doi.org/10.1007/978-1-4757-4015-8}

\bibitem{dblp:journals/jacm/kaminskikmo18}
Kaminski, B.L., Katoen, J., Matheja, C., Olmedo, F.: Weakest precondition
  reasoning for expected runtimes of randomized algorithms. J.\ ACM
  \textbf{65} (2018), \url{https://doi.org/10.1145/3208102}

\bibitem{FoundationsExpectedRuntime2020}
Kaminski\noopsort{2}, B.L., Katoen, J., Matheja, C.: Expected runtime analyis
  by program verification. In: Barthe, G., Katoen, J., Silva, A. (eds.)
  Foundations of Probabilistic Programming, pp. 185--–220. Cambridge
  University Press (2020), \url{https://doi.org/10.1017/9781108770750.007}

\bibitem{techreport}
\textsf{KoAT}: Web interface, binary, Docker image, and examples available at
  the web site \url{https://aprove-developers.github.io/ExpectedUpperBounds/}.
  The source code is available at \sourcecode{}.

\bibitem{dblp:journals/jcss/kozen81}
Kozen, D.: Semantics of probabilistic programs. J.\ Comput.\ Syst.\ Sci.
  \textbf{22}(3),  328--350 (1981),
  \url{https://doi.org/10.1016/0022-0000(81)90036-2}

\bibitem{dblp:series/mcs/mciverm05}
McIver, A., Morgan, C.: Abstraction, Refinement and Proof for Probabilistic
  Systems. Springer (2005), \url{https://doi.org/10.1007/b138392}

\bibitem{dblp:journals/pacmpl/mcivermkk18}
McIver, A., Morgan, C., Kaminski, B.L., Katoen, J.: A new proof rule for
  almost-sure termination. Proc.\ {ACM} Program.\ Lang.  \textbf{2}(POPL)
  (2018), \url{https://doi.org/10.1145/3158121}

\bibitem{amber}
Moosbrugger, M., Bartocci, E., Katoen, J., Kov{\'{a}}cs, L.: Automated
  termination analysis of polynomial probabilistic programs. In: Proc.\ {ESOP}\
  '21. LNCS (2021), to appear.

\bibitem{z3}
de~Moura, L., Bj{\o}rner, N.: \textsf{Z3}: An efficient {SMT} solver. In:
  Proc.\ TACAS~'08. LNCS, vol.~4963, pp. 337--340 (2008),
  \url{https://doi.org/10.1007/978-3-540-78800-3_24}

\bibitem{absynth}
Ngo, V.C., Carbonneaux, Q., Hoffmann, J.: Bounded expectations: Resource
  analysis for probabilistic programs. In: Proc.\ {PLDI}~'18. pp. 496--512
  (2018), \url{https://doi.org/10.1145/3192366.3192394}, tool artifact and
  benchmarks available from
  \url{https://channgo2203.github.io/zips/tool_benchmark.zip}

\bibitem{dp-framework}
Noschinski, L., Emmes, F., Giesl, J.: Analyzing innermost runtime complexity of
  term rewriting by dependency pairs. J.\ Autom.\ Reasoning  \textbf{51}(1),
  27--56 (2013), \url{https://doi.org/10.1007/s10817-013-9277-6}

\bibitem{dblp:conf/lics/olmedokkm16}
Olmedo, F., Kaminski, B.L., Katoen, J., Matheja, C.: Reasoning about recursive
  probabilistic programs. In: Proc.\ {LICS}~'16. pp. 672--681 (2016),
  \url{https://doi.org/10.1145/2933575.2935317}

\bibitem{dblp:conf/vmcai/podelskir04}
Podelski, A., Rybalchenko, A.: A complete method for the synthesis of linear
  ranking functions. In: Proc.\ {VMCAI}~'04. LNCS, vol.~2937, pp. 239--251
  (2004), \url{https://doi.org/10.1007/978-3-540-24622-0\_20}

\bibitem{putermanmdp}
Puterman, M.L.: Markov Decision Processes: Discrete Stochastic Dynamic
  Programming. John Wiley \& Sons (2005)

\bibitem{ramlweb}
{\textsf{RaML} (Resource Aware \textsf{ML})},
  \url{https://www.raml.co/interface/}

\bibitem{stanfordinvariantgenerator}
Sankaranarayanan, S., Sipma, H.B., Manna, Z.: Constraint-based linear-relations
  analysis. In: Proc.\ SAS~'04. LNCS, vol.~3148, pp. 53--68 (2004),
  \url{https://doi.org/10.1007/978-3-540-27864-1\_7}

\bibitem{dblp:conf/cade/sinnz10}
Sinn, M., Zuleger, F., Veith, H.: Complexity and resource bound analysis of
  imperative programs using difference constraints. J.\ Autom.\ Reasoning
  \textbf{59}(1),  3--45 (2017),
  \url{https://doi.org/10.1007/s10817-016-9402-4}

\bibitem{campy}
Srikanth, A., Sahin, B., Harris, W.R.: Complexity verification using guided
  theorem enumeration. In: Proc.\ POPL~'17. pp. 639--652 (2017),
  \url{https://doi.org/10.1145/3009837.3009864}

\bibitem{tpdb}
{TPDB (Termination Problems Data Base)},
  \url{http://termination-portal.org/wiki/TPDB}

\bibitem{cylindricalvardi}
Vardi, M.Y.: Automatic verification of probabilistic concurrent finite-state
  programs. In: Proc.\ FOCS~'85. pp. 327--338 (1985),
  \url{https://doi.org/10.1109/SFCS.1985.12}

\bibitem{hoffmannicfp2020}
Wang, D., Kahn, D.M., Hoffmann, J.: Raising expectations: automating expected
  cost analysis with types. Proc.\ ACM Program.\ Lang.  \textbf{4}(ICFP)
  (2020), \url{https://doi.org/10.1145/3408992}

\bibitem{dblp:conf/pldi/wang0gcqs19}
Wang, P., Fu, H., Goharshady, A.K., Chatterjee, K., Qin, X., Shi, W.: Cost
  analysis of nondeterministic probabilistic programs. In: Proc.\ PLDI~'19. pp.
  204--220 (2019), \url{https://doi.org/10.1145/3314221.3314581}

\end{thebibliography}
